\RequirePackage{fixltx2e}
\documentclass[aip,jmp,reprint,amsmath,amssymb,onecolumn,nofootinbib]{revtex4-1}
\pdfoutput=1
\usepackage[utf8]{inputenc}
\usepackage{graphicx}
\usepackage[colorlinks=true,urlcolor=blue,anchorcolor=blue,citecolor=blue,filecolor=blue,linkcolor=blue,menucolor=blue,linktocpage=true,unicode=true,bookmarksopen=true,pdfa=true]{hyperref}

\usepackage{cleveref}
\crefformat{equation}{(#2#1#3)}
\crefmultiformat{equation}{(#2#1#3)}{, (#2#1#3)}{, (#2#1#3)}{, (#2#1#3)}
\crefrangeformat{equation}{(#3#1#4)--(#5\crefstripprefix{#1}{#2}#6)}
\renewcommand{\eqref}{\cref}

\usepackage{float}
\usepackage{longtable}
\usepackage[labelformat=simple]{subcaption}

\makeatletter
\def\switch@longtable{%
  \class@info{Patching longtable package}%
  \let\longtable\longtable@new
  \let\endlongtable\endlongtable@new
  \let\LT@start\LT@start@new
  \let\LT@end@hd@ft\LT@end@hd@ft@new
  \let\LT@array\LT@array@new
  \renewenvironment{longtable*}{%
   \onecolumngrid@push
   \longtable
  }{%
   \endlongtable
   \onecolumngrid@pop
  }%
}%
\makeatother

\usepackage{mathtools}
\mathtoolsset{mathic} % apply italic correction for math as well
\allowdisplaybreaks

\let\originalleft\left
\let\originalright\right
\renewcommand{\left}{\mathopen{}\mathclose\bgroup\originalleft}
\renewcommand{\right}{\aftergroup\egroup\originalright}

\RequirePackage{lmodern}
\RequirePackage[T1]{fontenc}
\DeclareFontShape{OMX}{cmex}{m}{n}{
  <-7.5> cmex7
  <7.5-8.5> cmex8
  <8.5-9.5> cmex9
  <9.5-> cmex10
}{}
\SetSymbolFont{largesymbols}{normal}{OMX}{cmex}{m}{n}
\SetSymbolFont{largesymbols}{bold}  {OMX}{cmex}{m}{n}

\DeclareSymbolFontAlphabet{\mathbb}{AMSb}
\DeclareMathAlphabet{\mathsfi}{OT1}{cmss}{m}{sl}
\DeclareMathAlphabet{\mathbfi}{OML}{cmm}{b}{it}
\DeclareSymbolFont{rsfs}{U}{rsfs}{m}{n}
\DeclareSymbolFontAlphabet{\mathscr}{rsfs}

\renewcommand{\vec}[1]{{\ifnum9<1#1\mathbf{#1}\else\ifcat\noexpand#1\relax\boldsymbol{#1}\else\mathbfi{#1}\fi\fi}}

\newcommand{\mathe}{\mathrm{e}}
\newcommand{\mathi}{\mathrm{i}}
\newcommand{\total}{\mathop{}\!\mathrm{d}}
\newenvironment{equations}[1][]{\subequations\ifx\relax#1\relax\else\label{#1}\fi\align\ignorespaces}{\endalign\ignorespacesafterend\endsubequations}

\makeatletter

\def\@splitequation#1{\begin{equation}\begin{split}#1\end{split}\end{equation}}
\def\splitequation{\collect@body\@splitequation}

\DeclareRobustCommand{\abs}{\bgroup\@ifstar\@@abs\@abs}
\newcommand*{\@@abs}[2][]{{#1\lvert{#2}#1\rvert}\egroup}
\newcommand*{\@abs}[1]{{\left\lvert{#1}\right\rvert}\egroup}
\DeclareRobustCommand{\norm}{\bgroup\@ifstar\@@norm\@norm}
\newcommand*{\@@norm}[2][]{{#1\lVert{#2}#1\rVert}\egroup}
\newcommand*{\@norm}[1]{{\left\lVert{#1}\right\rVert}\egroup}

\DeclareRobustCommand{\bra}{\bgroup\@ifstar\@@bra\@bra}
\newcommand*{\@@bra}[2][]{{#1\langle{#2}#1\vert}\egroup}
\newcommand*{\@bra}[1]{{\left\langle{#1}\right\vert}\egroup}
\DeclareRobustCommand{\ket}{\bgroup\@ifstar\@@ket\@ket}
\newcommand*{\@@ket}[2][]{{#1\vert{#2}#1\rangle}\egroup}
\newcommand*{\@ket}[1]{{\left\vert{#1}\right\rangle}\egroup}

\DeclareRobustCommand{\bigo}{\bgroup\@ifstar\@@bigo\@bigo}
\newcommand*{\@@bigo}[2][]{{\mathcal{O}#1({#2}#1)}\egroup}
\newcommand*{\@bigo}[1]{{\mathcal{O}\left({#1}\right)}\egroup}

\makeatother

\newcommand{\eqend}[1]{\,\mathrm{#1}}

\newcommand{\subline}[1]{\smash[b]{\substack{#1}}}

\newcommand{\expect}[1]{{\left\langle{#1}\right\rangle}}
\newcommand{\laplace}{\mathop{}\!\bigtriangleup}
\newcommand{\sgn}{\mathop{\mathrm{sgn}}}
\newcommand{\tr}{\operatorname{tr}}
\newcommand{\op}{\mathcal{O}}

\newcommand{\brst}{\mathop{}\!\mathsf{s}\hskip 0.05em\relax}

\newcommand{\st}{\mathop{}\!\hat{\mathsf{s}}\hskip 0.05em\relax}
\newcommand{\stq}{\mathop{}\!\hat{\mathsf{q}}\hskip 0.05em\relax}

\newcommand{\conv}{\star}

\DeclareRobustCommand{\SkipTocEntry}[5]{}

\usepackage{xspace}

\newcommand{\ie}{\textit{i.\,e.}\xspace}
\newcommand{\eg}{\textit{e.\,g.}\xspace}
\newcommand{\etc}{\textit{etc.}\xspace}

\usepackage{amsthm}
\newtheorem{theorem}{Theorem}
\newtheorem{proposition}[theorem]{Proposition}
\newtheorem{corollary}[theorem]{Corollary}
\newtheorem{definition}{Definition}
\newtheorem{lemma}{Lemma}
\newtheorem*{lemma*}{Lemma}

\begin{document}

\title{All-order bounds for correlation functions of gauge-invariant operators in Yang-Mills theory}

\author{Markus B. Fr\"ob}
\email{mfroeb@itp.uni-leipzig.de}
\altaffiliation[Current address: ]{Department of Mathematics, University of York, Heslington, York, YO10 5DD, United Kingdom; Electronic mail: \href{mailto:mbf503@york.ac.uk}{mbf503@york.ac.uk}}

\author{Jan Holland}
\email{holland@itp.uni-leipzig.de}
\altaffiliation[Current address: ]{Springer-Verlag, Tiergartenstraße 17, 69121 Heidelberg, Germany; Electronic mail: \href{mailto:jan.holland@springer.com}{jan.holland@springer.com}}

\author{Stefan Hollands}
\email{hollands@uni-leipzig.de}

\affiliation{Institut f\"ur Theoretische Physik, Universit\"at Leipzig, Br\"uderstra\ss e 16, 04103 Leipzig, Germany}

\begin{abstract}
We give a complete, self-contained, and mathematically rigorous proof that Euclidean Yang-Mills theories are perturbatively renormalisable, in the sense that all correlation functions of arbitrary composite local operators fulfil suitable Ward identities. Our proof treats rigorously both all ultraviolet and infrared problems of the theory and provides, in the end, detailed analytical bounds on the correlation functions of an arbitrary number of composite local operators. These bounds are formulated in terms of certain weighted spanning trees extending between the insertion points of these operators. Our proofs are obtained within the framework of the Wilson-Wegner-Polchinski-Wetterich renormalisation group flow equations, combined with estimation techniques based on tree structures. Compared with previous mathematical treatments of massless theories without local gauge invariance [R.~Guida and Ch.~Kopper, \href{http://arxiv.org/abs/1103.5692}{\texttt{arXiv:1103.5692}}; J.~Holland, S.~Hollands, and Ch.~Kopper, \href{http://dx.doi.org/10.1007/s00220-015-2486-6}{Commun.~Math.~Phys. \textbf{342} (2016) 385}] our constructions require several technical advances; in particular, we need to fully control the BRST invariance of our correlation functions.
\end{abstract}

\date{25 January 2016}
\revised{25 October 2016}


\maketitle

\tableofcontents

\addtocontents{toc}{\SkipTocEntry}
\section*{Glossary of important symbols}

\begin{center}
\begin{longtable}{l|p{0.75\textwidth}}
\hline
$\phi_K$ & basic field, indexed by $K$, $L$, \ldots (distinguishing kind of field, tensor index and Lie algebra index) \\
$\phi_K^\ddag$ & antifield, indexed by $K$, $L$, \ldots (distinguishing kind of field, tensor index and Lie algebra index) \\
$\op_A$ & composite operator, indexed by $A$, $B$, \ldots (distinguishing number of fields/antifields and their kind, tensor index and Lie algebra index) \\
$\brst$ & classical BRST differential; $\brst_0$ is the corresponding free differential \\
$\st$ & classical Slavnov-Taylor differential, identical to $\brst$ when acting on functionals which do not depend on antifields; $\st_0$ is the corresponding free differential \\
$\stq$ & quantum BRST/Slavnov-Taylor differential \\
$\left(\cdot, \cdot\right)$, $\left(\cdot, \cdot\right)_\hbar$ & classical/quantum antibracket \\
$[ \cdot ]$ & engineering dimension \\
$\Lambda$, $\Lambda_0$, $\mu$ & infrared (IR) cutoff, ultraviolet (UV) cutoff, fixed renormalisation scale \\
$C_{KL}^{\Lambda, \Lambda_0}$ & regularised covariance (``propagator'') \\
$\nu^{\Lambda, \Lambda_0}$ & Gaussian measure in field space defined by the regularised covariance $C_{KL}^{\Lambda, \Lambda_0}$ \\
$\left\langle \cdot, \cdot \right\rangle$, $\ast$ & $L^2$ inner product, $L^2$ convolution \\
$\conv$ & convolution in field space (with the Gaussian measure) \\
$L^{\Lambda_0}$ & interaction part of the action, including counterterms (``bare interaction'') \\
$L^{\Lambda, \Lambda_0}$ & generating functional of regularised connected and amputated correlation functions (CACs) of basic fields \\
$L^{\Lambda, \Lambda_0}\left( \bigotimes_{k=1}^s \op_{A_k} \right)$ & generating functional of regularised CACs with insertions of composite operators \\
$L^{\Lambda, \Lambda_0}\left( \int\!\op_A \right)$ & generating functional of regularised CACs with insertion of an integrated composite operator \\
$\mathcal{L}^{\Lambda, \Lambda_0, l}_{\vec{K} \vec{L}^\ddag}\left( \vec{q} \right)$ & expansion coefficient (=CAC) of $L^{\Lambda, \Lambda_0}$ with external basic fields $\phi_{K_1}, \ldots, \phi^\ddag_{L_n}$ depending on external momenta $\vec{q}$, of formal perturbation order $\hbar^l$ \\
$\mathcal{L}^{\Lambda, \Lambda_0, l}_{\vec{K} \vec{L}^\ddag}\left( \bigotimes_{k=1}^s \op_{A_k}; \vec{q} \right)$ & expansion coefficient (=CAC) of $L^{\Lambda, \Lambda_0}\left( \bigotimes_{k=1}^s \op_{A_k} \right)$ with external basic fields $\phi_{K_1}, \ldots, \phi^\ddag_{L_n}$ depending on external momenta $\vec{q}$, of formal perturbation order $\hbar^l$ \\
$\mathcal{L}^{\Lambda, \Lambda_0, l}_{\vec{K} \vec{L}^\ddag}\left( \int\!\op_A; \vec{q} \right)$ & expansion coefficient (=CAC) of $L^{\Lambda, \Lambda_0}\left( \int\!\op_A \right)$ with external basic fields $\phi_{K_1}, \ldots, \phi^\ddag_{L_n}$ depending on external momenta $\vec{q}$, of formal perturbation order $\hbar^l$ \\
$T$, $T^*$ & tree, tree with special vertex \\
$\mathsf{G}^{T/T^*,\vec{w}}_{\vec{K} \vec{L}^\ddag; [v_p]}(\vec{q}; \mu, \Lambda)$ & weight factor associated to the tree $T/T^*$ with external basic fields $\phi_{K_1}, \ldots, \phi^\ddag_{L_n}$ depending on external momenta $\vec{q}$, with dimension $[v_p]$ of the special vertex (absent for a tree $T$ without special vertex), and with $\vec{w}$ momentum derivatives \\
$\abs{\vec{q}}$ & maximum possible sum of the momenta $\vec{q}$: supremum of the modulus over all subsums \\
$\eta\left(\vec{q}\right)$ & exceptionality of the momenta $\vec{q}$: infimum of the modulus over all strict subsums \\
$\bar{\eta}\left(\vec{q}\right)$ & exceptionality of the momenta $\vec{q}$: infimum of the modulus over all subsums \\
$\mathcal{P}$ & polynomial with positive coefficients \\
\hline
\end{longtable}
\end{center}

\section{Introduction}

Quantum field theories with local gauge invariance are at the core of the standard model of particle physics, which describes the known interactions between elementary particles. A major step forward in the understanding of these theories at the quantum level was the invention of dimensional regularisation~\cite{thooft1971,bollinigiambiagi1972,thooftveltman1972,leibbrandt1975,breitenlohnermaison1977}. It manifestly preserves the gauge symmetry in the regularised theory, and thus in principle leads to the correct Ward identities in the renormalised perturbation series, ensuring gauge independence of the S-Matrix~\cite{leezinnjustin1973}. It also is rather efficient from a computational viewpoint. Dimensional regularisation is not free of major subtleties, however. It was noted early on that chiral fermions cannot be unambiguously defined~\cite{thooft1971,chanowitzfurmanhinchliffe1979}, and this issue is closely related to the chiral anomaly~\cite{adler1969,bardeen1969}. A more serious issue, however, is the precise handling of the divergences in the renormalisation process. In the case of massless theories, including all gauge theories without spontaneous symmetry breaking, there are in principle both infrared (IR) and ultraviolet (UV) divergences. Dimensional regularisation treats both of them at the same time -- roughly speaking UV divergences are regulated if one continues the spacetime dimension $n$ to values $<4$, IR divergences are regulated for $n>4$~\cite{marcianosirlin1975}. As a consequence, integrals which are both UV- and IR-divergent are very difficult to handle, in particular since a given Feynman graph can have a very complicated nested structure of loop integrals of this sort. For this reason, the existing treatments remain somewhat formal for massless theories with regard to the IR problems. In fact, there is, to the best of our knowledge, no completely satisfactory systematic treatment of the UV and IR problem for arbitrary graphs in massless theories.

The discovery of a residual fermionic symmetry of the gauge-fixed Lagrangian~\cite{becchietal1975} (called BRST invariance) together with the corresponding BRST formalism then made it possible to study renormalisability and classify possible anomalies independently of the chosen regularisation scheme. Of central importance in this approach are the cohomology classes of the BRST operator that implements this symmetry. The ultimate goal in this approach is to show that correlation functions of gauge invariant operators can be renormalised using finitely many counterterms that are in the kernel of the BRST operator. If this can be done, then one can formally see that the correlation functions are gauge invariant in an appropriate sense, see, \eg, Refs.~[\onlinecite{weinberg_v2},\onlinecite{barnichetal2000}] and references therein for details. A particularly attractive feature of this approach is that it can be implemented iteratively in the expansion of the correlation functions in $\hbar$. Assuming the problem to be solved at a given order, one shows formally that this task can be accomplished also at the next order if a certain cohomology problem associated with the BRST operator has only trivial solutions. If this is the case -- as can in many cases be decided from the structure of the classical theory alone -- then it is shown that the regularisation scheme can be altered in such a way (if necessary) that the task is also accomplished at the next order, and so on~\cite{piguetsorella}. A later extension by Batalin and Vilkovisky~\cite{batalinvilkovisky1981,batalinvilkovisky1983,batalinvilkovisky1984} streamlined the BRST method and extended its range of applicability to a much wider class of gauge theories. In fact, their method is also a significant improvement in the special case of gauge theories of Yang-Mills type, because possible anomalies with regard to the equations of motion at the quantum level are systematically accounted for. The ``Quantum Master Equation''~\cite{weinberg_v2,henneauxteitelboim1992} of Batalin and Vilkovisky is also rather closely related to the Zinn-Justin equation~\cite{zinnjustin1974} for the effective action obtained by a Legendre transformation.

Unfortunately, just as dimensional regularisation, also the BRST method is not without problems. Just as there, the problem is that the derivations are usually only formal in as far as potential IR divergences are concerned. In Abelian gauge theories like quantum electrodynamics, one can regulate the IR divergences independently of the UV problem by introducing a photon mass, and the limit of zero photon mass can be taken for observable quantities such as cross sections~\cite{blochnordsieck1937}, taking into account quantum processes with photons that have an energy below the detection threshold. However, introducing by hand a mass for the gauge field in non-Abelian gauge theories leads to a violation of BRST invariance, thereby destroying the core principle underlying the entire method.

To summarise, while the gauge invariance, UV and IR problems have been treated in a satisfactory way separately, assuming that the respective other issue(s) can somehow be overcome in a suitable way, treating all at the same time is considerably more difficult and has not been done in full generality. In this article, we want to close this gap, considering Yang-Mills theories in Euclidean space in a fully mathematical rigorous way. We introduce momentum-space UV and IR regulators, and show that they can be removed for arbitrary correlation functions (including insertions of arbitrary composite operators) away from exceptional momentum configurations, to all orders in perturbation theory, in such a way that the resulting correlation is fully gauge invariant in the physically correct sense that we will explain. Our strategy for achieving this will be as follows:
\begin{enumerate}
\item (``Analytic part'') First, we need to analyse in sufficient detail and generality correlation functions of arbitrary composite operators in general massless theories. The theory that we have in mind is a gauge-fixed version of Yang-Mills theory, but the precise form of the interaction terms other than appropriate ``power counting'' properties will not matter. As in previous papers~\cite{hollandskopper2012,hollandetal2014,hollandhollands2015a,hollandhollands2015b}, our bounds will be derived using the renormalisation group (RG) flow equation approach to quantum field theory~\cite{polchinski1984,wetterich1993,kopper1998,mueller2003,kopper2007}. Our results establish appropriate behavior of correlation functions for both small and large values of the momenta and spacetime separation of the composite operators.
\item (``Algebraic part'') In the second step, we will then treat the gauge structure of the theory. The gauge fixed theory is set up, as usual, with appropriate ghost fields and auxiliary fields in such a way that the action remains BRST invariant. (Actually, in order to retain full control of the properties of the theory, we will use the BV formalism including also antifields.) BRST invariance of the theory at the quantum level is expressed by certain Ward identities connecting different correlation functions. The renormalisation scheme adopted in our construction in part 1) will violate these identities in general, but we derive suitable ``anomalous'' versions of these identities that retain sufficient information about the gauge structure in the form of an identity for the ``anomaly'' of cohomological nature. This information suffices to see how one can change our scheme in order to satisfy the desired identities without anomaly. To set up our anomalous Ward identities, we have to trace how these behave under the RG flow. At finite values of the UV cutoff only approximate equations are valid, and the task is to show that the approximation becomes exact as the cutoffs are removed.
\end{enumerate}

\subsection{Correlation functions in massless theories}

In this part, we consider an arbitrary massless Euclidean quantum field theory that is superficially renormalisable in the sense that the classical Lagrangian without sources only contains operators of engineering dimension less or equal to four and coupling constants of positive or zero (mass) dimension. The simplest operators $\op$ are the basic fields themselves (in Yang-Mills theory the gauge potential, a Lie algebra-valued one-form $A$), but of course in gauge theories one is primarily concerned with gauge-invariant \emph{composite operators} such as $\tr F^2$ with the field strength $F$. We are interested in bounds on their correlation functions
\begin{equation}
\expect{ \op_{A_1}(x_1) \cdots \op_{A_s}(x_s) } \eqend{,}
\end{equation}
for an arbitrary number $s$ of arbitrary composite operators $\op_{A_i}$. As usual, it is convenient to consider the connected correlation functions $\expect{ \cdot }_\text{c}$, related by
\begin{equation}
\expect{ \op_{A_1}(x_1) \cdots \op_{A_s}(x_s) }_\mathrm{c} \equiv \sum_{P_1, \ldots, P_p} (-1)^{p+1} \prod_{i=1}^p \expect{\prod_{j \in P_i} \op_{A_j}(x_j) } \eqend{,}
\end{equation}
where the sum is over all possible partitions of $\{1, \dots, n\}$ into an arbitrary number $p$ of disjoint nonempty subsets $P_1, \ldots, P_p$. Our first ``analytical'' result is:
\begin{theorem}
\label{thm1}
For an arbitrary massless, superficially renormalisable theory and up to an arbitrary, but fixed perturbation order, the connected correlation functions with composite operator insertions are tempered distributions. They can be written in the form
\begin{equation}
\expect{ \op_{A_1}(x_1) \cdots \op_{A_s}(x_s) }_\mathrm{c} = \sum_{\abs{w} \leq D+1} \mu^{-\abs{w}} \partial^w_{\vec{x}} \mathcal{K}\left( \bigotimes_{k=1}^s \op_{A_k}(x_k) \right) \eqend{,}
\end{equation}
where $w$ is a multiindex not involving the coordinate $x_s$, $D = \sum_{i=1}^s [\op_{A_i}]$ is the sum of the engineering dimensions of the composite operators, $\mu$ is the (fixed) renormalisation scale, and where the kernel $\mathcal{K}$ is a locally integrable function that fulfils the bound
\begin{equation}
\abs{ \mathcal{K}\left( \bigotimes_{k=1}^s \op_{A_k}(x_k) \right) } \leq C \mu^D \prod_{i=1}^{s-1} \Big[ 1 - \ln \inf(1, \mu \abs{x_i-x_s}) \Big]
\end{equation}
for some constant $C$ depending on the perturbation order and the renormalisation conditions.
\end{theorem}
To prove that an object is a tempered distribution, one has to smear it with test functions $f_i \in \mathcal{S}(\mathbb{R}^4)$, and bound the smeared object by some appropriate norm of the $f_i$. According to a theorem by Schwartz~\cite{schwartz}, tempered distributions can also be characterised by the fact that they can be written as a finite number of derivatives acting on a locally integrable function. The bound on the kernel $\mathcal{K}$ directly shows that $\mathcal{K}$ is an integrable function, since it has at most logarithmic singularities, and thus Theorem~\ref{thm1} in particular shows that the correlation functions are distributions, by Schwartz's theorem.

To shorten the notation, in the remaining part of the introduction we set the renormalisation scale $\mu = 1$, since the dependence on $\mu$ can be easily restored using dimensional analysis. Let us introduce test functions $f_i \in \mathcal{S}(\mathbb{R}^4)$ and the Schwartz norms
\begin{equation}
\label{schwartz_norm}
\norm{f}_k \equiv \sup_{x \in \mathbb{R}^4, \abs{w} \leq k} \abs{ (1+x^2)^4 \partial^w_x f(x) } \eqend{,}
\end{equation}
which are used to define ``averaged operators''
\begin{equation}
\op_A(f) \equiv \int \op_A(x) f(x) \total^4 x \eqend{.}
\end{equation}
The above theorem then gives the
\begin{corollary}
\label{thm2}
For an arbitrary massless, superficially renormalisable theory and up to an arbitrary, but fixed perturbation order, the connected correlation functions with composite operator insertions are tempered distributions. Smeared with test functions, they satisfy the bound
\begin{equation}
\expect{ \op_{A_1}(f_1) \cdots \op_{A_s}(f_s) }_\mathrm{c} \leq C \prod_{i=1}^s \norm{f_i}_{D+1}
\end{equation}
for some constant $C$ depending on the perturbation order and the renormalisation conditions, and where $D = \sum_{i=1}^s [\op_{A_i}]$.
\end{corollary}
This bound is quite weak compared to the one above, but shows explicitly that the correlation functions are in fact tempered distributions. If all the positions $x_i$ of the operator insertions are distinct, we also obtain rather sharp bounds on the behaviour when some of the $x_i$ approach each other. These bounds are best expressed in terms of trees, which are set up in the following
\begin{definition}
\label{def_tree_x}
A tree $\tau$ of order $s$ has $s$ numbered vertices connected by lines such that there are no closed loops. Choose a function $z$ with $z(i) \neq i$ such that the vertices $i$ and $z(i)$ are connected by a line. To each line of the tree connecting the vertices $i$ and $z(i)$ we then associate the weight factor
\begin{equation}
\label{tau_weightfactor}
\mathsf{W}_{i,z(i)}\left( x_i, x_{z(i)} \right) = \inf\left( 1, \abs{x_i - x_{z(i)}} \right)^{-[\op_{A_i}]-\epsilon} \eqend{,}
\end{equation}
where $[\op_{A_i}]$ is the engineering dimension of the operator $\op_{A_i}$ and $\epsilon > 0$ an arbitrary positive number, while all other lines are assigned weight factor $1$. The weight factor $\mathsf{W}^\tau(x_1,\ldots,x_s)$ associated to the tree is obtained as the product of the weight factors of all lines. We denote the set of all trees $\tau$ of order $s$ by $\mathscr{T}_s$, counting trees which only differ in the function $z$ as distinct.
\end{definition}
One easily sees that $\mathscr{T}_s$ is a finite set. We then prove
\begin{theorem}
\label{thm3}
For an arbitrary massless, superficially renormalisable theory and up to an arbitrary, but fixed perturbation order, the connected correlation functions with $s$ composite operator insertions at distinct points fulfil the bound
\begin{equation}
\label{thm3_bound_tree}
\abs{ \expect{ \op_{A_1}(x_1) \cdots \op_{A_s}(x_s) }_\mathrm{c} } \leq C \sum_{\tau \in \mathscr{T}_s} \mathsf{W}^\tau(x_1,\ldots,x_s)
\end{equation}
for some constant $C$ depending on the perturbation order and the renormalisation conditions.
\end{theorem}
\paragraph*{Remark.} These bounds are presumably nearly optimal regarding the scaling behaviour with respect to the $x_i$, since one expects that in perturbation theory the correlation functions scale in the same way as they do at tree order (thus according to the engineering dimension of the operators appearing) up to logarithmic corrections. These logarithmic corrections can be bounded by an arbitrary small power of the distance between two operator insertions, which are manifest in the terms involving $\epsilon$ that can be arbitrarily small. The tree structures also nicely express in a quantitative manner the experience that correlation functions diverge if a group of points is scaled together. The emergence of trees is of course evident in the lowest order (``tree level'') approximation of the correlation functions. The non-trivial statement is that this behaviour persists at higher loops, up to the logarithmic corrections already mentioned above. As an example, two of the trees contributing to the right-hand side of equation~\eqref{thm3_bound_tree} for a correlation function of four composite operators are shown in Figure~\ref{fig_bound_tree}.
\begin{figure}
\includegraphics[scale=1.0]{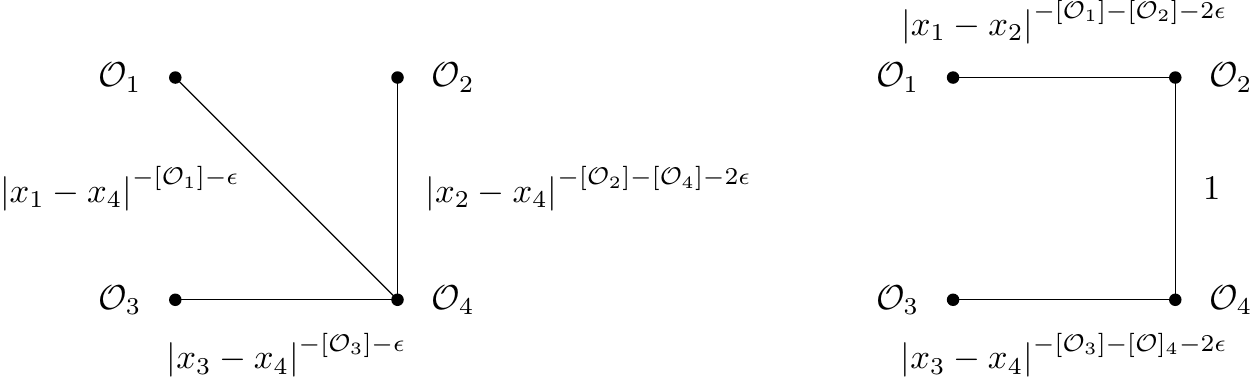}
\caption{Two trees contributing to the bound on a correlation function of four composite operators $\op_i$, together with the corresponding weight factors (assuming that $\abs{x_i - x_j} < 1$). For the first tree, the function $z(i)$ is given by $z(1) = z(2) = z(3) = 4$ and $z(4) = 2$, while for the second tree we have $z(1) = 2$, $z(2) = 1$, $z(3) = 4$ and $z(4) = 3$.}
\label{fig_bound_tree}
\end{figure}

\subsection{Gauge theories}
\label{sec_intro_gauge}

Finally, we also treat gauge theories of the Yang-Mills type. While the bounds given above are still valid, at intermediate steps of the construction gauge invariance is violated, and must be ``restored'' afterwards by an appropriate finite change in the ``renormalisation conditions'' implicit in the definition of the correlation functions. As we explain below, gauge invariance in the quantum field theory context is expressed by a set of Ward(-Takahashi-Slavnov-Taylor) identities~\cite{ward1950,takahashi1957,taylor1971,slavnov1972}. There are various ways to state such identities. We find it most convenient to use a version of the Batalin-Vilkovisky (or field-antifield) formalism~\cite{batalinvilkovisky1981,batalinvilkovisky1983,batalinvilkovisky1984,weinberg_v2}. In this formalism, one enlarges the field space of the original theory by additional dynamical fields, and replaces the original action by an enlarged one which contains gauge-fixing terms and Fadeev-Popov ghosts $c$ and antighosts $\bar{c}$, which are Lie-algebra valued Grassmann fields. The resulting action breaks gauge invariance, but is still invariant under the action of a nilpotent BRST differential $\brst$~\cite{becchietal1975}. One then introduces sources (``antifields'') for the BRST transformation of the fields and generalises the BRST differential $\brst$ to the Slavnov-Taylor differential $\st$ which also acts on the antifields. One shows that classically gauge-invariant observables can be identified with cohomology classes of $\st$, \ie, they are annihilated by $\st$, and two observables are identified if they differ by a $\st$-exact term. (This formalism is explained in more detail in Section~\ref{sec_gauge}.) Therefore, if the correlation functions of composite operators are also invariant under $\st$ in a suitable sense, they are well-defined on the \emph{classes}, and hence in this sense one can say that the theory is gauge invariant.

The main question is then whether the construction of the correlation functions can be made consistent with the symmetry $\st$. That this is indeed so is expressed by the Ward identity asserted in the following theorem:
\begin{theorem}
\label{thm4}
For an arbitrary massless, superficially renormalisable gauge theory and up to an arbitrary, but fixed perturbation order, there exist renormalisation conditions such that the CACs with insertions fulfil the Ward identity
\begin{splitequation}
\label{thm4_ward}
&\sum_{k=1}^s \expect{ \op_{A_1}(x_1) \cdots \left( \stq \op_{A_k} \right)(x_k) \cdots \op_{A_s}(x_s) }_\mathrm{c} \\
&\quad= \hbar \sum_{1 \leq k < l \leq s} \expect{ \op_{A_1}(x_1) \cdots \left( \op_{A_k}(x_k), \op_{A_l}(x_l) \right)_\hbar \cdots \op_{A_s}(x_s) }_\mathrm{c}
\end{splitequation}
if the equivariant classical cohomology of $\st$ at form degree $4$ and ghost number $1$ (and thus dimension $5$ with our conventions), is empty, $H^{1,4}_{\mathrm{E}(4)}(\st\vert\total) = 0$. In these equations, $\stq = \st + \bigo{\hbar}$ is the nilpotent quantum Slavnov-Taylor differential ($\stq^2 = 0$), which differs in higher orders in $\hbar$ from the classical Slavnov-Taylor differential $\st$. The contact terms are given by an associative quantum antibracket $\left( \op_{A_k}(x_k), \op_{A_l}(x_l) \right)_\hbar = \left( \op_{A_k}(x_k), \op_{A_l}(x_l) \right) + \bigo{\hbar}$ with the classical antibracket $( \cdot, \cdot )$, and they are supported on the diagonal $x_k = x_l$. In the given form, these Ward identities are valid for bosonic operators, while for fermionic operators additional minus signs appear. The proper minus signs can be obtained by introducing an auxiliary constant fermion $\epsilon_k$ for each fermionic operator $\op_{A_k}$, replacing $\op_{A_k} \to \epsilon_k \op_{A_k}$ and then taking derivatives with respect to the $\epsilon_k$.

The quantum antibracket fulfils the symmetry conditions
\begin{equation}
\label{bvq_symm}
\left( \op_A, \op_B \right)_\hbar = \pm \left( \op_B, \op_A \right)_\hbar
\end{equation}
with the minus sign if at least one of the operators is fermionic, and the graded Jacobi identity
\begin{equation}
\label{bvq_jacobi}
\left( \op_{A_1}, \left( \op_{A_2}, \op_{A_3} \right)_\hbar \right)_\hbar \pm \left( \op_{A_2}, \left( \op_{A_3}, \op_{A_1} \right)_\hbar \right)_\hbar \pm \left( \op_{A_3}, \left( \op_{A_1}, \op_{A_2} \right)_\hbar \right)_\hbar = 0
\end{equation}
with the appropriate minus signs if some of the operators are fermionic, and the quantum differential $\stq$ is compatible with the quantum antibracket $\left( \cdot, \cdot \right)_\hbar$ in the sense that
\begin{equation}
\label{bvq_stq_compat}
\stq \left( \op_A, \op_B \right)_\hbar = \left( \stq \op_A, \op_B \right)_\hbar \pm \left( \op_A, \stq \op_B \right)_\hbar \eqend{,}
\end{equation}
with the minus sign if $\op_A$ is bosonic. Furthermore, if the cohomology of $\st$ is empty at dimension $[\op_{A_k}]+1$ and ghost number one larger than the ghost number of $\op_{A_k}$, one can choose renormalisation conditions such that $\stq \op_{A_k} = \st \op_{A_k}$.
\end{theorem}
The physical meaning of the theorem is best explained if we neglect, for the moment, the distinction between classical and quantum BRST differential. In the classical case, local $p$-form valued gauge-invariant operators are of the form
\begin{equation}
\label{yangmills_invops}
\op_A = \prod_j q_j(F, \mathcal{D} F, \ldots, \mathcal{D}^l F)
\end{equation}
(where the $q_j$ are invariant polynomials of the Lie algebra, $F \equiv \total A + \mathi g [A,A]$ is the Lie-algebra valued field strength and $\mathcal{D}$ is the gauge-covariant derivative), and such $\op_A$'s are in the kernel of $\st$. In fact, it can be shown~\cite{barnichetal2000} that also the converse is true: any (local) expression in the kernel of $\st$ can be written in terms of such $\op_A$'s up to a term in the image of $\st$. Thus, the cohomology $H^{0,p}(\st)$ is generated by such expressions (at ghost number 0), and therefore the cohomology classes are in a one-to-one correspondence with the classically gauge invariant operators. It immediately follows from~\eqref{thm4_ward} that, for mutually distinct insertion points $x_i \neq x_j$, the correlation functions are well-defined on the classes, \ie, the elements in $\mathrm{Ker} \st / \mathrm{Im} \st$: They do not change if we add an operator in the image of $\st$ to another one in the kernel of $\st$. Theorem~\ref{thm4} is therefore indeed, in this sense, the statement that the theory is gauge invariant at the quantum level.

According to the theorem, for equation~\eqref{thm4_ward} to hold, we must however understand also the equivariant cohomologies at ghost number $1$. For Yang-Mills theories based on a semisimple Lie group, these have also been calculated quite some time ago~\cite{barnichetal1995a,barnichetal1995b,barnichetal2000}. The relevant cohomology class $H^{1,4}(\st\vert\total)$ of $\st$ modulo exact terms (\ie, modulo $\total$) at form degree $4$ and ghost number $1$ is generated by the unique element (often called the ``gauge anomaly'')
\begin{equation}
\label{cohom_parity}
\mathcal{A} = \total c^a \wedge \left[ d_{abe} A^b \wedge \total A^e - \frac{\mathi}{12} g\, d_{abcd} A^b \wedge A^c \wedge A^d \right] \eqend{.}
\end{equation}
The $\mathfrak{g}$-invariant totally symmetric tensors $d_{abc}$ and $d_{abcd}$ are defined as~\cite{piguetsorella}
\begin{equation}
d_{abc} \equiv \frac{1}{2} \tr \left( t_a t_b t_c + t_a t_c t_b \right) \eqend{,} \qquad d_{abcd} \equiv d_{abe} f_{cde} + d_{ace} f_{bde} + d_{ade} f_{bce} \eqend{,}
\end{equation}
where the trace is in some non-real representation of the Lie algebra $\mathfrak{g}$, and $\{t_a\}$ represent a Lie algebra basis in the given representation. For Lie algebras possessing only real or pseudo-real representations, we have $d_{abc} = 0 = d_{abcd}$ (see, \eg, Ref.~[\onlinecite{srednicki}]), and thus the above expression~\eqref{cohom_parity} vanishes. Consequently, for those Lie algebras the entire cohomology is trivial, which happens, \eg, for $\mathfrak{su}(2)$, $\mathfrak{so}(2n+1)$ or $\mathfrak{so}(4n)$ for $n \in \mathbb{N}$. The \emph{equivariant} cohomology consists of all elements that are invariant under the group $G$ leaving the classical action invariant. In our case, this is the Euclidean group $\mathrm{E}(4)$, including parity inversion (which entails the replacement $\epsilon^{\mu\nu\rho\sigma} \to -\epsilon^{\mu\nu\rho\sigma}$, with all other terms unchanged), and time reversal: $\mathrm{E}(4) = \mathrm{O}(4) \rtimes \mathbb{R}^4$. However, the anomaly~\eqref{cohom_parity} is odd under parity~\cite{bardeen1969,weinberg_v2,piguetsorella}, and therefore the equivariant cohomology $H^{1,4}_{\mathrm{E}(4)}(\st\vert\total)$ is \emph{always} empty, no matter what our choice of $\mathfrak{g}$ is. Said differently, the relevant cohomology is not the full $H^{1,4}(\st\vert\total)$ but only its subsector which is invariant under any global symmetries that are maintained by the regularisation. In the case of pure Yang-Mills theory -- but not, for instance, in the presence of chiral fermions in non-real representations -- this includes invariance under parity inversion. The above Ward identities~\eqref{thm4_ward} are thus fulfilled. Taking $\hbar$ into account does not change the basic story: the quantum observables are the operators which are in the cohomology of the quantum Slavnov-Taylor differential $\stq$, and equation~\eqref{thm4_ward} implies gauge invariance in the sense we have described.

\subsection{Organisation of the paper and comparison with other approaches}

The proof of the above theorems is based on the renormalisation group flow equation approach to quantum field theory~\cite{polchinski1984,wetterich1993,kopper1998,mueller2003,kopper2007}, which has proven to be a very powerful and mathematically rigorous framework to study properties of QFTs, bypassing completely the analysis of individual Feynman diagrams and complicated forest formulas. It provides extremely simple and short proofs of perturbative renormalisability for massive~\cite{kelleretal1992,koppersmirnov1993} and massless~\cite{kellerkopper1994,guidakopper2011,guidakopper2015} scalar field theories, including insertions of composite operators, Zimmermann identities, Lowenstein rules and large momentum bounds~\cite{kellerkopper1992,kellerkopper1993,koppermeunier2002}, as well as quantum electrodynamics~\cite{kellerkopper1991} and spontaneously broken SU(2) Yang-Mills-Higgs theory~\cite{koppermueller2000a,koppermueller2000b,mueller2003,koppermueller2009}.

The main technical results in the first part (Section~\ref{sec_bounds}) are bounds uniform in the UV cutoff for the connected, amputated correlation functions (CACs) with an arbitrary number of basic fields (Prop.~\ref{thm_l0}, p.~\pageref{thm_l0}), one insertion of a composite operator and an arbitrary number of basic fields (Prop.~\ref{thm_l1} and~\ref{thm_l1i}, p.~\pageref{thm_l1}) and at least two insertions of composite operators and an arbitrary number of basic fields (Prop.~\ref{thm_lsa}--\ref{thm_lsc}, p.~\pageref{thm_lsa}). Furthermore, for restricted boundary conditions, we prove related bounds (Prop.~\ref{thm_l0_van}--\ref{thm_ls_van}, pp.~\pageref{thm_l0_van}--\pageref{thm_ls_van}), which imply convergence of the CACs as the UV cutoff is sent to infinity. Theorem~\ref{thm1} is then just the special case of Proposition~\ref{thm_lsb} in the unregularised limit when no additional basic fields are present, Corollary~\ref{thm2} follows using Lemma~\ref{lemma_smearing}, and Theorem~\ref{thm3} is the special case of Proposition~\ref{thm_lsc} in the unregularised limit when no additional basic fields are present. Section~\ref{sec_gauge} fixes our notation for classical gauge theories and the flow equation framework is shortly reviewed in Section~\ref{sec_flow}. The bounds on the correlation functions are expressed using a refinement of the trees used for massless scalar field theory~\cite{koppermeunier2002,guidakopper2011,guidakopper2015,hollandetal2014}, which are defined in Section~\ref{sec_trees} where also some of their properties are derived. The mentioned bounds on the CACs are then finally proven in Section~\ref{sec_bounds}.

In the second part (Section~\ref{sec_brst}), we show that gauge invariance can be restored in the quantum theory, order by order in perturbation theory. To this end, we first derive an anomalous Ward identity including extra anomalous terms $\mathsf{A}_0$, $\mathsf{A}_1$ and $\mathsf{A}_2$ in addition to the ones given above, which quantify a possible violation of gauge symmetry for functionals with zero, one and two or more insertions of composite operators, respectively, and then derive the analogue of the Wess-Zumino consistency conditions in the flow equation framework. The solution of these consistency conditions is then used to perform a finite change of renormalisation conditions, after which the appropriate Ward identities are fulfilled. The general results are given by Propositions~\ref{thm_brst} and~\ref{thm_anomward} (p.~\pageref{thm_brst}) for functionals with and without insertions of composite operators and with an arbitrary number of basic fields; Theorem~\ref{thm4} is then just a special case of Proposition~\ref{thm_brst} when no additional basic fields (and antifields) are present, using the relation~\eqref{relation_l_opconn}.

We should also mention that our use of the BV formalism and of cohomological methods, while being very elegant and efficient, is in some sense optional and not mandatory. In fact, it is plausible that a proof of suitable Ward identites could be accomplished directly by analysing in an explicit manner the remaining freedom in the boundary conditions for the flow equations (although the treatment of correlation functions of arbitrary composite operators is presumably very challenging). This approach was adopted successfully in Refs.~[\onlinecite{koppermueller2000a,koppermueller2000b,mueller2003,koppermueller2009}] for the case of spontaneously broken SU(2) Yang-Mills-Higgs theory, where Ward identities for the one-particle-irreducible effective action were demonstrated. The generalisation to SU(2) gauge theories without spontaneous symmetry breaking is currently under investigation~\cite{efremovguidakopper2015}, and we are grateful to these authors for making available to us their forthcoming manuscript and for discussions on many issues related to renormalisation.

Let us briefly discuss the relation between our approach and other mathematically rigorous approaches to gauge theories. A perturbative construction of Yang-Mills gauge theories on Lorentzian curved spacetimes has been given in the ``perturbative Algebraic Quantum Field Theory'' (pAQFT) approach (see Refs.~[\onlinecite{hollandswald2015,fredenhagenrejzner2016,rejzner2016}] for a modern introduction). In this approach, one constructs the gauge (\ie, BRST-) invariant interacting field operators as formal power series taking values in some abstract $\star$-algebra, hence the name of the approach. To compute physical quantities like correlation functions one has to construct a representation of this algebra on a Hilbert space, which is in principle possible once a representation of the underlying free field theory has been chosen. The particular choice depends on the nature of the underlying spacetime, and for flat spacetime one would naturally choose a representation based on the Minkowski vacuum state. Consistency of the construction requires one to derive suitable Ward identities for the renormalised series of the interacting field operators, which in the end have a structure that is somewhat similar to the identities contained in our Theorem~\ref{thm4}. The first steps in this direction were in fact already taken by Refs.~[\onlinecite{duetschboas2002,duetschfredenhagen2003,brenneckeduetsch2008}], who proposed that a ``Master Ward Identity'' should be added to the list of axioms defining pAQFT. While this identity basically reduces to Theorem~\ref{thm4} in the antifield-free case, it could not be proven. A proof of the analogue of Theorem~\ref{thm4} for Yang-Mills theories was then first given in Ref.~[\onlinecite{hollands2008}]. Later, a similar analysis highlighting in particular the similarities with the classical work of Batalin and Vilkovisky was taken up in Refs.~[\onlinecite{rejzner2011,fredenhagenrejzner2013,brunettifredenhagenrejzner2013}]. Concerning our treatment of the BRST symmetry, our approach is closest to the one presented in Ref.~[\onlinecite{hollands2008}]; in particular, the anomalous Ward identity (Proposition~\ref{thm_anomward}) is very similar to the connected version of the anomalous Ward identity presented there. In line with the pAQFT approach, the anomaly $\mathsf{A}_0$ can be removed by a finite change of renormalisation/boundary conditions if the proper equivariant cohomology is empty, as stated in Theorem~\ref{thm4}. The same is true for the anomaly $\mathsf{A}_1$~\cite{hollands2008}, but \emph{not} for the anomaly $\mathsf{A}_2$. While in the pAQFT approach the possibility of removing $\mathsf{A}_2$ is again equivalent to a cohomological problem~\cite{fredenhagenrejzner2013}, this is not the case in the flow equation framework since all boundary conditions have been completely fixed at this stage. However, this is offset by the fact that in the flow equation framework the functionals with at least two insertions of composite operators are uniquely determined by the conditions for the functionals without and with one insertion of a composite operator, which is not the case in the pAQFT approach. One may see it as a matter of taste whether the correlation functions or the antibracket have to be changed to obtain the correct Ward identities, and we expect that the actual correlation functions are identical as distributions (for a suitable allowed choice of renormalisation/boundary conditions).

What is largely lacking in the program of pAQFT is the analysis and characterisation of correlation functions of the interacting composite field operators for massless field theories, such as their decay properties (even in flat spacetime). In particular, it is highly unclear how to deal with the severe infrared problems that occur in such theories on Minkowski space or more general asymptotically flat spacetimes. In this sense, the results obtained to date within pAQFT are considerably weaker than the ones obtained in the present paper.

Our work does not deal with gauge theories on general Riemannian manifolds, but is restricted to $\mathbb{R}^4$. It would clearly be desirable to generalise our analysis to include at least certain classes of Riemannian four-manifolds. In $\mathbb{R}^4$, the main difficulties derive from the subtle interplay between gauge invariance and infrared behaviour. In theories without gauge invariance, as shown in Ref.~[\onlinecite{koppermueller2006}], the method of flow equations can be adapted to rather general classes of non-compact Riemannian manifolds $(M,g)$ whose global geometric properties are controlled by various curvature conditions. These curvature conditions enter the IR behaviour of the heat kernel on $(M,g)$, see, \eg, Refs.~[\onlinecite{liyau1986},\onlinecite{davies1988}], and are hence important for the IR behaviour of the quantum field theory as well. On the other hand, in theories without an infrared problem -- when $(M,g)$ is chosen to be a compact Riemannian manifold without boundary, for instance -- gauge invariance poses no particular problem. In fact, a quite general framework for perturbative quantum field theories on general compact manifolds which also uses renormalisation group flow equations and, for gauge theories, the Batalin-Vilkovisky formalism, was presented in Ref.~[\onlinecite{costello2011}]. Actually, in this approach, the flow equations are only used to demonstrate that counterterms are local (and satisfy the usual power counting behaviour), while the analytic bounds, which in this case only involve the UV problem, are essentially dealt with graph by graph, using an analysis very similar to the classical one by BPHZ in flat space~\cite{bogoliubowparasiuk1957,hepp1966,zimmermann1969}. In fact, the difference from flat space analysis is rather minor, since to leading order the heat kernel basically behaves like the flat-space one, with all subleading orders well characterised via the usual heat kernel expansion. On the other hand, while the BPHZ method has been generalised to massless theories~\cite{lowensteinzimmermann1975a,lowensteinzimmermann1975b,lowenstein1976}, it has not been able to produce complete results on the infrared behaviour of correlation functions of arbitrary composite operators in massless theories even in flat space, by contrast with the flow equation method~\cite{kellerkopper1994,guidakopper2011,guidakopper2015}. In this sense, the approach of Ref.~[\onlinecite{costello2011}] appears to be a step backwards. Concerning the question of gauge invariance, Ref.~[\onlinecite{costello2011}] shows that this can be implemented at the level of the action. In our language, this corresponds to a proof that the anomaly $\mathsf{A}_0$ can be removed by a finite change of renormalisation/boundary conditions. However, the gauge invariance of correlation functions of composite operators, which we also examine, it not considered at all there. Summarising, we find it implausible that the method of Ref.~[\onlinecite{costello2011}] could give results on the IR behaviour -- or even the existence -- of correlation functions on non-compact Riemannian manifolds, at least not without highly non-trivial extensions.

\subsection{Notations}

We use a standard multiindex notation, where $\vec{w} = (w_1,\ldots,w_n)$ and $w_i = (w_i^1,\ldots,w_i^4)$ with $w_i^\alpha \geq 0$,
\begin{equation}
\partial^\vec{w} f(\vec{q}) \equiv \partial^{w_1^1}_{q_1^1} \cdots \partial^{w_n^4}_{q_n^4} f(q_1, \ldots, q_n) \eqend{,}
\end{equation}
\begin{equation}
\abs{\vec{w}} \equiv \sum_{i=1}^n \abs{w_i} \equiv \sum_{i=1}^n \sum_{\alpha=1}^4 w_i^\alpha \eqend{,}
\end{equation}
and
\begin{equation}
\vec{w}! \equiv \prod_{i=1}^n w_i! \equiv \prod_{i=1}^n \prod_{\alpha=1}^4 w_i^\alpha! \eqend{.}
\end{equation}
Furthermore, we write $\vec{w} > 0$ if for all $i$ we have $w_i^\alpha > 0$ for at least one $\alpha \in \{1,2,3,4\}$. To reduce notational clutter, we further stipulate that the vector $\vec{q}$ will always either have $n$ or $m+n$ entries, depending on context, while $\vec{q}_\rho \equiv \{ q_i \vert i \in \rho \}$. We use a condensed $L^2$ inner product notation, defining
\begin{equation}
\left\langle A, B \right\rangle \equiv \int A(x) B(x) \total^4 x
\end{equation}
and
\begin{equation}
A \ast B \equiv \int A(x-y) B(y) \total^4 y \eqend{.}
\end{equation}
We also use the positive part of the logarithm, defined as
\begin{equation}
\ln_+ x \equiv \sup( \ln x, 0 ) = \ln \sup(x,1) \eqend{,}
\end{equation}
which satisfies the useful identity
\begin{equation}
\ln_+ (a x + b y) \leq \ln_+ (a+b) + \ln_+ x + \ln_+ y
\end{equation}
for all $a,b \geq 0$ and $x,y \in \mathbb{R}$. For later use, we further define
\begin{equation}
\label{h_def}
\abs{\vec{q}} \equiv \sup_{Q \subseteq \{q_1,\ldots,q_n\}} \abs{\sum_{q \in Q} q} \eqend{,}
\end{equation}
which measures the maximum possible sum of a set of momenta,
\begin{equation}
\label{eta_i_def}
\eta_{q_i}(\vec{q}) \equiv \inf_{Q \subseteq \{q_1,\ldots,q_{n-1}\}\setminus\{q_i\}} \abs{q_i + \sum_{q \in Q} q}
\end{equation}
which measures the exceptionality of a set of momenta which includes $q_i$,
\begin{equation}
\label{eta_def}
\eta(\vec{q}) \equiv \inf_{q \in \{q_1,\ldots,q_{n-1}\}} \eta_q(q_1,\ldots,q_n) \eqend{,}
\end{equation}
which is the smallest such exceptionality,
\begin{equation}
\label{bareta_i_def}
\bar{\eta}_{q_i}(\vec{q}) \equiv \inf_{Q \subseteq \{q_1,\ldots,q_n\}\setminus\{q_i\}} \abs{q_i + \sum_{q \in Q} q}
\end{equation}
which is the same as $\eta_{q_i}$ except that the sum runs over all momenta including $q_n$, and
\begin{equation}
\label{bareta_def}
\bar{\eta}(\vec{q}) \equiv \inf_{q \in \{q_1,\ldots,q_n\}} \bar{\eta}_q(\vec{q}) \eqend{.}
\end{equation}
If $n = 1$, we define $\eta_q(q) \equiv \abs{q}$, and if $n = 0$ we define $\bar{\eta}() \equiv \mu$, $\eta() \equiv 0$ and $\abs{\cdot} \equiv 0$. With these definitions, the special cases are covered by the general estimates.

Throughout the paper, $\Lambda_0$ denotes the ultraviolet (UV) cutoff, $\Lambda$ the infrared (IR) cutoff and $\mu$ the renormalisation scale, which satisfy $\Lambda_0 \geq \Lambda \geq 0$ and $\Lambda_0 \geq \mu > 0$. Furthermore, $c$, $c'$, \etc, will denote arbitrary positive constants, which may change even within an equation, and $\mathcal{P}$ denotes a polynomial with positive coefficients, which also may change.

\section{The classical gauge theory, BRST transformations and cohomologies}
\label{sec_gauge}

The theory that we want to consider is pure Yang-Mills theory in Euclidean spacetime $\mathbb{R}^4$. It is classically described by the action
\begin{equation}
S = \frac{1}{2} \int \tr F \wedge \star F \eqend{,}
\end{equation}
where the two-form $F \equiv - \mathi/g [\mathcal{D}, \mathcal{D}]$ is the field strength of a connection $\mathcal{D}$ in the flat bundle $\mathbb{R}^4 \times G$, where $G$ is a compact real Lie group, $\star$ is the Hodge dual and $\tr$ denotes contraction with the positive definite Cartan-Killing metric on the real Lie algebra $\mathfrak{g}$ of $G$. Since the bundle is flat, we can decompose the connection $\mathcal{D}$ into the standard flat background connection $\partial$ of $\mathbb{R}^4$ and a real Lie-algebra valued one-form $A \in \mathcal{C}^\infty(\mathbb{R}^4, \Omega^1 \otimes \mathfrak{g})$ as
\begin{equation}
\mathcal{D} = \partial + \mathi g A \eqend{,}
\end{equation}
such that the field strength reads
\begin{equation}
F = \total A + \mathi g [A,A] \eqend{.}
\end{equation}
Here, $\Omega^p = \wedge^p T^*_x \mathbb{R}^4$ is the space of $p$-forms at $x \in \mathbb{R}^4$. At the non-perturbative level, the coupling constant $g$ could be absorbed into a redefinition of $A$, but we keep it as an explicit perturbation parameter to study the corresponding quantum field theory perturbatively. It is well known that straightforward perturbation theory does not work because the action is invariant under local gauge transformations
\begin{equation}
\label{a_gauge_trafo}
A \to A - \mathi \mathcal{D} f
\end{equation}
for any smooth Lie-algebra valued function $f \in \mathcal{C}^\infty(\mathbb{R}^4, \mathfrak{g})$, where the gauge-covariant derivative $\mathcal{D}$ acts via
\begin{equation}
\mathcal{D} f = \total f + \mathi g [A,f]
\end{equation}
on Lie-algebra valued functions, and extends to forms and products by linearity and the Leibniz rule. Therefore, the differential operator appearing in the free (\ie, quadratic in $A$) part of the Yang-Mills action, obtained by setting $g=0$, is not invertible, and no corresponding Gaussian measure exists that could be used to define the path integral even of the free theory. The freedom of making local gauge transformations could be used to set some of the components of $A$ to zero, such that the inverse of the differential operator is well-defined for the remaining components, but this is cumbersome, and either breaks $\mathrm{E}(4)$ covariance, or gauge invariance, or locality. It is much more convenient to use a manifestly covariant quantisation approach, which consists in first enlarging the field space of the theory, quantising the enlarged theory and then, at a final stage, remove the extra degrees of freedom for physical states and observables. We will also follow this well-known scheme.

In this approach, called the ``BRST method''~\cite{becchietal1975}, one adds additional dynamical Grassmann Lie-algebra valued fields $c \in \mathcal{C}^\infty\left( \mathbb{R}^4, \mathfrak{g} \otimes \Lambda \right)$ (where $\Lambda = \Lambda\left( \mathbb{R}^\infty \right)$ is an infinite-dimensional auxiliary Grassmann algebra) and $\bar{c} \in \mathcal{C}^\infty\left( \mathbb{R}^4, \mathfrak{g} \otimes \Lambda \right)$ (called ghost and antighost), and a Lie-algebra valued field $B \in \mathcal{C}^\infty\left( \mathbb{R}^4, \mathfrak{g} \right)$ (called auxiliary, or Nakanishi-Lautrup field~\cite{nakanishi1966,lautrup1967,kugoojima1979}) to the theory, and considers the ``gauge-fixed'' action
\begin{equation}
\label{gaugefixed}
S = \frac{1}{2} \int \tr F \wedge \star F + \int \tr \left[ B \wedge \star \left( \frac{B}{2} + \mathi \xi G[A] \right) \right] - \mathi \int \tr \left[ \mathcal{D} c \wedge \star \frac{\delta \left( G[A] \bar{c} \right)}{\delta A} \right] \eqend{,}
\end{equation}
where $G[A]$ is some local functional (see Definition~\ref{def_local_functional}) of the field $A$, called gauge-fixing functional, and $\xi$ is a real parameter. This functional is conveniently chosen so as to render the differential operator in the free part of the (new) theory invertible (our concrete choice is discussed below). To relate the theory described by this new action to the original theory, one introduces a (classical) BRST transformation $\brst$, which acts on the fields as
\begin{equation}
\label{brst_standard}
\brst A = \mathcal{D} c \eqend{,} \quad \brst c = - \frac{\mathi}{2} g [c,c] \eqend{,} \quad \brst \bar{c} = \xi B \eqend{,} \quad \brst B = 0 \eqend{,}
\end{equation}
and is extended to arbitrary smooth, local functionals of the fields by linearity and a graded Leibniz rule, where $\brst$ anticommutes with Grassmann variables and forms of odd degree, and commutes with everything else; especially, it anticommutes with the exterior differential $\total$. If one assigns ghost number $1$ to the ghost $c$, ghost number $-1$ to the antighost $\bar{c}$ and ghost number $0$ to $A$ and $B$, the BRST transformation $\brst$ raises the ghost number by $1$. A crucial property of $\brst$ is its nilpotency $\brst^2 = 0$, and the fact that $\brst S = 0$, which can be seen most easily by writing $S$ in the form
\begin{equation}
S = \frac{1}{2} \int \tr F \wedge \star F + \brst \int \tr \left[ \bar{c} \wedge \star \left( \frac{1}{2 \xi} B + \mathi G[A] \right) \right] \eqend{.}
\end{equation}
Indeed, since $\brst$ acts on $A$ like a linearised gauge transformation~\eqref{a_gauge_trafo} where the gauge function $f$ is replaced by the ghost field $c$, the first term is invariant under $\brst$ since it is gauge invariant, and the second term is invariant since $\brst^2 = 0$. The action is an example of a functional which is the integral of a ``local functional'' of the fields $\phi = ( A, c, \bar{c}, B )$, which can be thought of as an element of the space $\mathcal{C}^\infty\left( \mathbb{R}^4, \Omega^1 \otimes \mathfrak{g} \oplus (\mathfrak{g} \otimes \Lambda)^2 \oplus \mathfrak{g} \right)$. Since this notion is of key relevance throughout this paper, we formalise it in the following definition, which also sets up some relevant cohomologies.
\begin{definition}
\label{def_local_functional}
A functional $\op(x, \cdot)$
\begin{equation}
\phi \in \mathcal{C}^\infty\left( \mathbb{R}^4, \Omega^1 \otimes \mathfrak{g} \oplus (\mathfrak{g} \otimes \Lambda)^2 \oplus \mathfrak{g} \right) \mapsto \op(x, \phi) \in \Omega \otimes \mathfrak{g} \otimes \Lambda \eqend{,}
\end{equation}
where $\Omega = \bigoplus_{p=0}^4 \Omega^p$ are the exterior forms at $x$, is called ``local'' (and polynomial in the fields) if it is of the form $\op(x,\phi) = f(x, \phi(x), \partial \phi(x), \ldots, \partial^s \phi(x))$ for some $s$, where $f(x, \cdot)$ is smooth in $x$ and polynomial in the other entries. We denote by $\mathcal{F}^{g,p}$ the space of local $p$-form valued functionals with total ghost number $g$, such that
\begin{equation}
\op(x, \phi) \in \Omega \otimes \mathfrak{g} \otimes \Lambda = \bigoplus_{p=0}^4 \bigoplus_g \mathcal{F}^{g,p} \eqend{,}
\end{equation}
and, since $\op$ is polynomial in the fields, it decomposes into pieces of definite form degree and ghost number $\op(x, \phi) = \sum_{g,p} \op^{g,p}(x, \phi)$ with $\op^{g,p}(x, \phi) \in \mathcal{F}^{g,p}$.
\end{definition}
We define the cohomologies
\begin{equation}
\label{cohom_def}
H^{g,p}(\brst) = \frac{\mathrm{Ker}(\brst\colon \mathcal{F}^{g,p} \to \mathcal{F}^{g+1,p})}{\mathrm{Im}(\brst\colon \mathcal{F}^{g-1,p} \to \mathcal{F}^{g,p})} \eqend{.}
\end{equation}
On the spaces $\mathcal{F}^{g,p}$ we have an obvious action of the Euclidean group $\mathrm{E}(4)$, given by
\begin{equation}
(e \op)(\cdot, \phi) = (e^{-1})^* \op(\cdot, e^* \phi)
\end{equation}
for $e \in \mathrm{E}(4)$. The space of functionals invariant under the action of $\mathrm{E}(4)$ is denoted by $\mathcal{F}^{g,p}_{\mathrm{E}(4)}$. Since $\brst$ and $\total$ commute with the action of $\mathrm{E}(4)$, they restrict to the subspace of invariant functionals. The corresponding equivariant cohomologies
\begin{equation}
\label{cohom_def_eq}
H^{g,p}_{\mathrm{E}(4)}(\brst) = \frac{\mathrm{Ker}(\brst\colon \mathcal{F}^{g,p}_{\mathrm{E}(4)} \to \mathcal{F}_{\mathrm{E}(4)}^{g+1,p})}{\mathrm{Im}(\brst\colon \mathcal{F}^{g-1,p}_{\mathrm{E}(4)} \to \mathcal{F}^{g,p}_{\mathrm{E}(4)})}
\end{equation}
can therefore also be defined consistently. Most statements about $H^{g,p}$ have obvious analogues in the equivariant case, so we will usually not dwell on this. The spaces $H^{g,p}$ have a linear structure, and also a graded commutative multiplicative structure $H^{g,p}(\brst) \times H^{g',p'}(\brst) \to H^{g+g',p+p'}(\brst)$, where monomials commute if $g g' + p p' = 0 \operatorname{mod} 2$ and anticommute otherwise. It is clear that gauge-invariant observables of the original Yang-Mills theory [of the form given in equation~\eqref{yangmills_invops}] are annihilated by $\brst$, \ie, they are in the kernel of $\brst$. The following standard theorem in BRST cohomology~\cite{barnichetal2000} characterises completely the cohomology of $\brst$:
\begin{lemma}
$H^{g,p}(\brst)$ is generated (over $\mathcal{C}^\infty(\mathbb{R}^4)$) by elements of the form
\begin{equation}
\label{hgp_elements}
\prod_i p_i(c) \prod_j q_j(F, \mathcal{D} F, \ldots, \mathcal{D}^l F) \eqend{,}
\end{equation}
where $p_i$ and $q_j$ are invariant polynomials of the Lie algebra, $p_i$ is homogeneous of degree $g$ and the tensor structures are suitably contracted with the flat Euclidean metric to yield a $p$-form. In particular, the local and covariant functionals in $H^{0,p}(\brst)$, the cohomology of $\brst$ at ghost number $0$, are in one-to-one correspondence with the classical gauge-invariant observables.
\end{lemma}
This lemma is the raison d'{\^e}tre for the usefulness of the BRST formalism: Because the cohomology classes are in one-to-one correspondence with gauge-invariant fields (at ghost number 0), one can first quantise the gauge-fixed theory described by equation~\eqref{gaugefixed}, and then in the end pass to BRST equivalence classes. What one needs to show for this is that the correlation functions are well defined on such classes, \ie, that they are BRST invariant. That this is indeed so is the final statement of our paper, Theorem~\ref{thm4}.

Along the way of the long proof of this theorem, one also needs to consider various extension of the BRST formalism, which we now introduce. Besides the BRST differential $\brst\colon \mathcal{F}^{g,p} \to \mathcal{F}^{g+1,p}$, we also have the exterior differential $\total\colon \mathcal{F}^{g,p} \to \mathcal{F}^{g,p+1}$. The differentials and grading have been set up so that not only $\total^2 = 0 = \brst^2$, but even $\total \brst + \brst \total = 0$. Thus, $\mathcal{F}^{p,g}$ forms a bi-complex, and we can form the relative cohomologies $H^{g,p}(\brst\vert\total)$. Classes in this cohomology are in correspondence with form-valued functionals which are in the kernel of $\brst$ modulo exact forms. Concerning these spaces, one has the following theorem~\cite{barnichetal2000,piguetsorella,kugo}, which will be useful to characterise possible ``anomalies'' in the quantisation process:
\begin{lemma}
Each element of $H^{g,p}(\brst\vert\total)$ is given by a linear combination of elements of $H^{g,p}(\brst)$~\eqref{hgp_elements} and the $p$-form, ghost number $g$ part of
\begin{equation}
\prod_i f_i(F) \prod_j q_{r_j}(A + c, F) \eqend{,}
\end{equation}
where the $f_i$ are gauge invariant monomials containing only the field strength $F$ but not its derivatives, the $q_{r_j}$ are the generalised Chern-Simons forms
\begin{equation}
q_r(A + c, F) \equiv \int_0^1 \tr \left[ (A+c) \left( t F + \frac{\mathi}{2} g t (t-1) [A+c,A+c] \right)^{m(r)-1} \right] \total t \eqend{,}
\end{equation}
with the trace in some representation, and $m(r)$ is the order of the $r$-th independent Casimir of $G$.
\end{lemma}
The choice of the gauge fixing functional $G$ in the gauge-fixed action $S$~\eqref{gaugefixed} does not play a role in as far as the cohomologies are concerned -- it neither affects the definition of $\brst$, nor does a change of $G$ alter the fact that $S$ remains BRST-invariant. The choice of $G$ does, however, matter when quantising the theory. A choice which is often made is $G[A] = \star \total \star A$ (called $R_\xi$ gauges), which preserves $\mathrm{E}(4)$ covariance. However, in these gauges the differential operator appearing in the ``free action'' (\ie, the part of $S$ quadratic in the fields) is not positive definite [on the subspace corresponding to the Grassmann even fields $(A,B)$], and for $\xi^2 < 0$ even has imaginary eigenvalues. This is of no relevance in the usual formal perturbation theory, where only formal Gaussian functional integration is carried out regardless of the actual convergence of the integrals, but is becomes problematic in the flow equation framework where one needs a proper positive definite Gaussian measure in order to derive bounds on correlation functions. We can cure this issue by making the field redefinition
\begin{equation}
B \to B - \mathi \xi \star \total \star A \eqend{,}
\end{equation}
leading to the action
\begin{equation}
\label{yangmills_action}
S = \frac{1}{2} \int \tr F \wedge \star F + \frac{1}{2} \int \tr B \wedge \star B + \frac{\xi^2}{2} \int \tr \left( \star \total \star A \wedge \total \star A \right) + \mathi \int \tr \bar{c} \wedge \total \star \mathcal{D} c \eqend{,}
\end{equation}
which is positive definite on the subspace corresponding to $(A,B)$ as long as $\xi^2 > 0$. Since the new action was obtained simply by a field redefinition of $B$, it is still BRST-invariant, and the new BRST transformation is given by inserting the redefinition of $B$ in the old BRST transformation~\eqref{brst_standard}. After the field redefinition, the BRST transformation reads
\begin{equation}
\label{brst_new}
\brst A = \mathcal{D} c \eqend{,} \quad \brst c = - \frac{\mathi}{2} g [c,c] \eqend{,} \quad \brst \bar{c} = \xi B - \mathi \xi^2 \star \total \star A \eqend{,} \quad \brst B = \mathi \xi \star \total \star \mathcal{D} c \eqend{.}
\end{equation}
Since the new $\brst$ corresponds to the old one (denoted by the same symbol by abuse of notation) via a field redefinition, it still satisfies all the properties of the standard BRST transformation, and the above theorems on the BRST cohomologies remain valid also for the new $\brst$, which we will use in the following.

In the perturbatively quantised theory, it is not a priori clear whether BRST invariance can be maintained, and thus whether correlation functions of classically gauge-invariant observables, obtained using the above cohomological construction, are gauge-independent in the quantum theory. The outcome naturally depends on the form of the regulator that is used; in formal perturbation theory using dimensional regularisation~\cite{thooft1971,leibbrandt1975,breitenlohnermaison1977}, BRST invariance can be formally maintained for pure Yang-Mills theory, but including chiral fermions can lead to the well-known Adler-Bardeen axial anomaly~\cite{adler1969,bardeen1969} since dimensional regularisation does not preserve chiral symmetry. In the flow equation framework, a momentum cutoff is used as a regulator which manifestly breaks BRST invariance, and in order to classify possible violations of classical BRST symmetry (anomalies) in the renormalised quantum theory and to derive stringent consistency conditions on possible anomalies, it is very convenient to pass to the Batalin-Vilkovisky or field-antifield formalism~\cite{batalinvilkovisky1981,batalinvilkovisky1983,batalinvilkovisky1984,weinberg_v2} and introduce ``external sources'' for the BRST variations of the basic fields, called ``antifields''. They have Grassmann parity and form degree opposite to that of the corresponding basic fields, and for pure Yang-Mills theory they are Lie-algebra valued. We thus introduce the three-form $A^\ddag \in \mathcal{C}_0^\infty\left( \mathbb{R}^4, \Omega^3 \otimes \mathfrak{g} \otimes \Lambda \right)$ and the four-forms
$c^\ddag, \bar{c}^\ddag \in \mathcal{C}^\infty_0\left( \mathbb{R}^4, \Omega^4 \otimes \mathfrak{g} \right)$ and $B^\ddag \in \mathcal{C}^\infty_0\left( \mathbb{R}^4, \Omega^4 \otimes \mathfrak{g} \otimes \Lambda \right)$, denoted collectively by
\begin{equation}
\phi^\ddag \equiv ( A^\ddag, c^\ddag, \bar{c}^\ddag, B^\ddag ) \in \mathcal{C}^\infty_0\left( \mathbb{R}^4, \Omega^3 \otimes \mathfrak{g} \otimes \Lambda \oplus \left( \Omega^4 \otimes \mathfrak{g} \right)^{\otimes 2} \oplus \Omega^4 \otimes \mathfrak{g} \otimes \Lambda \right) \eqend{,}
\end{equation}
and augment the action to obtain
\begin{splitequation}
\label{action_field_antifield_coupling}
S_\text{total} &= S - \int \tr \left[ \left( \brst A \right) \wedge A^\ddag + \left( \brst c \right) \wedge c^\ddag + \left( \brst \bar{c} \right) \wedge \bar{c}^\ddag + \left( \brst B \right) \wedge B^\ddag \right] \\
&= S + \int \tr \left[ - \left( \mathcal{D} c \right) \wedge A^\ddag + \frac{\mathi}{2} g [c,c] \wedge c^\ddag - \xi \left( B - \mathi \xi \star \total \star A \right) \wedge \bar{c}^\ddag - \mathi \xi \left( \star \total \star \mathcal{D} c \right) \wedge B^\ddag \right] \eqend{.}
\end{splitequation}
Because only the total action will play a role in this paper, we will drop the subscript ``total'' on it to avoid clutter, \ie, we will write again $S$ for $S_\text{total}$ by abuse of notation. We then define the BV- or antibracket $(\cdot,\cdot)$ on arbitrary local functionals of fields and antifields by declaring fields and antifields conjugate to each other and extending it to general functionals by linearity and a graded Leibniz rule; the exact definition is given below in equation~\eqref{bv_def} in component notation. The BRST invariance of the action then results in $(S,S) = 0$, and we define the classical Slavnov-Taylor differential $\st$ acting on a functional of fields and antifields by
\begin{equation}
\st F \equiv (S,F) \eqend{,}
\end{equation}
which satisfies
\begin{equation}
\st^2 = (S, (S, \cdot) ) = 0
\end{equation}
as a consequence of the (graded) Jacobi identity fulfilled by the antibracket and $(S,S) = 0$; linearity and the graded Leibniz rule also extend to $\st$ from $(\cdot,\cdot)$ by definition. Furthermore, on functionals $F$ which do not depend on the antifields the Slavnov-Taylor differential reduces to the BRST differential $\st F = \brst F$.

The relevant cohomologies then also involve the Slavnov-Taylor differential $\st$ instead of the BRST differential $\brst$, and are defined analogously to~\eqref{cohom_def}. Of particular importance for us is the space $H^{1,4}(\st\vert\total)$ because it is related to the anomaly appearing in the quantisation of the theory described by $S$, and the spaces $H^{0/1,p}(\st)$ because they correspond to gauge invariant observables and their anomalies. For semisimple Lie groups, we have the following theorem~\cite{barnichetal2000}:
\begin{lemma}
For ghost numbers $g=0,1$ representatives of $H^{g,p}(\st)$ and $H^{g,p}(\st\vert\total)$ can be chosen to be independent of antifields, \ie,
\begin{equation}
H^{g,p}(\st) \cong H^{g,p}(\brst) \quad\text{ and }\quad H^{g,p}(\st\vert\total) \cong H^{g,p}(\brst\vert\total) \quad\text{ for }\quad g=0,1 \eqend{.}
\end{equation}
In other words, any representative $\op^{g,p}$ of a class in $H^{g,p}(\st\vert\total)$ can be written as a local functional in the kernel of $\brst$ modulo $\total$ of form/ghost degree $p$/$g$ not containing antifields, plus a functional of the form $\st \op^{g-1,p} + \total \op^{g,p-1}$, where $\op^{g-1,p}$ and $\op^{g,p-1}$ are smooth, local functionals of fields and antifields of form/ghost degree $p$/$(g-1)$ and $(p-1)$/$g$, respectively, and similarly for $H^{g,p}(\st)$.
\end{lemma}
An explicit representative of $H^{1,4}(\st\vert\total)$ is given in equation~\eqref{cohom_parity}, and representatives of $H^{g,p}(\brst)$ are given in equation~\eqref{hgp_elements}. Especially, since for semisimple Lie groups there are no invariant polynomials of the Lie algebra of degree $1$, we have
\begin{equation}
\label{cohom_op}
H^{1,p}(\st) = \emptyset \eqend{.}
\end{equation}

In order to avoid a cluttered notation in subsequent sections (to the extent possible), we will denote the fields collectively as $\phi = ( A, B, c, \bar{c} )$, and we denote their various components as $\phi_K$, where $K$ is an index that distinguishes the kind of field, the tensor index (if any) and the Lie algebra index relative to an arbitrarily chosen basis in $\mathfrak{g}$. We likewise denote $\phi^\ddag = ( A^\ddag, B^\ddag, c^\ddag, \bar{c}^\ddag )$ and denote its components as $\phi^\ddag_L$. As usual in quantum field theory, it also useful to assign ``engineering dimensions'' to the basic fields $\phi$ and $\phi^\dagger$. There is a certain amount of freedom to do this in the extended theory with ghosts and antifields. For technical reasons, we find it convenient to assign an engineering dimension $\geq 1$ to all basic fields, setting $[A] = [c] = [\bar{c}] = 1$ and $[B] = 2$, where here and in the following, square brackets indicate the dimension, and we define the dimension of arbitrary monomials of the fields by demanding the dimension to be additive. With this assignment, $\brst$ increases the dimension by one unit. We also assign an engineering dimension to antifields according to the rule
\begin{equation}
\label{antifield_dim}
[\phi^\ddag] = 3 - [\phi] \eqend{,}
\end{equation}
and then also $\st$ increases the dimension by one. As a further notation, we define
\begin{equation}
[\vec{K}] \equiv \sum_{i=1}^m [\phi_{K_i}] \eqend{,} \qquad [\vec{L}^\ddag] \equiv \sum_{j=1}^n [\phi_{L_j}^\ddag]
\end{equation}
for multiindices $\vec{K} = (K_1, \ldots, K_m)$ and $\vec{L}^\ddag = (L_1, \ldots, L_n^\ddag)$. The assignment of the various gradings of all fields, including our assignments of the engineering dimensions in the theory is summarised for convenience in Table~\ref{table_fields}.
\setcounter{table}{0}
\begin{table}
\begin{center}
\begin{tabular}{ccccc|ccccc}
\hline
field     & dim. & form deg. & gh. number & grading & antifield & dim. & form deg. & gh. number & grading \\
\hline
$A$       & $1$  & $1$ &  $0$ & $+1$ & $A^\ddag$       & $2$  & $3$ & $-1$ & $-1$ \\
$c$       & $1$  & $0$ &  $1$ & $-1$ & $c^\ddag$       & $2$  & $4$ & $-2$ & $+1$ \\
$\bar{c}$ & $1$  & $0$ & $-1$ & $-1$ & $\bar{c}^\ddag$ & $2$  & $4$ &  $0$ & $+1$ \\
$B$       & $2$  & $0$ &  $0$ & $+1$ & $B^\ddag$       & $1$  & $4$ & $-1$ & $-1$ \\
\hline
\end{tabular}
\end{center}
\caption{Fields and antifields, their engineering dimension, form degree, ghost number and Grassmann grading.}
\label{table_fields}
\end{table}

In our abstract component notation, the antibracket then reads
\begin{equation}
\label{bv_def}
(F,G) \equiv \int \left[ \frac{\delta_\text{R} F}{\delta \phi_K(x)\vphantom{\delta \phi_K^\ddag}} \frac{\delta_\text{L} G}{\delta \phi_K^\ddag(x)} - \frac{\delta_\text{R} F}{\delta \phi_K^\ddag(x)} \frac{\delta_\text{L} G}{\delta \phi_K(x)\vphantom{\delta \phi_K^\ddag}} \right] \total^4 x \eqend{,}
\end{equation}
with the right (left) derivative $\delta_\text{R}$ ($\delta_\text{L}$) given by
\begin{equation}
\label{deltarl_def}
\delta F = \frac{\delta_\text{R} F}{\delta \phi_K} \delta \phi_K = \delta \phi_K \frac{\delta_\text{L} F}{\delta \phi_K} \eqend{,}
\end{equation}
and equal indices $K$, $L$, \etc, are summed over. The Slavnov-Taylor differential takes the form
\begin{equation}
\label{st_def}
\st F = (S, F) = \int \left[ \frac{\delta_\text{R} S}{\delta \phi_K(x)\vphantom{\delta \phi_K^\ddag}} \frac{\delta_\text{L} F}{\delta \phi_K^\ddag(x)} + \left( \brst \phi_K(x) \right) \frac{\delta_\text{L} F}{\delta \phi_K(x)\vphantom{\delta \phi_K^\ddag}} \right] \total^4 x \eqend{,}
\end{equation}
and one easily checks that
\begin{equation}
\label{classical_brst_action}
\st S = (S,S) = 2 \int \left( \brst \phi_K(x) \right) \frac{\delta_\text{L} S}{\delta \phi_K(x)} \total^4 x = 2 \brst S = 0 \eqend{,}
\end{equation}
since $S$ was BRST-invariant. Splitting the action $S$ into a free part $S_0$ (comprising all terms quadratic in the basic fields and antifields), and an interaction $S_\text{int} \equiv S - S_0$, we also define a corresponding free ST differential
\begin{equation}
\label{st0_def}
\st_0 F \equiv (S_0, F) \eqend{,}
\end{equation}
which will play a role later in the derivation of Ward identities.

\section{The flow equation framework}
\label{sec_flow}

\subsection{Generating functionals of CACs}

The objects of interest in the QFT are (connected) correlation functions, defined as functional derivatives of the logarithm of the generating functional $Z(J)$
\begin{equation}
\expect{ \varphi_{K_1}(x_1) \cdots \varphi_{K_n}(x_n) }_\text{c} = \left. \frac{\hbar^n \delta^n}{\delta J_{K_1}(x_1) \cdots \delta J_{K_n}(x_n)} \ln Z(J) \right\rvert_{J=0} \eqend{,}
\end{equation}
which is given formally by the Euclidean path integral
\begin{equation}
\label{euclidean_pathintegral}
Z(J) = \frac{1}{Z(0)} \int \exp\left[ - \frac{1}{\hbar} S[\varphi] + \frac{1}{\hbar} \left\langle J_K, \varphi_K \right\rangle \right] \,\mathcal{D} \varphi
\end{equation}
over a space of ``functions'' $\varphi = ( \varphi_K )$ on $\mathbb{R}^4$. In our case these fields comprise the gauge-, ghost, and auxiliary field, 
$( \phi_K ) = (A, c, \bar{c}, B)$, as described in the preceding section. Of course it is not at all clear what such a functional integral is supposed to mean. 
To give a mathematically precise meaning to this formal expression, we start from the free (i.e., quadratic) part of the action which we write as
\begin{equation}
\label{free_action}
S_0[\varphi] = \frac{1}{2} \left\langle \varphi_K, \left( C^{0, \infty} \right)^{-1}_{KL} \ast \varphi_L \right\rangle - \left\langle \brst_0 \varphi_K, \phi_K^\ddag \right\rangle \eqend{,}
\end{equation}
where here and in the following we use the summation convention for the indices $K, L, \dots$, and where $\brst_0 $ is the linear part of the full BRST transformation~\eqref{brst_new}. This defines a free covariance $C_{KL}^{0, \infty}$ of test functions $\varphi_K \in \mathcal{S}(\mathbb{R}^4)$. Note that the antifields $\varphi^\ddag_L$ act only as external sources and do not participate in the formal functional integration in the generating functional~\eqref{euclidean_pathintegral}; we thus do not associate any covariance with them. We then introduce an UV cutoff $\Lambda_0$ and an IR cutoff $\Lambda$ and define a regularised covariance matrix (``propagator'')
\begin{equation}
\label{reg_covariance}
C_{KL}^{\Lambda, \Lambda_0} \equiv C_{KL}^{0, \infty} \ast \left( R^{\Lambda_0} - R^\Lambda \right) \eqend{,}
\end{equation}
where $R^\Lambda$ is a real, smooth, $\mathrm{E}(4)$ invariant regulator function $R^\Lambda$, in Fourier space analytic at $p=0$, which fulfils the properties
\begin{equations}[r_prop]
0 < R^\Lambda(p) &< 1 \quad\text{ for }\quad 0 < \Lambda < \infty \eqend{,} \\
R^0(p) &= 0 \eqend{,} \qquad R^\infty(p) = 1 \eqend{,} \label{r_prop_2} \\
R^\Lambda(p) &< R^{\Lambda_0}(p) \quad\text{ for }\quad \Lambda < \Lambda_0 \eqend{,} \\
\label{r_prop_bound} \abs{\partial^w \partial^k_\Lambda R^\Lambda(p)} &\leq c \sup(\abs{p},\Lambda)^{-k-\abs{w}} \mathe^{-\frac{\abs{p}^2}{2 \Lambda^2}} \eqend{.}
\end{equations}
A simple example of such a regulator is given by
\begin{equation}
\label{r_example}
R^\Lambda(p) = \mathe^{-\frac{\abs{p}^2}{\Lambda^2}} \eqend{,}
\end{equation}
which is even analytic for all $p$. For Yang-Mills theories with the action~\eqref{yangmills_action} expressed in component form and $(\phi_K) = (A, c, \bar c, B)$ as above, one checks that
\begin{equation}
\label{cov_def}
C^{\Lambda, \Lambda_0}_{KL}(p) = \delta_{ab} \begin{pmatrix} \delta_{\mu\nu} + (\xi^{-1}-1) p_\mu p_\nu / p^2 & 0 & 0 & 0 \\ 0 & 0 & - 1 & 0 \\ 0 & 1 & 0 & 0 \\ 0 & 0 & 0 & p^2 \end{pmatrix} \frac{R^{\Lambda_0}(p) - R^\Lambda(p)}{p^2} \eqend{,}
\end{equation}
which is positive definite on the subspace corresponding to the Grassmann even fields $(A,B)$. For later use, we note that from property~\eqref{r_prop_bound} of the regulator and the fact that $S_0$ has dimension $4$ we can get the crucial estimate
\begin{equation}
\label{prop_abl}
\abs{\partial^w \partial_\Lambda C^{\Lambda, \Lambda_0}_{KL}(p)} \leq c \sup(\abs{p}, \Lambda)^{-5+[\phi_K]+[\phi_L]-\abs{w}} \, \mathe^{-\frac{\abs{p}^2}{2 \Lambda^2}} \eqend{.}
\end{equation}
The positive definite property of the regularised covariance $C^{\Lambda, \Lambda_0}_{KL}$ implies that one can defines from it a unique corresponding Gaussian measure on the space of test functions $\varphi \in \mathcal{S}(\mathbb{R}^4)$, which we denote by $\total \nu^{\Lambda, \Lambda_0}$~\cite{glimmjaffe1987,salmhofer1999}. Note that while for bosonic fields this is the usual Lebesgue integral, for fermionic (Grassmann-valued) fields the integral symbol is only formal and the integral is defined via a linear map on the space of Grassmann-valued functions, see Ref.~[\onlinecite{salmhofer1999}] for details. Alternatively, one may define the fermionic integral via a functional determinant. As is customary, we will not make this distinction explicit. We are thus led to consider the regularised generating functional
\begin{equation}
\label{regularised_genfunc}
Z^{\Lambda, \Lambda_0}(J) \equiv \int \exp\left[ - \frac{1}{\hbar} L^{\Lambda_0}[\varphi] + \frac{1}{\hbar} \left\langle \brst_0 \varphi_K, \phi_K^\ddag \right\rangle + \frac{1}{\hbar} \left\langle J_K, \varphi_K \right\rangle \right] \total\nu^{\Lambda, \Lambda_0}(\varphi) \eqend{,}
\end{equation}
for finite cutoffs $\Lambda$ and $\Lambda_0$. Here, $L^{\Lambda_0}$ is the part $S_\text{int}$ of the action $S$ which contains the terms not already included 
in the quadratic part $S_0$, \emph{plus} additional ``counterterms'' (depending on the cutoff $\Lambda_0$) by means of which we hope to make the functional integral well-defined upon removal of the cutoffs. One formally sees that the counterterms are of order $\hbar$ (treating, as we will, $\hbar$ as an expansion parameter):
\begin{equation}
\label{lint}
L^{\Lambda_0} = S_\text{int} + \bigo{\hbar} \eqend{,}
\end{equation}
and it will follow that the counterterms must be local operators that cannot exceed a certain dimension. We refrain from giving an explicit listing of all possible terms appearing in $L^{\Lambda_0}$, since this is a simple bookkeeping exercise using the dimension assignment given in Table~\ref{table_fields}. If we knew that $L^{\Lambda_0}$ was bounded from below, then the Gaussian integral~\eqref{regularised_genfunc} would actually be completely well defined, and hence would rigorously define $Z^{\Lambda, \Lambda_0}(J)$ as functional of the source $J$. We could then obtain the regularised correlation functions by functional differentiation with respect to $J$, and if we could show that the cutoffs can removed as $\Lambda_0 \to \infty$, $\Lambda \to 0$ with a suitable choice of the counterterms, this would then rigorously define the theory. We will indeed essentially follow this procedure, but there are two problems:
\begin{enumerate}
\item Since the regulator $R^\Lambda$ violates BRST invariance, the interaction part $L^{\Lambda_0}$ must contain counterterms which are also not BRST invariant. We thus a priori have to take the interaction part as the most general polynomial in fields and antifields of dimension $4$ which is invariant under the remaining unbroken global symmetries (such as $\mathrm{E}(4)$ invariance and ghost number conservation). The question then arises whether and how BRST invariance is restored when we remove the cutoffs. This will in the end be achieved by the powerful machine of Ward identities for the BRST symmetry.
\item Because $L^{\Lambda_0}$ is not manifestly bounded from below (it even contains Grassmann-valued fields, so that it is actually meaningless to ask it to be bounded from below in a straightforward sense), we cannot give a precise meaning to the functional integral even with cutoffs. We effectively ignore this issue and instead proceed by formally expanding out the exponential. Each term in the resulting series is then well defined, and we will proceed by giving mathematical meaning to each such term as the cutoffs are removed. The disadvantage is that we will only be able to define the theory in the sense of a formal power series (in $\hbar$, as it turns out).
\end{enumerate}
To proceed, it is convenient to pass to connected amputated correlation functions (CACs), which are obtained by taking derivatives of $\ln Z^{\Lambda, \Lambda_0}(J)$ with respect to the source $J$ and acting with the inverse of the free covariance $\left( C^{\Lambda, \Lambda_0} \right)^{-1}$ on each field. In the BV formalism, it is also convenient to remove the free part (\ie, quadratic in fields/antifields) from the field-antifield coupling displayed in equation~\eqref{action_field_antifield_coupling}. Both operations can be done by replacing the source by
\begin{equation}
J_K \to \frac{\delta_\text{L}}{\delta \phi_K} \left[ \frac{1}{2} \left\langle \phi_K, \left( C^{\Lambda, \Lambda_0} \right)^{-1}_{KL} \ast \phi_L \right\rangle - \left\langle \brst_0 \phi_K, \phi_K^\ddag \right\rangle \right]
\end{equation}
and shifting the integration variable $\varphi \to \varphi + \phi$ in the generating functional. In the flow equation framework, it is further necessary to remove the free part from the generating functional, and we denote the resulting generating functional for the CACS by $L^{\Lambda, \Lambda_0}\left( \phi, \phi^\ddag \right)$, where $\phi, \phi^\ddag \in \mathcal{S}(\mathbb{R}^4)$ are Schwartz functions. These manipulations in fact amount to defining
\begin{equation}
\label{l_0op_def}
L^{\Lambda, \Lambda_0} + I^{\Lambda, \Lambda_0} \equiv - \hbar \ln \left[ \nu^{\Lambda, \Lambda_0} \conv \exp\left( - \frac{1}{\hbar} L^{\Lambda_0} \right) \right] \eqend{,}
\end{equation}
with the convolution $\conv$ with the Gaussian measure defined by
\begin{equation}
\label{conv_def}
\left( \nu^{\Lambda, \Lambda_0} \conv F \right)(\phi) \equiv \int F(\phi + \varphi) \total\nu^{\Lambda, \Lambda_0}(\varphi)
\end{equation}
for $\phi, \varphi \in \mathcal{S}(\mathbb{R}^4)$, and where $I^{\Lambda, \Lambda_0}$, which only depends on antifields, is given by
\begin{equation}
I^{\Lambda, \Lambda_0} \equiv - \hbar \ln Z^{\Lambda, \Lambda_0}(0) \eqend{.}
\end{equation}
The antifield-independent part of $I^{\Lambda, \Lambda_0}$ is proportional to the volume of space and thus needs a finite volume to be well-defined. We do not make this explicit, since we are only interested in the quantities obtained from $\ln Z^{\Lambda,\Lambda_0}(\phi, \phi^\ddag)$ by functional differentiation in $\phi$ and $\phi^\ddag$, from which the antifield-independent part of $I^{\Lambda, \Lambda_0}$ anyhow drops out.

Taking a $\Lambda$ derivative of equation~\eqref{l_0op_def} and using the properties of Gaussian measures~\cite{glimmjaffe1987,mueller2003}, we obtain the equation
\begin{splitequation}
\label{l_0op_flow}
\partial_\Lambda L^{\Lambda, \Lambda_0} + \partial_\Lambda I^{\Lambda, \Lambda_0} &= \frac{\hbar}{2} \left\langle \frac{\delta}{\delta \phi_K}, \left( \partial_\Lambda C^{\Lambda, \Lambda_0}_{KL} \right) \ast \frac{\delta}{\delta \phi_L} \right\rangle L^{\Lambda, \Lambda_0} \\
&\quad- \frac{1}{2} \left\langle \frac{\delta}{\delta \phi_K} L^{\Lambda, \Lambda_0} , \left( \partial_\Lambda C^{\Lambda, \Lambda_0}_{KL} \right) \ast \frac{\delta}{\delta \phi_L} L^{\Lambda, \Lambda_0} \right\rangle \eqend{.}
\end{splitequation}
This differential equation, called the \emph{renormalisation group flow equation}, is the starting point in our analysis. If we expand $L^{\Lambda, \Lambda_0}$ in powers of $\hbar$, then each order is completely well defined, since it has a well-defined functional integral representation. However, rather than going back to that functional integral, we instead want to solve the flow equation (order-by-order in $\hbar$, see Subsection~\ref{sec_pert_theory}), and to do this we need to state the boundary conditions. As $\Lambda \to \Lambda_0$, the Gaussian measure $\total\nu^{\Lambda, \Lambda_0}$ reduces to a $\delta$ measure and we get
\begin{equation}
\label{l_0op_lambda0}
L^{\Lambda_0, \Lambda_0} = L^{\Lambda_0} \eqend{,} \qquad I^{\Lambda_0, \Lambda_0} = 0 \eqend{.}
\end{equation}
However, these boundary conditions at $\Lambda = \Lambda_0$ are not very practical, since a priori one does not know the form of the counterterms in $L^{\Lambda_0}$ that are required in order to make $L^{\Lambda, \Lambda_0}$ finite in the limit as the cutoffs are removed. What one knows in any case is that the counterterms must be local functionals not exceeding a certain dimension. Thus, this aspect of the boundary conditions \emph{can} be imposed at $\Lambda = \Lambda_0$. The remaining information only concerns a finite number of ``relevant'' counterterms, and it is much more convenient to encode it in a boundary condition at $\Lambda = 0$. The proof that this is possible, and the precise form of the boundary conditions are presented in full detail in Subsection~\ref{sec_pert_theory}.

We end this subsection noting that $I^{\Lambda, \Lambda_0}$ does not appear on the right-hand side of the flow equation, and its antifield-independent part is simply a constant times the total volume of spacetime. Since this part does not contribute to correlation functions, it is of no interest to us, and for all other parts we can pass without problems to the infinite volume limit (see Ref.~[\onlinecite{kellerkopper1992}] for details). In summary, given a solution to the flow equation subject to the boundary conditions presented below, we can obtain the physical (unregularised) CACs in the limit $\Lambda_0 \to \infty$ and $\Lambda \to 0$, provided we can show that these limits exist.

\subsection{Insertions of composite operators}

We are also interested in correlation functions with insertions of local composite operators, given by monomials
\begin{equation}
\label{op_def}
\op_A(x) \equiv \left( \prod_{i=1}^m \partial^{w_i} \phi_{K_i}(x) \right) \left( \prod_{j=1}^n \partial^{w^\ddag_j} \phi^\ddag_{L_j}(x) \right) \in \mathcal{F} \eqend{,}
\end{equation}
which are indexed by $A = \{ \vec{K}, \vec{L}^\ddag, \vec{w}, \vec{w}^\ddag \}$, and which have the engineering dimension $[\op_A] = [\vec{K}] + [\vec{L}^\ddag] + \abs{\vec{w}} + \abs{\vec{w}^\ddag}$. Let us denote by $\Delta$ the smallest difference in operator dimensions, \ie, $\Delta$ is the smallest number such that if $[\op_A] > [\op_B]$ we have $[\op_A] \geq [\op_B] + \Delta$ for any two operators $\op_A$ and $\op_B$ (including the basic fields themselves). For technical reasons, we restrict to strictly positive $\Delta > 0$. This is fulfilled both for pure YM theory, where we have $\Delta = 1$ since all basic fields have integer dimensions, and for YM theory including fermionic matter (with engineering dimension $3/2$), where $\Delta = 1/2$ [since derivatives are included in the definition of composite operators~\eqref{op_def}, we always have $\Delta \leq 1$]. The corresponding generating functional of the CACs with insertions is obtained by replacing
\begin{equation}
L^{\Lambda_0} \to L^{\Lambda_0} + \sum_{k=1}^s \left\langle \chi_k, \op_{A_k} + \delta^{\Lambda_0} \op_{A_k} \right\rangle \eqend{,}
\end{equation}
where $\chi_k \in \mathcal{S}(\mathbb{R}^4)$, and where
\begin{equation}
\label{op_ct_def}
\delta^{\Lambda_0} \colon \mathcal{F} \to \mathcal{F} \eqend{,} \quad \delta^{\Lambda_0} \op_{A_k} = \bigo{\hbar}
\end{equation}
is a map representing the counterterms (depending on the cutoff $\Lambda_0$) that are necessary to make the insertion finite, whose precise form is unimportant for our purposes. We then take variational derivatives with respect to the $\chi_k$ to define the generating functional of CACs with insertions:
\begin{splitequation}
\label{l_sop_def}
&L^{\Lambda, \Lambda_0}\left( \bigotimes_{k=1}^s \op_{A_k}(x_k) \right) + I^{\Lambda, \Lambda_0}\left( \bigotimes_{k=1}^s \op_{A_k}(x_k) \right) \\
&\quad\equiv - \hbar \left( \prod_{k=1}^s \frac{\delta}{\delta \chi_k(x_k)} \right) \ln \left[ \nu^{\Lambda, \Lambda_0} \conv \exp\left( - \frac{1}{\hbar} L^{\Lambda_0} - \frac{1}{\hbar} \sum_{k=1}^s \left\langle \chi_k, \op_{A_k} + \delta^{\Lambda_0} \op_{A_k} \right\rangle \right) \right]_{\chi_k = 0} \eqend{.}
\end{splitequation}
For the sake of brevity, we will in the following suppress the coordinate space dependence of $\op_{A_k}$ when no confusion can arise. The corresponding flow equation then reads
\begin{splitequation}
\label{l_sop_flow}
&\partial_\Lambda L^{\Lambda, \Lambda_0}\left( \bigotimes_{k=1}^s \op_{A_k} \right) + \partial_\Lambda I^{\Lambda, \Lambda_0}\left( \bigotimes_{k=1}^s \op_{A_k} \right) = \frac{\hbar}{2} \left\langle \frac{\delta}{\delta \phi_K}, \left( \partial_\Lambda C^{\Lambda, \Lambda_0}_{KL} \right) \ast \frac{\delta}{\delta \phi_L} \right\rangle L^{\Lambda, \Lambda_0}\left( \bigotimes_{k=1}^s \op_{A_k} \right) \\
&\qquad- \left\langle \frac{\delta}{\delta \phi_K} L^{\Lambda, \Lambda_0}, \left( \partial_\Lambda C^{\Lambda, \Lambda_0}_{KL} \right) \ast \frac{\delta}{\delta \phi_L} L^{\Lambda, \Lambda_0}\left( \bigotimes_{k=1}^s \op_{A_k} \right) \right\rangle \\
&\qquad- \hspace{-1em}\sum_{\subline{\alpha \cup \beta = \{1, \ldots, s\} \\ \alpha \neq \emptyset \neq \beta}} \left\langle \frac{\delta}{\delta \phi_K} L^{\Lambda, \Lambda_0}\left( \bigotimes_{k\in\alpha} \op_{A_k} \right) , \left( \partial_\Lambda C^{\Lambda, \Lambda_0}_{KL} \right) \ast \frac{\delta}{\delta \phi_L} L^{\Lambda, \Lambda_0}\left( \bigotimes_{k\in\beta} \op_{A_k} \right) \right\rangle \eqend{,}
\end{splitequation}
which in comparison to the flow equation for functionals without insertions~\eqref{l_0op_flow} has an additional source term depending on functionals with a smaller number of operator insertions. Again, as $\Lambda \to \Lambda_0$ the Gaussian measure reduces to a $\delta$ measure and we have the boundary conditions
\begin{equation}
\label{l_1op_lambda0}
L^{\Lambda_0, \Lambda_0}\left( \op_A \right) = \op_A + \delta^{\Lambda_0} \op_A
\end{equation}
for the insertion of one operator and
\begin{equation}
\label{l_sop_lambda0}
L^{\Lambda_0, \Lambda_0}\left( \bigotimes_{k=1}^s \op_{A_k} \right) = 0
\end{equation}
for the insertion of $s \geq 2$ operators, while always $I^{\Lambda_0, \Lambda_0}\left( \bigotimes_{k=1}^s \op_{A_k} \right) = 0$. As in the case of no insertions, this form of the boundary conditions is actually unpractical, and to make efficient use of the flow equation we need to reformulate them by imposing part of the boundary conditions at $\Lambda = \Lambda_0$. Again, this will be discussed in detail in Subsection~\ref{sec_pert_theory}.

It can be seen that for finite values of the cutoffs, the CACs with insertions are smooth (in fact real-analytic) as functions of the points $x_i$~\cite{kellerkopper1993,kopper1998}. For more than one insertion, this ceases to be true in the limit $\Lambda \to 0$, $\Lambda_0 \to \infty$, and the CACs with insertions become distributions on Schwartz space~\cite{kellerkopper1993,kopper1998}. Our bounds will imply that these distributions are in fact represented by smooth functions as long as all the $x_i$ are distinct: combining Propositions~\ref{thm_lowenstein_1} and~\ref{thm_lsc}, we see that all derivatives are bounded in this case. For later applications, one would also like to smear the CACs with one insertion of a composite operator of dimension $[\op_A] \leq 4$ in $x$ against the test function $1$. While this is legal as long as there is a finite IR cutoff $\Lambda$, it is of course not a priori possible to do this after the limit $\Lambda \to 0$ has been taken, unless one has sufficient information about the decay properties of CACs with insertions as the points $x_i$ tend to infinity. Explicit bounds incorporating such information were provided in previous works~\cite{hollandhollands2015a}, but in this paper we find it easier to derive the separate flow equation
\begin{splitequation}
\label{l_intop_flow}
\partial_\Lambda L^{\Lambda, \Lambda_0}\left( \int\!\op_A \right) + \partial_\Lambda I^{\Lambda, \Lambda_0}\left( \bigotimes_{k=1}^s \op_{A_k} \right) &= \frac{\hbar}{2} \left\langle \frac{\delta}{\delta \phi_K}, \left( \partial_\Lambda C^{\Lambda, \Lambda_0}_{KL} \right) \ast \frac{\delta}{\delta \phi_L} \right\rangle L^{\Lambda, \Lambda_0}\left( \int\!\op_A \right) \\
&\quad- \left\langle \frac{\delta}{\delta \phi_K} L^{\Lambda, \Lambda_0}, \left( \partial_\Lambda C^{\Lambda, \Lambda_0}_{KL} \right) \ast \frac{\delta}{\delta \phi_L} L^{\Lambda, \Lambda_0}\left( \int\!\op_A \right) \right\rangle
\end{splitequation}
for the integrated functionals, denoted by $L^{\Lambda, \Lambda_0}\left( \int\!\op_A \right)$, together with separate boundary conditions, which are then analysed in their own right.

We close this section by noting that for zero external fields/antifields $\phi$ and $\phi^\ddag$, we obtain the connected correlation functions with insertions of composite operators, which are ultimately the objects of interest for us, via:
\begin{equation}
\label{relation_l_opconn}
\left. L^{0,\infty}\left( \bigotimes_{k=1}^s \op_{A_k} \right) \right\rvert_{\phi = \phi^\ddag = 0} = (-\hbar)^{1-s} \expect{ \op_{A_1} \cdots \op_{A_s} }_\text{c} \eqend{.}
\end{equation}
This relation follows directly from the definition~\eqref{l_sop_def}. Since the expansion of the functionals $L^{0,\infty}$ in $\hbar$ starts at order $\hbar^0$, this shows that the expansion of the connected correlation functions of $s$ composite operators starts at order $\hbar^{s-1}$, and vice versa.

Note that in deriving the above equations, we have tacitly assumed that all operators $\op_{A_i}$ are bosonic, since then the functionals with operator insertions are totally symmetric in the insertions. To treat fermionic operators, we introduce for each of them an auxiliary constant fermion, such that all formulas are valid for the product of the auxiliary fermion and the fermionic operator. The correlation functions, including the correct minus signs, are then obtained by taking derivatives w.r.t. these constant fermions in the final results.

\subsection{Perturbation theory and boundary conditions on CACs}
\label{sec_pert_theory}

As already mentioned, in order to prove that the physical limit $\Lambda_0 \to \infty$, $\Lambda \to 0$ exists we have to resort to perturbation theory, expanding the generating functionals in the number of external fields and in a formal power series in $\hbar$. The expansion coefficients (which are also called functionals) are denoted by $\mathcal{L}^{\Lambda, \Lambda_0, l}_{\vec{K} \vec{L}^\ddag}$, and are defined by
\begin{equations}[field_expansion]
\mathcal{L}^{\Lambda, \Lambda_0, l}_{\vec{K} \vec{L}^\ddag}(\vec{x}) &\equiv \left. \frac{1}{l!} \frac{\partial^l}{\partial \hbar^l} \left( \prod_{i=1}^m \frac{\delta_\text{L}}{\delta \Phi_{K_i}(x_i)} \right) \left( \prod_{j=1}^n \frac{\delta_\text{L}}{\delta \phi^\ddag_{L_j}(x_{m+j})} \right) L^{\Lambda, \Lambda_0} \right\rvert_{\phi = \phi^\ddag = 0, \hbar = 0} \eqend{,} \\
\mathcal{L}^{\Lambda, \Lambda_0, l}_{\vec{L}^\ddag}(\vec{x}) &\equiv \left. \frac{1}{l!} \frac{\partial^l}{\partial \hbar^l} \left( \prod_{j=1}^n \frac{\delta_\text{L}}{\delta \phi^\ddag_{L_j}(x_j)} \right) I^{\Lambda, \Lambda_0} \right\rvert_{\phi^\ddag = 0, \hbar = 0} \eqend{,}
\end{equations}
and similarly for the functionals with insertions of composite operators. It is furthermore advantageous to pass to momentum space. Because of translation invariance, overall momentum conservation holds for functionals without insertions of composite operators, and we define their Fourier transform with the momentum-conserving $\delta$ taken out. To reduce notational clutter, we do not introduce new notation, and simply set
\begin{equation}
\label{func_ft_0op}
\mathcal{L}^{\Lambda, \Lambda_0, l}_{\vec{K} \vec{L}^\ddag}(\vec{x}) \equiv \int (2\pi)^4 \delta\left( \sum_{i=1}^{m+n} q_i \right) \mathcal{L}^{\Lambda, \Lambda_0, l}_{\vec{K} \vec{L}^\ddag}(\vec{q}) \prod_{i=1}^{m+n} \mathe^{\mathi q_i x_i} \frac{\total^4 q_i}{(2\pi)^4} \eqend{.}
\end{equation}
Note that even though we include $q_{m+n}$ in the list of arguments, it is not an independent variable, but instead determined as a function of the other $q_i$. For functionals with insertion of composite operators, we do not have overall momentum conservation, and thus define
\begin{equation}
\label{func_ft_sop}
\mathcal{L}^{\Lambda, \Lambda_0, l}_{\vec{K} \vec{L}^\ddag}\left( \bigotimes_{k=1}^s \op_{A_k}; \vec{x} \right) \equiv \int \mathcal{L}^{\Lambda, \Lambda_0, l}_{\vec{K} \vec{L}^\ddag}\left( \bigotimes_{k=1}^s \op_{A_k}; \vec{q} \right) \prod_{i=1}^{m+n} \mathe^{\mathi q_i x_i} \frac{\total^4 q_i}{(2\pi)^4} \eqend{.}
\end{equation}
However, translation invariance still tells us the functionals with insertions have a shift property
\begin{equation}
\label{func_sop_shift}
\mathcal{L}^{\Lambda, \Lambda_0, l}_{\vec{K} \vec{L}^\ddag}\left( \bigotimes_{k=1}^s \op_{A_k}(x_k); \vec{q} \right) = \mathe^{- \mathi y \sum_{i=1}^{m+n} q_i} \mathcal{L}^{\Lambda, \Lambda_0, l}_{\vec{K} \vec{L}^\ddag}\left( \bigotimes_{k=1}^s \op_{A_k}(x_k - y); \vec{q} \right) \eqend{.}
\end{equation}
Functionals which only contain integrated insertions again have overall momentum conservation, and thus we define their Fourier transform also with the momentum-conserving $\delta$ taken out
\begin{equation}
\label{func_ft_iop}
\mathcal{L}^{\Lambda, \Lambda_0, l}_{\vec{K} \vec{L}^\ddag}\left( \bigotimes_{k=1}^s \int\!\op_{A_k}; \vec{x} \right) \equiv \int (2\pi)^4 \delta\left( \sum_{i=1}^{m+n} q_i \right) \mathcal{L}^{\Lambda, \Lambda_0, l}_{\vec{K} \vec{L}^\ddag}\left( \bigotimes_{k=1}^s \int\!\op_{A_k}; \vec{q} \right) \prod_{i=1}^{m+n} \mathe^{\mathi q_i x_i} \frac{\total^4 q_i}{(2\pi)^4} \eqend{.}
\end{equation}

The flow equation~\eqref{l_0op_flow} then gives the hierarchy of perturbative flow equations
\begin{splitequation}
\label{l_0op_flow_hierarchy}
\partial_\Lambda \mathcal{L}^{\Lambda, \Lambda_0, l}_{\vec{K} \vec{L}^\ddag}(\vec{q}) &= \frac{c}{2} \int \left( \partial_\Lambda C^{\Lambda, \Lambda_0}_{MN}(-p) \right) \mathcal{L}^{\Lambda, \Lambda_0, l-1}_{MN \vec{K} \vec{L}^\ddag}(p, -p, \vec{q}) \frac{\total^4 p}{(2\pi)^4} \\
&\quad- \sum_{\subline{\sigma \cup \tau = \{1, \ldots, m\} \\ \rho \cup \varsigma = \{1, \ldots, n\}}} \sum_{l'=0}^l \frac{c_{\sigma\tau\rho\varsigma}}{2} \mathcal{L}^{\Lambda, \Lambda_0, l'}_{\vec{K}_\sigma \vec{L}_\rho^\ddag M}(\vec{q}_\sigma,\vec{q}_\rho,-k) \left( \partial_\Lambda C^{\Lambda, \Lambda_0}_{MN}(k) \right) \mathcal{L}^{\Lambda, \Lambda_0, l-l'}_{N \vec{K}_\tau \vec{L}_\varsigma^\ddag}(k,\vec{q}_\tau,\vec{q}_\varsigma)
\end{splitequation}
with
\begin{equation}
\label{k_def}
k \equiv \sum_{i \in \sigma \cup \rho} q_i = - \sum_{i \in \tau \cup \varsigma} q_i \eqend{,}
\end{equation}
where $c$ and $c_{\sigma\tau\rho\varsigma}$ are some constants stemming from the anticommutating nature of fermionic fields (and antifields). This hierarchy is now suited to inductive proofs: the functional in the first line on the right-hand side and the functionals in the second line have a lower order in $\hbar$ if $0 < l' < l$. If $l' = 0$ (or $l' = l$), most functionals in the second line have a smaller number of external fields and antifields, except when $l' = 0$ and the first functional has only one or two external fields or antifields, or when $l' = l$ and the second functional has only one or two external fields or antifields. However, these functionals vanish by definition since we removed the free part from $L^{\Lambda, \Lambda_0}$, which are exactly the terms linear and quadratic in the external fields and antifields at order $\hbar^0$. We can thus ascend in $m+n+2l$, where $m+n$ is the number of external fields and antifields, and for fixed $m+n+2l$, ascend in $l$.

For functionals with insertions of composite operators, the flow equation~\eqref{l_sop_flow} gives the hierarchy
\begin{splitequation}
\label{l_sop_flow_hierarchy}
&\partial_\Lambda \mathcal{L}^{\Lambda, \Lambda_0, l}_{\vec{K} \vec{L}^\ddag}\left( \bigotimes_{k=1}^s \op_{A_k}; \vec{q} \right) = \frac{c}{2} \int \left( \partial_\Lambda C^{\Lambda, \Lambda_0}_{MN}(-p) \right) \mathcal{L}^{\Lambda, \Lambda_0, l-1}_{MN \vec{K} \vec{L}^\ddag}\left( \bigotimes_{k=1}^s \op_{A_k}; p,-p,\vec{q} \right) \frac{\total^4 p}{(2\pi)^4} \\
&\qquad- \sum_{\subline{\sigma \cup \tau = \{1, \ldots, m\} \\ \rho \cup \varsigma = \{1, \ldots, n\} }} \sum_{l'=0}^l c_{\sigma\tau\rho\varsigma} \mathcal{L}^{\Lambda, \Lambda_0, l'}_{\vec{K}_\sigma \vec{L}_\rho^\ddag M}(\vec{q}_\sigma,\vec{q}_\rho,-k) \left( \partial_\Lambda C^{\Lambda, \Lambda_0}_{MN}(k) \right) \mathcal{L}^{\Lambda, \Lambda_0, l-l'}_{N \vec{K}_\tau \vec{L}_\varsigma^\ddag}\left( \bigotimes_{k=1}^s \op_{A_k}; k,\vec{q}_\tau,\vec{q}_\varsigma \right) \\
&\qquad- \sum_{\subline{\alpha\cup\beta = \{1, \ldots, s\} \\ \alpha \neq \emptyset \neq \beta}} \sum_{\subline{\sigma \cup \tau = \{1, \ldots, m\} \\ \rho \cup \varsigma = \{1, \ldots, n\} }} \sum_{l'=0}^l c_{\sigma\tau\rho\varsigma} \int \mathcal{L}^{\Lambda, \Lambda_0, l'}_{\vec{K}_\sigma \vec{L}_\rho^\ddag M}\left( \bigotimes_{k\in\alpha} \op_{A_k}; \vec{q}_\sigma,\vec{q}_\rho,p \right) \left( \partial_\Lambda C^{\Lambda, \Lambda_0}_{MN}(-p) \right) \\
&\hspace{18em}\times \mathcal{L}^{\Lambda, \Lambda_0, l-l'}_{N \vec{K}_\tau \vec{L}_\varsigma^\ddag}\left( \bigotimes_{k\in\beta} \op_{A_k};-p,\vec{q}_\tau,\vec{q}_\varsigma \right) \frac{\total^4 p}{(2\pi)^4} \eqend{.}
\end{splitequation}
The term in the second line depends on the functionals without insertions, such that these have to be bounded first. The consistency of the induction scheme is again assured by that fact that the functionals without insertions vanish at lowest loop order $l=0$ for one or two external fields, such that the terms on the right-hand side are either of lower loop order or, for the same loop order, have a lower number of external fields. The source term in the last line appears for $s \geq 2$, and depends on functionals with a lower number of insertions, such that the induction has to go up in the number of insertions. Furthermore, for functionals with an insertion of an integrated operator, the flow equation~\eqref{l_intop_flow} gives the hierarchy
\begin{splitequation}
\label{l_intop_flow_hierarchy}
&\partial_\Lambda \mathcal{L}^{\Lambda, \Lambda_0, l}_{\vec{K} \vec{L}^\ddag}\left( \int\!\op_A; \vec{q} \right) = \frac{c}{2} \int \left( \partial_\Lambda C^{\Lambda, \Lambda_0}_{MN}(-p) \right) \mathcal{L}^{\Lambda, \Lambda_0, l-1}_{MN \vec{K} \vec{L}^\ddag}\left( \int\!\op_A; p,-p,\vec{q} \right) \frac{\total^4 p}{(2\pi)^4} \\
&\qquad- \sum_{\subline{\sigma \cup \tau = \{1, \ldots, m\} \\ \rho \cup \varsigma = \{1, \ldots, n\} }} \sum_{l'=0}^l c_{\sigma\tau\rho\varsigma} \mathcal{L}^{\Lambda, \Lambda_0, l'}_{\vec{K}_\sigma \vec{L}_\rho^\ddag M}(\vec{q}_\sigma,\vec{q}_\rho,-k) \left( \partial_\Lambda C^{\Lambda, \Lambda_0}_{MN}(k) \right) \mathcal{L}^{\Lambda, \Lambda_0, l-l'}_{N \vec{K}_\tau \vec{L}_\varsigma^\ddag}\left( \int\!\op_A; k,\vec{q}_\tau,\vec{q}_\varsigma \right) \eqend{.}
\end{splitequation}

To close the induction we will also need flow equations for momentum derivatives of functionals, which, since the regulator is smooth, can be taken without problems and then distributed over the terms on the right-hand side. For the terms which are quadratic in functionals and where momentum is conserved, we have to view the last momentum $q_{m+n}$ as a function of the first $m+n-1$ momenta $q_i$ and of $k$, such that momentum derivatives can also act on $k$ (and thus on the regulated covariance when it depends on $k$). We refrain from writing out the corresponding flow equations in detail.

As we have already mentioned, instead of fixing all boundary conditions at $\Lambda = \Lambda_0$, it is much more convenient to give some boundary conditions at $\Lambda = 0$ and some at $\Lambda = \mu$, where $\mu$ is some renormalisation scale. This is possible because there is a one-to-one correspondence between conditions at $\Lambda = \Lambda_0$ and conditions given for some other value of $\Lambda$, namely
\begin{equation}
\mathcal{L}^{\Lambda_0, \Lambda_0, l}_{\vec{K} \vec{L}^\ddag}(\vec{q}) = \mathcal{L}^{\Lambda, \Lambda_0, l}_{\vec{K} \vec{L}^\ddag}(\vec{q}) + \int_{\Lambda}^{\Lambda_0} \partial_\lambda \mathcal{L}^{\lambda, \Lambda_0, l}_{\vec{K} \vec{L}^\ddag}(\vec{q}) \total \lambda \eqend{,}
\end{equation}
where the $\lambda$ derivative is given by the right-hand side of the flow equation~\eqref{l_0op_flow_hierarchy}, which is already determined previously in the induction, and thus fixed. The boundary conditions differ depending on the type of functional (with or without insertion) and whether it is irrelevant, marginal or relevant, shown in Table~\ref{table_boundary}. While most of the boundary conditions vanish (for relevant functionals, vanishing boundary conditions are necessary to obtain IR-finite results for non-exceptional momenta~\cite{kopper1998}), taking \emph{all} of them to vanish would result in a trivial non-interacting theory. (Typically, one takes non-zero constants at order $\hbar^0$ (\ie, $l = 0$) for all functionals which correspond to an interaction in the original Lagrangian, and vanishing boundary conditions for $l > 0$.) Furthermore, if one wants to preserve invariance of the correlation functions under symmetries which are not explicitly broken by the regulator~\eqref{reg_covariance}, such as manifest $\mathrm{E}(4)$ invariance in gauge theories, also the boundary conditions must be chosen to be invariant under this symmetry.
\begin{table}
\begin{center}
\begin{tabular}{lcc}
\hline
functional & type & \hspace{-2em} boundary condition \\
\hline
$\partial^w \mathcal{L}^{0, \Lambda_0, l}_{\vec{K} \vec{L}^\ddag}(\vec{0})$ & relevant: $[\vec{K}] + [\vec{L}^\ddag] + \abs{w} < 4$ & $0$ \\[0.5em]
$\partial^w \mathcal{L}^{\mu, \Lambda_0, l}_{\vec{K} \vec{L}^\ddag}(\vec{0})$ & marginal: $[\vec{K}] + [\vec{L}^\ddag] + \abs{w} = 4$ & arbitrary \\[0.5em]
$\partial^w \mathcal{L}^{\Lambda_0, \Lambda_0, l}_{\vec{K} \vec{L}^\ddag}(\vec{q})$ & irrelevant: $[\vec{K}] + [\vec{L}^\ddag] + \abs{w} > 4$ & $0$ \\[0.2em]
\hline
$\partial^w \mathcal{L}^{\mu, \Lambda_0, l}_{\vec{K} \vec{L}^\ddag}\left( \op_A; \vec{0} \right)$ & relevant/marginal: $[\vec{K}] + [\vec{L}^\ddag] + \abs{w} \leq [\op_A]$ & arbitrary \\[0.5em]
$\partial^w \mathcal{L}^{\Lambda_0, \Lambda_0, l}_{\vec{K} \vec{L}^\ddag}\left( \op_A; \vec{q} \right)$ & irrelevant: $[\vec{K}] + [\vec{L}^\ddag] + \abs{w} > [\op_A]$ & $0$ \\[0.2em]
\hline
$\partial^w \mathcal{L}^{0, \Lambda_0, l}_{\vec{K} \vec{L}^\ddag}\left( \int\!\op_A; \vec{0} \right)$ & relevant: $[\vec{K}] + [\vec{L}^\ddag] + \abs{w} < 4$ & $0$ \\[0.5em]
$\partial^w \mathcal{L}^{\mu, \Lambda_0, l}_{\vec{K} \vec{L}^\ddag}\left( \int\!\op_A; \vec{0} \right)$ & relevant/marginal: $4 \leq [\vec{K}] + [\vec{L}^\ddag] + \abs{w} \leq [\op_A]$ & arbitrary \\[0.5em]
$\partial^w \mathcal{L}^{\Lambda_0, \Lambda_0, l}_{\vec{K} \vec{L}^\ddag}\left( \int\!\op_A; \vec{q} \right)$ & irrelevant: $[\vec{K}] + [\vec{L}^\ddag] + \abs{w} > [\op_A]$ & $0$ \\[0.2em]
\hline
$\mathcal{L}^{\Lambda_0, \Lambda_0, l}_{\vec{K} \vec{L}^\ddag}\left( \bigotimes_{k=1}^s \op_{A_k}; \vec{q} \right)$ & always irrelevant & $0$ \\[0.2em]
\hline
\end{tabular}
\end{center}
\caption{Boundary conditions for the flow equation hierarchy. The ``arbitrary'' conditions should not all be vanishing (since otherwise the corresponding theory is trivial), and naturally must be invariant under the action of all symmetries which one wants to preserve manifestly, \eg, in our case they should be $\mathrm{E}(4)$-invariant.}
\label{table_boundary}
\end{table}
Furthermore, one could even give conditions for marginal functionals at $\Lambda = 0$, but then one has to restrict to non-exceptional momenta since the functionals are otherwise IR divergent. One first determines the value of a marginal functional at $\Lambda = \mu$ and non-exceptional momenta using the Taylor formula with integral remainder from the boundary conditions at $\Lambda = \mu$ and zero momenta. Since only irrelevant functionals appear in the Taylor formula which have their boundary conditions fixed, this is a one-to-one correspondence. In the second step, one then integrates the flow equation downwards to determine the value of the marginal functional at $\Lambda = 0$ and non-exceptional momenta from its value at $\Lambda = \mu$. Since all functionals appearing on the right-hand side of the flow equation have already been determined and fixed in the induction, the correspondence is again one-to-one. Thus although the exact correspondence will be in general very complicated, it is unique.

\section{Trees}
\label{sec_trees}

The bounds on functionals with and without insertions are specified in terms of fully reduced weighted trees, which we define in the following, and for which we also derive some properties that will be important later. The trees are motivated by the fact that in perturbation theory, the scaling behaviour of correlation functions including loop corrections is only modified logarithmically with respect to the tree level, and thus our trees basically represent tree level Feynman graphs. We stress, however, that this analogy must not be taken too literally; the trees and the valence of vertices are independent of the detailed form of the $n$-point interactions in the theory, and in particular do not depend in any way on spin, colour or Lorentz indices. Furthermore, internal lines (which would correspond to propagators in a tree-level Feynman diagram) and vertices are always assigned integer dimensions independent of the engineering dimension of the basic fields or the dimension of $n$-point interactions in the theory; the dimension of the basic fields only appears in the factors associated to external vertices.

\subsection{Weighted Trees}

Trees are connected graphs without loops; \ie, we define
\begin{definition}
A graph $G \equiv (V,L)$ is a finite set of vertices $V$ and lines (edges) $L$, which are unordered pairs of two elements of $V$, i.e., $e \in L \Leftrightarrow e = \{v,w\} \colon v,w \in V$ (for which we say that the line $e$ connects the vertices $v$ and $w$). A tree $T$ is a graph where
\begin{itemize}
\item Every vertex is connected by a line to some other vertex: either $\abs{V} = 1$ (there is only one vertex), or for all $v \in V$ there exists $e \in L$ such that $e = \{v,w\}$ for some $w \in V$,
\item There are no loops (the graph is acyclic): there exists no subset of $L$ of the form $\{ \{v_1,v_2\}$, $\{v_2,v_3\}$, \ldots, $\{ v_k,v_1 \} \}$ for all $k \in \mathbb{N}$ (including self-loops or ``tadpoles'' with $k=1$).
\end{itemize}
\end{definition}
We then further define
\begin{definition}
A weighted tree $T$ of order $(m+n,r)$ has $m+n$ external and $r$ internal vertices, and a tree $T^*$ of order $(m+n,r)$ has $m+n$ external, $r$ internal and one special vertex. External vertices have valency $1$ (\ie, exactly one line is incident to them), and may only be connected by a line to internal or special vertices. Internal vertices have valency between $1$ and $4$, and special vertices may have any valency. We require that $m+n \geq 1$ and $r \geq 1$ for trees $T$ (\ie, at least one external and one internal vertex), but do not impose further conditions on trees $T^*$ (\ie, a tree may consist only of the special vertex). The external vertices are numbered from $1$ to $m+n$, and a momentum $q$ is assigned to each of them. The momenta assigned to the lines are determined by imposing momentum conservation at each vertex (which is allowed because later on we will have $q_{m+n} = - \sum_{i=1}^{m+n-1} q_i$), except for the special vertex. Afterwards, internal vertices are assigned the momentum with highest absolute value among the momenta assigned to all lines incident to that vertex. Furthermore, we associate to each external vertex $v_e$ an index $K_e$ (or alternatively $L^\ddag_e$) and a dimension $[v_e] = [\phi_{K_e}] \in [1,3]$ (or $[v_e] = [\phi^\ddag_{L_e}] \in [1,3]$), and an overall derivative multiindex $\vec{w}$ to the tree.
\end{definition}
To reduce notational clutter, we will use $T$ also for a generic tree, when it is clear from the context which tree is meant (\ie, with or without a special vertex), or if a formula applies to all trees. To each tree, we assign a weight factor which appears in the bounds.
\begin{definition}
The weight factor $\mathsf{G}^{T,\vec{w}}_{\vec{K} \vec{L}^\ddag; [v_p]}(\vec{q}; \mu, \Lambda)$ associated to a tree $T$ is given by multiplying the weight factors assigned to each vertex and line of $T$ given in Table~\ref{table_weights}, the particular weight factor given by
\begin{equation}
\label{particular_weight}
\mathsf{G}^p(\vec{q}; \Lambda) = \sup(\abs{\vec{q}}, \mu, \Lambda)^{[v_p]}
\end{equation}
for a dimension $[v_p] \in \mathbb{R}$, and the derivative weight factor $\mathsf{G}^\vec{w}(\vec{q}; \Lambda)$, given by
\begin{equation}
\label{gw_def}
\mathsf{G}^\vec{w}(\vec{q}; \Lambda) \equiv \prod_{i=1}^{m+n} \begin{cases} \sup(\eta_{q_i}(\vec{q}), \Lambda)^{-\abs{w_i}} & \text{for trees $T$} \\ \sup(\bar{\eta}_{q_i}(\vec{q}), \Lambda)^{-\abs{w_i}} & \text{for trees $T^*$} \eqend{.} \end{cases}
\end{equation}
Note that for the trees $T$ with momentum conservation we require $w_{m+n} = 0$ in order to be consistent with the definition of $\eta_{q_i}$~\eqref{eta_i_def}. Since in this case the last momentum is determined by overall momentum conservation $q_{m+n} = - \sum_{i=1}^{m+n-1} q_i$, derivatives with respect to $q_n$ can be converted into derivatives with respect to the other $q_i$, and no problem arises.
\end{definition}
\begin{table}
\begin{center}
\begin{tabular}{cll}
\hline
\multicolumn{2}{c}{Component} & \multicolumn{1}{c}{Associated weight} \\
\hline
\raisebox{.3em}{\includegraphics{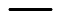}} & line $l$ & $\mathsf{G}^l(q; \mu, \Lambda) = \sup(\abs{q},\Lambda)^{-2}$ \\
\includegraphics{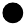} & external vertex $v_e$ & $\mathsf{G}^{v_e}(q; \mu, \Lambda) = \sup(\abs{q},\Lambda)^{3-[v_e]}$ \\
\includegraphics{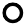} & internal vertex $v_i$ of valence $k$ & $\mathsf{G}^{v_i}(q; \mu, \Lambda) = \sup(\abs{q},\Lambda)^{4-k}$ \\
\raisebox{-.25\height}{\includegraphics{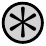}} & special vertex $v_s$ of valence $k$ & $\mathsf{G}^{v_s}(q; \mu, \Lambda) = \sup(\mu,\Lambda)^{-k}$ \\
\hline
\end{tabular}
\end{center}
\caption{Weights $\mathsf{G}$ assigned to components of a tree of order $(n,r)$. $q$ always refers to the momentum associated to the component.}
\label{table_weights}
\end{table}
The tree itself is also assigned a dimension, given by the sum of the exponents of all weight factors, which gives
\begin{definition}
To each tree $T$ we associate an overall dimension $[T]$, given by
\begin{equation}
[T] \equiv \sum_{v_e} ( 3 - [v_e] ) + \sum_{v_i} ( 4 - k_i ) + [v_p] - k_s - 2 N_l - \abs{\vec{w}} \eqend{,}
\end{equation}
where the sums run over all external vertices $v_e$ and all internal vertices $v_i$ with $k_i$ the valency of $v_i$, and $N_l$ is the number of lines and $k_s$ the valency of the special vertex (or $k_s = 0$ if no special vertex exists).
\end{definition}
This dimension measures the scaling of the tree weight, \ie, we have
\begin{equation}
\label{tree_scaling}
\lim_{\Lambda \to \infty} \mathsf{G}^{T,\vec{w}}_{\vec{K} \vec{L}^\ddag; [v_p]}(\vec{q}; \mu, \Lambda) \Lambda^{-[T]} = 1 \eqend{.}
\end{equation}
It is possible to obtain a simpler expression for the tree dimension, given by
\begin{lemma}
The tree dimension $[T]$ can be expressed as
\begin{equations}[t_dim_def]
[T] &= 4 + [v_p] - \sum_{v_e} [v_e] - \abs{\vec{w}} = 4 + [v_p] - [\vec{K}] - [\vec{L}^\ddag] - \abs{\vec{w}} \\
[T^*] &= [v_p] - \sum_{v_e} [v_e] - \abs{\vec{w}} = [v_p] - [\vec{K}] - [\vec{L}^\ddag] - \abs{\vec{w}} \eqend{.}
\end{equations}
\end{lemma}
\begin{proof}
A basic result from graph theory states that in a connected graph without loops, the number of vertices $N_v$ is one bigger than the number of lines, $N_v = N_l + 1$. This is easy to see: the simplest tree has one vertex and no lines, and for each subsequent line a new vertex must be added. Let us denote by $N_i$ the number of internal vertices and by $N_e$ the number of external vertices. The sum over $v_i$, with $k_i$ the valency of the internal vertex $v_i$, then counts all lines twice, except for the ones which are incident to the special vertex or external vertices, such that
\begin{equation}
\sum_{v_i} ( 4 - k_i ) = 4 N_i - \left( 2 N_l - k_s - N_e \right)
\end{equation}
and thus
\begin{equation}
[T] = 4 ( N_i + N_e - N_v + 1 ) + [v_p] - \sum_{v_e} [v_e] - \abs{\vec{w}} \eqend{.}
\end{equation}
If the tree has a special vertex, we have $N_v = N_i + N_e + 1$, while otherwise $N_v = N_i + N_e$. In total, we thus obtain the Lemma.
\end{proof}

A tree is said to be relevant if $[T] > 0$, marginal if $[T] = 0$ and irrelevant if $[T] < 0$, and the bounds for irrelevant, marginal and relevant functionals (with or without operator insertions) are given in terms of irrelevant, marginal and relevant trees, respectively. Obviously if all dimensions of external, particular and special vertices are a multiple of $\Delta$ (which will be the case later on), $[T]$ is a multiple of $\Delta$, such that we even have the stronger bounds $[T] \geq \Delta$ for relevant and $[T] \leq - \Delta$ for irrelevant trees.

To illustrate these definitions, an example of a tree $T$ of order $(3,5,0)$ is given in Figure~\ref{fig_tree_example} together with the associated weight $\mathsf{G}$.
\begin{figure}
\includegraphics[scale=0.9]{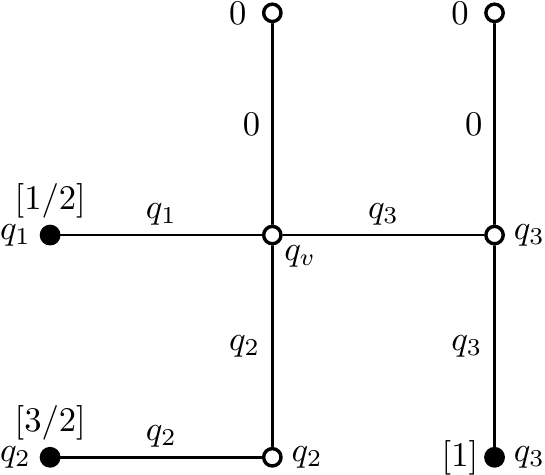}
\caption{Example of a tree $T$ of order $(3,5)$, displaying the momentum assignment and the dimensions of the external vertices. Since this tree does not contain any special vertex, we have overall momentum conservation such that $q_1 + q_2 + q_3 = 0$, and thus momentum is conserved also at the central vertex $v$. Its associated momentum $q_v$ is the one from $\{q_1, q_2, q_3\}$ with the largest absolute value. For $\vec{w} = ((1,0,0,0),(3,0,0,0),(0,0,0,0))$, the total weight of the tree reads $\mathsf{G}^{T,\vec{w}}(q_1, q_2, q_3; \mu, \Lambda) = \sqrt{\sup(\abs{q_1},\Lambda)}\, \Lambda^2 / \big[ \sqrt{\sup(\abs{q_2},\Lambda)}\, \sup(\abs{q_3},\Lambda) \sup(\inf(\abs{q_1},\abs{q_3}), \Lambda) \sup(\inf(\abs{q_2},\abs{q_3}), \Lambda)^3 \big]$, and we can read off that $[T] = -3 = 4 - (1/2+3/2+1) - (1+3+0)$.}
\label{fig_tree_example}
\end{figure}

\subsection{Reduction, Fusion and Amputation}
\label{sec_tree_fuse}

\paragraph*{Reduction.} Starting from some tree $T$, we obtain a reduced tree $T'$ by performing one of the following reduction operations:
\begin{itemize}
\item Remove an internal vertex of valence $2$, and fuse the incident lines into one line
\item Remove an internal vertex of valence $1$ and the incident line if it is connected to an internal vertex
\item Remove an internal vertex of valence $1$ and the incident line if it is connected to a particular or special vertex
\end{itemize}
Reduction operations may change the associated weight, but can only increase it as detailed in Table~\ref{table_reduction}; furthermore the order of a tree can only decrease (\ie, if the order of $T$ is $(n,r)$, the order of $T'$ is $(n,r')$ with $r' \leq r$). However, the dimension of the tree is unchanged under reduction, since it depends only on the dimensions of the external and particular vertices and on the particular weight factor, which are unchanged under reduction.
\begin{table}
\begin{center}
\begin{tabular}{cl}
\hline
\multicolumn{1}{c}{Reduction operation} & \multicolumn{1}{c}{Change in weights} \\
\hline
\raisebox{-2.3em}{\includegraphics{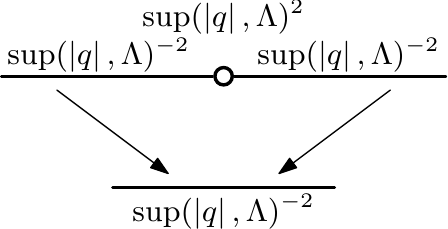}} & $\mathsf{G}^{T,w}_{\vec{K} \vec{L}^\ddag; [v_p]}(\vec{q}; \mu, \Lambda) = \mathsf{G}^{T',w}_{\vec{K} \vec{L}^\ddag; [v_p]}(\vec{q}; \mu, \Lambda)$ \\[3em]
\raisebox{-2.3em}{\includegraphics{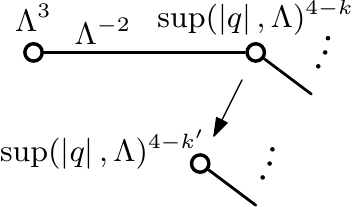}} & \begin{tabular}[c]{@{}c@{}}$\mathsf{G}^{T,w}_{\vec{K} \vec{L}^\ddag; [v_p]}(\vec{q}; \mu, \Lambda) = \dfrac{\Lambda}{\sup(\abs{q},\Lambda)} \mathsf{G}^{T',w}_{\vec{K} \vec{L}^\ddag; [v_p]}(\vec{q}; \mu, \Lambda)$\\for $k' = k-1$, since the valence is smaller by one\end{tabular} \\[3em]
\raisebox{-2.3em}{\includegraphics{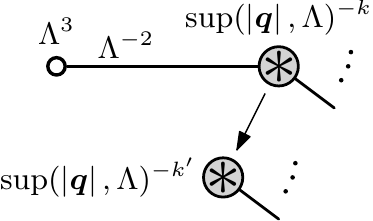}} & \begin{tabular}[c]{@{}c@{}}$\mathsf{G}^{T,w}_{\vec{K} \vec{L}^\ddag; [v_p]}(\vec{q}; \mu, \Lambda) = \dfrac{\Lambda}{\sup(\mu,\Lambda)} \mathsf{G}^{T',w}_{\vec{K} \vec{L}^\ddag; [v_p]}(\vec{q}; \mu, \Lambda)$\\for $k' = k-1$, since the valence is smaller by one\end{tabular} \\[0.1em]
\hline
\end{tabular}
\end{center}
\caption{Reduction operations that can be performed on a tree $T$ to obtain a reduced tree $T'$, and the corresponding change in weights.}
\label{table_reduction}
\end{table}
We thus obtain
\begin{definition}
A fully reduced (weighted) tree $T$ is a weighted tree $T$ where all possible reduction operations have been performed. It is easy to see that a fully reduced tree of order $(n,r)$ with $n \geq 3$ only contains internal vertices of valence $3$ and $4$.
\end{definition}

\paragraph*{Fusion.} A tree $T_1$ of order $(n_1,r_1)$ and another tree $T_2$ of order $(n_2,r_2)$ can also be fused in two different ways. If both trees have a special vertex, fusing is done by merging the two special vertices into one, creating a tree $T^*$ of order $(n_1+n_2, r_1+r_2)$. If only one or neither of both trees has a special vertex, fusion is only possible if $T_1$ has an external vertex $v_1$ with momentum $-k$ and $T_2$ has an external vertex $v_2$ with momentum $k$. If $T_1$ has no special vertex, $v_1$ must be the last external vertex of $T_1$, and if $T_2$ has no special vertex, $v_2$ must be the first external vertex of $T_2$ (remember that the external vertices are numbered). The trees are then fused by removing both external vertices and combining the incident lines into one line, which creates a tree $T$ of order $(n_1+n_2-2,r_1+r_2)$. For the change in weights, we obtain
\begin{lemma}
Fusing two trees $T_1^*$ and $T_2^*$ with special vertices into a tree $T^*$, the change in weight factors can be estimated by
\begin{equation}
\label{gw_fused_1_est}
\mathsf{G}^{T_1^*,\vec{w}_1}_{\vec{K}_1 \vec{L}^\ddag_1; [v_{p1}]}(\vec{q}_1; \mu, \Lambda) \mathsf{G}^{T_2^*,\vec{w}_2}_{\vec{K}_2 \vec{L}^\ddag_2; [v_{p2}]}(\vec{q}_2; \mu, \Lambda) \leq \mathsf{G}^{T^*,\vec{w}_1+\vec{w}_2}_{\vec{K}_1 \vec{L}^\ddag_1 \vec{K}_2 \vec{L}^\ddag_2; [v_{p1}]+[v_{p2}]}(\vec{q}_1, \vec{q}_2; \mu, \Lambda) \eqend{.}
\end{equation}
Fusing two trees $T_1$ and $T_2$ where at most one has a special vertex, under the assumption that the derivative weight factor acting on the momentum $k$ (of the line of $T_2$ which is combined in the fusion) was obtained by deriving the functional that is bounded by $T_2$ w.r.t. some specific $q_j$ (since $k = \sum_{i=1}^{n_1} q_i$), the change in weight factors is given by
\begin{equation}
\label{gw_fused_2_est}
\mathsf{G}^{T_1,\vec{w}_1}_{\vec{K}_1 \vec{L}^\ddag_1 M; [v_{p1}]}(\vec{q}_1, - k; \mu, \Lambda) \mathsf{G}^{T_2,\vec{w}_2}_{N \vec{K}_2 \vec{L}^\ddag_2; [v_{p2}]}(k, \vec{q}_2; \mu, \Lambda) \leq \frac{\mathsf{G}^{T,\vec{w}_1+\vec{w}_2}_{\vec{K}_1 \vec{L}^\ddag_1 \vec{K}_2 \vec{L}^\ddag_2; [v_{p1}]+[v_{p2}]}(\vec{q}_1, \vec{q}_2; \mu, \Lambda)}{\sup(\abs{k},\Lambda)^{[v_M]+[v_N]-4}} \eqend{.}
\end{equation}
\end{lemma}
\begin{proof}
For the fusion of two trees with special vertices, the change in tree weights can be inferred from Table~\ref{table_weights}, and is given by
\begin{splitequation}
\label{gw_fused_1a}
&\mathsf{G}^{T,\vec{w}_1+\vec{w}_2}_{\vec{K}_1 \vec{L}^\ddag_1 \vec{K}_2 \vec{L}^\ddag_2; [v_{p1}]+[v_{p2}]}(\vec{q}_1, \vec{q}_2; \mu, \Lambda) = \mathsf{G}^{T_1,\vec{w}_1}_{\vec{K}_1 \vec{L}^\ddag_1; [v_{p1}]}(\vec{q}_1; \mu, \Lambda) \mathsf{G}^{T_2,\vec{w}_2}_{\vec{K}_2 \vec{L}^\ddag_2; [v_{p2}]}(\vec{q}_2; \mu, \Lambda) \\
&\qquad\times \frac{\sup(\abs{\vec{q}_1, \vec{q}_2}, \mu, \Lambda)^{[v_{p1}]+[v_{p2}]}}{\sup(\abs{\vec{q}_1}, \mu, \Lambda)^{[v_{p1}]} \sup(\abs{\vec{q}_2}, \mu, \Lambda)^{[v_{p2}]}} \frac{\mathsf{G}^{\vec{w}_1+\vec{w}_2}(\vec{q}_1, \vec{q}_2; \Lambda)}{\mathsf{G}^{\vec{w}_1}(\vec{q}_1; \Lambda) \mathsf{G}^{\vec{w}_2}(\vec{q}_2; \Lambda)} \eqend{.}
\end{splitequation}
Since $\abs{\vec{q}_1, \vec{q}_2} \geq \abs{\vec{q}_1}$ and $\abs{\vec{q}_1, \vec{q}_2} \geq \abs{\vec{q}_2}$, we can estimate this by
\begin{splitequation}
\label{gw_fused_1}
\mathsf{G}^{T,\vec{w}_1+\vec{w}_2}_{\vec{K}_1 \vec{L}^\ddag_1 \vec{K}_2 \vec{L}^\ddag_2; [v_{p1}]+[v_{p2}]}(\vec{q}_1, \vec{q}_2; \mu, \Lambda) &\geq \mathsf{G}^{T_1,\vec{w}_1}_{\vec{K}_1 \vec{L}^\ddag_1; [v_{p1}]}(\vec{q}_1; \mu, \Lambda) \mathsf{G}^{T_2,\vec{w}_2}_{\vec{K}_2 \vec{L}^\ddag_2; [v_{p2}]}(\vec{q}_2; \mu, \Lambda) \\
&\quad\times \frac{\mathsf{G}^{\vec{w}_1+\vec{w}_2}(\vec{q}_1, \vec{q}_2; \Lambda)}{\mathsf{G}^{\vec{w}_1}(\vec{q}_1; \Lambda) \mathsf{G}^{\vec{w}_2}(\vec{q}_2; \Lambda)} \eqend{.}
\end{splitequation}
For the fusion of two trees where at most one has a special vertex, the change in weights is given by
\begin{splitequation}
\label{gw_fused_2}
\mathsf{G}^{T,\vec{w}_1+\vec{w}_2}_{\vec{K}_1 \vec{L}^\ddag_1 \vec{K}_2 \vec{L}^\ddag_2; [v_{p1}]+[v_{p2}]}(\vec{q}_1, \vec{q}_2; \mu, \Lambda) &\geq \mathsf{G}^{T_1,\vec{w}_1}_{\vec{K}_1 \vec{L}^\ddag_1 M; [v_{p1}]}(\vec{q}_1, - k; \mu, \Lambda) \mathsf{G}^{T_2,\vec{w}_2}_{N \vec{K}_2 \vec{L}^\ddag_2; [v_{p2}]}(k, \vec{q}_2; \mu, \Lambda) \\
&\quad\times \sup(\abs{k},\Lambda)^{[v_M]+[v_N]-4} \frac{\mathsf{G}^{\vec{w}_1+\vec{w}_2}(\vec{q}_1, \vec{q}_2; \Lambda)}{\mathsf{G}^{\vec{w}_1}(\vec{q}_1, -k; \Lambda) \mathsf{G}^{\vec{w}_2}(k, \vec{q}_2; \Lambda)} \eqend{.}
\end{splitequation}
If both trees do not contain special vertices, overall momentum conservation tells us that $k = \sum_{i=1}^{n_1} q_i$ and $p_{n_2} = - k - \sum_{i=1}^{n_2-1} p_i$. According to the definition of $\eta_{q_i}$~\eqref{eta_i_def} we then have
\begin{equations}
\eta_{q_i}(\vec{q}_1, -k) &\geq \eta_{q_i}(\vec{q}_1, \vec{q}_2) \eqend{,} \\
\eta_{q_i}(k, \vec{q}_2) &\geq \eta_{q_i}(\vec{q}_1, \vec{q}_2) \eqend{,} \\
\eta_k(k, \vec{q}_2) &\geq \eta_q(\vec{q}_1, \vec{q}_2) \quad \text{ for any } q \in \vec{q}_1 \eqend{.}
\end{equations}
The assumption we made is that the derivative weight factor for the momentum $k$ in the second tree $T_2$ was obtained by deriving the functional that is bounded by $T_2$ w.r.t. some specific $q_j$ (since $k = \sum_{i=1}^{n_1} q_i$), and then we choose this $q_j$ in the last inequality (the extension to more than one derivative is straightforward). From the definition of $\mathsf{G}^\vec{w}$~\eqref{gw_def} and the fact that $w_i^\alpha \geq 0$ we thus obtain
\begin{equations}
\mathsf{G}^{\vec{w}_1}(\vec{q}_1, -k; \Lambda) &\leq \prod_{q_i \in \vec{q}_1} \sup(\eta_{q_i}(\vec{q}_1, \vec{q}_2), \Lambda)^{-\abs{w_{1,i}}} \eqend{,} \\
\begin{split}
\mathsf{G}^{\vec{w}_2}(k, \vec{q}_2; \Lambda) &= \sup(\eta_k(k, \vec{q}_2), \Lambda)^{-\abs{w_{2,*}}} \prod_{q_i \in \vec{q}_2} \sup(\eta_{q_i}(k, \vec{q}_2), \Lambda)^{-\abs{w_{2,i}}} \\
&\leq \sup(\eta_{q_j}(\vec{q}_1, \vec{q}_2), \Lambda)^{-\abs{w_{2,*}}} \prod_{q_i \in \vec{q}_2} \sup(\eta_{q_i}(\vec{q}_1, \vec{q}_2), \Lambda)^{-\abs{w_{2,i}}} \eqend{,}
\end{split}
\end{equations}
and it follows that
\begin{equation}
\label{gw_fused_est}
\frac{\mathsf{G}^{\vec{w}_1+\vec{w}_2}(\vec{q}_1, \vec{q}_2; \Lambda)}{\mathsf{G}^{\vec{w}_1}(\vec{q}_1, -k; \Lambda) \mathsf{G}^{\vec{w}_2}(k, \vec{q}_2; \Lambda)} \geq 1
\end{equation}
with $w_{2,*}$ added at the appropriate place:
\begin{equation}
\vec{w}_1+\vec{w}_2 = (w_{1,1}, \ldots, w_{1,j-1}, w_{1,j}+w_{2,*}, w_{1,j+1}, \ldots, w_{1,n_1}, w_{2,1}, \ldots, w_{2,n_2}) \eqend{.}
\end{equation}
In the case that one or both trees contain a special vertex, the argumentation is similar, with $\bar{\eta}_{q_i}$~\eqref{bareta_i_def} used instead of $\eta_{q_i}$ at the appropriate places, and we also obtain equation~\eqref{gw_fused_est}. For a fusion of two trees with special vertices, from equation~\eqref{gw_fused_1} we thus obtain equation~\eqref{gw_fused_1_est} of the Lemma, while in the other cases equation~\eqref{gw_fused_2} gives equation~\eqref{gw_fused_2_est}, and the Lemma is proven.
\end{proof}

\paragraph*{Amputation.} If an external momentum of a tree $T$ of order $(m+n+1,r)$ vanishes and no derivatives act on this momentum, we obtain a new tree $T'$ of order $(m+n,r)$ by amputating the corresponding external vertex $v$ and the incident line (w.l.o.g. we can assume that this external momentum is the first one). We treat first the case without special vertex, where we have
\begin{lemma}
Amputating an external vertex $v$ from a tree $T$ without a special vertex, the change in weights can be estimated by
\begin{equation}
\label{amputate}
\mathsf{G}^{T,\vec{w}}_{M \vec{K} \vec{L}^\ddag; [v_p]}(0, \vec{q}; \mu, \Lambda) \leq \frac{\Lambda^{1-[v]}}{\sup(\inf(\mu, \eta(\vec{q})),\Lambda)} \mathsf{G}^{T',\vec{w}}_{\vec{K} \vec{L}^\ddag; [v_p]}(\vec{q}; \mu, \Lambda) \eqend{,}
\end{equation}
while for trees $T^*$ with a special vertex we have the same estimate with $\eta$ replaced by $\bar{\eta}$.
\end{lemma}
\begin{proof}
Since no derivatives act on the amputated external vertex, the derivative weight factor does not change due to $\eta_{q_i}(0, \vec{q}) = \eta_{q_i}(\vec{q})$ for $i \in \{1,\ldots,n-1\}$. In order to determine the change in weights coming from the amputation, we first convert the external vertex into an internal one, and then perform a reduction operation. The conversion gives an extra factor of $\Lambda^{-[v]}$ (since the vertex has valence $1$), and the change in weights for the reduction is given in Table~\ref{table_reduction}. Since $\abs{q_i} \geq \inf(\mu, \eta(\vec{q}))$ for any $q_i$, the Lemma follows. For trees $T^*$ where momentum is not conserved at the special vertex, the same procedure applies, and we obtain the same estimate with $\eta$ replaced by $\bar{\eta}$.
\end{proof}

If necessary, one has to perform additional reduction operations afterwards to again obtain a fully reduced tree. Since reduction operations can only increase the tree weight, the estimate~\eqref{amputate} stays valid. We then define
\begin{definition}
The set of all fully reduced trees $T$ of order $(m+n,r)$ with arbitrary $r$ is denoted by $\mathcal{T}_{m+n}$, and the set of all fully reduced trees $T^*$ of order $(m+n,r)$ with arbitrary $r$ is denoted by $\mathcal{T}^*_{m+n}$.
\end{definition}
Note that these sets are finite since only three- and four-valent internal vertices are allowed (except for $\mathcal{T}_1$, which consists of one tree with a one-valent vertex, and $\mathcal{T}_2$, which consists of one tree with a two-valent vertex).

\subsection{Inequalities}
\label{sec_trees_ineqs}

For the proofs in the next section, we need to estimate tree weight factors for fully reduced trees for larger or smaller $\Lambda$ and for larger momenta, depending on whether the tree is irrelevant, marginal or relevant.

\subsubsection{Irrelevant and marginal trees}

\begin{lemma}
For $\lambda \geq \Lambda$ and any tree $T$ with $[T] \leq 0$, we have
\begin{equation}
\label{t_irr_ineq1}
\mathsf{G}^{T,\vec{w}}_{\vec{K} \vec{L}^\ddag}(\vec{q}; \mu, \lambda) \leq \mathsf{G}^{T,\vec{w}}_{\vec{K} \vec{L}^\ddag}(\vec{q}; \mu, \Lambda) \eqend{,}
\end{equation}
while for $[T] < 0$ we even have
\begin{equation}
\label{t_irr_ineq2}
\mathsf{G}^{T,\vec{w}}_{\vec{K} \vec{L}^\ddag}(\vec{q}; \mu, \lambda) \leq \left( \frac{\sup(\inf(\mu, \eta(\vec{q})), \Lambda)}{\sup(\inf(\mu, \eta(\vec{q})), \lambda)} \right)^\epsilon \mathsf{G}^{T,\vec{w}}_{\vec{K} \vec{L}^\ddag}(\vec{q}; \mu, \Lambda)
\end{equation}
for any $0 \leq \epsilon \leq -[T]$, in particular for $\epsilon = \Delta$. For trees $T^*$ with a special vertex where momentum is not conserved, the same estimates are valid, with $\eta$ replaced by $\bar{\eta}$.
\end{lemma}
\begin{proof}
The various weight factors contained in $\mathsf{G}^{T,\vec{w}}$ (mostly of the form $\sup(a, \lambda)^c$) come with positive and negative powers $c$, and we would like to estimate the whole tree weight factor at $\lambda = \Lambda$. This is trivial if $c \leq 0$, but we first have to extract positive weight factors using Lemma~\ref{lemma_biggermomentum}. For a tree without a particular weight factor, positive weight factors can come from $3$-valent internal vertices and external vertices. However, the weight factor of an external vertex $v_e$ together with the weight factor of the adjacent line of momentum $q$ give a factor of $\sup(\abs{q}, \Lambda)^{1-[v_e]}$, and since all $[v_e] \geq 1$ this has always a negative power. Since the momentum $q_v$ associated to an internal $3$-valent vertex $v$ is always the one of largest absolute value among the momenta of the incident lines $l$ ($\abs{q_v} \geq \abs{q_l}$) and the dimension of $v$ is $1$, we take the weight factor of this vertex together with half of the weight factor of any adjacent line $l$, giving a factor of $\sup(\abs{q_v},\lambda) / \sup(\abs{q_l},\lambda)$, which can be estimated using Lemma~\ref{lemma_biggermomentum} (taking $k = \Lambda$, $K = \lambda$). This only cannot work if we have three external vertices connected to a single $3$-valent internal vertex, since then we would need to use the weight factor of one internal line twice. However, since $[T] \leq 0$, either the dimensions of the external vertices are large enough (and thus provide the necessary negative weight), or we have at least one derivative weight factor, and use this instead of the weight factor of the corresponding internal line. If the tree contains a particular weight factor, the corresponding momentum is larger than any momentum associated to any other element of the tree, and if $[v_p] > 0$ we can use Lemma~\ref{lemma_biggermomentum} for the particular weight factor together with any other weight factor. Again, since $[T] \leq 0$, we can extract all positive weight factors. The remaining weight factors are all of the form $\sup(a, \lambda)^c$ with $c \leq 0$ and can thus be estimated at $\lambda = \Lambda$, such that we obtain equation~\eqref{t_irr_ineq1} of the Lemma. If the tree is strictly irrelevant, $[T] < 0$, there will be strictly more negative than positive weight factors, and since all $a$ that appear in the weight factors of the form $\sup(a,\Lambda)$ fulfil $a \geq \inf(\mu, \eta(\vec{q}))$, we obtain equation~\eqref{t_irr_ineq2} of the Lemma. For trees $T^*$ with a special vertex where momentum is not conserved, the same proof works, with $\eta$ replaced by $\bar{\eta}$, and the Lemma follows.
\end{proof}

\subsubsection{Relevant and marginal trees}

\begin{lemma}
For $0 \leq t \leq 1$ and any tree $T$ with $[T] \geq 0$, we have
\begin{equation}
\label{t_rel_ineq1}
\mathsf{G}^{T,\vec{w}}_{\vec{K} \vec{L}^\ddag}(t \vec{q}; \Lambda, \Lambda) \leq \mathsf{G}^{T,\vec{w}}_{\vec{K} \vec{L}^\ddag}(\vec{q}; \Lambda, \Lambda) \eqend{.}
\end{equation}
\end{lemma}
\begin{proof}
For any weight factors with a positive power, the inequality is immediate, and we thus only need to extract weight factors with a negative power first. This can also be done using Lemma~\ref{lemma_biggermomentum}, but in a different way. First we note that since $[T] \geq 0$, for each weight factor with a negative power, we have a weight factor with positive power but larger momentum (including the particular weight factor). Such a pair of weight factors can then be estimated as follows (with $\abs{q_1} > \abs{q_2}$ and $c \geq 0$):
\begin{equation}
\left( \frac{\sup(t \abs{q_1}, \Lambda)}{\sup(t \abs{q_2}, \Lambda)} \right)^c = \left( \frac{\sup\left( \abs{q_1}, \frac{\Lambda}{t} \right)}{\sup\left( \abs{q_2}, \frac{\Lambda}{t} \right)} \right)^c \leq \left( \frac{\sup\left( \abs{q_1}, \Lambda \right)}{\sup\left( \abs{q_2}, \Lambda \right)} \right)^c \eqend{,}
\end{equation}
using Lemma~\ref{lemma_biggermomentum} with $K = \Lambda/t$ and $k = \Lambda$. We can thus extract all negative weight factors, and the remaining weight factors are all of the form $\sup(t a, \Lambda)^c$ with $a,c \geq 0$ and can thus be estimated at $t = 1$, such that the Lemma follows.
\end{proof}

\begin{lemma}
For $\Lambda \leq \lambda \leq \mu$ and any tree $T^*$ with a special vertex with $[T] \geq 0$, we have
\begin{equation}
\label{t_rel_ineq3}
\mathsf{G}^{T^*,\vec{w}}_{\vec{K} \vec{L}^\ddag}(\vec{q}; \mu, \lambda) \leq \mathsf{G}^{T^*,\vec{w}}_{\vec{K} \vec{L}^\ddag}(\vec{q}; \mu, \Lambda) \eqend{.}
\end{equation}
For trees $T$ without a special vertex but with a particular weight factor, the same estimate is valid for $[\vec{K}] + [\vec{L}^\ddag] + \abs{\vec{w}} \geq 4$, \ie, only for trees which are ``not too relevant''.
\end{lemma}
\begin{proof}
Since $\mu \geq \lambda$, the special vertex factor and the particular weight factor do not depend on $\lambda$. Consider then the trees $\{T_k\}$ which result by removing the special vertex and the particular weight factor (and which thus have lines connected to only one vertex). Since the special vertex factor contributes $- k$ to the tree dimension $[T^*]$, where $k$ is the valency of the special vertex, and overall we have $[T^*] = [v_p] - [\vec{K}] - [\vec{L}^\ddag] - \abs{\vec{w}}$~\eqref{t_dim_def}, each tree $T_k$ obtained after removing the special vertex (with $m_k + n_k$ external vertices and $\abs{\vec{w}_k}$ derivatives) and the particular weight factor has dimension $[T_k] = 1 - [\vec{K}] - [\vec{L}^\ddag] - \abs{\vec{w}_k}$. Since all $[\phi_K], [\phi_L^\ddag] \geq 1$ we especially have $[T_k] \leq 0$, and can thus apply the inequality~\eqref{t_irr_ineq1} to each $T_k$. Adding the special vertex and the particular weight factor back we thus obtain equation~\eqref{t_rel_ineq3} of the Lemma. For trees $T$ without a special vertex but with a particular weight factor, the expression for the tree dimension $[T] = 4 + [v_p] - [\vec{K}] - [\vec{L}^\ddag] - \abs{\vec{w}}$~\eqref{t_dim_def} shows that the same argument works for $[\vec{K}] + [\vec{L}^\ddag] + \abs{\vec{w}} \geq 4$, and the Lemma is proven.
\end{proof}

\section{Bounds on functionals}
\label{sec_bounds}

The regularised functionals with and without composite operator insertions are bounded uniformly in $\Lambda_0$, and furthermore their derivative with respect to $\Lambda_0$ is bounded in a way that implies the existence of the unregularised limit $\Lambda_0 \to \infty$. These bounds can be expressed using fully reduced weighted trees, and concretely we prove:
\begin{proposition}
\label{thm_l0}
For all multiindices $\vec{w}$, at each order $l$ in perturbation theory and for an arbitrary number $m$ of external fields $\vec{K}$ and $n$ antifields $\vec{L}^\ddag$, we have the bound
\begin{equation}
\label{bound_l0}
\abs{ \partial^\vec{w} \mathcal{L}^{\Lambda, \Lambda_0, l}_{\vec{K} \vec{L}^\ddag}(\vec{q}) } \leq \sum_{T \in \mathcal{T}_{m+n}} \mathsf{G}^{T,\vec{w}}_{\vec{K} \vec{L}^\ddag}(\vec{q}; \mu, \Lambda) \,\mathcal{P}\left( \ln_+ \frac{\sup\left( \abs{\vec{q}}, \mu \right)}{\sup(\inf(\mu, \eta(\vec{q})), \Lambda)}, \ln_+ \frac{\Lambda}{\mu} \right) \eqend{,}
\end{equation}
where $\mathcal{P}$ is a polynomial with non-negative coefficients (depending on $m,n,l,\abs{\vec{w}}$ and the renormalisation conditions). The sum runs over all fully reduced trees $T$ of order $(m+n,r)$ with arbitrary $r$ (the number of internal vertices), where the dimension of the external vertices is given by the dimension of the corresponding operator (\ie, $[v_1] = [\phi_{K_1}]$, \etc).
\end{proposition}
For non-exceptional external momenta where $\eta(\vec{q}) > 0$, this shows uniform boundedness. For the existence of the unregularised limit, we also need a bound on the $\Lambda_0$ derivative, which is given in Subsection~\ref{bounds_lambda0}.

For functionals with one insertion of a composite operator, we can always use the shift property~\eqref{func_sop_shift} with $y = x$ to obtain a functional with one insertion at $x = 0$. We then prove
\begin{proposition}
\label{thm_l1}
For all multiindices $\vec{w}$, at each order $l$ in perturbation theory and for an arbitrary number $m$ of external fields $\vec{K}$ and $n$ antifields $\vec{L}^\ddag$ and any composite operator $\op_A$, we have the bound
\begin{splitequation}
\label{bound_l1}
\abs{ \partial^\vec{w} \mathcal{L}^{\Lambda, \Lambda_0, l}_{\vec{K} \vec{L}^\ddag}\left( \op_A(0); \vec{q} \right) } &\leq \sup\left( 1, \frac{\abs{\vec{q}}}{\sup(\mu, \Lambda)} \right)^{g^{(1)}([\op_A],m+n+2l,\abs{\vec{w}})} \\
&\quad\times \sum_{T^* \in \mathcal{T}^*_{m+n}} \mathsf{G}^{T^*,\vec{w}}_{\vec{K} \vec{L}^\ddag; [\op_A]}(\vec{q}; \mu, \Lambda) \, \mathcal{P}\left( \ln_+ \frac{\sup\left( \abs{\vec{q}}, \mu \right)}{\sup(\inf(\mu, \bar{\eta}(\vec{q})), \Lambda)}, \ln_+ \frac{\Lambda}{\mu} \right) \eqend{,}
\end{splitequation}
where the sum runs over all fully reduced trees $T^*$ with one special vertex where momentum is not conserved.
\end{proposition}
This bound shows uniform boundedness, and to show convergence we again need a bound on the $\Lambda_0$ derivative, given in Subsection~\ref{bounds_lambda0}. In contrast to the case without insertions, this bound involves a loop-order dependent ``large momentum factor'' with a function $g^{(s)}$ defined for $[\op] \geq 0$, $r \geq 0$ and $s \geq 1$ by
\begin{equation}
\label{gs_def}
g^{(s)}([\op],r,\abs{\vec{w}}) \equiv ([\op]+s) (r+3s-3) + \sup([\op]+s-\abs{\vec{w}},0) \eqend{,}
\end{equation}
which for all $\vec{u}+\vec{v} \leq \vec{w}$ and all $\vec{w}'$ fulfils the properties
\begin{equations}
g^{(s)}([\op],r,\abs{\vec{v}}) &\leq g^{(s)}([\op],r+1,\abs{\vec{w}}) \eqend{,} \label{gs_prop_1} \\
g^{(s)}([\op],r,\abs{\vec{w}}+1) + 1 &\leq g^{(s)}([\op],r,\abs{\vec{w}}) \qquad\text{ for } \abs{\vec{w}} \leq [\op]+s-1 \eqend{,} \label{gs_prop_2} \\
g^{(s)}([\op],r,\abs{\vec{u}}) + g^{(s')}([\op'],r',\abs{\vec{v}}) &\leq g^{(s+s')}([\op]+[\op'], r+r'-2, \abs{\vec{w}'}) - ([\op]+[\op']+s+s') \eqend{.} \label{gs_prop_3}
\end{equations}
Functionals with one insertion of an integrated composite operator are proven to fulfil
\begin{proposition}
\label{thm_l1i}
For all multiindices $\vec{w}$, at each order $l$ in perturbation theory and for an arbitrary number $m$ of external fields $\vec{K}$ and $n$ antifields $\vec{L}^\ddag$ and any integrated composite operator $\op_A$ of dimension $[\op_A] \geq 4$, we have the bound
\begin{splitequation}
\label{bound_l1i}
\abs{ \partial^\vec{w} \mathcal{L}^{\Lambda, \Lambda_0, l}_{\vec{K} \vec{L}^\ddag}\left( \int\!\op_A; \vec{q} \right) } &\leq \sup\left( 1, \frac{\abs{\vec{q}}}{\sup(\mu, \Lambda)} \right)^{g^{(1)}([\op_A]-4,m+n+2l,\abs{\vec{w}})} \\
&\quad\times \sum_{T \in \mathcal{T}_{m+n}} \mathsf{G}^{T,\vec{w}}_{\vec{K} \vec{L}^\ddag; [\op_A]-4}(\vec{q}; \mu, \Lambda) \, \mathcal{P}\left( \ln_+ \frac{\sup\left( \abs{\vec{q}}, \mu \right)}{\sup(\inf(\mu, \eta(\vec{q})), \Lambda)}, \ln_+ \frac{\Lambda}{\mu} \right) \eqend{,}
\end{splitequation}
which only differs in the type of trees summed over, where momentum is again conserved (and thus also involves $\eta$ instead of $\bar{\eta}$).
\end{proposition}
We recall that functionals with integrated insertions are \emph{not} defined by first taking a functional with a non-integrated insertion of $\op_A(x)$ and then integrating over $x$ as done previously~\cite{hollandhollands2015a}, where one needs to show in addition that the integral over $x$ is convergent for multiple insertions, but are defined in their own right by the corresponding flow equation and boundary conditions, as suggested by the notation. Since the boundary conditions for non-integrated and integrated composite operators are distinct (see Table~\ref{table_boundary}), this has the unfortunate side effect that
\begin{equation}
(2\pi)^4 \delta\left( \sum_{i=1}^{m+n} q_i \right) \mathcal{L}^{\Lambda, \Lambda_0, l}_{\vec{K} \vec{L}^\ddag}\left( \int\!\op_A; \vec{q} \right) \neq \int \mathcal{L}^{\Lambda, \Lambda_0, l}_{\vec{K} \vec{L}^\ddag}\left( \op_A(x); \vec{q} \right) \total^4 x
\end{equation}
(we recall that the functionals with an integrated insertion are defined with the momentum-conserving $\delta$ taken out~\eqref{func_ft_iop}). However, we have
\begin{proposition}
\label{thm_l1i_d}
For any composite operator $\op_A$ of dimension $[\op_A] \geq 4$, there exists a uniquely defined composite operator $\tilde{\op}_A = \sum_k \op_{A,k}$ of lower or equal dimension $[\op_{A,k}] \leq [\op_A]$ and of the same ghost number, such that for all multiindices $\vec{w}$, at each order $l$ in perturbation theory and for an arbitrary number $m$ of external fields $\vec{K}$ and $n$ antifields $\vec{L}^\ddag$ we have
\begin{equation}
\label{func_1iop_int}
(2\pi)^4 \delta\left( \sum_{i=1}^{m+n} q_i \right) \mathcal{L}^{\Lambda, \Lambda_0, l}_{\vec{K} \vec{L}^\ddag}\left( \int\!\op_A; \vec{q} \right) = \int \mathcal{L}^{\Lambda, \Lambda_0, l}_{\vec{K} \vec{L}^\ddag}\left( \tilde{\op}_A(x); \vec{q} \right) \total^4 x \eqend{.}
\end{equation}
\end{proposition}
Furthermore, for a proper choice of boundary conditions the functionals with one operator insertion fulfil
\begin{proposition}
\label{thm_lowenstein_1}
For all multiindices $\vec{w}$, at each order $l$ in perturbation theory, for an arbitrary number $m$ of external fields $\vec{K}$ and $n$ antifields $\vec{L}^\ddag$ and for any composite operator $\op_A$ and any multiindex $a$, there exist boundary conditions (\eg, the ones given in equation~\eqref{func_1op_bphz}) such that
\begin{splitequation}
\label{lowenstein_1}
\partial^a_x \mathcal{L}^{\Lambda, \Lambda_0, l}_{\vec{K} \vec{L}^\ddag}\left( \op_A(x); \vec{q} \right) = \mathcal{L}^{\Lambda, \Lambda_0, l}_{\vec{K} \vec{L}^\ddag}\left( \partial^a_x \op_A(x); \vec{q} \right) \eqend{.}
\end{splitequation}
\end{proposition}
Since the right-hand side of equation~\eqref{func_1iop_int} thus vanishes when $\tilde{\op}_A = \partial^a_x \tilde{\op}_B$, and $\tilde{\op}_A$ depends obviously  linearly on $\op_A$, it is seen to be consistent to choose boundary conditions for the functionals with one integrated operator insertion such that
\begin{equation}
\label{func_1iop_d}
\mathcal{L}^{\Lambda, \Lambda_0, l}_{\vec{K} \vec{L}^\ddag}\left( \int\!\op_A; \vec{q} \right) = 0
\end{equation}
if $\op_A(x) = \partial^a_x \op_B(x)$ for a composite operator $\op_B$ and a multiindex $a > 0$.

For functionals with more than one insertion of a composite operator, we have to distinguish between functionals with one and zero integrated insertions. First we note that we can always choose the last operator insertion to be at $x_s = 0$, since the general case can be recovered using the shift property~\eqref{func_sop_shift}, and this will be understood from now on even if not shown explicitly. We then prove
\begin{proposition}
\label{thm_lsa}
At each order $l$ in perturbation theory and for an arbitrary number $m$ of external fields $\vec{K}$ and $n$ antifields $\vec{L}^\ddag$, and for $\Lambda \geq \mu$, the functionals with $s \geq 2$ insertions of arbitrary (non-integrated) composite operators $\op_{A_i}$ can be written in the form
\begin{equation}
\label{bound_ls_lambdamu}
\partial^\vec{w} \mathcal{L}^{\Lambda, \Lambda_0, l}_{\vec{K} \vec{L}^\ddag}\left( \bigotimes_{k=1}^s \op_{A_k}(x_k); \vec{q} \right) = \sum_{\vec{a} > 0, \abs{\vec{a}} = D+[\op_\vec{A}]} \partial^\vec{a}_{\vec{x}} \partial^\vec{w} \mathcal{K}^{\Lambda, \Lambda_0, l}_{\vec{K} \vec{L}^\ddag; D}\left( \bigotimes_{k=1}^s \op_{A_k}(x_k); \vec{q} \right) \eqend{,}
\end{equation}
where $a$ is a multiindex not involving the last coordinate $x_s$,
\begin{equation}
[\op_\vec{A}] \equiv \sum_{k=1}^s [\op_{A_k}] \eqend{,}
\end{equation}
the parameter $D \in \{0,1\}$ for $m+n > 0$ and $D = 1$ for $m+n=0$. The kernel $\mathcal{K}^{\Lambda, \Lambda_0, l}$ satisfies the bound
\begin{splitequation}
\label{bound_ks_lambdamu}
\abs{ \partial^\vec{w} \mathcal{K}^{\Lambda, \Lambda_0, l}_{\vec{K} \vec{L}^\ddag; D}\left( \bigotimes_{k=1}^s \op_{A_k}(x_k); \vec{q} \right) } &\leq \prod_{i=1}^{s-1} \left( 1 + \ln_+ \frac{1}{\mu \abs{x_i}} \right) \sup\left( 1, \frac{\abs{\vec{q}}}{\Lambda} \right)^{g^{(s)}([\op_\vec{A}],m+n+2l,\abs{\vec{w}})} \\
&\quad\times \sum_{T^* \in \mathcal{T}^*_{m+n}} \mathsf{G}^{T^*,\vec{w}}_{\vec{K} \vec{L}^\ddag; -D}(\vec{q}; \mu, \Lambda) \, \mathcal{P}\left( \ln_+ \frac{\sup\left( \abs{\vec{q}}, \mu \right)}{\Lambda}, \ln_+ \frac{\Lambda}{\mu} \right) \eqend{.}
\end{splitequation}
\end{proposition}

For $\Lambda < \mu$, we show
\begin{proposition}
\label{thm_lsb}
At each order $l$ in perturbation theory and for an arbitrary number $m$ of external fields $\vec{K}$ and $n$ antifields $\vec{L}^\ddag$, and for $\Lambda < \mu$, the functionals with $s \geq 2$ insertions of arbitrary (non-integrated) composite operators $\op_{A_i}$ can be written in the form
\begin{equation}
\label{bound_ls_lambda}
\partial^\vec{w} \mathcal{L}^{\Lambda, \Lambda_0, l}_{\vec{K} \vec{L}^\ddag}\left( \bigotimes_{k=1}^s \op_{A_k}(x_k); \vec{q} \right) = \sum_{\abs{\vec{a}} \leq D+[\op_\vec{A}]} \mu^{D+[\op_\vec{A}]-\abs{\vec{a}}} \partial^\vec{a}_{\vec{x}} \partial^\vec{w} \mathcal{K}^{\Lambda, \Lambda_0, l}_{\vec{K} \vec{L}^\ddag}\left( \bigotimes_{k=1}^s \op_{A_k}(x_k); \vec{q} \right) \eqend{,}
\end{equation}
where the kernel $\mathcal{K}^{\Lambda, \Lambda_0, l}$ satisfies the bounds
\begin{splitequation}
\label{bound_ks_lambda}
\abs{ \partial^\vec{w} \mathcal{K}^{\Lambda, \Lambda_0, l}_{\vec{K} \vec{L}^\ddag}\left( \bigotimes_{k=1}^s \op_{A_k}(x_k); \vec{q} \right) } &\leq \prod_{i=1}^{s-1} \left( 1 + \ln_+ \frac{1}{\mu \abs{x_i}} \right) \sup\left( 1, \frac{\abs{\vec{q}}}{\mu} \right)^{g^{(s)}([\op_\vec{A}],m+n+2l,\abs{\vec{w}})} \\
&\quad\times \sum_{T^* \in \mathcal{T}^*_{m+n}} \mathsf{G}^{T^*,\vec{w}}_{\vec{K} \vec{L}^\ddag; -D}(\vec{q}; \mu, \Lambda) \, \mathcal{P}\left( \ln_+ \frac{\sup\left( \abs{\vec{q}}, \mu \right)}{\sup(\inf(\mu,\bar{\eta}(\vec{q})),\Lambda)} \right) \eqend{.}
\end{splitequation}
\end{proposition}

If all $x_i$ are distinct, we can also prove bounds which do not involve any derivatives and which are given by
\begin{proposition}
\label{thm_lsc}
At each order $l$ in perturbation theory and for an arbitrary number $m$ of external fields $\vec{K}$ and $n$ antifields $\vec{L}^\ddag$, the functionals with $s \geq 2$ insertions of arbitrary (non-integrated) composite operators $\op_{A_i}$ fulfil the bound
\begin{splitequation}
\label{bound_ls_x}
&\abs{ \partial^\vec{w} \mathcal{L}^{\Lambda, \Lambda_0, l}_{\vec{K} \vec{L}^\ddag}\left( \bigotimes_{k=1}^s \op_{A_k}(x_k); \vec{q} \right) } \leq \mu^{[\op_\vec{A}]+s \epsilon} \sup\left( 1, \frac{\abs{\vec{q}}}{\sup(\mu,\Lambda)} \right)^{g^{(s)}([\op_\vec{A}],m+n+2l,\abs{\vec{w}})} \\
&\qquad\times \sum_{\tau \in \mathscr{T}_s} \mathsf{W}^\tau(x_1,\ldots,x_s) \sum_{T^* \in \mathcal{T}^*_{m+n}} \mathsf{G}^{T^*,\vec{w}}_{\vec{K} \vec{L}^\ddag; -s \epsilon}(\vec{q}; \mu, \Lambda) \, \mathcal{P}\left( \ln_+ \frac{\sup\left( \abs{\vec{q}}, \mu \right)}{\sup(\inf(\mu,\bar{\eta}(\vec{q})),\Lambda)}, \ln_+ \frac{\Lambda}{\mu} \right) \eqend{,}
\end{splitequation}
where the set $\mathscr{T}_s$ and the weight factors $\mathsf{W}^\tau$ are given in Definition~\ref{def_tree_x}, and where $\epsilon \in [0,\infty)$ with $m+n+\epsilon>0$ (\ie, for $m+n > 0$ we may take $\epsilon = 0$, but for $m+n = 0$ we need $\epsilon > 0$).
\end{proposition}

All these bounds can be proven in the same way: first we bound the right-hand side of the corresponding flow equation using the induction hypothesis, and then we integrate over $\lambda$ using the boundary conditions given in Table~\ref{table_boundary}. Since the necessary estimates are basically the same for all functionals we can do everything together: the right-hand side of the flow equation is bounded in Subsection~\ref{sec_bounds_rhs}, and the integration over $\lambda$ is done in Subsections~\ref{sec_bounds_irr} -- \ref{sec_bounds_relvan}. This thus proves the Propositions~\ref{thm_l0}, \ref{thm_l1} and~\ref{thm_l1i}. Propositions~\ref{thm_l1i_d} and~\ref{thm_lowenstein_1} are then proven in Subsection~\ref{sec_bounds_l1iprops}. Functionals with more than one insertion are sufficiently different that we prefer to treat them only afterwards, in Subsection~\ref{sec_bounds_func_sop}, where Propositions~\ref{thm_lsa}, \ref{thm_lsb} and~\ref{thm_lsc} are proven.

\subsection{The right-hand side}
\label{sec_bounds_rhs}

The bounds that we want to obtain for the right-hand side of the flow equation are of the form
\begin{equation}
\label{bound_l0_lambdaderiv}
\abs{ \partial_\Lambda \partial^\vec{w} \mathcal{L}^{\Lambda, \Lambda_0, l}_{\vec{K} \vec{L}^\ddag}(\vec{q}) } \leq \frac{1}{\sup(\inf(\mu, \eta(\vec{q})), \Lambda)} \sum_{T \in \mathcal{T}_{m+n}} \mathsf{G}^{T,\vec{w}}_{\vec{K} \vec{L}^\ddag}(\vec{q}; \mu, \Lambda) \,\mathcal{P}\left( \ln_+ \frac{\sup\left( \abs{\vec{q}}, \mu \right)}{\Lambda}, \ln_+ \frac{\Lambda}{\mu} \right) \eqend{,}
\end{equation}
\ie, the bound~\eqref{bound_l0} for $\partial^\vec{w} \mathcal{L}^{\Lambda, \Lambda_0, l}_{\vec{K} \vec{L}^\ddag}(\vec{q})$ divided by $\sup(\inf(\mu, \eta(\vec{q})), \Lambda)$, and with $\ln_+ \sup\left( \abs{\vec{q}}, \mu \right) / \sup(\inf(\mu, \eta(\vec{q})), \Lambda)$ replaced by $\ln_+ \sup\left( \abs{\vec{q}}, \mu \right) / \Lambda$ in the polynomial. For functionals with insertions, we similarly want to prove
\begin{splitequation}
\label{bound_l1_lambdaderiv}
\abs{ \partial_\Lambda \partial^\vec{w} \mathcal{L}^{\Lambda, \Lambda_0, l}_{\vec{K} \vec{L}^\ddag}\left( \op_A(0); \vec{q} \right) } &\leq \frac{1}{\sup(\inf(\mu, \bar{\eta}(\vec{q})), \Lambda)} \sup\left( 1, \frac{\abs{\vec{q}}}{\sup(\mu, \Lambda)} \right)^{g^{(1)}([\op_A],m+n+2l,\abs{\vec{w}})} \\
&\quad\times \sum_{T^* \in \mathcal{T}^*_{m+n}} \mathsf{G}^{T^*,\vec{w}}_{\vec{K} \vec{L}^\ddag; [\op_A]}(\vec{q}; \mu, \Lambda) \, \mathcal{P}\left( \ln_+ \frac{\sup\left( \abs{\vec{q}}, \mu \right)}{\Lambda}, \ln_+ \frac{\Lambda}{\mu} \right) \eqend{,}
\end{splitequation}
and for functionals with an integrated insertion we need
\begin{splitequation}
\label{bound_l1i_lambdaderiv}
\abs{ \partial_\Lambda \partial^\vec{w} \mathcal{L}^{\Lambda, \Lambda_0, l}_{\vec{K} \vec{L}^\ddag}\left( \int\!\op_A; \vec{q} \right) } &\leq \frac{1}{\sup(\inf(\mu, \eta(\vec{q})), \Lambda)} \sup\left( 1, \frac{\abs{\vec{q}}}{\sup(\mu, \Lambda)} \right)^{g^{(1)}([\op_A]-4,m+n+2l,\abs{\vec{w}})} \\
&\quad\times \sum_{T \in \mathcal{T}_{m+n}} \mathsf{G}^{T,\vec{w}}_{\vec{K} \vec{L}^\ddag; [\op_A]-4}(\vec{q}; \mu, \Lambda) \, \mathcal{P}\left( \ln_+ \frac{\sup\left( \abs{\vec{q}}, \mu \right)}{\Lambda}, \ln_+ \frac{\Lambda}{\mu} \right) \eqend{.}
\end{splitequation}

\subsubsection{The linear term}

Let us start with the first term on the right-hand side of the flow equation~\eqref{l_0op_flow_hierarchy}, which taking $\vec{w}$ momentum derivatives reads
\begin{equation}
F_1 \equiv \frac{c}{2} \int \left( \partial_\Lambda C^{\Lambda, \Lambda_0}_{MN}(-p) \right) \partial^\vec{w} \mathcal{L}^{\Lambda, \Lambda_0, l-1}_{MN \vec{K} \vec{L}^\ddag}(p, -p, \vec{q}) \frac{\total^4 p}{(2\pi)^4} \eqend{.}
\end{equation}
From the estimate on the covariance~\eqref{prop_abl}, we have the bound
\begin{equation}
\abs{F_1} \leq c \int \sup(\abs{p},\Lambda)^{-5+[\phi_M]+[\phi_N]} \, \mathe^{-\frac{\abs{p}^2}{2\Lambda^2}} \abs{\partial^\vec{w} \mathcal{L}^{\Lambda, \Lambda_0, l-1}_{MN \vec{K} \vec{L}^\ddag}(p, -p, \vec{q})} \frac{\total^4 p}{(2\pi)^4} \eqend{.}
\end{equation}
We now insert the induction hypothesis~\eqref{bound_l0}, using that
\begin{equations}[func_0op_estlog]
\eta(p, -p, \vec{q}) &= 0 \eqend{,} \\
\abs{p, -p, \vec{q}} &\leq 2 \abs{p} + \abs{\vec{q}} \eqend{,} \\
\ln_+ \frac{\sup\left( \abs{p, -p, \vec{q}}, \mu \right)}{\Lambda} &\leq \ln_+ \frac{\sup\left( \abs{\vec{q}}, \mu \right)}{\Lambda} + \ln_+ \frac{\abs{p}}{\Lambda} + \ln 2 \eqend{,}
\end{equations}
to obtain
\begin{splitequation}
\abs{F_1} \leq \sum_{T \in \mathcal{T}_{m+n+2}} \int &\sup(\abs{p},\Lambda)^{-5+[\phi_M]+[\phi_N]} \mathe^{-\frac{\abs{p}^2}{2 \Lambda^2}} \mathsf{G}^{T,\vec{w}}_{MN \vec{K} \vec{L}^\ddag}(p, -p, \vec{q}; \mu, \Lambda) \\
&\quad\times \mathcal{P}\left( \ln_+ \frac{\sup\left( \abs{\vec{q}}, \mu \right)}{\Lambda}, \ln_+ \frac{\abs{p}}{\Lambda}, \ln_+ \frac{\Lambda}{\mu} \right) \frac{\total^4 p}{(2\pi)^4} \eqend{.}
\end{splitequation}
Rescaling $p = x \Lambda$ and applying Lemma~\ref{lemma_pint2} with $\beta_i = \gamma_i = 1$, we obtain the bound
\begin{splitequation}
\abs{F_1} \leq \Lambda^{[\phi_M]+[\phi_N]-1} \sum_{T \in \mathcal{T}_{m+n+2}} &\mathsf{G}^{T,\vec{w}}_{MN \vec{K} \vec{L}^\ddag}(0, 0, \vec{q}; \mu, \Lambda) \, \mathcal{P}\left( \ln_+ \frac{\sup\left( \abs{\vec{q}}, \mu \right)}{\Lambda}, \ln_+ \frac{\Lambda}{\mu} \right) \eqend{.}
\end{splitequation}
For each tree $T$ in the sum, we now amputate the first two external vertices with zero momentum corresponding to $M$ and $N$. The amputation gives us an extra factor~\eqref{amputate}
\begin{equation}
\label{func_0op_estamputate}
\frac{\Lambda^{2-[\phi_M]-[\phi_N]}}{\sup(\inf(\mu, \eta(\vec{q})),\Lambda)^2} \leq \frac{\Lambda^{1-[\phi_M]-[\phi_N]}}{\sup(\inf(\mu, \eta(\vec{q})),\Lambda)} \eqend{,}
\end{equation}
and the new tree $T'$ has $m+n$ external vertices such that $T' \in \mathcal{T}_{m+n}$, and thus we obtain a bound of the form~\eqref{bound_l0_lambdaderiv} for $F_1$.

For functionals with one insertion of a composite operator, the first term of the right-hand side of the flow equation~\eqref{l_sop_flow_hierarchy} is given by
\begin{equation}
F_1 \equiv \frac{c}{2} \int \left( \partial_\Lambda C^{\Lambda, \Lambda_0}_{MN}(-p) \right) \partial^\vec{w} \mathcal{L}^{\Lambda, \Lambda_0, l-1}_{MN \vec{K} \vec{L}^\ddag}\left( \op_A(0); p, -p, \vec{q} \right) \frac{\total^4 p}{(2\pi)^4} \eqend{.}
\end{equation}
Inserting the bounds on the covariance~\eqref{prop_abl} and the induction hypothesis~\eqref{bound_l1}, we obtain in the same way as before
\begin{splitequation}
\abs{F_1} \leq &\sum_{T^* \in \mathcal{T}^*_{m+n+2}} \int \sup(\abs{p},\Lambda)^{-5+[\phi_M]+[\phi_N]} \mathe^{-\frac{\abs{p}^2}{2 \Lambda^2}} \mathsf{G}^{T*,\vec{w}}_{MN \vec{K} \vec{L}^\ddag; [\op_A]}(p, -p, \vec{q}; \mu, \Lambda) \\
&\quad\times \sup\left( 1, \frac{\abs{p,-p,\vec{q}}}{\sup(\mu, \Lambda)} \right)^{g^{(1)}([\op_A],m+n+2l,\abs{\vec{w}})} \, \mathcal{P}\left( \ln_+ \frac{\sup\left( \abs{\vec{q}}, \mu \right)}{\Lambda}, \ln_+ \frac{\abs{p}}{\Lambda}, \ln_+ \frac{\Lambda}{\mu} \right) \frac{\total^4 p}{(2\pi)^4} \eqend{.}
\end{splitequation}
We then rescale $p = x \Lambda$ as before, but now use Lemma~\ref{lemma_pint2} with $\beta_i = 1$ and $\gamma_i = \sup(\mu,\Lambda)/\Lambda$ to get
\begin{splitequation}
\abs{F_1} \leq \sum_{T^* \in \mathcal{T}^*_{m+n+2}} &\Lambda^{-1+[\phi_M]+[\phi_N]} \mathsf{G}^{T*,\vec{w}}_{MN \vec{K} \vec{L}^\ddag; [\op_A]}(0, 0, \vec{q}; \mu, \Lambda) \\
&\quad\times \sup\left( 1, \frac{\abs{\vec{q}}}{\sup(\mu, \Lambda)} \right)^{g^{(1)}([\op_A],m+n+2l,\abs{\vec{w}})} \, \mathcal{P}\left( \ln_+ \frac{\sup\left( \abs{\vec{q}}, \mu \right)}{\Lambda}, \ln_+ \frac{\Lambda}{\mu} \right) \eqend{.}
\end{splitequation}
The amputation of the external legs corresponding to $M$ and $N$ now gives an extra factor of
\begin{equation}
\frac{\Lambda^{2-[\phi_M]-[\phi_N]}}{\sup(\inf(\mu, \bar{\eta}(\vec{q})),\Lambda)^2} \leq \frac{\Lambda^{1-[\phi_M]-[\phi_N]}}{\sup(\inf(\mu, \bar{\eta}(\vec{q})),\Lambda)} \eqend{,}
\end{equation}
and we obtain a bound of the form~\eqref{bound_l1_lambdaderiv}. For functionals with an insertion of an integrated composite operator, we follow the same steps and obtain a bound of the form~\eqref{bound_l1i_lambdaderiv} for the first term on the right-hand side of the flow equation~\eqref{l_intop_flow_hierarchy}.

\subsubsection{The quadratic term}

Let us now turn to the second term on the right-hand of the flow equation~\eqref{l_0op_flow_hierarchy} (with $w$ momentum derivatives), which reads
\begin{splitequation}
F_2 \equiv - \sum_{\subline{\sigma \cup \tau = \{1, \ldots, m\} \\ \rho \cup \varsigma = \{1, \ldots, n\}}} \sum_{l'=0}^l &\sum_{\vec{u}+\vec{v} \leq \vec{w}} \frac{c_{\sigma\tau\rho\varsigma} c_{uvw}}{2} \left( \partial^\vec{u} \mathcal{L}^{\Lambda, \Lambda_0, l'}_{\vec{K}_\sigma \vec{L}_\rho^\ddag M}(\vec{q}_\sigma,\vec{q}_\rho,-k) \right) \\
&\quad\times \left( \partial^{\vec{w}-\vec{u}-\vec{v}} \partial_\Lambda C^{\Lambda, \Lambda_0}_{MN}(k) \right) \left( \partial^\vec{v} \mathcal{L}^{\Lambda, \Lambda_0, l-l'}_{N \vec{K}_\tau \vec{L}_\varsigma^\ddag}(k,\vec{q}_\tau,\vec{q}_\varsigma) \right) \eqend{,}
\end{splitequation}
with the momentum $k$ defined by equation~\eqref{k_def}, and where $c_{uvw}$ are some constants coming from the Leibniz rule for derivatives. Since for both functionals overall momentum is conserved, the last momentum $q_{m+n}$ is determined in terms of the other $q_i$ and $k$. We then insert the bound on the covariance~\eqref{prop_abl} and the induction hypothesis~\eqref{bound_l0} to obtain
\begin{splitequation}
\abs{F_2} &\leq \sum_{\subline{\sigma \cup \tau = \{1, \ldots, m\} \\ \rho \cup \varsigma = \{1, \ldots, n\}}} \sum_{\vec{u}+\vec{v} \leq \vec{w}} \sup(\abs{k}, \Lambda)^{-5+[\phi_M]+[\phi_N]-\abs{\vec{w}}+\abs{\vec{u}}+\abs{\vec{v}}} \, \mathe^{-\frac{\abs{k}^2}{2 \Lambda^2}} \\
&\quad\times\sum_{T_1 \in \mathcal{T}_{\abs{\sigma}+\abs{\rho}+1}} \mathsf{G}^{T_1,\vec{u}}_{\vec{K}_\sigma \vec{L}_\rho^\ddag M}(\vec{q}_\sigma,\vec{q}_\rho,-k; \mu, \Lambda) \,\mathcal{P}\left( \ln_+ \frac{\sup\left( \abs{\vec{q}_\sigma,\vec{q}_\rho,-k}, \mu \right)}{\sup(\inf(\mu, \eta(\vec{q}_\sigma,\vec{q}_\rho,-k)), \Lambda)}, \ln_+ \frac{\Lambda}{\mu} \right) \\
&\quad\times\sum_{T_2 \in \mathcal{T}_{\abs{\tau}+\abs{\varsigma}+1}} \mathsf{G}^{T_2,\vec{v}}_{N \vec{K}_\tau \vec{L}_\varsigma^\ddag}(k,\vec{q}_\tau,\vec{q}_\varsigma; \mu, \Lambda) \,\mathcal{P}\left( \ln_+ \frac{\sup\left( \abs{k,\vec{q}_\tau,\vec{q}_\varsigma}, \mu \right)}{\sup(\inf(\mu, \eta(k,\vec{q}_\tau,\vec{q}_\varsigma)), \Lambda)}, \ln_+ \frac{\Lambda}{\mu} \right) \eqend{.}
\end{splitequation}
We then estimate
\begin{equations}[func_0op_estlog2]
\abs{\vec{q}_\sigma,\vec{q}_\rho,-k} &\leq \abs{\vec{q}} \eqend{,} \\
\abs{k,\vec{q}_\tau,\vec{q}_\varsigma} &\leq \abs{\vec{q}} \eqend{,} \label{func_0op_estlog2_b} \\
\eta(\vec{q}_\sigma,\vec{q}_\rho,-k) &\geq 0 \eqend{,} \\
\eta(k,\vec{q}_\tau,\vec{q}_\varsigma) & \geq 0 \eqend{,}
\end{equations}
which enables us to merge the polynomials in logarithms. For each tree $T_1 \in \mathcal{T}_{\abs{\sigma}+\abs{\rho}+1}$ and $T_2 \in \mathcal{T}_{\abs{\tau}+\abs{\varsigma}+1}$, we fuse them as detailed in Section~\ref{sec_tree_fuse} to obtain a tree $T \in \mathcal{T}_{m+n}$, and can estimate the weight factors according to the estimate~\eqref{gw_fused_2_est}. This gives the bound
\begin{splitequation}
\abs{F_2} &\leq \sum_{\vec{u}+\vec{v} \leq \vec{w}} \sup(\abs{k}, \Lambda)^{-1-\abs{\vec{w}}+\abs{\vec{u}}+\abs{\vec{v}}} \, \mathe^{-\frac{\abs{k}^2}{2 \Lambda^2}} \\
&\quad\times \sum_{T \in \mathcal{T}_{m+n}} \mathsf{G}^{T,\vec{u}+\vec{v}}_{\vec{K} \vec{L}^\ddag}(\vec{q}; \mu, \Lambda) \, \mathcal{P}\left( \ln_+ \frac{\sup\left( \abs{\vec{q}}, \mu \right)}{\Lambda}, \ln_+ \frac{\Lambda}{\mu} \right) \eqend{.}
\end{splitequation}
The last step is to change the $\vec{u}+\vec{v}$ derivatives acting on the tree to $\vec{w}$ derivatives. Since $\eta_{q_i}(\vec{q}) \leq \abs{k}$ for any $i$, we have
\begin{equation}
\label{func_0op_estfuse_a}
\sup(\abs{k},\Lambda)^{-\abs{\vec{w}}+\abs{\vec{u}}+\abs{\vec{v}}} \leq \prod_{i=1}^{m+n} \sup(\eta_{q_i}(\vec{q}),\Lambda)^{-\abs{(\vec{w}-\vec{u}-\vec{v})_i}}
\end{equation}
and thus (remembering the definition of the derivative weight factor~\eqref{gw_def})
\begin{equation}
\label{func_0op_estfuse_b}
\sup(\abs{k},\Lambda)^{-\abs{\vec{w}}+\abs{\vec{u}}+\abs{\vec{v}}} \mathsf{G}^{T,\vec{u}+\vec{v}}_{\vec{K} \vec{L}^\ddag}(\vec{q}; \mu, \Lambda) \leq \mathsf{G}^{T,\vec{w}}_{\vec{K} \vec{L}^\ddag}(\vec{q}; \mu, \Lambda) \eqend{.}
\end{equation}
The last estimate
\begin{equation}
\label{func_0op_estexp}
\sup(\abs{k}, \Lambda)^{-1} \, \mathe^{-\frac{\abs{k}^2}{2 \Lambda^2}} \leq \sup(\inf(\mu, \eta(\vec{q})), \Lambda)^{-1}
\end{equation}
then gives us the required bound~\eqref{bound_l0_lambdaderiv}.

For functionals with one insertion, the corresponding term of the flow equation~\eqref{l_sop_flow_hierarchy} reads
\begin{splitequation}
\label{func_1op_f2_est}
F_2 \equiv - \sum_{\subline{\sigma \cup \tau = \{1, \ldots, m\} \\ \rho \cup \varsigma = \{1, \ldots, n\}}} \sum_{l'=0}^l \sum_{\vec{u}+\vec{v} \leq \vec{w}} &c_{\sigma\tau\rho\varsigma} c_{uvw} \left( \partial^\vec{u} \mathcal{L}^{\Lambda, \Lambda_0, l'}_{\vec{K}_\sigma \vec{L}_\rho^\ddag M}(\vec{q}_\sigma,\vec{q}_\rho,-k) \right) \\
&\quad\times \left( \partial^{\vec{w}-\vec{u}-\vec{v}} \partial_\Lambda C^{\Lambda, \Lambda_0}_{MN}(k) \right) \left( \partial^\vec{v} \mathcal{L}^{\Lambda, \Lambda_0, l-l'}_{N \vec{K}_\tau \vec{L}_\varsigma^\ddag}\left( \op_A(0); k,\vec{q}_\tau,\vec{q}_\varsigma \right) \right) \eqend{,}
\end{splitequation}
where now overall momentum is not conserved anymore (\ie, $q_{m+n}$ is an independent variable), but $k$ is still given by~\eqref{k_def} and thus momentum derivatives also act on the covariance. Inserting the estimates on the covariance~\eqref{prop_abl} and the induction hypotheses~\eqref{bound_l0} and~\eqref{bound_l1}, and fusing the polynomials in logarithms in the same way as above, we obtain
\begin{splitequation}
\abs{F_2} &\leq \sum_{\subline{\sigma \cup \tau = \{1, \ldots, m\} \\ \rho \cup \varsigma = \{1, \ldots, n\}}} \sum_{l'=0}^l \sum_{\vec{u}+\vec{v} \leq \vec{w}} \sup(\abs{k}, \Lambda)^{-5+[\phi_M]+[\phi_N]-\abs{\vec{w}}+\abs{\vec{u}}+\abs{\vec{v}}} \, \mathe^{-\frac{\abs{k}^2}{2 \Lambda^2}} \\
&\quad\times \sup\left( 1, \frac{\abs{k,\vec{q}_\tau,\vec{q}_\varsigma}}{\sup(\mu, \Lambda)} \right)^{g^{(1)}([\op_A],\abs{\tau}+\abs{\varsigma}+1+2(l-l'),\abs{\vec{v}})} \sum_{T \in \mathcal{T}_{\abs{\sigma}+\abs{\rho}+1}} \mathsf{G}^{T,\vec{u}}_{\vec{K}_\sigma \vec{L}_\rho^\ddag M}(\vec{q}_\sigma,\vec{q}_\rho,-k; \mu, \Lambda) \\
&\quad\times \sum_{T^* \in \mathcal{T}^*_{\abs{\tau}+\abs{\varsigma}+1}} \mathsf{G}^{T^*,\vec{v}}_{N \vec{K}_\tau \vec{L}_\varsigma^\ddag; [\op_A]}(k,\vec{q}_\tau,\vec{q}_\varsigma; \mu, \Lambda) \, \mathcal{P}\left( \ln_+ \frac{\sup\left( \abs{\vec{q}}, \mu \right)}{\Lambda}, \ln_+ \frac{\Lambda}{\mu} \right) \eqend{.}
\end{splitequation}
To estimate the large momentum factor, we use the bound~\eqref{func_0op_estlog2_b}, and then estimate
\begin{equation}
\label{func_1op_est_order}
\abs{\tau}+\abs{\varsigma}+1+2(l-l') \leq m+n+2l-1
\end{equation}
(obvious for $l' > 0$, while for $l' = 0$ we can assume that $\abs{\tau}+\abs{\varsigma} \leq m+n-1$ since otherwise the first functional in~\eqref{func_1op_f2_est}, and thus $F_2$, vanishes). Property~\eqref{gs_prop_1} of $g^{(s)}$ then shows that
\begin{equation}
\sup\left( 1, \frac{\abs{k,\vec{q}_\tau,\vec{q}_\varsigma}}{\sup(\mu, \Lambda)} \right)^{g^{(1)}([\op_A],\abs{\tau}+\abs{\varsigma}+1+2(l-l'),\abs{\vec{v}})} \leq \sup\left( 1, \frac{\abs{\vec{q}}}{\sup(\mu, \Lambda)} \right)^{g^{(1)}([\op_A],m+n+2l,\abs{\vec{w}})} \eqend{.}
\end{equation}
For each tree $T \in \mathcal{T}_{\abs{\sigma}+\abs{\rho}+1}$ and $T^* \in \mathcal{T}^*_{\abs{\tau}+\abs{\varsigma}+1}$, we fuse them as detailed in Section~\ref{sec_tree_fuse} (estimating the weight factors according to~\eqref{gw_fused_2_est}), and then change the $\vec{u}+\vec{v}$ derivatives acting on the fused tree to $\vec{w}$ derivatives in the same way as for functionals without insertions (using equations~\eqref{func_0op_estfuse_a} and~\eqref{func_0op_estfuse_b} with $\bar{\eta}$ instead of $\eta$). Finally, we also use the estimate~\eqref{func_0op_estexp} with $\bar{\eta}$ instead of $\eta$ and obtain a bound of the form~\eqref{bound_l1_lambdaderiv} for $F_2$. For functionals with one integrated insertion, we do the same steps, and arrive at a bound of the form~\eqref{bound_l1i_lambdaderiv}.

\subsection{Irrelevant functionals}
\label{sec_bounds_irr}

Irrelevant functionals are integrated downwards with vanishing boundary conditions at $\Lambda = \Lambda_0$, as specified in Table~\ref{table_boundary}. However, to derive the bounds we may admit non-vanishing boundary conditions as long as they are compatible with the bounds~\eqref{bound_l0,bound_l1,bound_l1i} evaluated at $\Lambda = \Lambda_0$, a freedom which we will need to exploit later on to prove the anomalous Ward identities in Section~\ref{sec_brst}.

For the functionals without insertions, which are irrelevant when $[\vec{K}] + [\vec{L}^\ddag] + \abs{\vec{w}} > 4$, we thus have
\begin{splitequation}
\abs{ \partial^\vec{w} \mathcal{L}^{\Lambda, \Lambda_0, l}_{\vec{K} \vec{L}^\ddag}(\vec{q}) } &\leq \sum_{T \in \mathcal{T}_{m+n}} \mathsf{G}^{T,\vec{w}}_{\vec{K} \vec{L}^\ddag}(\vec{q}; \mu, \Lambda_0) \,\mathcal{P}\left( \ln_+ \frac{\sup\left( \abs{\vec{q}}, \mu \right)}{\Lambda_0}, \ln_+ \frac{\Lambda_0}{\mu} \right) \\
&\quad+ \int_\Lambda^{\Lambda_0} \abs{ \partial_\lambda \partial^\vec{w} \mathcal{L}^{\lambda, \Lambda_0, l}_{\vec{K} \vec{L}^\ddag}(\vec{q}) } \total \lambda \eqend{,}
\end{splitequation}
where the first term is absent for vanishing boundary conditions. We then insert the bound for the $\Lambda$ derivative~\eqref{bound_l0_lambdaderiv} and estimate the trees using the inequality~\eqref{t_irr_ineq2} with $\epsilon = \Delta$ to obtain
\begin{splitequation}
\label{func_0op_irrelevant}
\abs{ \partial^\vec{w} \mathcal{L}^{\Lambda, \Lambda_0, l}_{\vec{K} \vec{L}^\ddag}(\vec{q}) } &\leq \sum_{T \in \mathcal{T}_{m+n}} \mathsf{G}^{T,\vec{w}}_{\vec{K} \vec{L}^\ddag}(\vec{q}; \mu, \Lambda) \Bigg[ \frac{\sup(\inf(\mu, \eta(\vec{q})), \Lambda)^\Delta}{\sup(\inf(\mu, \eta(\vec{q})), \Lambda_0)^\Delta} \,\mathcal{P}\left( \ln_+ \frac{\sup\left( \abs{\vec{q}}, \mu \right)}{\Lambda_0}, \ln_+ \frac{\Lambda_0}{\mu} \right) \\
&\qquad+ \int_\Lambda^{\Lambda_0} \frac{\sup(\inf(\mu, \eta(\vec{q})), \Lambda)^\Delta}{\sup(\inf(\mu, \eta(\vec{q})), \lambda)^{\Delta+1}} \,\mathcal{P}\left( \ln_+ \frac{\sup\left( \abs{\vec{q}}, \mu \right)}{\lambda}, \ln_+ \frac{\lambda}{\mu} \right) \total \lambda \Bigg] \eqend{.}
\end{splitequation}
We then use Lemma~\ref{lemma_largerlog} to obtain the estimate
\begin{splitequation}
\label{func_0op_deltalog}
\left( \frac{\sup(\inf(\mu,\eta(\vec{q})), \Lambda)}{\sup(\inf(\mu,\eta(\vec{q})),\Lambda_0)} \right)^\Delta \ln_+ \frac{\Lambda_0}{\mu} &\leq \left( \frac{\sup(\mu, \Lambda)}{\Lambda_0} \right)^\Delta \ln_+ \left( \frac{\Lambda_0}{\sup(\mu, \Lambda)} \frac{\sup(\mu, \Lambda)}{\mu} \right) \\
&\leq \mathcal{P}\left( \ln_+ \frac{\sup(\mu, \Lambda)}{\mu} \right) \leq \mathcal{P}\left( \ln_+ \frac{\Lambda}{\mu} \right) \eqend{,}
\end{splitequation}
and the subsequent estimate
\begin{equation}
\ln_+ \frac{\sup\left( \abs{\vec{q}}, \mu \right)}{\Lambda_0} \leq \ln_+ \frac{\sup\left( \abs{\vec{q}}, \mu \right)}{\sup(\inf(\mu, \eta(\vec{q})), \Lambda)}
\end{equation}
allows us to bound the first term. The second term can be estimated using Lemma~\ref{lemma_lambdaint}, and we obtain the bounds~\eqref{bound_l0}.

For functionals with one insertion of a composite operator, which are irrelevant when $[\vec{K}] + [\vec{L}^\ddag] + \abs{\vec{w}} > [\op_A]$, the same estimates apply (with $\bar{\eta}$ instead of $\eta$), and we only need to estimate the large momentum factor using
\begin{equation}
\sup\left( 1, \frac{\abs{\vec{q}}}{\sup(\mu, \Lambda_0)} \right)^{g^{(1)}([\op_A],m+n+2l,\abs{\vec{w}})} \leq \sup\left( 1, \frac{\abs{\vec{q}}}{\sup(\mu, \Lambda)} \right)^{g^{(1)}([\op_A],m+n+2l,\abs{\vec{w}})} \eqend{,}
\end{equation}
and then obtain the bounds~\eqref{bound_l1}. Using the same estimates, for functionals with one insertion of an integrated composite operator we obtain the bounds~\eqref{bound_l1i}.

\subsection{Relevant and marginal functionals with arbitrary boundary conditions}

For functionals without insertions, we can choose arbitrary boundary conditions at $\Lambda = \mu$ and vanishing momenta for marginal functionals with $[\vec{K}] + [\vec{L}^\ddag] + \abs{\vec{w}} = 4$. For functionals with an insertion of a composite operator, we have this freedom for relevant and marginal functionals which have $[\vec{K}] + [\vec{L}^\ddag] + \abs{\vec{w}} \leq [\op_A]$, but for functionals with an integrated insertion, we can choose only the conditions for marginal and some relevant functionals, which have $4 \leq [\vec{K}] + [\vec{L}^\ddag] + \abs{\vec{w}} \leq [\op_A]$. The proof of the bounds in these cases first extends the boundary conditions to $\Lambda \geq \mu$ at vanishing momenta using the flow equation, then to arbitrary momenta using Taylor's theorem with integral remainder, and then finally to all $\Lambda < \mu$. Since the integral remainder in Taylor's theorem involves the functionals with one additional momentum derivative, to close the induction we have to bound the marginal functionals first (where the integral remainder involves irrelevant functionals), followed by the least relevant (where the integral remainder involves marginal functionals) and then ascend in the order of relevancy.

\subsubsection{Vanishing momenta}

Starting with the functionals without insertions, we get for $\Lambda \geq \mu$
\begin{equation}
\abs{ \partial^\vec{w} \mathcal{L}^{\Lambda, \Lambda_0, l}_{\vec{K} \vec{L}^\ddag}(\vec{0}) } \leq c + \int_\mu^\Lambda \abs{ \partial_\lambda \partial^\vec{w} \mathcal{L}^{\lambda, \Lambda_0, l}_{\vec{K} \vec{L}^\ddag}(\vec{0}) } \total \lambda \eqend{,}
\end{equation}
where the constant subsumes that arbitrary renormalisation conditions. Inserting the bounds~\eqref{bound_l0_lambdaderiv} (where $\ln_+ \mu/\lambda = 0$ since $\lambda \geq \mu$), we get
\begin{equation}
\abs{ \partial^\vec{w} \mathcal{L}^{\Lambda, \Lambda_0, l}_{\vec{K} \vec{L}^\ddag}(\vec{0}) } \leq c + \int_\mu^\Lambda \frac{1}{\lambda} \sum_{T \in \mathcal{T}_{m+n}} \mathsf{G}^{T,\vec{w}}_{\vec{K} \vec{L}^\ddag}(\vec{0}; \mu, \lambda) \,\mathcal{P}\left( \ln_+ \frac{\lambda}{\mu} \right) \total \lambda = c + \int_\mu^\Lambda \frac{1}{\lambda} \, \mathcal{P}\left( \ln_+ \frac{\lambda}{\mu} \right) \total \lambda \eqend{,}
\end{equation}
since at zero momentum the tree weight factor $\mathsf{G}^{T,\vec{w}}$ is simply given by $\lambda^{[T]}$ (which can be read off from the Table~\ref{table_weights}), where $[T]$ is the dimension of the tree given by~\eqref{t_dim_def}. Marginal trees have $[T] = 0$, and a simple integration yields
\begin{equation}
\label{bound_l0_zeromomentum}
\abs{ \partial^\vec{w} \mathcal{L}^{\Lambda, \Lambda_0, l}_{\vec{K} \vec{L}^\ddag}(\vec{0}) } \leq c + \mathcal{P}\left( \ln_+ \frac{\Lambda}{\mu} \right) \int_\mu^\Lambda \frac{1}{\lambda} \total \lambda = \mathcal{P}\left( \ln_+ \frac{\Lambda}{\mu} \right) \eqend{.}
\end{equation}
The boundary conditions for relevant functionals with one insertion must involve a factor of $\mu^{[\op_A] - [\vec{K}] - [\vec{L}^\ddag] - \abs{\vec{w}}}$ in order to be compatible with the bounds. Since the flow equation~\eqref{l_sop_flow_hierarchy} is linear for $s=1$ insertion of a composite operator, the functionals with one insertion of $\op_A$ form a vector space indexed by $A$ and the imposed boundary conditions. A suitable basis is then given by taking BPHZ-like renormalisation conditions, \ie,
\begin{equation}
\label{func_1op_bphz}
\partial^\vec{w} \mathcal{L}^{\mu, \Lambda_0, l}_{\vec{K} \vec{L}^\ddag}\left( \op_A(0), \vec{0} \right) = (-\mathi)^\abs{\vec{w}} \vec{w}! \delta_{\vec{w},\vec{v}} \delta_{\vec{K} \vec{M}} \delta_{\vec{L}^\ddag \vec{N}^\ddag} \delta_{l,0}
\end{equation}
for the operator $\op_A$ indexed by $A = (\vec{M} \vec{N}^\ddag, \vec{v})$, and in the following we will make this choice (at least for the lowest loop order $l=0$). We then obtain in the same way as before the bounds
\begin{splitequation}
\label{bound_l1_zeromomentum}
\abs{ \partial^\vec{w} \mathcal{L}^{\Lambda, \Lambda_0, l}_{\vec{K} \vec{L}^\ddag}\left( \op_A(0), \vec{0} \right) } &\leq c \mu^{[\op_A] - [\vec{K}] - [\vec{L}^\ddag] - \abs{\vec{w}}} + \mathcal{P}\left( \ln_+ \frac{\Lambda}{\mu} \right) \int_\mu^\Lambda \lambda^{[\op_A] - [\vec{K}] - [\vec{L}^\ddag] - \abs{\vec{w}}-1} \total \lambda \\
&\leq \Lambda^{[\op_A] - [\vec{K}] - [\vec{L}^\ddag] - \abs{\vec{w}}} \, \mathcal{P}\left( \ln_+ \frac{\Lambda}{\mu} \right) \\
&= \sum_{T^* \in \mathcal{T}^*_{m+n}} \mathsf{G}^{T^*,\vec{w}}_{\vec{K} \vec{L}^\ddag; [\op_A]}(\vec{0}; \mu, \Lambda) \, \mathcal{P}\left( \ln_+ \frac{\Lambda}{\mu} \right) \eqend{,}
\end{splitequation}
for $\Lambda \geq \mu$, and the same bound for functionals with one integrated insertion.

\subsubsection{Extension to general momenta}

\begin{figure}
\setbox9=\hbox{\includegraphics[scale=0.9]{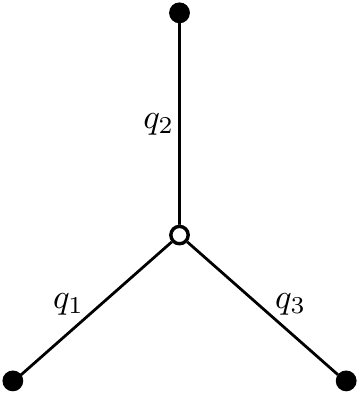}}
\begin{subfigure}[b]{.33\linewidth}
\centering\raisebox{\dimexpr(\ht9-\height)/2}{\includegraphics[scale=0.9]{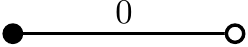}}
\caption{The unique tree $T \in \mathcal{T}_1$.}\label{fig_tree_1pf}
\end{subfigure}%
\begin{subfigure}[b]{.33\linewidth}
\centering\raisebox{\dimexpr(\ht9-\height)/2}{\includegraphics[scale=0.9]{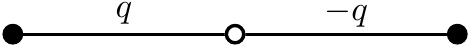}}
\caption{The unique tree $T \in \mathcal{T}_2$.}\label{fig_tree_2pf}
\end{subfigure}%
\begin{subfigure}[b]{.33\linewidth}
\centering\includegraphics[scale=0.9]{fig/tree_3pf}
\caption{The unique tree $T \in \mathcal{T}_3$, with $q_1 + q_2 + q_3 = 0$.}\label{fig_tree_3pf}
\end{subfigure}
\caption{The only fully reduced trees for functionals with one, two and three external legs.}
\end{figure}
\begin{figure}
\includegraphics[scale=0.9]{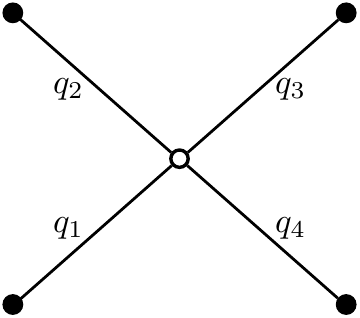}\hfil\includegraphics[scale=0.9]{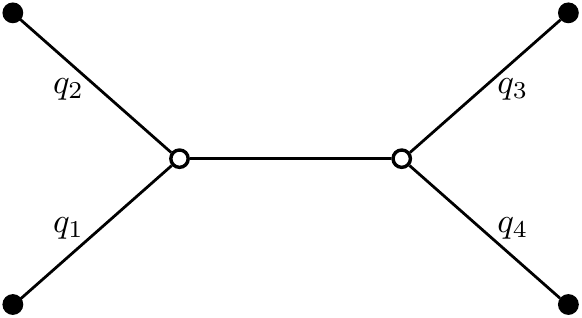}
\caption{The only fully reduced trees for functionals with four external legs, with $q_1+q_2+q_3+q_4 = 0$.}\label{fig_tree_4pf}
\end{figure}

We now extend these bounds to general momenta using the Taylor formula with integral remainder, which reads
\begin{equation}
\label{func_0op_taylor}
\partial^\vec{w} \mathcal{L}^{\Lambda, \Lambda_0, l}_{\vec{K} \vec{L}^\ddag}(\vec{q}) = \partial^\vec{w} \mathcal{L}^{\Lambda, \Lambda_0, l}_{\vec{K} \vec{L}^\ddag}(\vec{0}) + \sum_{i=1}^{m+n-1} \sum_{\alpha=1}^4 \int_0^1 \frac{\partial k_i^\alpha(t)}{\partial t} \partial_{k_i^\alpha} \partial^\vec{w} \mathcal{L}^{\Lambda, \Lambda_0, l}_{\vec{K} \vec{L}^\ddag}(\vec{k}(t)) \total t
\end{equation}
for some path $\vec{k}(t)$ with $\vec{k}(0) = 0$ and $\vec{k}(1) = \vec{q}$. The result is independent of the path taken, since for finite cutoffs all functionals are smooth in momenta, and only depend on $m+n-1$ momenta such that a generalised Stokes theorem holds. This is a freedom we need to exploit. Furthermore, the second term involves a functional which has one momentum derivative more and thus is irrelevant and has already been bounded. Since the functionals with one external leg do not depend on momenta, there is nothing left to do for them, and we just note that the bound~\eqref{bound_l0_zeromomentum} is already the needed one~\eqref{bound_l0}, with the only fully reduced tree $T \in \mathcal{T}_1$ shown in Figure~\ref{fig_tree_1pf}. Inserting the bounds~\eqref{bound_l0_zeromomentum} into the Taylor formula~\eqref{func_0op_taylor} and using the induction hypothesis~\eqref{bound_l0}, we obtain
\begin{splitequation}
\label{func_marginal_taylor}
\abs{ \partial^\vec{w} \mathcal{L}^{\Lambda, \Lambda_0, l}_{\vec{K} \vec{L}^\ddag}(\vec{q}) } &\leq \mathcal{P}\left( \ln_+ \frac{\Lambda}{\mu} \right) + \sum_{i=1}^{m+n-1} \sum_{\alpha=1}^4 \int_0^1 \abs{k_i'(t)} \\
&\qquad\times \sum_{T \in \mathcal{T}_{m+n}} \mathsf{G}^{T,\vec{w}+\tilde{\vec{w}}}_{\vec{K} \vec{L}^\ddag}(\vec{k}(t); \mu, \Lambda) \,\mathcal{P}\left( \ln_+ \frac{\sup\left( \abs{\vec{k}(t)}, \mu \right)}{\sup(\inf(\mu, \eta(\vec{k}(t))), \Lambda)}, \ln_+ \frac{\Lambda}{\mu} \right) \total t \eqend{,}
\end{splitequation}
where $\tilde{\vec{w}}$ is a multiindex which has only one non-zero entry corresponding to the additional $k_i^\alpha$ derivative. For the functionals with two external legs, we take the simple path
\begin{equation}
k(t) = t q
\end{equation}
and estimate the polynomial in logarithms using (since $\Lambda \geq \mu$)
\begin{equation}
\ln_+ \frac{\sup\left( \abs{\vec{k}(t)}, \mu \right)}{\sup(\inf(\mu, \eta(\vec{k}(t))), \Lambda)} \leq \ln_+ \frac{\sup\left( \abs{\vec{q}}, \mu \right)}{\Lambda} = \ln_+ \frac{\sup\left( \abs{\vec{q}}, \mu \right)}{\sup(\inf(\mu, \eta(\vec{q})), \Lambda)} \eqend{.}
\end{equation}
There is only one contributing tree, shown in Figure~\ref{fig_tree_2pf}, which has
\begin{equation}
\label{func_marginal_2pf_tree}
\mathsf{G}^{T,\vec{w}+\tilde{\vec{w}}}_{\vec{K} \vec{L}^\ddag}(\vec{k}(t); \mu, \Lambda) = \sup(\abs{k(t)}, \Lambda)^{-1} = \sup(t \abs{q}, \Lambda)^{-1} \eqend{.}
\end{equation}
We then estimate using Lemma~\ref{lemma_slalom}
\begin{equation}
\int_0^1 \frac{\abs{q}}{\sup(t \abs{q}, \Lambda)} \total t \leq 2 \int_0^1 \frac{\abs{q}}{t \abs{q} + \Lambda} \total t \leq \mathcal{P}\left( \ln_+ \frac{\abs{q}}{\Lambda} \right) \leq \mathcal{P}\left( \ln_+ \frac{\sup\left( \abs{q}, \mu \right)}{\sup(\inf(\mu, \eta(q)), \Lambda)} \right) \eqend{,}
\end{equation}
and so for $\Lambda \geq \mu$ the bounds~\eqref{bound_l0} follow by combining all of the above (the first term in the Taylor formula~\eqref{func_marginal_taylor} combines with the second). For the functionals with three external legs, we extend the summation also to $k_3$ (setting $\eta_{k_3} = \eta$) and may then w.l.o.g. assume $\abs{q_1} \geq \abs{q_2} \geq \abs{q_3}$ (otherwise we just relabel), and that $\abs{q_1} > 0$ (otherwise there is nothing to do). Define then
\begin{equation}
q_2^\bot \equiv q_2 + \alpha q_1 \eqend{,}
\end{equation}
with $\alpha = - (q_1 q_2) / \abs{q_1}^2$, and note that we have $\abs{q_2}^2 = \alpha^2 \abs{q_1}^2 + \abs{q_2^\bot}^2$ and, because of momentum conservation, $q_3 = - q_2^\bot + (\alpha-1) q_1$. Since furthermore $\abs{q_2} \geq \abs{q_3}$, it follows that $\alpha \geq 1/2$. The path that we choose is given by
\begin{equations}
k_1(t) &= \begin{cases} 2 t q_1 & 0 \leq t \leq \frac{1}{2} \\ q_1 & \frac{1}{2} \leq t \leq 1 \eqend{,} \end{cases} \\
k_2(t) &= \begin{cases} - t q_1 + \sqrt{3} \, t \dfrac{\abs{q_1}}{\abs{q_2^\bot}} q_2^\bot & 0 \leq t \leq \frac{1}{2} \\ 2 (1-t) k_2\left( \frac{1}{2} \right) + (2t-1) q_2 & \frac{1}{2} \leq t \leq 1 \eqend{,} \end{cases}
\end{equations}
and always $k_3(t) = - k_1(t) - k_2(t)$. This path is chosen such that on the first half ($0 \leq t \leq 1/2$), we have
\begin{equation}
\abs{k_1(t)} = 2 t \abs{q_1} = \abs{k_2(t)} = \abs{k_3(t)} \eqend{,}
\end{equation}
and on the second half ($1/2 \leq t \leq 1$) we have
\begin{equations}[rel_3pf_ineq]
\frac{1}{2} \abs{q_i} &\leq \abs{k_i(t)} \leq 2 \abs{q_1} \qquad\text{ for } i \in \{2,3\} \eqend{,} \\
2 (1-t) \abs{q_1} &\leq \abs{k_3(t)} \leq \abs{k_2(t)} \eqend{,} \\
\abs{k_2'(t)} &= \abs{k_3'(t)} \leq 2 \abs{q_1} \eqend{.}
\end{equations}
There is also only one contributing tree, which has the weight factor
\begin{equation}
\mathsf{G}^{T,\vec{w}+\tilde{\vec{w}}}_{\vec{K} \vec{L}^\ddag}(\vec{k}(t); \mu, \Lambda) = \frac{F(t)}{\sup(\eta_{k_i}(\vec{k}(t)),\Lambda)}
\end{equation}
with
\begin{equation}
F(t) \equiv \frac{\sup(\abs{\vec{k}(t)},\Lambda)}{\sup(\abs{k_1(t)},\Lambda)^{\alpha_1} \sup(\abs{k_2(t)},\Lambda)^{\alpha_2} \sup(\abs{k_3(t)},\Lambda)^{\alpha_3}}
\end{equation}
for some $\alpha_i \geq 0$ with $\alpha_1 + \alpha_2 + \alpha_3 = 1$ (since all fields have $[\phi_K], [\phi_L^\ddag] \geq 1$). On the first half of the path, all momenta have the same size and we obtain
\begin{equation}
F(t) = 1 \leq \frac{\sup(\abs{q_1},\Lambda)}{\sup(\abs{q_1},\Lambda)^{\alpha_1} \sup(\abs{q_2},\Lambda)^{\alpha_2} \sup(\abs{q_3},\Lambda)^{\alpha_3}} = F(1) = \mathsf{G}^{T,\vec{w}}_{\vec{K} \vec{L}^\ddag}(\vec{q}; \mu, \Lambda)\eqend{,}
\end{equation}
while on the second half we get
\begin{equation}
F(t) \leq \frac{\sup(2 \abs{q_1},\Lambda)}{\sup(\abs{q_1},\Lambda)^{\alpha_1} \sup(1/2 \abs{q_2},\Lambda)^{\alpha_2} \sup(1/2 \abs{q_3},\Lambda)^{\alpha_3}} \leq 2^{1+\alpha_2+\alpha_3} F(1) \eqend{.}
\end{equation}
For the polynomial in logarithms, we estimate
\begin{splitequation}
\ln_+ \frac{\sup\left( \abs{\vec{k}(t)}, \mu \right)}{\sup(\inf(\mu, \eta(\vec{k}(t))), \Lambda)} &= \ln_+ \frac{\sup\left( \abs{\vec{k}(t)}, \mu \right)}{\Lambda} \leq \ln_+ \frac{\sup\left( 2 \abs{q_1}, \mu \right)}{\Lambda} \\
&\leq \ln_+ \frac{\sup\left( \abs{q_1}, \mu \right)}{\Lambda} + 2 \ln 2 = \ln_+ \frac{\sup\left( \abs{\vec{q}}, \mu \right)}{\sup(\inf(\mu, \eta(\vec{q})), \Lambda)} + 2 \ln 2 \eqend{.}
\end{splitequation}
The remaining integral can then be estimated on the first half of the path by
\begin{equation}
\sum_{i=1}^3 \int_0^{1/2} \frac{\abs{k_i'(t)}}{\sup(\eta_{k_i}(\vec{k}(t)),\Lambda)} \total t = \int_0^{1/2} \frac{6 \abs{q_1}}{\sup(2 t \abs{q_1},\Lambda)} \total t \leq \int_0^1 \frac{6 \abs{q_1}}{s \abs{q_1} + \Lambda} \total s
\end{equation}
with the change of variables $s = 2 t$, and on the second half by
\begin{equation}
\sum_{i=1}^3 \int_{1/2}^1 \frac{\abs{k_i'(t)}}{\sup(\eta_{k_i}(\vec{k}(t)),\Lambda)} \total t \leq \int_{1/2}^1 \frac{4 \abs{q_1}}{\sup(2(1-t) \abs{q_1},\Lambda)} \total t \leq \int_0^1 \frac{4 \abs{q_1}}{s \abs{q_1} + \Lambda} \total s
\end{equation}
with the change of variables $s = 2(1-t)$. A subsequent application of Lemma~\ref{lemma_slalom} gives the result
\begin{equation}
\mathcal{P}\left( \ln_+ \frac{\abs{q_1}}{\Lambda} \right) \leq \mathcal{P}\left( \ln_+ \frac{\sup\left( \abs{\vec{q}}, \mu \right)}{\sup(\inf(\mu, \eta(\vec{q})), \Lambda)} \right) \eqend{,}
\end{equation}
and thus the bounds~\eqref{bound_l0} are proven also for this case (for $\Lambda \geq \mu$). Lastly, functionals with four external legs necessarily have external fields of dimension $1$ to be marginal, and no momentum derivatives. The fully reduced trees that contribute are shown in Figure~\ref{fig_tree_4pf}, and the path we take reads
\begin{equation}
k_{\pi(i)}(t) = q_{\pi(i)} \begin{cases} 0 & 0 \leq t \leq \frac{i-1}{3} \\ 3t-i+1 & \frac{i-1}{3} \leq t \leq \frac{i}{3} \\ 1 & \frac{i}{3} \leq t \leq 1 \eqend{,} \end{cases}
\end{equation}
for some permutation $\pi(i)$ of $\{1,2,3\}$ (\ie, on part $i$ of the path the momentum $k_{\pi(i)}$ changes linearly from $0$ to $q_{\pi(i)}$), and we set always $k_4(t) = - k_1(t) - k_2(t) - k_3(t)$. For the first tree of Figure~\ref{fig_tree_4pf}, we have
\begin{equation}
\mathsf{G}^{T,\vec{w}+\tilde{\vec{w}}}_{\vec{K} \vec{L}^\ddag}(\vec{k}(t); \mu, \Lambda) = \frac{1}{\sup(\eta_{k_{\pi(i)}(t)}(\vec{k}(t)),\Lambda)}
\end{equation}
and can take the permutation $\pi(i) = i$. The weight factor of the second tree depends on the exact labelling of the external vertices, and is of the form
\begin{equation}
\mathsf{G}^{T,\vec{w}+\tilde{\vec{w}}}_{\vec{K} \vec{L}^\ddag}(\vec{k}(t); \mu, \Lambda) = \frac{F(t)}{\sup(\eta_{k_{\pi(i)}(t)}(\vec{k}(t)),\Lambda)}
\end{equation}
for
\begin{equation}
F(t) \equiv \sup\left( 1, \frac{\sup(\abs{k_{j_1}(t)}, \abs{k_{j_2}(t)})}{\sup(\abs{k_{j_1}(t) + k_{j_2}(t)},\Lambda)} \right) \sup\left( 1, \frac{\sup(\abs{k_{j_3}(t)}, \abs{k_{j_1}(t) + k_{j_2}(t) + k_{j_3}(t)})}{\sup(\abs{k_{j_1}(t) + k_{j_2}(t)},\Lambda)} \right)
\end{equation}
with $j_i$ a permutation of $\{1,2,3\}$ such that $\abs{q_{j_2}} \leq \abs{q_{j_1}}$. To bound such a factor, we take the permutation $\pi(i) = j_i$ to get
\begin{equation}
F(t) = \begin{cases} 1 & 0 \leq t \leq \frac{1}{3} \\ \sup\left( 1, \dfrac{\sup(\abs{q_{j_1}}, (3t-1) \abs{q_{j_2}})}{\sup(\abs{q_{j_1} + (3t-1) q_{j_2}},\Lambda)} \right) & \frac{1}{3} \leq t \leq \frac{2}{3} \\ \sup\left( 1, \dfrac{\sup(\abs{q_{j_1}}, \abs{q_{j_2}})}{\sup(\abs{q_{j_1} + q_{j_2}},\Lambda)} \right) \sup\left( 1, \dfrac{\sup((3t-2)\abs{q_{j_3}}, \abs{q_{j_1} + q_{j_2} + (3t-2) q_{j_3}})}{\sup(\abs{q_{j_1} + q_{j_2}},\Lambda)} \right) & \frac{2}{3} \leq t \leq 1 \eqend{.} \end{cases}
\end{equation}
For $0 \leq t \leq 1/3$, obviously $F(t) = 1 \leq F(1)$. For the other parts of the path, we need some auxiliary bounds. For $0 \leq s \leq 1$ and $\abs{p} \geq \abs{q}$, define
\begin{equation}
f(p,q,s) \equiv 4 \abs{p + s q}^2 - \abs{p + q}^2 = 3 \abs{p}^2 + (4 s^2-1) \abs{q}^2 + 2 (4s-1) (p q) \eqend{.}
\end{equation}
For $0 \leq s \leq 1/2$, we estimate
\begin{equation}
f(s) \geq 3 \abs{p}^2 + (4 s^2-1) \abs{q}^2 + 2 (4s-1) \abs{p} \abs{q} = 4 \left( \abs{p} + s \abs{q} \right)^2 - \left( \abs{p} + \abs{q} \right)^2 \geq 4 \abs{p}^2 - \left( \abs{p} + \abs{q} \right)^2 \geq 0
\end{equation}
(since $\abs{q} \leq \abs{p}$), while for $1/2 \leq s \leq 1$ we have
\begin{equation}
f(s) \geq 3 \abs{p}^2 + (4 s^2-1) \abs{q}^2 - 2 (4s-1) \abs{p} \abs{q} = 4 \left( \abs{p} - s \abs{q} \right)^2 - \left( \abs{p} - \abs{q} \right)^2 \geq 3 \left( \abs{p} - \abs{q} \right)^2 \geq 0 \eqend{.}
\end{equation}
For $1/3 \leq t \leq 2/3$, we thus have
\begin{equation}
0 \leq f(q_{j_1},q_{j_2},3t-1) = 4 \abs{q_{j_1} + (3t-1) q_{j_2}}^2 - \abs{q_{j_1} + q_{j_2}}^2 \eqend{,}
\end{equation}
and therefore
\begin{equation}
\sup\left( 1, \frac{\sup(\abs{q_{j_1}}, (3t-1) \abs{q_{j_2}})}{\sup(\abs{q_{j_1} + (3t-1) q_{j_2}},\Lambda)} \right) \leq 2 \sup\left( 1, \frac{\sup(\abs{q_{j_1}}, \abs{q_{j_2}})}{\sup(\abs{q_{j_1} + q_{j_2}},\Lambda)} \right) \leq 2 F(1) \eqend{.}
\end{equation}
For the last term, we have
\begin{splitequation}
&\sup\left( 1, \frac{\sup((3t-2) \abs{q_{j_3}}, \abs{q_{j_1} + q_{j_2} + (3t-2) q_{j_3}})}{\sup(\abs{q_{j_1} + q_{j_2}},\Lambda)} \right) \leq \sup\left( 1, \frac{\sup(\abs{q_{j_3}}, \abs{q_{j_1} + q_{j_2}} + \abs{q_{j_3}})}{\sup(\abs{q_{j_1} + q_{j_2}},\Lambda)} \right) \\
&\qquad\leq 2 + \frac{\abs{q_{j_3}}}{\sup(\abs{q_{j_1} + q_{j_2}},\Lambda)} \leq 3 \sup\left( 1, \frac{\abs{q_{j_3}}}{\sup(\abs{q_{j_1} + q_{j_2}},\Lambda)} \right) \eqend{,}
\end{splitequation}
and thus on the whole path $F(t) \leq 3 F(1)$ such that
\begin{equation}
\sum_{T \in \mathcal{T}_4} \mathsf{G}^{T,\vec{w}+\tilde{\vec{w}}}_{\vec{K} \vec{L}^\ddag}(\vec{k}(t); \mu, \Lambda) \leq \frac{3}{\sup(\eta_{k_{\pi(i)}(t)}(\vec{k}(t)),\Lambda)} \sum_{T \in \mathcal{T}_4} \mathsf{G}^{T,\vec{w}}_{\vec{K} \vec{L}^\ddag}(\vec{q}; \mu, \Lambda) \eqend{.}
\end{equation}
Since on the whole path $\abs{k_i(t)} \leq \abs{q_i}$, we can bound the polynomial in logarithms as before, and since on each part of the path only one of the $k_i(t)$ changes, the remaining integral reads
\begin{equation}
\sum_{i=1}^3 \int_0^1 \frac{\abs{k_i'(t)}}{\sup(\eta_{k_i(t)}(\vec{k}(t)),\Lambda)} \total t \leq \sum_{i=1}^3 \sum_{Q \subseteq \{k_1(t),k_2(t),k_3(t)\}\setminus\{k_i(t)\}} \int_0^1 \frac{2 \abs{k_i'(t)}}{\abs{k_i(t) + \sum_{q \in Q} q} + \Lambda} \total t \eqend{.}
\end{equation}
An application of Lemma~\ref{lemma_slalom} to the integral then gives
\begin{equation}
\sum_{i=1}^3 \, \mathcal{P}\left( \ln_+ \frac{\abs{q_i}}{\Lambda} \right) \leq \mathcal{P}\left( \ln_+ \frac{\sup\left( \abs{\vec{q}}, \mu \right)}{\sup(\inf(\mu, \eta(\vec{q})), \Lambda)} \right) \eqend{,}
\end{equation}
and combining all of the above the bounds~\eqref{bound_l0} follow for $\Lambda \geq \mu$.

For functionals with one operator insertion, the Taylor formula reads
\begin{splitequation}
\label{func_1op_taylor}
\partial^\vec{w} \mathcal{L}^{\Lambda, \Lambda_0, l}_{\vec{K} \vec{L}^\ddag}\left( \op_A(0); \vec{q} \right) &= \partial^\vec{w} \mathcal{L}^{\Lambda, \Lambda_0, l}_{\vec{K} \vec{L}^\ddag}\left( \op_A(0); \vec{0} \right) \\
&\quad+ \sum_{i=1}^{m+n} \sum_{\alpha=1}^4 \int_0^1 \frac{\partial k_i^\alpha(t)}{\partial t} \partial_{k_i^\alpha} \partial^\vec{w} \mathcal{L}^{\Lambda, \Lambda_0, l}_{\vec{K} \vec{L}^\ddag}\left( \op_A(0); \vec{k}(t) \right) \total t \eqend{,}
\end{splitequation}
but this time the sum ranges over all momenta (since overall momentum is not conserved anymore), and we can take the simple path $k_i(t) = t q_i$ for all functionals. Inserting the bound~\eqref{bound_l1_zeromomentum} and the induction hypothesis~\eqref{bound_l1}, we obtain
\begin{splitequation}
\abs{ \partial^\vec{w} \mathcal{L}^{\Lambda, \Lambda_0, l}_{\vec{K} \vec{L}^\ddag}\left( \op_A(0); \vec{q} \right) } &\leq \sum_{T^* \in \mathcal{T}^*_{m+n}} \mathsf{G}^{T^*,\vec{w}}_{\vec{K} \vec{L}^\ddag; [\op_A]}(\vec{0}; \mu, \Lambda) \, \mathcal{P}\left( \ln_+ \frac{\Lambda}{\mu} \right) \\
&\quad+ \sum_{i=1}^{m+n} \int_0^1 \frac{\abs{q_i}}{\sup(t \bar{\eta}_{q_i}(\vec{q}), \Lambda)} \sup\left( 1, \frac{t \abs{\vec{q}}}{\Lambda} \right)^{g^{(1)}([\op_A],m+n+2l,\abs{\vec{w}}+1)} \\
&\qquad\quad\times \sum_{T^* \in \mathcal{T}^*_{m+n}} \mathsf{G}^{T^*,\vec{w}}_{\vec{K} \vec{L}^\ddag; [\op_A]}(t \vec{q}; \mu, \Lambda) \, \mathcal{P}\left( \ln_+ \frac{\sup\left( t \abs{\vec{q}}, \mu \right)}{\Lambda}, \ln_+ \frac{\Lambda}{\mu} \right) \total t \eqend{,}
\end{splitequation}
where we extracted the weight factor corresponding to the additional derivative from the tree weight factor. Since the functionals are marginal or relevant, we especially have $\abs{\vec{w}} \leq [\op_A]$ and can thus use property~\eqref{gs_prop_2} of $g^{(s)}$ to estimate
\begin{splitequation}
\frac{\abs{q_i}}{\sup(t \bar{\eta}_{q_i}(\vec{q}), \Lambda)} \sup\left( 1, \frac{t \abs{\vec{q}}}{\Lambda} \right)^{g^{(1)}([\op_A],m+n+2l,\abs{\vec{w}}+1)} &\leq \frac{\abs{\vec{q}}}{\Lambda} \sup\left( 1, \frac{\abs{\vec{q}}}{\Lambda} \right)^{g^{(1)}([\op_A],m+n+2l,\abs{\vec{w}}+1)} \\
&\leq \sup\left( 1, \frac{\abs{\vec{q}}}{\Lambda} \right)^{g^{(1)}([\op_A],m+n+2l,\abs{\vec{w}})} \eqend{.}
\end{splitequation}
The polynomial in logarithms can be estimated trivially at $t = 1$, and since for $\Lambda \geq \mu$ the tree weight factor does not depend on $\mu$, we can use the inequality~\eqref{t_rel_ineq1} to estimate the trees at $t = 1$. The remaining $t$ integral is trivial, and we obtain the bounds~\eqref{bound_l1}. The same estimates apply for functionals with one integrated insertion (only replacing $\bar{\eta}$ by $\eta$), and we obtain correspondingly the bounds~\eqref{bound_l1i}.

\subsubsection{Extension to small \texorpdfstring{$\Lambda$}{Λ}}

In the last step we use these bounds as boundary conditions at $\Lambda = \mu$ and integrate the flow equation downwards. For the functionals without insertions, we have (using the bounds~\eqref{bound_l0} at $\Lambda = \mu$ and~\eqref{bound_l0_lambdaderiv} for the $\Lambda$ derivative)
\begin{splitequation}
\abs{ \partial^\vec{w} \mathcal{L}^{\Lambda, \Lambda_0, l}_{\vec{K} \vec{L}^\ddag}(\vec{q}) } &\leq \abs{ \partial^\vec{w} \mathcal{L}^{\mu, \Lambda_0, l}_{\vec{K} \vec{L}^\ddag}(\vec{q}) } + \int_\Lambda^\mu \abs{\partial_\lambda \partial^\vec{w} \mathcal{L}^{\lambda, \Lambda_0, l}_{\vec{K} \vec{L}^\ddag}(\vec{q}) } \total \lambda \\
&= \sum_{T \in \mathcal{T}_{m+n}} \mathsf{G}^{T,\vec{w}}_{\vec{K} \vec{L}^\ddag}(\vec{q}; \mu, \mu) \,\mathcal{P}\left( \ln_+ \frac{\sup\left( \abs{\vec{q}}, \mu \right)}{\mu} \right) \\
&\quad+ \int_\Lambda^\mu \frac{1}{\sup(\inf(\mu, \eta(\vec{q})), \lambda)} \sum_{T \in \mathcal{T}_{m+n}} \mathsf{G}^{T,\vec{w}}_{\vec{K} \vec{L}^\ddag}(\vec{q}; \mu, \lambda) \,\mathcal{P}\left( \ln_+ \frac{\sup\left( \abs{\vec{q}}, \mu \right)}{\lambda} \right) \total \lambda \eqend{.}
\end{splitequation}
Since all trees are marginal, we use the inequality~\eqref{t_irr_ineq1} to estimate them at $\lambda = \Lambda$. For the first polynomial in logarithms we use
\begin{equation}
\label{func_marginal_polylog}
\ln_+ \frac{\sup\left( \abs{\vec{q}}, \mu \right)}{\mu} = \ln_+ \frac{\sup\left( \abs{\vec{q}}, \mu \right)}{\sup(\mu,\Lambda)} \leq \ln_+ \frac{\sup\left( \abs{\vec{q}}, \mu \right)}{\sup(\inf(\mu, \eta(\vec{q})),\Lambda)} \eqend{,}
\end{equation}
and the remaining $\lambda$ integral can be done using Lemma~\ref{lemma_lambdaint3}, which then gives the needed bound~\eqref{bound_l0} also for $\Lambda < \mu$. For the functionals with one insertion of a composite operator, we have (using the bounds~\eqref{bound_l1} at $\Lambda = \mu$ and~\eqref{bound_l1_lambdaderiv} for the $\Lambda$ derivative)
\begin{splitequation}
&\abs{ \partial^\vec{w} \mathcal{L}^{\Lambda, \Lambda_0, l}_{\vec{K} \vec{L}^\ddag}\left( \op_A(0); \vec{q} \right) } \leq \abs{ \partial^\vec{w} \mathcal{L}^{\mu, \Lambda_0, l}_{\vec{K} \vec{L}^\ddag}\left( \op_A(0); \vec{q} \right) } + \int_\Lambda^\mu \abs{\partial_\lambda \partial^\vec{w} \mathcal{L}^{\lambda, \Lambda_0, l}_{\vec{K} \vec{L}^\ddag}\left( \op_A(0); \vec{q} \right) } \total \lambda \\
&\quad= \sup\left( 1, \frac{\abs{\vec{q}}}{\mu} \right)^{g^{(1)}([\op_A],m+n+2l,\abs{\vec{w}})} \sum_{T^* \in \mathcal{T}^*_{m+n}} \mathsf{G}^{T^*,\vec{w}}_{\vec{K} \vec{L}^\ddag; [\op_A]}(\vec{q}; \mu, \mu) \, \mathcal{P}\left( \ln_+ \frac{\sup\left( \abs{\vec{q}}, \mu \right)}{\mu} \right) \\
&\qquad\quad+ \int_\Lambda^\mu \frac{1}{\sup(\inf(\mu, \bar{\eta}(\vec{q})), \lambda)} \sup\left( 1, \frac{\abs{\vec{q}}}{\mu} \right)^{g^{(1)}([\op_A],m+n+2l,\abs{\vec{w}})} \\
&\qquad\qquad\qquad\times \sum_{T^* \in \mathcal{T}^*_{m+n}} \mathsf{G}^{T^*,\vec{w}}_{\vec{K} \vec{L}^\ddag; [\op_A]}(\vec{q}; \mu, \lambda) \, \mathcal{P}\left( \ln_+ \frac{\sup\left( \abs{\vec{q}}, \mu \right)}{\lambda} \right) \total \lambda \eqend{.}
\end{splitequation}
The first polynomial in logarithms can be estimated using~\eqref{func_marginal_polylog}, and the large momentum factors can be trivially estimated. For the trees we use the inequality~\eqref{t_rel_ineq3}, and the remaining $\lambda$ integral can again be done using Lemma~\ref{lemma_lambdaint3}, which results in the needed bound~\eqref{bound_l1} also for $\Lambda < \mu$. For the functionals with one integrated insertion, we use the same estimates (where the condition $[\vec{K}] + [\vec{L}^\ddag] + \abs{\vec{w}} \geq 4$ is needed in order to use~\eqref{t_rel_ineq3}), and obtain the bounds~\eqref{bound_l1i}.

\subsection{Relevant functionals with vanishing boundary conditions}
\label{sec_bounds_relvan}

For the same reason as in the previous subsection, we have to bound the functionals in increasing order of relevancy. We first derive bounds for arbitrary $\Lambda$ at vanishing momentum, and then extend these bounds to general momenta using Taylor's formula with integral remainder, as in the previous section. For functionals without insertions, we thus have (using the bounds~\eqref{bound_l0_lambdaderiv}
\begin{splitequation}
\abs{ \partial^\vec{w} \mathcal{L}^{\Lambda, \Lambda_0, l}_{\vec{K} \vec{L}^\ddag}(\vec{0}) } &\leq \int_0^\Lambda \abs{ \partial_\lambda \partial^\vec{w} \mathcal{L}^{\lambda, \Lambda_0, l}_{\vec{K} \vec{L}^\ddag}(\vec{0})} \total \lambda \leq \int_0^\Lambda \frac{1}{\lambda} \sum_{T \in \mathcal{T}_{m+n}} \mathsf{G}^{T,\vec{w}}_{\vec{K} \vec{L}^\ddag}(\vec{0}; \mu, \lambda) \,\mathcal{P}\left( \ln_+ \frac{\mu}{\lambda}, \ln_+ \frac{\lambda}{\mu} \right) \total \lambda \\
&= \int_0^\Lambda \sum_{T \in \mathcal{T}_{m+n}} \lambda^{[T]-1} \,\mathcal{P}\left( \ln_+ \frac{\mu}{\lambda}, \ln_+ \frac{\lambda}{\mu} \right) \total \lambda \eqend{.}
\end{splitequation}
Since $[T] > 0$ for relevant trees, we can apply Lemma~\ref{lemma_lambdaint2} with $b = 0$ and obtain
\begin{equation}
\label{relevant_zeromomentum}
\abs{\partial^\vec{w} \mathcal{L}^{\Lambda, \Lambda_0, l}_{\vec{K} \vec{L}^\ddag}(\vec{0})} \leq \sum_{T \in \mathcal{T}_{m+n}} \sup(c,\Lambda)^{[T]} \, \mathcal{P}\left( \ln_+ \frac{\mu}{\sup(c,\Lambda)}, \ln_+ \frac{\Lambda}{\mu} \right)
\end{equation}
for any $c \geq 0$. The functionals with one external leg are independent of momenta, and taking $c = 0$ we already have the correct bound~\eqref{bound_l0}. For the functionals with two or three external legs we extend the bounds to general momenta using the Taylor formula~\eqref{func_0op_taylor}. Inserting the bounds~\eqref{relevant_zeromomentum} and the induction hypothesis~\eqref{bound_l0} in the Taylor formula, we obtain
\begin{splitequation}
\label{func_relevant_taylor}
\abs{ \partial^\vec{w} \mathcal{L}^{\Lambda, \Lambda_0, l}_{\vec{K} \vec{L}^\ddag}(\vec{q}) } &\leq \sum_{T \in \mathcal{T}_{m+n}} \bigg[ \sup(c,\Lambda)^{[T]} \, \mathcal{P}\left( \ln_+ \frac{\mu}{\sup(c,\Lambda)}, \ln_+ \frac{\Lambda}{\mu} \right) \\
&\qquad+ \sum_{i=1}^{m+n-1} \int_0^1 \abs{k_i'(t)} \mathsf{G}^{T,\vec{w}+\tilde{\vec{w}}}_{\vec{K} \vec{L}^\ddag}(\vec{k}(t); \mu, \Lambda) \,\mathcal{P}\left( \ln_+ \frac{\sup\left( \abs{\vec{k}(t)}, \mu \right)}{\sup(\inf(\mu, \eta(\vec{k}(t))), \Lambda)}, \ln_+ \frac{\Lambda}{\mu} \right) \total t \bigg] \eqend{,}
\end{splitequation}
where again $\tilde{\vec{w}}$ is a multiindex with only one non-vanishing entry corresponding to the additional $k_i^\alpha$ derivative. For the functionals with two external legs, the only contributing tree is shown in Figure~\ref{fig_tree_2pf}, and taking $c = \abs{q}$ and the path $k(t) = t q$ we get
\begin{splitequation}
\abs{ \partial^\vec{w} \mathcal{L}^{\Lambda, \Lambda_0, l}_{\vec{K} \vec{L}^\ddag}(\vec{q}) } \leq \sum_{T \in \mathcal{T}_{m+n}} &\bigg[ \sup(\abs{q},\Lambda)^{[T]} \, \mathcal{P}\left( \ln_+ \frac{\mu}{\sup(\abs{q},\Lambda)}, \ln_+ \frac{\Lambda}{\mu} \right) \\
&\quad+ \int_0^1 \abs{q} \sup( t \abs{q}, \Lambda )^{[T]-1} \,\mathcal{P}\left( \ln_+ \frac{\sup\left( t \abs{q}, \mu \right)}{\sup(\inf(\mu, t \abs{q}), \Lambda)}, \ln_+ \frac{\Lambda}{\mu} \right) \total t \bigg] \eqend{.}
\end{splitequation}
The estimate
\begin{equation}
\ln_+ \frac{\mu}{\sup(\abs{q},\Lambda)} \leq \ln_+ \frac{\sup(\abs{q}, \mu)}{\sup(\inf(\mu,\abs{q}),\Lambda)} = \ln_+ \frac{\sup(\abs{q}, \mu)}{\sup(\inf(\mu,\eta(q)),\Lambda)}
\end{equation}
gives the correct bound~\eqref{bound_l0} for the first term, while for the integral we have to change variables to $x = t \abs{q}$ and apply Lemma~\ref{lemma_taylor}. For functionals with three external legs, we again extend the summation to include $k_3$ (with $\eta_{k_3} = \eta$) and take the same path as for marginal functionals. Then we have as before
\begin{equation}
\mathsf{G}^{T,\vec{w}+\tilde{\vec{w}}}_{\vec{K} \vec{L}^\ddag}(\vec{k}(t); \mu, \Lambda) = \frac{1}{\sup(\eta_{k_i}(\vec{k}(t)),\Lambda)} \frac{\sup(\abs{\vec{k}(t)},\Lambda)}{\sup(\abs{k_1(t)},\Lambda)^{\alpha_1} \sup(\abs{k_2(t)},\Lambda)^{\alpha_2} \sup(\abs{k_3(t)},\Lambda)^{\alpha_3}}
\end{equation}
for some $\alpha_i \geq 0$, but now $\alpha_1 + \alpha_2 + \alpha_3 = 1 - [T] < 1$ since we have a relevant functional. On the first half of the path, the integral in~\eqref{func_relevant_taylor} then reduces to
\begin{equation}
6 \int_0^{1/2} \abs{q_1} \sup(2 t \abs{q_1},\Lambda)^{[T]-1} \,\mathcal{P}\left( \ln_+ \frac{\sup\left( 2 t \abs{q_1}, \mu \right)}{\sup(\inf(\mu, 2 t \abs{q_1}), \Lambda)}, \ln_+ \frac{\Lambda}{\mu} \right) \total t \eqend{,}
\end{equation}
and a change of variables $x = 2 t \abs{q_1}$ and an application of Lemma~\ref{lemma_taylor} (which can be used since $[T] > 0$) gives
\begin{equation}
\sup(\abs{q_1},\Lambda)^{[T]} \,\mathcal{P}\left( \ln_+ \frac{\sup\left( \abs{q_1}, \mu \right)}{\sup(\inf(\mu, \abs{q_1}), \Lambda)}, \ln_+ \frac{\Lambda}{\mu} \right) \eqend{.}
\end{equation}
Since $q_1$ is the largest momentum, we have
\begin{equation}
\label{func_relevant_g3}
\sup(\abs{q_1},\Lambda)^{[T]} \leq \frac{\sup(\abs{\vec{q}(t)},\Lambda)}{\sup(\abs{q_1},\Lambda)^{\alpha_1} \sup(\abs{q_2},\Lambda)^{\alpha_2} \sup(\abs{q_3},\Lambda)^{\alpha_3}} = \mathsf{G}^{T,\vec{w}}_{\vec{K} \vec{L}^\ddag}(\vec{q}; \mu, \Lambda) \eqend{,}
\end{equation}
and since $\abs{q_1} \geq \eta(\vec{q})$ we obtain the correct bound~\eqref{bound_l0} for this part. On the second half of the path, we use the inequalities~\eqref{rel_3pf_ineq} to estimate the integral by
\begin{splitequation}
&4 \int_{1/2}^1 \frac{\abs{q_1}}{\sup(\abs{k_3(t)},\Lambda)} \frac{\sup(2\abs{q_1},\Lambda)}{\sup(\abs{q_1},\Lambda)^{\alpha_1} \sup(1/2 \abs{q_2},\Lambda)^{\alpha_2} \sup(1/2 \abs{q_3},\Lambda)^{\alpha_3}} \\
&\hspace{8em}\times \mathcal{P}\left( \ln_+ \frac{\sup\left( 2 \abs{q_1}, \mu \right)}{\sup(\inf(\mu, 1/2 \abs{q_3}), \Lambda)}, \ln_+ \frac{\Lambda}{\mu} \right) \total t \\
&\quad\leq \mathsf{G}^{T,\vec{w}}_{\vec{K} \vec{L}^\ddag}(\vec{q}; \mu, \Lambda) \, \mathcal{P}\left( \ln_+ \frac{\sup\left( \abs{\vec{q}}, \mu \right)}{\sup(\inf(\mu, \eta(\vec{q})), \Lambda)}, \ln_+ \frac{\Lambda}{\mu} \right) \int_{1/2}^1 \frac{\abs{q_1}}{\sup(\abs{k_3(t)},\Lambda)} \total t \eqend{.} \raisetag{1.6\baselineskip}
\end{splitequation}
The subsequent estimate (using again the inequalities~\eqref{rel_3pf_ineq})
\begin{equation}
\frac{\abs{q_1}}{\sup(\abs{k_3(t)},\Lambda)} \leq \frac{4 \abs{q_1}}{3 \abs{k_3(t)} + \Lambda} \leq \frac{4 \abs{q_1}}{2 (1-t) \abs{q_1} + \abs{q_3} + \Lambda} \eqend{,}
\end{equation}
an application of Lemma~\ref{lemma_taylor} and the following
\begin{equation}
\ln_+ \frac{\abs{q_1}}{\abs{q_3} + \Lambda} \leq \ln_+ \frac{\abs{q_1}}{\sup(\abs{q_3},\Lambda)} \leq \ln_+ \frac{\sup(\abs{\vec{q}},\mu)}{\sup(\inf(\mu,\eta(\vec{q})),\Lambda)}
\end{equation}
then also give the bounds~\eqref{bound_l0}. It remains to bound the first term in~\eqref{func_relevant_taylor}, for which we take $c = \abs{q_1}$ and then use the estimates~\eqref{func_relevant_g3} and
\begin{equation}
\ln_+ \frac{\mu}{\sup(\abs{q_1},\Lambda)} \leq \ln_+ \frac{\sup(\abs{\vec{q}},\mu)}{\sup(\inf(\mu,\eta(\vec{q})),\Lambda)}
\end{equation}
to obtain the correct bounds.

For functionals with one integrated insertion, we use the same estimates. The only new term in the bounds that we want to prove is the large momentum factor, which is equal to one at vanishing momenta (and thus we do not need any new estimates in that case), while along the path we take in the Taylor formula we always have $\abs{k_i(t)} \leq 2 \abs{\vec{q}}$ and thus can estimate
\begin{equation}
\sup\left( 1, \frac{\abs{\vec{k}(t)}}{\sup(\mu, \Lambda)} \right) \leq 2 \sup\left( 1, \frac{\abs{\vec{q}}}{\sup(\mu, \Lambda)} \right) \eqend{,}
\end{equation}
such that the bounds~\eqref{bound_l1i} follow.

\subsection{Properties of functionals with one insertion of a composite operator}
\label{sec_bounds_l1iprops}

To prove Proposition~\ref{thm_l1i_d}, we first use the shift property~\eqref{func_sop_shift} to perform the integral over $x$ on the right-hand side of equation~\eqref{func_1iop_int}, which gives
\begin{equation}
\int \mathcal{L}^{\Lambda, \Lambda_0, l}_{\vec{K} \vec{L}^\ddag}\left( \tilde{\op}_A(x); \vec{q} \right) = (2\pi)^4 \delta\left( \sum_{i=1}^{m+n} q_i \right) \mathcal{L}^{\Lambda, \Lambda_0, l}_{\vec{K} \vec{L}^\ddag}\left( \tilde{\op}_A(0); \vec{q} \right) \eqend{.}
\end{equation}
We can thus take the $\Lambda$ derivative of an integral over functionals with one insertion without worrying about convergence issues, and then both sides of equation~\eqref{func_1iop_int} fulfil the same linear flow equation. Since the solution of this equation is unique, we can establish equality by imposing the same boundary conditions. This can be done by evaluating the functional with an insertion of $\int\!\op_A$ at $\Lambda = \mu$ and vanishing momenta, where the functionals with an insertion of $\op_{A,k}$ have their boundary conditions fixed. By ascending in the dimension of the $\op_{A,k}$, we also uniquely determine the required $\op_{A,k}$, and since the boundary conditions for irrelevant functionals are already the same on both sides it is clear that $[\op_{A,k}] \leq [\op_A]$.

To prove Proposition~\ref{thm_lowenstein_1}, we take the boundary conditions given in equation~\eqref{func_1op_bphz} for the functionals with one composite operator insertion. Note that on the left-hand side of equation~\eqref{lowenstein_1} we take derivatives with respect to $x$ of a functional with one operator insertion, while on the right-hand side we have a functional with a different operator insertion, which makes Proposition~\ref{thm_lowenstein_1} non-trivial. However, the flow equations are obviously the same on both sides, and only the boundary conditions need to be compared to prove equality. We then can use the shift property~\eqref{func_sop_shift} to calculate
\begin{splitequation}
\partial^\vec{w} \partial^a_x \mathcal{L}^{\Lambda, \Lambda_0, l}_{\vec{K} \vec{L}^\ddag}\left( \op_A(x); \vec{q} \right) &= \partial^\vec{w} \partial^a_x \left[ \mathe^{- \mathi x \sum_{i=1}^{m+n} q_i} \mathcal{L}^{\Lambda, \Lambda_0, l}_{\vec{K} \vec{L}^\ddag}\left( \op_A(0); \vec{q} \right) \right] \\
&= \partial^\vec{w} \left[ \left( - \mathi \sum_{i=1}^{m+n} q_i \right)^a \mathcal{L}^{\Lambda, \Lambda_0, l}_{\vec{K} \vec{L}^\ddag}\left( \op_A(x); \vec{q} \right) \right] \\
&= \sum_{\vec{u} \leq \vec{w}} \frac{\vec{w}!}{\vec{u}! (\vec{w}-\vec{u})!} \left[ \partial^\vec{u} \left( - \mathi \sum_{i=1}^{m+n} q_i \right)^a \right] \partial^{\vec{w}-\vec{u}} \mathcal{L}^{\Lambda, \Lambda_0, l}_{\vec{K} \vec{L}^\ddag}\left( \op_A(x); \vec{q} \right) \eqend{,}
\end{splitequation}
where we recall that
\begin{equation}
\op_A(x) = \left( \prod_{i=1}^m \partial^{v_i} \phi_{M_i}(x) \right) \left( \prod_{j=1}^n \partial^{v_{m+j}} \phi^\ddag_{N_j}(x) \right) \eqend{.}
\end{equation}
Setting $\Lambda = \mu$, $x = 0$ and $\vec{q} = \vec{0}$ to obtain the boundary conditions for relevant functionals, the right-hand side vanishes except if $u_1 + \cdots + u_{m+n} = a$ and $\vec{w}-\vec{u} = \vec{v}$, where we get
\begin{equation}
\label{lowenstein_bdry1}
\left[ \partial^\vec{w} \partial^a_x \mathcal{L}^{\mu, \Lambda_0, l}_{\vec{K} \vec{L}^\ddag}\left( \op_A(x); \vec{q} \right) \right]_{x = 0, \vec{q} = \vec{0}} = \sum_{u_1 + \cdots + u_{m+n} = a} \frac{\vec{w}!}{\vec{u}!} \left[ \partial^\vec{u} \left( \sum_{i=1}^{m+n} q_i \right)^a \right] (-\mathi)^\abs{\vec{w}} \delta_{\vec{K},\vec{M}} \delta_{\vec{L}^\ddag,\vec{N}^\ddag} \delta_{\vec{w},\vec{v}+\vec{u}} \delta_{l,0} \eqend{.}
\end{equation}
Since $u_1 + \cdots + u_{m+n} = a$, we get
\begin{equation}
\label{lowenstein_bdry2}
\partial^\vec{u} \left( \sum_{i=1}^{m+n} q_i \right)^a = \prod_{\alpha=1}^4 \partial^{u_1^\alpha + \cdots + u_{m+n}^\alpha} \left( \sum_{i=1}^{m+n} q_i \right)^{a^\alpha} = \prod_{\alpha=1}^4 \partial^{a^\alpha} \left( \sum_{i=1}^{m+n} q_i \right)^{a^\alpha} = \prod_{\alpha=1}^4 (a^\alpha)! = a! \eqend{.}
\end{equation}
On the other hand, we have
\begin{splitequation}
\partial^a_x \op_A(x) &= a! \sum_{u_1 + \cdots + u_{m+n} = a} \left( \prod_{i=1}^m \frac{1}{u_i!} \partial^{v_i+u_i} \phi_{M_i}(x) \right) \left( \prod_{j=1}^n \frac{1}{u_{m+j}!} \partial^{v_{m+j}+u_{m+j}} \phi^\ddag_{N_j}(x) \right) \\
&= \sum_{u_1 + \cdots + u_{m+n} = a} \frac{a!}{\vec{u}!} \op_{A;\vec{u}}(x) \eqend{,}
\end{splitequation}
where $\vec{u} = (u_1, \ldots, u_{m+n})$ and $\op_{A;\vec{u}}$ is defined by adding the appropriate $u_i$ derivatives on the fields appearing in $\op_A$. The relevant and marginal functionals with one insertion of $\op_{A;\vec{u}}$ then have the boundary conditions
\begin{equation}
\partial^\vec{w} \mathcal{L}^{\mu, \Lambda_0, l}_{\vec{K} \vec{L}^\ddag}\left( \op_{A;\vec{u}}(0); \vec{0} \right) = (-\mathi)^\abs{\vec{w}} \vec{w}! \delta_{\vec{K},\vec{M}} \delta_{\vec{L}^\ddag,\vec{N}^\ddag} \delta_{\vec{w},\vec{v}+\vec{u}} \delta_{l,0} \eqend{,}
\end{equation}
and since the functionals with one operator insertion are linear in the insertion, we get
\begin{equation}
\partial^\vec{w} \mathcal{L}^{\mu, \Lambda_0, l}_{\vec{K} \vec{L}^\ddag}\left( \partial^a \op_A(0); \vec{0} \right) = \sum_{u_1 + \cdots + u_{m+n} = a} \frac{a!}{\vec{u}!} (-\mathi)^\abs{\vec{w}} \vec{w}! \delta_{\vec{K},\vec{M}} \delta_{\vec{L}^\ddag,\vec{N}^\ddag} \delta_{\vec{w},\vec{v}+\vec{u}} \delta_{l,0} \eqend{.}
\end{equation}
This is thus seen to coincide with Equations~\eqref{lowenstein_bdry1} and~\eqref{lowenstein_bdry2}, and the proposition follows.

\subsection{Functionals with more than one insertion}
\label{sec_bounds_func_sop}

Since all functionals with more than one insertion of a composite operator have vanishing boundary conditions at $\Lambda = \Lambda_0$ (see Table~\ref{table_boundary}), we have to integrate the flow equation~\eqref{l_sop_flow_hierarchy} downwards. However, this is incompatible with the naive scaling of the source term (the last term in the flow equation~\eqref{l_sop_flow_hierarchy}), where the functionals with one operator insertion grow like a positive power of $\Lambda$ as $\Lambda \to \infty$. The way out of this is to generate additional negative powers of $\Lambda$ in the source term, which comes at the expense of derivatives with respect to the points at which the composite operators are inserted, or, if all points are distinct, at the expense of negative powers of the distance between two points. We remind the reader that we always take the last operator insertion to be at $x_s = 0$ (even if we do not show this explicitly in notation), since the general case can be recovered using the shift property~\eqref{func_sop_shift}.

Note that in this case, there is an additional induction step which involves the parameter $D$ in the bounds~\eqref{bound_ls_lambdamu}--\eqref{bound_ls_x}. For a given number $s$ of operator insertions we first have to prove the bounds for the functionals with at least one external leg (\ie, $m+n > 0$) for $D = 0$, which is possible because the right-hand side of the corresponding flow equation~\eqref{l_sop_flow_hierarchy} only involves functionals with at least one external leg. Then the bounds are proven for $m+n > 0$ and $D = 1$, and in the last step for $m+n = 0$ (and $D = 1$).

\subsubsection{Bounds for \texorpdfstring{$\Lambda \geq \mu$}{Λ≥μ}}

To prove equations~\eqref{bound_ls_lambdamu} and~\eqref{bound_ks_lambdamu}, we again first need to bound the right-hand side of the flow equation~\eqref{l_sop_flow_hierarchy}. The first two terms can already be written in the form~\eqref{bound_ls_lambdamu} by the induction, and then we want to prove the bounds
\begin{equation}
\label{bound_ks_lambdamu_lambdaderiv}
\abs{ \partial_\Lambda \partial^\vec{w} \mathcal{K}^{\Lambda, \Lambda_0, l}_{\vec{K} \vec{L}^\ddag; D}\left( \bigotimes_{k=1}^s \op_{A_k}(x_k); \vec{q} \right) } \leq \frac{1}{\Lambda} \abs{ \partial^\vec{w} \mathcal{K}^{\Lambda, \Lambda_0, l}_{\vec{K} \vec{L}^\ddag; D}\left( \bigotimes_{k=1}^s \op_{A_k}(x_k); \vec{q} \right) } \eqend{,}
\end{equation}
with the bound~\eqref{bound_ks_lambdamu} understood on the right-hand side. For the first (linear) term, we thus insert the induction hypothesis~\eqref{bound_ls_lambdamu} and obtain
\begin{equation}
\frac{c}{2} \int \left( \partial_\Lambda C^{\Lambda, \Lambda_0}_{MN}(-p) \right) \partial^\vec{w} \mathcal{L}^{\Lambda, \Lambda_0, l-1}_{MN \vec{K} \vec{L}^\ddag}\left( \bigotimes_{k=1}^s \op_{A_k}; p,-p,\vec{q} \right) \frac{\total^4 p}{(2\pi)^4} = \sum_{\vec{a} > 0, \abs{\vec{a}} = D+[\op_\vec{A}]} \partial^\vec{a}_{\vec{x}} F_1
\end{equation}
with
\begin{equation}
F_1 \equiv \frac{c}{2} \int \left( \partial_\Lambda C^{\Lambda, \Lambda_0}_{MN}(-p) \right) \partial^\vec{w} \mathcal{K}^{\Lambda, \Lambda_0, l-1}_{MN \vec{K} \vec{L}^\ddag; D}\left( \bigotimes_{k=1}^s \op_{A_k}; p,-p,\vec{q} \right) \frac{\total^4 p}{(2\pi)^4} \eqend{.}
\end{equation}
We then insert the induction hypothesis~\eqref{bound_ks_lambdamu} for $\mathcal{K}$ and the bounds on the covariance~\eqref{prop_abl} to get
\begin{splitequation}
\abs{F_1} &\leq \prod_{i=1}^{s-1} \left( 1 + \ln_+ \frac{1}{\mu \abs{x_i}} \right) \int \sup(\abs{p}, \Lambda)^{-5+[\phi_M]+[\phi_N]} \, \mathe^{-\frac{\abs{p}^2}{2 \Lambda^2}} \sup\left( 1, \frac{\abs{\vec{q},p,-p}}{\Lambda} \right)^{g^{(s)}([\op_\vec{A}],m+n+2l,\abs{\vec{w}})} \\
&\qquad\times \sum_{T^* \in \mathcal{T}^*_{m+n+2}} \mathsf{G}^{T^*,\vec{w}}_{MN \vec{K} \vec{L}; -D}(p,-p,\vec{q}; \mu, \Lambda) \, \mathcal{P}\left( \ln_+ \frac{\sup\left( \abs{\vec{q},p,-p}, \mu \right)}{\Lambda}, \ln_+ \frac{\Lambda}{\mu} \right) \frac{\total^4 p}{(2\pi)^4} \eqend{.}
\end{splitequation}
After using the estimates~\eqref{func_0op_estlog}, we rescale $p = x \Lambda$ and perform the $p$ integral applying Lemma~\ref{lemma_pint2}. A subsequent amputation of the external legs corresponding to $M$ and $N$ from each tree (which gives an additional factor according to the estimate~\eqref{amputate}) then gives a bound of the form~\eqref{bound_ks_lambdamu_lambdaderiv} for $F_1$. For the second (quadratic) term, we get after inserting the induction hypothesis~\eqref{bound_ls_lambdamu}
\begin{equation}
\sum_{\vec{a} > 0, \abs{\vec{a}} = D+[\op_\vec{A}]} \partial^\vec{a}_{\vec{x}} F_2
\end{equation}
with
\begin{splitequation}
F_2 \equiv - \sum_{\subline{\sigma \cup \tau = \{1, \ldots, m\} \\ \rho \cup \varsigma = \{1, \ldots, n\} }} \sum_{l'=0}^l \sum_{\vec{u}+\vec{v}\leq \vec{w}} &c_{\sigma\tau\rho\varsigma} c_{uvw} \left( \partial^\vec{u} \mathcal{L}^{\Lambda, \Lambda_0, l'}_{\vec{K}_\sigma \vec{L}_\rho^\ddag M}(\vec{q}_\sigma,\vec{q}_\rho,-k) \right) \left( \partial^{\vec{w}-\vec{u}-\vec{v}} \partial_\Lambda C^{\Lambda, \Lambda_0}_{MN}(k) \right) \\
&\quad\times \left[ \partial^\vec{v} \mathcal{K}^{\Lambda, \Lambda_0, l-l'}_{N \vec{K}_\tau \vec{L}_\varsigma^\ddag}\left( \bigotimes_{k=1}^s \op_{A_k}; k,\vec{q}_\tau,\vec{q}_\varsigma \right) \right] \eqend{,}
\end{splitequation}
where the momentum $k$ is defined by equation~\eqref{k_def}. Inserting the bound~\eqref{bound_l0} for the functional without insertions, the induction hypothesis~\eqref{bound_ks_lambdamu} for $\mathcal{K}$ and the bounds on the covariance~\eqref{prop_abl} it follows that
\begin{splitequation}
\abs{F_2} &\leq \sum_{\subline{\sigma \cup \tau = \{1, \ldots, m\} \\ \rho \cup \varsigma = \{1, \ldots, n\} }} \sum_{l'=0}^l \sum_{\vec{u}+\vec{v}\leq \vec{w}} \sum_{T \in \mathcal{T}_{\abs{\sigma}+\abs{\rho}+1}} \mathsf{G}^{T,\vec{u}}_{\vec{K}_\sigma \vec{L}_\rho^\ddag M}(\vec{q}_\sigma,\vec{q}_\rho,-k; \mu, \Lambda) \, \mathe^{-\frac{\abs{k}^2}{2 \Lambda^2}} \\
&\quad\times \sup(\abs{k}, \Lambda)^{-5+[\phi_M]+[\phi_N]-\abs{\vec{w}}+\abs{\vec{u}}+\abs{\vec{v}}} \, \mathcal{P}\left( \ln_+ \frac{\sup\left( \abs{\vec{q}_\sigma,\vec{q}_\rho,-k}, \mu \right)}{\Lambda}, \ln_+ \frac{\Lambda}{\mu} \right) \\
&\quad\times \prod_{i=1}^{s-1} \left( 1 + \ln_+ \frac{1}{\mu \abs{x_i}} \right) \sup\left( 1, \frac{\abs{k,\vec{q}_\tau,\vec{q}_\varsigma}}{\Lambda} \right)^{g^{(s)}([\op_\vec{A}],\abs{\tau}+\abs{\varsigma}+1+2(l-l'),\abs{\vec{v}})} \\
&\quad\times \sum_{T^* \in \mathcal{T}^*_{\abs{\tau}+\abs{\varsigma}+1}} \mathsf{G}^{T^*,\vec{v}}_{\vec{K} \vec{L}; -D}(k,\vec{q}_\tau,\vec{q}_\varsigma; \mu, \Lambda) \, \mathcal{P}\left( \ln_+ \frac{\sup\left( \abs{k,\vec{q}_\tau,\vec{q}_\varsigma}, \mu \right)}{\Lambda}, \ln_+ \frac{\Lambda}{\mu} \right) \eqend{.} \raisetag{1.9\baselineskip}
\end{splitequation}
Using the estimates~\eqref{func_0op_estlog2} and~\eqref{func_0op_estlog} we can fuse the polynomials in logarithms. For the large momentum factor, we use~\eqref{func_0op_estlog2_b} and~\eqref{func_1op_est_order}. The trees are fused as detailed in Section~\ref{sec_tree_fuse}, and we obtain an extra factor of $\sup(\abs{k}, \Lambda)^{4-[\phi_M]-[\phi_N]}$ according to the estimate~\eqref{gw_fused_2_est}. We then change the $\vec{u}+\vec{v}$ derivatives acting on the fused tree to $\vec{w}$ derivatives in the same way as for functionals without insertions (using equations~\eqref{func_0op_estfuse_a} and~\eqref{func_0op_estfuse_b} with $\bar{\eta}$ instead of $\eta$), such that taking everything together a bound of the form~\eqref{bound_ks_lambdamu_lambdaderiv} is achieved also for $F_2$.

For the source term in the flow equation~\eqref{l_sop_flow_hierarchy}, we have to distinguish three cases: first, both functionals contain only one operator insertion (this obviously can only happen if the functional on the left-hand side has exactly two operator insertions); second, one of the functionals contains one operator insertion and the other one contains at least two; and third, both functionals contain at least two operator insertions. Let us start with the first case, where the source term reads
\begin{splitequation}
F_3 \equiv - \sum_{\subline{\sigma \cup \tau = \{1, \ldots, m\} \\ \rho \cup \varsigma = \{1, \ldots, n\} }} &\sum_{l'=0}^l \sum_{\vec{u}\leq \vec{w}} c_{\sigma\tau\rho\varsigma} c_{uw} \int \left( \partial^\vec{u} \mathcal{L}^{\Lambda, \Lambda_0, l'}_{\vec{K}_\sigma \vec{L}_\rho^\ddag M}\left( \op_{A_1}(x_1); \vec{q}_\sigma,\vec{q}_\rho,p \right) \right) \\
&\qquad\times \left( \partial_\Lambda C^{\Lambda, \Lambda_0}_{MN}(-p) \right) \left( \partial^{\vec{w}-\vec{u}} \mathcal{L}^{\Lambda, \Lambda_0, l-l'}_{N \vec{K}_\tau \vec{L}_\varsigma^\ddag}\left( \op_{A_2}(0);-p,\vec{q}_\tau,\vec{q}_\varsigma \right) \right) \frac{\total^4 p}{(2\pi)^4} \eqend{.}
\end{splitequation}
While the second functional has an operator insertion at $x_2 = 0$ and can thus be bounded using the bound~\eqref{bound_l1}, for the first functional we have to use the shift property~\eqref{func_sop_shift} which gives an extra factor
\begin{equation}
\exp\left( - \mathi x_1 (p+k) \right) = \prod_{\alpha=1}^4 \mathcal{E}_0\left( x_1^\alpha (p+k)^\alpha \right) \eqend{,}
\end{equation}
with the function $\mathcal{E}_0$ defined in equation~\eqref{expint_func_def} and with $k$ defined in equation~\eqref{k_def}. Since the property~\eqref{expint_abl} is equivalent to
\begin{equation}
\mathcal{E}_k\left( x_1^\alpha (p+k)^\alpha \right) = - \mathi \partial_{p^\alpha} \partial_{x_1^\alpha} \mathcal{E}_{k+1}\left( x_1^\alpha (p+k)^\alpha \right)
\end{equation}
for any direction $\alpha \in \{1,2,3,4\}$, we choose a direction $\alpha$ such that $\abs{x_1^\alpha} \geq \abs{x_1}/2$ (which is always possible), and take the multiindex $a = (a^1, a^2, a^3, a^4)$ with $a^\beta = \abs{a} \delta^\beta_\alpha$ to obtain
\begin{equation}
\label{func_sop_genderivs}
\exp\left( - \mathi x_1 (p+k) \right) = \prod_{\alpha=1}^4 \left( - \mathi \partial_{p^\alpha} \partial_{x_1^\alpha} \right)^{a^\alpha} \mathcal{E}_{a^\alpha}\left( x_1^\alpha (p+k)^\alpha \right) = (-\mathi)^\abs{a} \partial_{x_1}^a \partial_p^a \prod_{\alpha=1}^4 \mathcal{E}_{a^\alpha}\left( x_1^\alpha (p+k)^\alpha \right) \eqend{.}
\end{equation}
Taking $\abs{a} = [\op_{A_1}] + [\op_{A_2}] + D$, we thus have obtained the necessary $x$ derivatives which can be taken out of the integral to obtain
\begin{equation}
F_3 = \partial_{x_1}^a F_4 \eqend{,}
\end{equation}
where $F_4$ denotes all the remaining terms for which we need to prove the bounds~\eqref{bound_ks_lambdamu_lambdaderiv}. We integrate the $p$ derivatives by parts (which does not give boundary terms because the covariance is rapidly decreasing as $\abs{p} \to \infty$~\eqref{r_prop_bound}) and insert the bound on the covariance~\eqref{prop_abl}, for the functionals with one insertion~\eqref{bound_l1} and for the function $\mathcal{E}_k$~\eqref{expint_bounds} to obtain (remember that $\Lambda \geq \mu$)
\begin{splitequation}
\abs{F_4} &\leq \sum_{\subline{\sigma \cup \tau = \{1, \ldots, m\} \\ \rho \cup \varsigma = \{1, \ldots, n\} }} \sum_{l'=0}^l \sum_{\vec{u} \leq \vec{w}} \sum_{v+v'\leq a} \int \left( 1 + \ln_+ \frac{1}{\abs{x_1^\alpha (p+k)^\alpha}} \right) \sup(\abs{p}, \Lambda)^{-5+[\phi_M]+[\phi_N]-\abs{a}+\abs{v}+\abs{v'}} \\
&\qquad\times \mathe^{-\frac{\abs{p}^2}{2 \Lambda^2}} \sup\left( 1, \frac{\abs{\vec{q}_\sigma,\vec{q}_\rho,p}}{\Lambda} \right)^{g^{(1)}([\op_{A_1}],\abs{\sigma}+\abs{\rho}+1+2l',\abs{\vec{u}}+\abs{v})} \\
&\qquad\times \sup\left( 1, \frac{\abs{-p,\vec{q}_\tau,\vec{q}_\varsigma}}{\Lambda} \right)^{g^{(1)}([\op_{A_2}],\abs{\tau}+\abs{\varsigma}+1+2(l-l'),\abs{\vec{w}}-\abs{\vec{u}}+\abs{v'})} \\
&\qquad\times \sum_{T^* \in \mathcal{T}^*_{\abs{\sigma}+\abs{\rho}+1}} \mathsf{G}^{T^*,\vec{u}+(\vec{0},v)}_{\vec{K}_\sigma \vec{L}_\rho^\ddag M; [\op_{A_1}]}(\vec{q}_\sigma, \vec{q}_\rho, p; \mu, \Lambda) \sum_{T^* \in \mathcal{T}^*_{\abs{\tau}+\abs{\varsigma}+1}} \mathsf{G}^{T^*,\vec{w}-\vec{u}+(v',\vec{0})}_{N \vec{K}_\tau \vec{L}_\varsigma^\ddag; [\op_{A_2}]}(-p, \vec{q}_\tau, \vec{q}_\varsigma; \mu, \Lambda) \\
&\qquad\times \mathcal{P}\left( \ln_+ \frac{\sup\left( \abs{\vec{q}_\sigma,\vec{q}_\rho,p}, \mu \right)}{\Lambda}, \ln_+ \frac{\Lambda}{\mu} \right) \mathcal{P}\left( \ln_+ \frac{\sup\left( \abs{-p,\vec{q}_\tau,\vec{q}_\varsigma}, \mu \right)}{\Lambda}, \ln_+ \frac{\Lambda}{\mu} \right) \frac{\total^4 p}{(2\pi)^4} \eqend{.} \raisetag{1.4\baselineskip}
\end{splitequation}
Using that
\begin{equation}
\label{func_2op_estlog}
\abs{\vec{q}_\sigma,\vec{q}_\rho,p} \leq \abs{\vec{q},p,-p} \eqend{,} \qquad \abs{-p,\vec{q}_\tau,\vec{q}_\varsigma} \leq \abs{\vec{q},p,-p}
\end{equation}
and the estimates~\eqref{func_0op_estlog}, we can fuse the polynomials in logarithms and the large momentum factors, and then we use property~\eqref{gs_prop_3} of $g^{(s)}$ to estimate
\begin{splitequation}
\label{func_sop_g1g1est}
&g^{(1)}([\op_{A_1}],\abs{\sigma}+\abs{\rho}+1+2l',\abs{\vec{u}}+\abs{v}) + g^{(1)}([\op_{A_2}],\abs{\tau}+\abs{\varsigma}+1+2(l-l'),\abs{\vec{w}}-\abs{\vec{u}}+\abs{v'}) \\
&\quad\leq g^{(2)}([\op_{A_1}]+[\op_{A_2}],m+n+2l,\abs{\vec{w}}) - ([\op_{A_1}]+[\op_{A_2}]+D) \eqend{.}
\end{splitequation}
The trees are fused according to the estimate~\eqref{gw_fused_1_est}, and for the logarithm coming from the $\mathcal{E}_k$ functions we use the estimate
\begin{equation}
\label{func_sop_logest}
\ln_+ \frac{1}{\abs{x_1^\alpha (p+k)^\alpha}} \leq \ln_+ \frac{1}{\mu \abs{x_1^\alpha}} + \ln_+ \frac{\Lambda}{\abs{(p+k)^\alpha}} \eqend{,}
\end{equation}
which then gives
\begin{splitequation}
\abs{F_4} \leq \sum_{v+v'\leq a} \int &\left( 1 + \ln_+ \frac{1}{\mu \abs{x_1^\alpha}} + \ln_+ \frac{\Lambda}{\abs{(p+k)^\alpha}} \right) \, \mathe^{-\frac{\abs{p}^2}{2 \Lambda^2}} \sup(\abs{p}, \Lambda)^{-5+[\phi_M]+[\phi_N]-\abs{a}+\abs{v}+\abs{v'}} \\
&\quad\times \sup\left( 1, \frac{\abs{\vec{q},p,-p}}{\Lambda} \right)^{g^{(2)}([\op_{A_1}]+[\op_{A_2}],m+n+2l,\abs{\vec{w}}) - ([\op_{A_1}]+[\op_{A_2}]+D)} \\
&\quad\times \sum_{T^* \in \mathcal{T}^*_{m+n+2}} \mathsf{G}^{T^*,\vec{w}+(\vec{0},v,v')}_{\vec{K} \vec{L}^\ddag M N; [\op_{A_1}]+[\op_{A_2}]}(\vec{q}, p, -p; \mu, \Lambda) \\
&\quad\times \mathcal{P}\left( \ln_+ \frac{\sup\left( \abs{\vec{q}}, \mu \right)}{\Lambda}, \ln_+ \frac{\abs{p}}{\Lambda}, \ln_+ \frac{\Lambda}{\mu} \right) \frac{\total^4 p}{(2\pi)^4} \eqend{.}
\end{splitequation}
The $p$ integral can now be done using Lemma~\ref{lemma_pint3}, and we obtain the estimate
\begin{splitequation}
\abs{F_4} &\leq \left( 1 + \ln_+ \frac{1}{\mu \abs{x_1^\alpha}} \right) \sum_{v+v' \leq a} \Lambda^{-1+[\phi_M]+[\phi_N]-[\op_{A_1}]-[\op_{A_2}]-D-\abs{a}+\abs{v}+\abs{v'}} \\
&\quad\times \sum_{T^* \in \mathcal{T}^*_{m+n+2}} \mathsf{G}^{T^*,\vec{w}+(\vec{0},v,v')}_{\vec{K} \vec{L}^\ddag M N; [\op_{A_1}]+[\op_{A_2}]}(\vec{q}, 0, 0; \mu, \Lambda) \, \mathcal{P}\left( \ln_+ \frac{\sup\left( \abs{\vec{q}}, \mu \right)}{\Lambda}, \ln_+ \frac{\Lambda}{\mu} \right) \\
&\qquad\times \sup\left( 1, \frac{\abs{\vec{q}}}{\Lambda} \right)^{g^{(2)}([\op_{A_1}]+[\op_{A_2}],m+n+2l,\abs{\vec{w}}) - [\op_{A_1}]-[\op_{A_2}]-D} \eqend{.}
\end{splitequation}
Since $\abs{x_1^\alpha} \geq \abs{x_1}/2$ by our choice of $\alpha$, we can estimate
\begin{equation}
\left( 1 + \ln_+ \frac{1}{\mu \abs{x_1^\alpha}} \right) \leq 1 + \ln_+ \frac{1}{\mu \abs{x_1^\alpha}} + \ln_+ 2 \leq 2 \left( 1 + \ln_+ \frac{1}{\mu \abs{x_1}} \right) \eqend{.}
\end{equation}
We then remove the derivative weight factor corresponding to the $v+v'$ derivatives that acted on the momentum $p$, which gives a factor~\eqref{gw_def} $\Lambda^{-\abs{v}-\abs{v'}}$ (since the momentum is now zero), and amputating then the external legs corresponding to $M$ and $N$ from the tree, we obtain factors according to the estimate~\eqref{amputate}, which (since $\bar{\eta}(\vec{q},0,0) = 0$) gives $\Lambda^{-[\phi_M]-[\phi_N]}$. We furthermore change the particular dimension associated to the tree from $[\op_{A_1}]+[\op_{A_2}]$ to $-D$, which according to~\eqref{particular_weight} gives an extra factor
\begin{equation}
\label{func_sop_treedim_est}
\sup(\abs{\vec{q}}, \Lambda)^{[\op_{A_1}]+[\op_{A_2}]+D} = \Lambda^{[\op_{A_1}]+[\op_{A_2}]+D} \sup\left( 1, \frac{\abs{\vec{q}}}{\Lambda} \right)^{[\op_{A_1}]+[\op_{A_2}]+D} \eqend{,}
\end{equation}
and then obtain a bound of the form~\eqref{bound_ks_lambdamu_lambdaderiv} for $F_4$.

The second case, where one functional contains $s' > 1$ insertions (of $\op_\vec{A}$), and the other one only one insertion of $\op_B$, is estimated in a similar way. We use the shift property~\eqref{func_sop_shift} to bring the last insertion of the first functional to the origin, and can then use the representation~\eqref{bound_ls_lambdamu} for the functional with more than one insertion, which already gives some of the necessary $x$ derivatives. We then perform the same estimates, using equation~\eqref{func_sop_genderivs} to generate the remaining $[\op_B]$ derivatives (note that the first functional has a lower number of insertions, and has thus already been estimated using the correct parameter $D$). Instead of equation~\eqref{func_sop_g1g1est}, property~\eqref{gs_prop_3} of $g^{(s)}$ gives
\begin{splitequation}
&g^{(s')}([\op_\vec{A}],\abs{\sigma}+\abs{\rho}+1+2l',\abs{\vec{u}}) + g^{(1)}([\op_B],\abs{\tau}+\abs{\varsigma}+1+2(l-l'),\abs{\vec{v}}) \\
&\quad\leq g^{(s'+1)}([\op_\vec{A}]+[\op_B],m+n+2l,\abs{\vec{w}}) - ([\op_\vec{A}]+[\op_B]+s'+1) \\
&\quad\leq g^{(s'+1)}([\op_\vec{A}]+[\op_B],m+n+2l,\abs{\vec{w}}) - [\op_B] \eqend{,}
\end{splitequation}
which is enough to bring the power of the particular weight factor of the fused tree from $-D+[\op_B]$ to $-D$, and cancel the powers of $\Lambda$ coming from the newly generated $x$ derivatives. The $p$ integral and the amputation are done as before, and we obtain a bound of the form~\eqref{bound_ks_lambdamu_lambdaderiv} as required. In the last case, where both functionals contain at least two insertions (of $s'$ operators $\op_\vec{A}$ and $s''$ operators $\op_\vec{B}$), we again use the shift property~\eqref{func_sop_shift} to bring the last operator insertion of the first functional to the origin, and can then use the representation~\eqref{bound_ls_lambdamu}, which already gives the correct number of derivatives. In this case we use property~\eqref{gs_prop_3} of $g^{(s)}$ to obtain
\begin{splitequation}
&g^{(s')}([\op_\vec{A}],\abs{\sigma}+\abs{\rho}+1+2l',\abs{\vec{u}}) + g^{(s'')}([\op_\vec{B}],\abs{\tau}+\abs{\varsigma}+1+2(l-l'),\abs{\vec{v}}) \\
&\qquad\leq g^{(s'+s'')}([\op_\vec{A}]+[\op_\vec{B}],m+n+2l,\abs{\vec{w}}) \eqend{.}
\end{splitequation}
For $D = 0$, the fused tree has then already the correct particular weight factor, while for $D = 1$ we need to insert the bounds with $D = 0$ for one of the source terms to obtain the correct particular weight factor, and thus need to ascend in $D$ as stated in the introduction. Everything else is done as before, and we obtain a bound of the form~\eqref{bound_ks_lambdamu_lambdaderiv} also in this case.

Since all functionals with more than one insertion are irrelevant with vanishing boundary conditions at $\Lambda = \Lambda_0$, we now simply integrate the bound~\eqref{bound_ks_lambdamu_lambdaderiv} over $\Lambda$ to obtain
\begin{splitequation}
\label{func_sop_lambdamu_int}
\abs{ \partial^\vec{w} \mathcal{K}^{\Lambda, \Lambda_0, l}_{\vec{K} \vec{L}^\ddag; D}\left( \bigotimes_{k=1}^s \op_{A_k}(x_k); \vec{q} \right) } &\leq \int_\Lambda^{\Lambda_0} \abs{ \partial^\vec{w} \partial_\lambda \mathcal{K}^{\lambda, \Lambda_0, l}_{\vec{K} \vec{L}^\ddag; D}\left( \bigotimes_{k=1}^s \op_{A_k}(x_k); \vec{q} \right) } \total \lambda \\
&\leq \prod_{i=1}^{s-1} \left( 1 + \ln_+ \frac{1}{\mu \abs{x_i}} \right) \int_\Lambda^{\Lambda_0} \frac{1}{\lambda} \sup\left( 1, \frac{\abs{\vec{q}}}{\lambda} \right)^{g^{(s)}([\op_\vec{A}],m+n+2l,\abs{\vec{w}})} \\
&\quad\times \sum_{T^* \in \mathcal{T}^*_{m+n}} \mathsf{G}^{T^*,\vec{w}}_{\vec{K} \vec{L}^\ddag; -D}(\vec{q}; \mu, \lambda) \, \mathcal{P}\left( \ln_+ \frac{\sup\left( \abs{\vec{q}}, \mu \right)}{\lambda}, \ln_+ \frac{\lambda}{\mu} \right) \total \lambda \eqend{.}
\end{splitequation}
The tree can be estimated at the lower bound $\lambda = \Lambda$ using the inequality~\eqref{t_irr_ineq2} with $\epsilon = 1$, and the large momentum factor is trivially estimated there, such that we get
\begin{splitequation}
&\abs{ \partial^\vec{w} \mathcal{K}^{\Lambda, \Lambda_0, l}_{\vec{K} \vec{L}^\ddag; D}\left( \bigotimes_{k=1}^s \op_{A_k}(x_k); \vec{q} \right) } \leq \prod_{i=1}^{s-1} \left( 1 + \ln_+ \frac{1}{\mu \abs{x_i}} \right) \sup\left( 1, \frac{\abs{\vec{q}}}{\Lambda} \right)^{g^{(s)}([\op_\vec{A}],m+n+2l,\abs{\vec{w}})} \\
&\hspace{10em}\times \sum_{T^* \in \mathcal{T}^*_{m+n}} \mathsf{G}^{T^*,\vec{w}}_{\vec{K} \vec{L}^\ddag; -D}(\vec{q}; \mu, \lambda) \int_\Lambda^{\Lambda_0} \frac{\Lambda^\epsilon}{\lambda^{1+\epsilon}} \, \mathcal{P}\left( \ln_+ \frac{\sup\left( \abs{\vec{q}}, \mu \right)}{\lambda}, \ln_+ \frac{\lambda}{\mu} \right) \total \lambda \eqend{.}
\end{splitequation}
An application of Lemma~\ref{lemma_lambdaint} to the remaining integral then gives the required bound~\eqref{bound_ks_lambdamu}. It is in this step that the restriction $D = 1$ enters for $m+n = 0$, because we need to have a negative tree weight factor~\eqref{t_dim_def} to apply Lemma~\ref{lemma_lambdaint}, which is ensured for either $D = 1$ or, since all basic fields have $[\phi_M] \geq 1$, by $m+n > 0$.

\subsubsection{Bounds for \texorpdfstring{$\Lambda < \mu$}{Λ<μ}}

Again, to prove the bounds~\eqref{bound_ls_lambda}, we first need to bound the right-hand side of the flow equation~\eqref{l_sop_flow_hierarchy}, where we want to prove the representation~\eqref{bound_ls_lambda} and
\begin{splitequation}
\label{bound_ks_lambda_lambdaderiv}
&\abs{ \partial_\Lambda \partial^\vec{w} \mathcal{K}^{\Lambda, \Lambda_0, l}_{\vec{K} \vec{L}^\ddag}\left( \bigotimes_{k=1}^s \op_{A_k}(x_k); \vec{q} \right) } \leq \frac{1}{\sup(\inf(\mu,\bar{\eta}(\vec{q})),\Lambda)} \prod_{i=1}^{s-1} \left( 1 + \ln_+ \frac{1}{\mu \abs{x_i}} \right) \\
&\qquad\qquad\times \sup\left( 1, \frac{\abs{\vec{q}}}{\mu} \right)^{g^{(s+1)}([\op_\vec{A}],m+n+2l,\abs{\vec{w}})} \sum_{T^* \in \mathcal{T}^*_{m+n}} \mathsf{G}^{T^*,\vec{w}}_{\vec{K} \vec{L}^\ddag; -D}(\vec{q}; \mu, \Lambda) \, \mathcal{P}\left( \ln_+ \frac{\sup\left( \abs{\vec{q}}, \mu \right)}{\Lambda} \right) \eqend{.}
\end{splitequation}
The first (linear) term on the right-hand side of the flow equation~\eqref{l_sop_flow_hierarchy} already fulfils the decomposition~\eqref{bound_ls_lambda} by hypothesis. Separating the $x$ derivatives, we can bound its contribution to $\partial_\Lambda \mathcal{K}$ (which we call $F_1$ as usual) by inserting the induction hypothesis~\eqref{bound_ks_lambda} and the bound on the covariance~\eqref{prop_abl} to obtain
\begin{splitequation}
\abs{F_1} &\leq \prod_{i=1}^{s-1} \left( 1 + \ln_+ \frac{1}{\mu \abs{x_i}} \right) \int \mathe^{-\frac{\abs{p}^2}{2 \Lambda^2}} \sup\left( 1, \frac{\abs{\vec{q},p,-p}}{\mu} \right)^{g^{(s+1)}([\op_\vec{A}],m+n+2l,\abs{\vec{w}})} \\
&\times \sup(\abs{p}, \Lambda)^{-5+[\phi_M]+[\phi_N]} \hspace{-1em}\sum_{T^* \in \mathcal{T}^*_{m+n+2}}\hspace{-1em} \mathsf{G}^{T^*,\vec{w}}_{MN \vec{K} \vec{L}^\ddag; -D}(p,-p,\vec{q}; \mu, \Lambda) \, \mathcal{P}\left( \ln_+ \frac{\sup\left( \abs{\vec{q},p,-p}, \mu \right)}{\Lambda} \right) \frac{\total^4 p}{(2\pi)^4} \eqend{.}
\end{splitequation}
We then rescale $p = x \Lambda$, use the estimates~\eqref{func_0op_estlog} for the polynomial in logarithms and perform the $p$ integral using Lemma~\ref{lemma_pint2}. Afterwards we amputate the external legs corresponding to $M$ and $N$ from each tree, which gives an extra factor according to equation~\eqref{amputate}, and the subsequent estimate~\eqref{func_0op_estamputate} (with $\bar{\eta}$ instead of $\eta$) gives the bound~\eqref{bound_ks_lambda_lambdaderiv} for $F_1$.

The second (quadratic) term (whose contribution to $\partial_\Lambda \mathcal{K}$ is called $F_2$) is bounded similarly: inserting the induction hypothesis~\eqref{bound_ks_lambda}, the bound on the covariance~\eqref{prop_abl} and the bounds on functionals without insertions~\eqref{bound_l0}, we get
\begin{splitequation}
\abs{F_2} &\leq \sum_{\subline{\sigma \cup \tau = \{1, \ldots, m\} \\ \rho \cup \varsigma = \{1, \ldots, n\} }} \sum_{l'=0}^l \sum_{\vec{u}+\vec{v}\leq \vec{w}} \prod_{i=1}^{s-1} \left( 1 + \ln_+ \frac{1}{\mu \abs{x_i}} \right) \sup(\abs{k}, \Lambda)^{-5+[\phi_M]+[\phi_N]-\abs{\vec{w}}+\abs{\vec{u}}+\abs{\vec{v}}} \\
&\qquad\times \mathe^{-\frac{\abs{k}^2}{2 \Lambda^2}} \sup\left( 1, \frac{\abs{k,\vec{q}_\tau,\vec{q}_\varsigma}}{\mu} \right)^{g^{(s+1)}([\op_\vec{A}],\abs{\tau}+\abs{\varsigma}+1+2(l-l'),\abs{\vec{v}})} \\
&\qquad\times \sum_{T \in \mathcal{T}_{\abs{\sigma}+\abs{\rho}+1}} \mathsf{G}^{T,\vec{u}}_{\vec{K}_\sigma \vec{L}_\rho^\ddag M}(\vec{q}_\sigma,\vec{q}_\rho,-k; \mu, \Lambda) \sum_{T^* \in \mathcal{T}^*_{\abs{\tau}+\abs{\varsigma}+1}} \mathsf{G}^{T^*,\vec{v}}_{N \vec{K}_\tau \vec{L}_\varsigma^\ddag; -D}(k,\vec{q}_\tau,\vec{q}_\varsigma; \mu, \Lambda) \\
&\qquad\times \mathcal{P}\left( \ln_+ \frac{\sup\left( \abs{\vec{q}_\sigma,\vec{q}_\rho,-k}, \mu \right)}{\Lambda} \right) \, \mathcal{P}\left( \ln_+ \frac{\sup\left( \abs{k,\vec{q}_\tau,\vec{q}_\varsigma}, \mu \right)}{\Lambda} \right) \eqend{.}
\end{splitequation}
We use the estimates~\eqref{func_0op_estlog2} and~\eqref{func_1op_est_order} to fuse the polynomials in logarithms and estimate the large momentum factor, and fuse the trees according to the estimate~\eqref{gw_fused_2_est}, changing the $\vec{u}+\vec{v}$ derivatives acting on the fused tree to $\vec{w}$ derivatives in the same way as for functionals without insertions (using equations~\eqref{func_0op_estfuse_a} and~\eqref{func_0op_estfuse_b} with $\bar{\eta}$ instead of $\eta$). The last estimate~\eqref{func_0op_estexp} (with $\bar{\eta}$ instead of $\eta$) then gives the required bound~\eqref{bound_ks_lambda_lambdaderiv}.

For the source term, the last term on the right-hand side of the flow equation~\eqref{l_sop_flow_hierarchy}, we first note that it is summed over all non-empty subsets $\alpha$, $\beta$ of $\{1,\ldots,s\}$, but we only display the contribution of one such set. Note then further that the bounds for functionals with one insertion~\eqref{bound_l1} (without derivatives, and for $\Lambda < \mu$) are compatible with the bounds~\eqref{bound_ks_lambda} (for $s=1$), except for the particular dimension of the tree, which must be changed from $[\op_A]$ to $-D$. This gives an extra factor of~\eqref{particular_weight}
\begin{equation}
\label{func_sop_1op_mu}
\sup\left( \abs{-p, \vec{q}_\tau, \vec{q}_\varsigma}, \mu, \Lambda \right)^{[\op_A]+D} = \mu^{[\op_A]+D} \sup\left( 1, \frac{\abs{-p,\vec{q}_\tau,\vec{q}_\varsigma}}{\mu} \right)^{[\op_A]+D} \eqend{,}
\end{equation}
but otherwise we do not need to treat separately the case where one of the functionals in the source term has only one insertion. We then first use the shift property~\eqref{func_sop_shift} to bring the position of the last operator insertion of the first functional to the origin. In the case where any of the functionals has more than one insertion, we then use the representation~\eqref{bound_ls_lambda}, which gives $x$ derivatives and powers of $\mu$ compatible with the representation~\eqref{bound_ls_lambda} (together with the power of $\mu$ coming from equation~\eqref{func_sop_1op_mu} if the second functional has only one insertion). The remaining terms then give a contribution to $\partial_\Lambda \mathcal{K}$, denoted by $F_3$, which can be estimated using the bound~\eqref{bound_ks_lambda} and the bound on the covariance~\eqref{prop_abl}. We thus obtain in that case
\begin{splitequation}
\abs{F_3} &\leq \prod_{i \in \{1, \ldots, s-1\} \setminus \{ \abs{\alpha} \} } \left( 1 + \ln_+ \frac{1}{\mu \abs{x_i}} \right) \sum_{\subline{\sigma \cup \tau = \{1, \ldots, m\} \\ \rho \cup \varsigma = \{1, \ldots, n\} }} \sum_{l'=0}^l \sum_{\vec{u}\leq \vec{w}} \int \sup(\abs{p}, \Lambda)^{-5+[\phi_M]+[\phi_N]} \, \mathe^{-\frac{\abs{p}^2}{2 \Lambda^2}} \\
&\quad\times \sup\left( 1, \frac{\abs{\vec{q}_\sigma,\vec{q}_\rho,p}}{\mu} \right)^{g^{(\abs{\alpha}+1)}\left( \sum_{k \in \alpha}[\op_{A_k}]+[\op_B]-4,\abs{\sigma}+\abs{\rho}+1+2l',\abs{\vec{u}} \right)} \\
&\quad\times \sup\left( 1, \frac{\abs{-p,\vec{q}_\tau,\vec{q}_\varsigma}}{\mu} \right)^{g^{(\abs{\beta})}\left( \sum_{k \in \beta}[\op_{A_k}],\abs{\tau}+\abs{\varsigma}+1+2(l-l'),\abs{\vec{w}}-\abs{\vec{u}} \right)} \\
&\quad\times \sum_{T^* \in \mathcal{T}^*_{\abs{\sigma}+\abs{\rho}+1}} \mathsf{G}^{T^*,\vec{u}}_{\vec{K}_\sigma \vec{L}_\rho^\ddag M; -D}(\vec{q}_\sigma,\vec{q}_\rho,p; \mu, \Lambda) \sum_{T^* \in \mathcal{T}^*_{\abs{\tau}+\abs{\varsigma}+1}} \mathsf{G}^{T^*,\vec{w}-\vec{u}}_{N \vec{K}_\tau \vec{L}_\varsigma^\ddag; -D}(-p,\vec{q}_\tau,\vec{q}_\varsigma; \mu, \Lambda) \\
&\quad\times \mathcal{P}\left( \ln_+ \frac{\sup\left( \abs{\vec{q}_\sigma,\vec{q}_\rho,p}, \mu \right)}{\Lambda} \right) \, \mathcal{P}\left( \ln_+ \frac{\sup\left( \abs{-p,\vec{q}_\tau,\vec{q}_\varsigma}, \mu \right)}{\Lambda} \right) \frac{\total^4 p}{(2\pi)^4} \eqend{,}
\end{splitequation}
and the same estimate with an additional factor of $\sup\left( 1, \abs{-p,\vec{q}_\tau,\vec{q}_\varsigma}/\mu \right)^{[\op_A]+D}$ if the second functional has only one insertion. Note that for $D = 1$, one of the source terms needs to be bounded with $D = 0$ as before to obtain the correct particular tree weight. We now use the estimates~\eqref{func_0op_estlog2} to fuse the polynomials in logarithms and the large momentum factor, and property~\eqref{gs_prop_3} of $g^{(s)}$ to obtain
\begin{splitequation}
\label{func_sop_gagbest}
&g^{(\abs{\alpha}+1)}\left( \sum_{k \in \alpha}[\op_{A_k}],\abs{\sigma}+\abs{\rho}+1+2l',\abs{\vec{u}} \right) + g^{(\abs{\beta})}\left( \sum_{k \in \beta}[\op_{A_k}],\abs{\tau}+\abs{\varsigma}+1+2(l-l'),\abs{\vec{w}}-\abs{\vec{u}} \right) \\
&\quad\leq g^{(s+1)}([\op_\vec{A}], m+n+2l,\abs{\vec{w}}) - ([\op_\vec{A}]+s+1) \eqend{,}
\end{splitequation}
and then estimate $- ([\op_\vec{A}]+s+1) \leq -[\op_A]-D$ in the case that the second functional has only one insertion, in order to cancel the additional large momentum factor, and $- ([\op_\vec{A}]+s+1) \leq 0$ otherwise. We furthermore insert a factor of $\left[ 1 + \ln_+ 1/\left( \mu \abs{x_\abs{\alpha}} \right) \right]$ and fuse the trees according to the estimate~\eqref{gw_fused_1_est}. A rescaling $p = x \Lambda$ and an application of Lemma~\ref{lemma_pint2} allows to perform the $p$ integral, and it follows that
\begin{splitequation}
\abs{F_3} &\leq \Lambda^{-1+[\phi_M]+[\phi_N]} \prod_{i=1}^{s-1} \left( 1 + \ln_+ \frac{1}{\mu \abs{x_i}} \right) \sup\left( 1, \frac{\abs{\vec{q}}}{\mu} \right)^{g^{(s+1)}([\op_\vec{A}],m+n+2l,\abs{\vec{w}})} \\
&\quad\times \sum_{T^* \in \mathcal{T}^*_{m+n+2}} \mathsf{G}^{T^*,\vec{w}}_{MN \vec{K} \vec{L}^\ddag; -D}(0,0,\vec{q}; \mu, \Lambda) \, \mathcal{P}\left( \ln_+ \frac{\sup\left( \abs{\vec{q}}, \mu \right)}{\Lambda} \right) \eqend{.}
\end{splitequation}
Amputating the external legs corresponding to $M$ and $N$ then gives an extra factor according to equation~\eqref{amputate}, and the subsequent estimate~\eqref{func_0op_estamputate} (with $\bar{\eta}$ instead of $\eta$) gives the bound~\eqref{bound_ks_lambda_lambdaderiv} also in this case. We now integrate the bound~\eqref{bound_ks_lambda_lambdaderiv} in $\Lambda$ to obtain
\begin{splitequation}
&\abs{ \partial^\vec{w} \mathcal{L}^{\Lambda, \Lambda_0, l}_{\vec{K} \vec{L}^\ddag}\left( \bigotimes_{k=1}^s \op_{A_k}(x_k); \vec{q} \right) } \leq \abs{ \partial^\vec{w} \mathcal{L}^{\mu, \Lambda_0, l}_{\vec{K} \vec{L}^\ddag}\left( \bigotimes_{k=1}^s \op_{A_k}(x_k); \vec{q} \right) } \\
&\quad+ \sup\left( 1, \frac{\abs{\vec{q}}}{\mu} \right)^{g^{(s+1)}([\op_\vec{A}],m+n+2l-\abs{\vec{w}})} \prod_{i=1}^{s-1} \left( 1 + \ln_+ \frac{1}{\mu \abs{x_i}} \right) \\
&\qquad\times \int_\Lambda^\mu \frac{1}{\sup(\inf(\mu,\bar{\eta}(\vec{q})),\lambda)} \sum_{T^* \in \mathcal{T}^*_{m+n}} \mathsf{G}^{T^*,\vec{w}}_{\vec{K} \vec{L}^\ddag; -D}(\vec{q}; \mu, \lambda) \, \mathcal{P}\left( \ln_+ \frac{\sup\left( \abs{\vec{q}}, \mu \right)}{\lambda} \right) \total \lambda \eqend{.}
\end{splitequation}
The trees appearing in the $\lambda$ integral can be estimated at $\lambda = \Lambda$ by the inequality~\eqref{t_rel_ineq3}, and an application of Lemma~\ref{lemma_lambdaint3} to the remaining $\lambda$ integral gives the bounds~\eqref{bound_ks_lambda} for the integral. The boundary conditions at $\Lambda = \mu$ are given by the bound~\eqref{bound_ks_lambdamu}, which reads
\begin{splitequation}
&\prod_{i=1}^{s-1} \left( 1 + \ln_+ \frac{1}{\mu \abs{x_i}} \right) \sup\left( 1, \frac{\abs{\vec{q}}}{\mu} \right)^{g^{(s)}([\op_\vec{A}],m+n+2l-\abs{\vec{w}})} \\
&\qquad\times \sum_{T^* \in \mathcal{T}^*_{m+n}} \mathsf{G}^{T^*,\vec{w}}_{\vec{K} \vec{L}^\ddag; -D}(\vec{q}; \mu, \mu) \, \mathcal{P}\left( \ln_+ \frac{\sup\left( \abs{\vec{q}}, \mu \right)}{\mu} \right) \eqend{.}
\end{splitequation}
To estimate their contribution, we use the inequality~\eqref{t_rel_ineq3} to estimate the trees at $\Lambda$, and the final estimate
\begin{equation}
\ln_+ \frac{\sup\left( \abs{\vec{q}}, \mu \right)}{\mu} = \ln_+ \frac{\sup\left( \abs{\vec{q}}, \mu \right)}{\sup(\mu, \Lambda)} \leq \ln_+ \frac{\sup\left( \abs{\vec{q}}, \mu \right)}{\sup(\inf(\mu,\bar{\eta}(\vec{q})),\Lambda)}
\end{equation}
then gives the bound~\eqref{bound_ks_lambda} also for the contribution from the boundary condition.

\subsubsection{Bounds for distinct points}

If all the $x_i$ are distinct, we also prove bounds where no $x$ derivatives appear anymore. The proof works almost in the same way as before, and so we do not need to spell it out in full but show only the differences. Instead of equation~\eqref{func_sop_genderivs}, we could use
\begin{equation}
\exp\left( - \mathi x_1 (p+k) \right) = \frac{\mathi^{\abs{a}}}{\mu^\abs{a} x_1^a} \mu^\abs{a} \partial_p^a \exp\left( - \mathi x_1 (p+k) \right)
\end{equation}
for $\Lambda \geq \mu$ to obtain negative powers of $x$ instead of derivatives and a logarithmically divergent kernel. Again, we can choose a direction $\alpha \in \{1,2,3,4\}$ such that $\abs{x_1^\alpha} \geq \abs{x_1}/2$, and take $a^\beta = \abs{a} \delta^\beta_\alpha$. For $\Lambda \leq \mu$, we do not create additional $p$ derivatives (which would lead to additional IR divergences), and thus the $x$-dependent factor that we obtain after taking absolute values would be given by
\begin{equation}
\label{bound_x_weight}
\sup\left( 1, \frac{1}{\mu^\abs{a} \abs{x_1^\alpha}^\abs{a}} \right) = \inf\left( 1, \mu \abs{x_1^\alpha} \right)^{-\abs{a}} \leq 2^\abs{a} \inf\left( 1, \mu \abs{x_1} \right)^{-\abs{a}} \eqend{.}
\end{equation}
The rest of the proof then proceeds as before, and we obtain the same estimates, with the $x$ derivatives and the logarithmic kernel replaced by products of the weight factors~\eqref{bound_x_weight}, which can be organised into the tree form given in Definition~\ref{def_tree_x}, with $\epsilon = 1$. For some functionals, the source term in the corresponding flow equation~\eqref{l_sop_flow_hierarchy} joins two functionals which already have the correct number of weight factors such that we do not need any additional factor, and thus some lines do not have any weight factors associated to them. However, since each operator insertion contributes $[\op_A]$ to the tree weight~\eqref{t_dim_def}, which must be compensated by the same number of $p$ derivatives, for each vertex there is one weight factor of the form~\eqref{bound_x_weight} associated to a line incident to that vertex, which is exactly the form given in Definition~\ref{def_tree_x}, such that we obtain the bounds~\eqref{bound_ls_x} with $\epsilon = 1$. In order to get to arbitrary small $\epsilon$, we use Lemma~\ref{lemma_frac} for the $p$ integrals appearing in the source term instead of the above, and then obtain the bounds~\eqref{bound_ls_x} with arbitrary $\epsilon > 0$.

\subsection{Additional bounds}
\label{sec_bounds_additional}

\subsubsection{Functionals which vanish in the limit \texorpdfstring{$\Lambda_0 \to \infty$}{Λ₀ → ∞}}

If one imposes vanishing boundary conditions everywhere, both the functionals without and with one insertion of a composite operator vanish, since their flow equation does not involve any source term. However, if the boundary conditions only vanish in the limit $\Lambda_0 \to \infty$ (which will be the case later on), vanishing of the functionals in this limit is a nontrivial fact, since their definition involves an integration over $\Lambda$. Concretely, we prove
\begin{proposition}
\label{thm_l0_van}
For all multiindices $\vec{w}$, at each order $l$ in perturbation theory and for an arbitrary number $m$ of external fields $\vec{K}$ and $n$ antifields $\vec{L}^\ddag$, if the boundary conditions given at $\Lambda = \mu$ (for relevant and marginal functionals) vanish, and the ones given at $\Lambda = \Lambda_0$ (for irrelevant functionals) are non-vanishing, but compatible with the bounds~\eqref{bound_l0} evaluated at $\Lambda = \Lambda_0$, we have the bound
\begin{splitequation}
\label{bound_l0_delta}
\abs{ \partial^\vec{w} \mathcal{L}^{\Lambda, \Lambda_0, l}_{\vec{K} \vec{L}^\ddag}(\vec{q}) } &\leq \left( \frac{\sup(\mu, \Lambda)}{\Lambda_0} \right)^\frac{\Delta}{2} \sum_{T \in \mathcal{T}_{m+n}} \mathsf{G}^{T,\vec{w}}_{\vec{K} \vec{L}^\ddag}(\vec{q}; \mu, \Lambda) \, \mathcal{P}\left( \ln_+ \frac{\sup\left( \abs{\vec{q}}, \mu \right)}{\sup(\inf(\mu, \eta(\vec{q})), \Lambda)}, \ln_+ \frac{\Lambda}{\mu} \right) \eqend{.}
\end{splitequation}
\end{proposition}
This bound differs from the bounds~\eqref{bound_l0} only in the additional factor $\left( \sup(\mu, \Lambda) / \Lambda_0 \right)^\frac{\Delta}{2}$, which only introduces small modifications in the inductive proof, and it is shorter to just give these modifications instead of spelling out the full proof. Since the right-hand side of the flow equation is estimated at fixed $\Lambda$ and the additional factor is momentum independent, the right-hand side of the flow equation can be bounded just as before and we obtain the estimate~\eqref{bound_l0_lambdaderiv} with this additional factor. Note that for the second (quadratic) term on the right-hand side of the flow equation~\eqref{l_0op_flow_hierarchy}, since there are two of these additional factors we need to estimate
\begin{equation}
\label{bound_vanish_delta}
\left( \frac{\sup(\mu, \Lambda)}{\Lambda_0} \right)^\Delta \leq \left( \frac{\sup(\mu, \Lambda)}{\Lambda_0} \right)^\frac{\Delta}{2} \eqend{.}
\end{equation}
For irrelevant terms, we proceed as before, and instead of equation~\eqref{func_0op_irrelevant} obtain
\begin{splitequation}
\abs{ \partial^\vec{w} \mathcal{L}^{\Lambda, \Lambda_0, l}_{\vec{K} \vec{L}^\ddag}(\vec{q}) } &\leq \sum_{T \in \mathcal{T}_{m+n}} \mathsf{G}^{T,\vec{w}}_{\vec{K} \vec{L}^\ddag}(\vec{q}; \mu, \Lambda)  \Bigg[ \frac{\sup(\inf(\mu, \eta(\vec{q})), \Lambda)^\Delta}{\sup(\inf(\mu, \eta(\vec{q})), \Lambda_0)^\Delta} \,\mathcal{P}\left( \ln_+ \frac{\sup\left( \abs{\vec{q}}, \mu \right)}{\Lambda_0}, \ln_+ \frac{\Lambda_0}{\mu} \right) \\
&+ \int_\Lambda^{\Lambda_0} \frac{\sup(\inf(\mu, \eta(\vec{q})), \Lambda)^\Delta}{\sup(\inf(\mu, \eta(\vec{q})), \lambda)^{\Delta+1}} \left( \frac{\sup(\mu, \lambda)}{\Lambda_0} \right)^\frac{\Delta}{2} \,\mathcal{P}\left( \ln_+ \frac{\sup\left( \abs{\vec{q}}, \mu \right)}{\lambda}, \ln_+ \frac{\lambda}{\mu} \right) \total \lambda \Bigg] \eqend{.}
\end{splitequation}
We then use the estimate~\eqref{func_0op_deltalog} with $\Delta/2$ instead of $\Delta$ to bound the first polynomial in logarithms, and estimate
\begin{equation}
\left( \frac{\sup(\inf(\mu, \eta(\vec{q})), \Lambda)}{\sup(\inf(\mu, \eta(\vec{q})), \Lambda_0)} \right)^\frac{\Delta}{2} \leq \left( \frac{\sup(\mu, \Lambda)}{\Lambda_0} \right)^\frac{\Delta}{2}
\end{equation}
to obtain the correct additional factor for the first term. For the second term, we estimate by a long but straightforward case-by-case analysis
\begin{equation}
\left( \frac{\sup(\inf(\mu, \eta(\vec{q})), \Lambda)}{\sup(\inf(\mu, \eta(\vec{q})), \lambda)} \frac{\sup(\mu, \lambda)}{\Lambda_0} \right)^\frac{\Delta}{2} \leq \left( \frac{\sup(\mu, \Lambda)}{\Lambda_0} \right)^\frac{\Delta}{2} \eqend{,}
\end{equation}
and then apply Lemma~\ref{lemma_lambdaint} to bound the remaining integral and obtain the bounds~\eqref{bound_l0_delta}. For marginal terms which now have vanishing boundary conditions at $\Lambda = \mu$, instead of equation~\eqref{bound_l0_zeromomentum} we get at zero momentum for $\Lambda \geq \mu$
\begin{splitequation}
\abs{ \partial^\vec{w} \mathcal{L}^{\Lambda, \Lambda_0, l}_{\vec{K} \vec{L}^\ddag}(\vec{0}) } &\leq \left( \frac{\sup(\mu, \Lambda)}{\Lambda_0} \right)^\frac{\Delta}{2} \mathcal{P}\left( \ln_+ \frac{\Lambda}{\mu} \right) \int_\mu^\Lambda \lambda^{4 - [\vec{K}] - [\vec{L}^\ddag] - \abs{\vec{w}}-1} \total \lambda \\
&\leq \left( \frac{\sup(\mu, \Lambda)}{\Lambda_0} \right)^\frac{\Delta}{2} \sum_{T \in \mathcal{T}_{m+n}} \mathsf{G}^{T,\vec{w}}_{\vec{K} \vec{L}^\ddag}(\vec{0}; \mu, \Lambda) \, \mathcal{P}\left( \ln_+ \frac{\Lambda}{\mu} \right) \eqend{.}
\end{splitequation}
The Taylor expansion needed to extend this bound to general momenta is done at fixed $\Lambda$, and the additional factor thus stays unchanged, and for the subsequent extension to $\Lambda < \mu$ achieved by integration of the flow equation we note that the additional factor becomes independent of $\Lambda$, such that we obtain the bounds~\eqref{bound_l0_delta} also in this cases. Lastly, for the relevant functionals with have vanishing boundary condition at $\Lambda = 0$ and vanishing momenta, the additional factor can be trivially estimated since the integration only proceeds upwards in $\Lambda$, such that instead of equation~\eqref{relevant_zeromomentum} we obtain
\begin{equation}
\abs{\partial^\vec{w} \mathcal{L}^{\Lambda, \Lambda_0, l}_{\vec{K} \vec{L}^\ddag}(\vec{0})} \leq \left( \frac{\sup(\mu, \Lambda)}{\Lambda_0} \right)^\frac{\Delta}{2} \sum_{T \in \mathcal{T}_{m+n}} \sup(c,\Lambda)^{[T]} \, \mathcal{P}\left( \ln_+ \frac{\mu}{\sup(c,\Lambda)}, \ln_+ \frac{\Lambda}{\mu} \right) \eqend{.}
\end{equation}
Again, the Taylor expansion that we use the extend this bound to general momenta is done for fixed $\Lambda$ such that the additional factor does not change, and we also obtain the bounds~\eqref{bound_l0_delta}.

For functionals with one operator insertion, we prove
\begin{proposition}
\label{thm_l1_van}
For all multiindices $\vec{w}$, at each order $l$ in perturbation theory and for an arbitrary number $m$ of external fields $\vec{K}$ and $n$ antifields $\vec{L}^\ddag$ and any composite operator $\op_A$, we have the bound
\begin{splitequation}
\label{bound_l1_delta}
\abs{ \partial^\vec{w} \mathcal{L}^{\Lambda, \Lambda_0, l}_{\vec{K} \vec{L}^\ddag}\left( \op_A(0); \vec{q} \right) } &\leq \sup\left( 1, \frac{\abs{\vec{q}}}{\sup(\mu, \Lambda)} \right)^{g^{(1)}([\op_A],m+n+2l,\abs{\vec{w}})} \sum_{T^* \in \mathcal{T}^*_{m+n}} \mathsf{G}^{T^*,\vec{w}}_{\vec{K} \vec{L}^\ddag; [\op_A]}(\vec{q}; \mu, \Lambda) \\
&\quad\times \left( \frac{\sup(\mu, \Lambda)}{\Lambda_0} \right)^\frac{\Delta}{2} \, \mathcal{P}\left( \ln_+ \frac{\sup\left( \abs{\vec{q}}, \mu \right)}{\sup(\inf(\mu, \bar{\eta}(\vec{q})), \Lambda)}, \ln_+ \frac{\Lambda}{\mu} \right) \eqend{,}
\end{splitequation}
if the boundary conditions given at $\Lambda = \mu$ (for relevant and marginal functionals) vanish, and the ones given at $\Lambda = \Lambda_0$ (for irrelevant functionals) are non-vanishing, but compatible with the bounds~\eqref{bound_l1}.
\end{proposition}
Again, this differs from the bounds~\eqref{bound_l1} only in the additional factor $\left( \sup(\mu, \Lambda) / \Lambda_0 \right)^\frac{\Delta}{2}$, and the proof is the same as for functionals without insertions, with the appropriate change in notation. For functionals with one integrated insertion, we have
\begin{proposition}
\label{thm_l1i_van}
For all multiindices $\vec{w}$, at each order $l$ in perturbation theory and for an arbitrary number $m$ of external fields $\vec{K}$ and $n$ antifields $\vec{L}^\ddag$ and any integrated composite operator $\op_A$ with $[\op_A] \geq 4$, when the boundary conditions given at $\Lambda = \mu$ (for marginal and some relevant functionals) vanish, and the ones given at $\Lambda = \Lambda_0$ (for irrelevant functionals) are non-vanishing, but compatible with the bounds~\eqref{bound_l1i}, we have the bound
\begin{splitequation}
\label{bound_l1i_delta}
\abs{ \partial^\vec{w} \mathcal{L}^{\Lambda, \Lambda_0, l}_{\vec{K} \vec{L}^\ddag}\left( \int\!\op_A; \vec{q} \right) } &\leq \sup\left( 1, \frac{\abs{\vec{q}}}{\sup(\mu, \Lambda)} \right)^{g^{(1)}([\op_A]-4,m+n+2l,\abs{\vec{w}})} \left( \frac{\sup(\mu, \Lambda)}{\Lambda_0} \right)^\frac{\Delta}{2} \\
&\times \sum_{T \in \mathcal{T}_{m+n}} \mathsf{G}^{T,\vec{w}}_{\vec{K} \vec{L}^\ddag; [\op_A]-4}(\vec{q}; \mu, \Lambda) \, \mathcal{P}\left( \ln_+ \frac{\sup\left( \abs{\vec{q}}, \mu \right)}{\sup(\inf(\mu, \eta(\vec{q})), \Lambda)}, \ln_+ \frac{\Lambda}{\mu} \right) \eqend{.}
\end{splitequation}
\end{proposition}

Functionals with more than one insertion already have vanishing boundary conditions for all functionals at $\Lambda = \Lambda_0$, and the non-vanishing contribution comes from the source term which depends on functionals with a lower number of insertions. To make these functionals vanish as $\Lambda_0 \to \infty$, it is thus necessary that the source term vanishes in that limit. Again, we have to ascend in the number of operator insertions, starting with two insertions, and we prove
\begin{proposition}
\label{thm_ls_van}
For all multiindices $\vec{w}$, at each order $l$ in perturbation theory, for an arbitrary number $m$ of external fields $\vec{K}$ and $n$ antifields $\vec{L}^\ddag$, and for an arbitrary number $s$ of non-integrated composite operators $\op_{A_i}$, the functionals with at least two insertions of the $\op_{A_i}$ can be written in the form given in Proposition~\ref{thm_lsa} (for $\Lambda \geq \mu$) or in the form given in Proposition~\ref{thm_lsb} (for $\Lambda < \mu$), and the kernels $\mathcal{K}^{\Lambda, \Lambda_0, l}$ satisfy a bound of the form given in these propositions with an additional factor of
\begin{equation}
\label{lambda0_vanishing}
\left( \frac{\sup(\Lambda,\mu)}{\Lambda_0} \right)^\frac{\Delta}{2} \eqend{,}
\end{equation}
if at least one of the two functionals appearing in the source term of the corresponding flow equation~\eqref{l_sop_flow} fulfils a bound of the form~\eqref{bound_l1_delta} or \eqref{bound_l1i_delta} (for functionals with one insertion), or the kernels $\mathcal{K}^{\Lambda, \Lambda_0, l}$ are bounded with an extra factor~\eqref{lambda0_vanishing} (for functionals with more than one insertion). Furthermore, the boundary conditions at $\Lambda = \Lambda_0$ need not vanish, but must be compatible with the bounds~\eqref{bound_ks_lambdamu} or~\eqref{bound_ks_lambda}.
\end{proposition}
This shows that the factor~\eqref{lambda0_vanishing} propagates, first for functionals with two insertions where the source term only contains functionals with one insertion, and then ascending in the number of insertions. Again, it is simpler to just detail the modifications of the proof that have to be done. Let us start with the case $\Lambda \geq \mu$. Since the right-hand side of the flow equation~\eqref{l_sop_flow_hierarchy} is estimated at fixed $\Lambda$, nothing changes, and we obtain the estimate~\eqref{bound_ks_lambdamu_lambdaderiv} with an additional factor of~\eqref{lambda0_vanishing}. Note that if both functionals in the source term come with this factor in their respective bounds, we just use the estimate~\eqref{bound_vanish_delta}. Instead of equation~\eqref{func_sop_lambdamu_int} we then obtain
\begin{splitequation}
&\abs{ \mathcal{K}^{\Lambda, \Lambda_0, l}_{\vec{K} \vec{L}^\ddag}\left( \bigotimes_{k=1}^s \op_{A_k}(x_k); \vec{q} \right) } \leq \prod_{i=1}^{s-1} \left( 1 + \ln_+ \frac{1}{\mu \abs{x_i}} \right) \sup\left( 1, \frac{\abs{\vec{q}}}{\Lambda_0} \right)^{g^{(s)}([\op_\vec{A}],m+n+2l,\abs{\vec{w}})} \\
&\hspace{10em}\times \sum_{T^* \in \mathcal{T}^*_{m+n}} \mathsf{G}^{T^*,\vec{w}}_{\vec{K} \vec{L}^\ddag; -D}(\vec{q}; \mu, \Lambda_0) \, \mathcal{P}\left( \ln_+ \frac{\sup\left( \abs{\vec{q}}, \mu \right)}{\Lambda_0}, \ln_+ \frac{\Lambda_0}{\mu} \right) \\
&\hspace{8em}+ \prod_{i=1}^{s-1} \left( 1 + \ln_+ \frac{1}{\mu \abs{x_i}} \right) \int_\Lambda^{\Lambda_0} \frac{1}{\lambda} \left( \frac{\lambda}{\Lambda_0} \right)^\frac{\Delta}{2} \sup\left( 1, \frac{\abs{\vec{q}}}{\lambda} \right)^{g^{(s)}([\op_\vec{A}],m+n+2l,\abs{\vec{w}})} \\
&\hspace{10em}\times \sum_{T^* \in \mathcal{T}^*_{m+n}} \mathsf{G}^{T^*,\vec{w}}_{\vec{K} \vec{L}; -D}(\vec{q}; \mu, \lambda) \, \mathcal{P}\left( \ln_+ \frac{\sup\left( \abs{\vec{q}}, \mu \right)}{\lambda}, \ln_+ \frac{\lambda}{\mu} \right) \total \lambda \eqend{.}
\end{splitequation}
The trees can be estimated at $\lambda = \Lambda$ using the inequality~\eqref{t_irr_ineq2} with $\epsilon = 1$, and the large momentum factor is trivially estimated there, such that we get
\begin{splitequation}
&\abs{ \mathcal{K}^{\Lambda, \Lambda_0, l}_{\vec{K} \vec{L}^\ddag}\left( \bigotimes_{k=1}^s \op_{A_k}(x_k); \vec{q} \right) } \leq \prod_{i=1}^{s-1} \left( 1 + \ln_+ \frac{1}{\mu \abs{x_i}} \right) \sup\left( 1, \frac{\abs{\vec{q}}}{\Lambda} \right)^{g^{(s)}([\op_\vec{A}],m+n+2l,\abs{\vec{w}})} \\
&\hspace{12em}\times \sum_{T^* \in \mathcal{T}^*_{m+n}} \mathsf{G}^{T^*,\vec{w}}_{\vec{K} \vec{L}^\ddag; -D}(\vec{q}; \mu, \Lambda) \left( \frac{\Lambda}{\Lambda_0} \right)^\frac{\Delta}{2} \\
&\quad\times \left[ \left( \frac{\Lambda}{\Lambda_0} \right)^{1-\frac{\Delta}{2}} \, \mathcal{P}\left( \ln_+ \frac{\sup\left( \abs{\vec{q}}, \mu \right)}{\Lambda_0}, \ln_+ \frac{\Lambda_0}{\mu} \right) + \int_\Lambda^{\Lambda_0} \frac{\Lambda^{1-\frac{\Delta}{2}}}{\lambda^{2-\frac{\Delta}{2}}} \, \mathcal{P}\left( \ln_+ \frac{\sup\left( \abs{\vec{q}}, \mu \right)}{\lambda}, \ln_+ \frac{\lambda}{\mu} \right) \total \lambda \right] \eqend{.}
\end{splitequation}
Since $\Delta \leq 1$, Lemma~\ref{lemma_largerlog} can be used to estimate the second logarithm in the first polynomial at $\Lambda$, while the first logarithm is trivially estimated there. An application of Lemma~\ref{lemma_lambdaint} to the remaining integral then gives the required bound, equation~\eqref{bound_ks_lambdamu} with an additional factor of~\eqref{lambda0_vanishing}. For $\Lambda < \mu$, the extra factor does not depend on $\Lambda$, and thus nothing at all changes in the proof.

\subsubsection{Convergence}
\label{bounds_lambda0}

Up to now, we only have shown boundedness of the functionals with and without operator insertions. To show convergence in the limit $\Lambda_0 \to \infty$, we also need to prove bounds on the derivative of the functionals with respect to $\Lambda_0$. The corresponding flow equation is obtained by taking a $\Lambda_0$ derivative of the respective flow equation~\eqref{l_0op_flow_hierarchy,l_sop_flow_hierarchy}. The boundary conditions at $\Lambda = 0$ or $\Lambda = \mu$ can be simply obtained by taking a $\Lambda_0$ derivative of the boundary conditions for the respective functional, and since these boundary conditions do not depend on $\Lambda_0$ we have vanishing boundary conditions in that case. For the irrelevant functionals where one imposes vanishing boundary conditions at $\Lambda = \Lambda_0$ one has to be more careful. Starting with the functionals without insertions, one integrates the flow equation~\eqref{l_0op_flow_hierarchy} with respect to $\Lambda$ and obtains for irrelevant functionals
\begin{equation}
\mathcal{L}^{\Lambda, \Lambda_0, l}_{\vec{K} \vec{L}^\ddag}(\vec{q}) = \int_\Lambda^{\Lambda_0} \partial_\lambda \mathcal{L}^{\lambda, \Lambda_0, l}_{\vec{K} \vec{L}^\ddag}(\vec{q}) \total \lambda \eqend{,}
\end{equation}
where $\partial_\lambda \mathcal{L}^{\lambda, \Lambda_0, l}_{\vec{K} \vec{L}^\ddag}(\vec{q}) \equiv F(\lambda)$ is given by the right-hand side of the flow equation~\eqref{l_0op_flow_hierarchy}. Taking a $\Lambda_0$ derivative, we obtain
\begin{equation}
\partial_{\Lambda_0} \mathcal{L}^{\Lambda, \Lambda_0, l}_{\vec{K} \vec{L}^\ddag}(\vec{q}) = F(\Lambda_0) + \int_\Lambda^{\Lambda_0} \partial_{\Lambda_0} F(\lambda) \total \lambda \eqend{,}
\end{equation}
and taking the limit $\Lambda \to \Lambda_0$ the second term vanishes. The right boundary conditions for the $\Lambda_0$ derivative of irrelevant functionals are thus given by the right-hand side of the flow equation~\eqref{l_0op_flow_hierarchy} evaluated at $\Lambda = \Lambda_0$, which satisfies the bound~\eqref{bound_l0_lambdaderiv}. We thus have
\begin{equation}
\label{bound_lambda0}
\Lambda_0 \abs{ \partial_{\Lambda_0} \partial^\vec{w} \mathcal{L}^{\Lambda, \Lambda_0, l}_{\vec{K} \vec{L}^\ddag}(\vec{q}) }_{\Lambda = \Lambda_0} \leq \abs{ \partial^\vec{w} \mathcal{L}^{\Lambda_0, \Lambda_0, l}_{\vec{K} \vec{L}^\ddag}(\vec{q}) } \eqend{,}
\end{equation}
and in the same way one finds the boundary conditions for the $\Lambda_0$ derivative of irrelevant functionals with insertions, by integrating the flow equations~\eqref{l_sop_flow_hierarchy} over $\Lambda$, taking a $\Lambda_0$ derivative and then the limit $\Lambda \to \Lambda_0$.

For the $\Lambda_0$ derivative of any functional, multiplied by $\Lambda_0$, we then obtain almost the same flow equation as for the functional itself. The only difference is the appearance of new source terms, which come from distributing the $\Lambda_0$ derivative. Furthermore, the boundary conditions~\eqref{bound_lambda0} (and similarly for functionals with insertions) are compatible with the appropriate bounds evaluated at $\Lambda = \Lambda_0$ for irrelevant functionals (given by~\eqref{bound_l0,bound_l1,bound_ls_lambdamu,bound_ks_lambdamu} and vanishing boundary conditions at $\Lambda = \mu$ (or, for some relevant functionals, at $\Lambda = 0$) and vanishing momenta. One then notes that the additional source terms can be estimated in the same way as the existing ones if the induction hypothesis~\eqref{bound_l1_delta} (resp.~\eqref{bound_l1i_delta} or the appropriate generalisation to functionals with more than one insertion) of the last subsection is made, since it only introduces an additional factor of~\eqref{lambda0_vanishing}. We can thus reuse the proof of the last subsection, which therefore tells us that the $\Lambda_0$ derivative of any functional satisfies the same bound as the corresponding functional, multiplied by
\begin{equation}
\frac{1}{\Lambda_0} \left( \frac{\sup(\Lambda, \mu)}{\Lambda_0} \right)^\frac{\Delta}{2} \eqend{.}
\end{equation}
An integration over $\Lambda_0$ then gives the bound
\begin{splitequation}
\abs{ \mathcal{L}^{\Lambda, \Lambda_1, l}_{\vec{K} \vec{L}^\ddag}(\vec{q}) } &\leq \abs{ \mathcal{L}^{\Lambda, \Lambda_0, l}_{\vec{K} \vec{L}^\ddag}(\vec{q}) } + \abs{ \mathcal{L}^{\Lambda, \Lambda_0, l}_{\vec{K} \vec{L}^\ddag}(\vec{q}) } \int_{\Lambda_0}^{\Lambda_1} \frac{1}{\lambda} \left( \frac{\sup(\Lambda, \mu)}{\lambda} \right)^\frac{\Delta}{2} \total \lambda \\
&\leq \abs{ \mathcal{L}^{\Lambda, \Lambda_0, l}_{\vec{K} \vec{L}^\ddag}(\vec{q}) } + \frac{2}{\Delta} \left( \frac{\sup(\Lambda, \mu)}{\Lambda_0} \right)^\frac{\Delta}{2} \abs{ \mathcal{L}^{\Lambda, \Lambda_0, l}_{\vec{K} \vec{L}^\ddag}(\vec{q}) }
\end{splitequation}
for all $\Lambda_1 \geq \Lambda_0$ since the bound~\eqref{bound_l1} for $\mathcal{L}^{\Lambda, \Lambda_0, l}_{\vec{K} \vec{L}^\ddag}(\vec{q})$ is independent of $\Lambda_0$, and similar bounds for functionals with insertions. Especially, we can take the limit $\Lambda_1 \to \infty$.

\section{Restoration of BRST symmetry}
\label{sec_brst}

In the previous sections we have given a rather exhaustive treatment of bounds on correlation functions without paying any attention to gauge/BRST invariance. The presence of the cutoffs $\Lambda$ and $\Lambda_0$ necessarily breaks this invariance, but a naive hope might be that it will be restored in the physical limit $\Lambda_0 \to \infty$ and $\Lambda \to 0$ (which we have shown to exist at least for non-exceptional momenta). However, for a general set of boundary conditions on the CACs (see Table~\ref{table_boundary} and equation~\eqref{func_1op_bphz}), this is simply false. One's first reaction might be that this renders our method basically unsuitable for theories with local gauge invariance, as was indeed the view taken in the early days of quantum Yang-Mills theory. This is however not so, because one can restore gauge invariance by a set of ``finite renormalisation'' changes (\ie, making certain specific choices of the boundary conditions on the CACs) while leaving all our analytical bounds intact, as we will see in this section. With that choice of boundary conditions, the final result is that gauge invariance holds in the sense expressed in Theorem~\ref{thm4}, as we explained in Sections~\ref{sec_intro_gauge} and~\ref{sec_gauge}.

Our main tools in the proof of Theorem~\ref{thm4} will be essentially algebraic methods from BRST/BV theory. In fact, in order to apply these methods, it turns out that rather than proving Theorem~\ref{thm4}, it is actually better to prove in one stroke a more general version of this result which we now present. Since in most of the following section we work exclusively in the physical limit $\Lambda \to 0$, $\Lambda_0 \to \infty$, to shorten the notation, we set
\begin{equation}
\label{abbreviation}
L(\cdots) \equiv L^{0,\infty}(\cdots) \eqend{.}
\end{equation}
The more general version of Theorem~\ref{thm4} then reads:
\begin{proposition}
\label{thm_brst}
One can choose $\mathrm{E}(4)$-covariant boundary/renormalisation conditions for the generating functionals with and without operator insertions, such that
\begin{equation}
\label{ward_0op}
\st_0 L = 0 \eqend{,}
\end{equation}
\begin{equation}
\label{ward_1op}
\st_0 L\left( \op_A(x) \right) = L\left( (\stq \op_A)(x) \right) \eqend{.}
\end{equation}
and (for $s \geq 2$)
\begin{splitequation}
\label{ward_sop}
\st_0 L\left( \bigotimes_{k=1}^s \op_{A_k}(x_k) \right) &= \sum_{l=1}^s L\left( \bigotimes_{k\in\{1,\ldots,s\}\setminus\{l\}} \op_{A_k}(x_k) \otimes (\stq \op_{A_l})(x_l) \right) \\
&\quad+ \sum_{1 \leq l < l' \leq s} L\left( \bigotimes_{k\in\{1,\ldots,s\}\setminus\{l,l'\}} \op_{A_k}(x_k) \otimes \left( \op_{A_l}(x_l), \op_{A_{l'}}(x_{l'}) \right)_\hbar \right) \eqend{,}
\end{splitequation}
understood as a shorthand for the hierarchy of identities obtained when we expand the above equations in $\hbar$, in external fields $\phi_K$ and antifields $\phi^\ddag_L$ as in equation~\eqref{field_expansion}. Here, $\st_0$ is the free part of $\st$ defined in equation~\eqref{st0_def}, $\stq$ is a differential and $\left( \cdot, \cdot \right)_\hbar$ a bracket which fulfil the properties stated in Theorem~\ref{thm4}.

In the given form, these identities are valid for bosonic operators, while for fermionic operators additional minus signs appear. As explained in Theorem~\ref{thm4}, the correct minus signs are obtained by introducing auxiliary constant fermions $\epsilon_k$ for each fermionic operator $\op_{A_k}$, replacing $\op_{A_k} \to \epsilon_k \op_{A_k}$ and then taking derivatives with respect to the $\epsilon_k$, using that the functionals are multilinear in the operator insertions.
\end{proposition}
Theorem~\ref{thm4} is a trivial corollary of this proposition. Indeed, is is just the first term of the expansion of equation~\eqref{ward_sop} where no external fields and antifields are present, with the relation between the functionals with multiple operator insertions and the connected correlation functions of these operators given by equation~\eqref{relation_l_opconn}.

Thus what we need to show is Proposition~\ref{thm_brst}. In general, the boundary conditions chosen so far -- summarised in Table~\ref{table_boundary} and equation~\eqref{func_1op_bphz} -- will not lead to the fulfilment of the identities stated in Proposition~\ref{thm_brst}. Naturally, one way to obtain the correct boundary conditions would be essentially by explicit calculation: one imposes boundary conditions containing sufficiently many free parameters and makes the dependence of the functionals on these parameters explicit throughout the calculation, as done for QED in Ref.~[\onlinecite{kellerkopper1991}] and for spontaneously broken SU(2) theory in Refs.~[\onlinecite{koppermueller2000a,koppermueller2000b,mueller2003,koppermueller2009}]. However, this is quite laborious and ultimately not very practical in our case if we want to consider arbitrary insertions of composition operators as in equation~\eqref{ward_sop}. We therefore proceed here by another method, which consists in first writing down ``anomalous'' versions of the identities stated in Proposition~\ref{thm_brst}. The ``anomaly'' quantifies the extent to which the identities are violated. Actually, there are three kinds of anomalies:
\begin{enumerate}
\item An anomaly $\mathsf{A}_0$, quantifying the violation of equation~\eqref{ward_0op}. This anomaly can be understood as governing local gauge invariance at the level of the renormalised (effective) action. It is often called ``gauge anomaly''. We show that $\mathsf{A}_0$ is of the form $\int\!\mathcal{A}$ where $\mathcal{A} \in \mathcal{F}^{1,4}$ is a local composite operator of form degree $4$ and ghost number $1$.
\item An anomaly $\mathsf{A}_1$, quantifying the difference between equation~\eqref{ward_1op} and the naively expected one $\st_0 L\left( \op_A \right) = L\left( \st \op_A \right)$. This anomaly can be understood as governing the local gauge invariance of renormalised composite operators inside a correlation function. Technically, it is a map $\mathsf{A}_1\colon \mathcal{F}^{g,p} \to \mathcal{F}^{g+1,p}$. $\mathsf{A}_1$ combines with the classical BRST differential $\st$ to the ``quantum'' BRST differential $\stq = \st + \mathsf{A}_1(\cdot)$, where the qualifier ``quantum'' refers to the fact that $\mathsf{A}_1$ is at least of order $\hbar$.
\item An anomaly $\mathsf{A}_2$, quantifying again the difference between equation~\eqref{ward_sop} and the naively expected one (which would be~\eqref{ward_sop} with $\st$ instead of $\stq$ and the classical BV bracket $(\cdot,\cdot)$ instead of $(\cdot,\cdot)_\hbar$). This anomaly can be understood as governing the local gauge invariance of ``contact type terms'' of renormalised composite operators inside a correlation function. Technically, it is a map $\mathsf{A}_2: \mathcal{F}^{g,p} \otimes \mathcal{F}^{g',p'} \to \mathcal{F}^{g+g'+1,p+p'}_2$. $\mathsf{A}_2$ combines with the classical BV bracket $(\cdot, \cdot)$~\cite{batalinvilkovisky1981,batalinvilkovisky1983,batalinvilkovisky1984,weinberg_v2} to the ``quantum'' BV bracket $(\cdot, \cdot)_\hbar = (\cdot, \cdot) + \mathsf{A}_2(\cdot, \cdot)$, where the qualifier ``quantum'' refers again to the fact that $\mathsf{A}_2$ is at least of order $\hbar$.
\end{enumerate}

In most of the literature, only the first type of anomaly $\mathsf{A}_0$ is discussed. Since its vanishing is a prerequisite for the Ward identities~\eqref{ward_1op} and~\eqref{ward_sop} to hold, it is indeed the most important one, while the other anomalies $\mathsf{A}_1$ and $\mathsf{A}_2$ ``just'' give the proper extension of the BRST differential and the BV bracket to the quantum theory. We now present in detail our result, which is the ``anomalous version'' of Proposition~\ref{thm_brst}.
\begin{proposition}
\label{thm_anomward}
For a general choice of boundary conditions (as given in Table~\ref{table_boundary}), the following is true:
\begin{enumerate}
\item There exists a composite operator $\mathcal{A} \in \mathcal{F}^{1,4}_{\mathrm{E}(4)}$ of dimension $[\mathcal{A}] = 5$ and (at least) order $\hbar$ such that, with $\mathsf{A}_0 = \int\!\mathcal{A}$ in the sense of Proposition~\ref{thm_l1i_d},
\begin{equation}
\label{anom_ward_0op}
\st_0 L = L\left( \mathsf{A}_0 \right) \eqend{.}
\end{equation}
\item There exists a linear map $\mathsf{A}_1\colon \mathcal{F}^{g,p} \to \mathcal{F}^{g+1,p}$ commuting with $\mathrm{E}(4)$ and of dimension $[\mathsf{A}_1(\op)] = [\op]+1$, such that
\begin{equation}
\label{anom_ward_1op_a}
\st_0 L\left( \mathsf{A}_0 \right) = L\left( \st \mathsf{A}_0 \right) + L\left( \mathsf{A}_1\left( \mathsf{A}_0 \right) \right) \eqend{,}
\end{equation}
with the integrated composite operators $\st \mathsf{A}_0 \equiv \int \st \mathcal{A}$ and $\mathsf{A}_1\left( \mathsf{A}_0 \right) \equiv \int \mathsf{A}_1\left( \mathcal{A} \right)$. Moreover, $\mathsf{A}_1(\op)$ is of higher order in $\hbar$ than $\op$.
\end{enumerate}
For a choice of boundary conditions such that $\mathsf{A}_0 = 0$, the following is true:
\begin{enumerate}
\item Defining
\begin{equation}
\label{stq_def}
\stq \op_A \equiv \st \op_A + \mathsf{A}_1(\op_A)
\end{equation}
with the same map $\mathsf{A}_1$ as before, we have $\stq^2 = 0$ and
\begin{equation}
\label{anom_ward_1op}
\st_0 L\left( \op(x) \right) = L\left( \stq \op(x) \right) \eqend{.}
\end{equation}
\item There exists a bilinear map $\mathsf{A}_2\colon \mathcal{F}^{p,g} \otimes \mathcal{F}^{p',g'} \to \mathcal{F}^{p+p',g+g'+1}_2$ of dimension \\ $[\mathsf{A}_2\left( \op_{A_k} \otimes \op_{A_l} \right)] = [\op_{A_k}]+[\op_{A_l}]-3$, supported on the diagonal $x_k = x_l$ and of higher order in $\hbar$ than $\left( \op_{A_k},\op_{A_l} \right)$, such that we have
\begin{splitequation}
\label{anom_ward_sop}
\st_0 L\left( \bigotimes_{k=1}^s \op_{A_k}(x_k) \right) &= \sum_{l=1}^s L\left( \bigotimes_{k\in\{1,\ldots,s\}\setminus\{l\}} \op_{A_k}(x_k) \otimes (\stq \op_{A_l})(x_l) \right) \\
&\quad+ \sum_{1 \leq l < l' \leq s} L\left( \bigotimes_{k\in\{1,\ldots,s\}\setminus\{l,l'\}} \op_{A_k}(x_k) \otimes \left( \op_{A_l}(x_l), \op_{A_{l'}}(x_{l'}) \right)_\hbar \right)
\end{splitequation}
with the quantum antibracket $\left( \cdot, \cdot \right)_\hbar$ defined by
\begin{equation}
\label{bvq_def}
\left( \op_{A_k},\op_{A_l} \right)_\hbar \equiv \left( \op_{A_k},\op_{A_l} \right) + \mathsf{A}_2\left( \op_{A_k} \otimes \op_{A_l} \right) \eqend{.}
\end{equation}
The bracket satisfies the symmetry~\eqref{bvq_symm}, the graded Jacobi identity~\eqref{bvq_jacobi} and the compatibility condition~\eqref{bvq_stq_compat}.
\end{enumerate}
Again, all identities should be understood as a shorthand for the hierarchy of identities obtained when we expand the above equations in $\hbar$ and in external fields and antifields.
\end{proposition}
Proposition~\ref{thm_brst} immediately follows from this proposition if we can choose renormalisation conditions such that $\mathsf{A}_0 = 0$, and we prove that such a choice is indeed possible. The essential point is that the Wess-Zumino consistency conditions for the anomaly $\mathsf{A}_0$~\cite{wesszumino1971,barnichetal2000}, which follow from Proposition~\ref{thm_anomward}, are sufficiently strong such as to reduce this proof to purely algebraic manipulations. This is the essential advantage of the BRST method, and the detailed argument is given in Subsection~\ref{sec_brst_consistency}. Furthermore, the consistency conditions for the anomalies $\mathsf{A}_1$ and $\mathsf{A}_2$, which also follow from Proposition~\ref{thm_anomward}, give the properties of the quantum Slavnov-Taylor differential $\stq$ and the quantum antibracket $(\cdot,\cdot)_\hbar$ stated in Theorem~\ref{thm4}.

The natural way to prove the anomalous Ward identities in Proposition~\ref{thm_anomward} in the flow equation setup is by an argument involving a flow equation. For this, we need to go back to the regularised quantities with finite cutoffs $\Lambda, \Lambda_0$, but as we have already mentioned, the proposition is not expected to hold for these. To make progress, we need to consider yet another type of identity valid for finite cutoffs, which we call ``regularised Ward identity''. The proof of Proposition~\ref{thm_anomward} will ultimately follow from this regularised identity. Thus our chain of implications is altogether
\begin{quote}
Regularised Ward identity (Proposition~\ref{thm_anomward_reg}, Subsection~\ref{sec_brst_anomward}) $\Rightarrow$

Unregulated, anomalous Ward identity (Proposition~\ref{thm_anomward}, Subsections~\ref{sec_brst_anomward1}--\ref{sec_brst_anomward3}) $\Rightarrow$

Consistency conditions on the anomaly (Subsection~\ref{sec_brst_consistency}) $\Rightarrow$

Ward identity (Proposition~\ref{thm_brst}) $\Rightarrow$

Theorem~\ref{thm4}.
\end{quote}

\subsection{Regularised Ward identity}
\label{sec_brst_anomward}

To set up our regularised identity and machinery, we first introduce a regularised free action
\begin{equation}
\label{s0_lambda0_def}
S_0^{\Lambda_0} \equiv \frac{1}{2} \left\langle \phi_K. \left( C^{0, \Lambda_0} \right)^{-1}_{KL} \ast \phi_L \right\rangle - \left\langle \brst_0 \phi_K, \phi_K^\ddag \right\rangle \eqend{,}
\end{equation}
a regularised antibracket
\begin{equation}
\label{bvreg_def}
(F, G)^{\Lambda_0} \equiv \left\langle \frac{\delta_\text{R} F}{\delta \phi_K\vphantom{\delta \phi_K^\ddag}}, R^{\Lambda_0} \ast \frac{\delta_\text{L} G}{\delta \phi_K^\ddag} \right\rangle - \left\langle \frac{\delta_\text{R} F}{\delta \phi_K^\ddag}, R^{\Lambda_0} \ast \frac{\delta_\text{L} G}{\delta \phi_K\vphantom{\delta \phi_K^\ddag}} \right\rangle \eqend{,}
\end{equation}
which still satisfies the graded Jacobi identity exactly, and a regularised free Slavnov-Taylor differential
\begin{equation}
\st_0^{\Lambda_0} F \equiv \left( S_0^{\Lambda_0}, F \right)^{\Lambda_0} \eqend{.}
\end{equation}
Since in the unregularised limit $\Lambda_0 \to \infty$ we have $R^{\Lambda_0}(x-y) \to \delta^4(x-y)$, the regularised antibracket becomes the usual one, and since the free Slavnov-Taylor differential is linear, $\st_0^{\Lambda_0} F$ converges to $\st_0 F$ whenever $F$ is suitably convergent in that limit. Furthermore, we still have
\begin{equation}
\label{bvreg_s0_vanish}
\left( S_0^{\Lambda_0}, S_0^{\Lambda_0} \right)^{\Lambda_0} = 0
\end{equation}
exactly since the regulators cancel out. We also define a regulated ``BV Laplacian''
\begin{equation}
\label{bv_laplace_def}
\laplace^{\Lambda_0} F \equiv \left\langle \frac{\delta_\text{L}}{\delta \phi_K\vphantom{\delta \phi_K^\ddag}}, R^{\Lambda_0} \ast \frac{\delta_\text{R}}{\delta \phi_K^\ddag} \right\rangle F \eqend{,}
\end{equation}
and since $\st_0$ is linear and nilpotent we have
\begin{equation}
\label{bv_laplace_s0_vanish}
\laplace^{\Lambda_0} S_0^{\Lambda_0} = 0 \eqend{.}
\end{equation}

Using the properties of Gaussian measures (and the rules for fermionic ``integration'')~\cite{glimmjaffe1987,salmhofer1999,mueller2003} and integrating by parts, one can then easily establish a regularised Ward identity:
\begin{proposition}
\label{thm_anomward_reg}
For an arbitrary bosonic functional $B^{\Lambda_0}$ which is polynomial in the fields and antifields, depending on $\Lambda_0$ but not on $\Lambda$, and for non-exceptional momenta, we have
\begin{splitequation}
\label{anom_ward_general}
&\st_0^{\Lambda_0} \left[ \nu^{0, \Lambda_0} \conv \exp\left( - \frac{1}{\hbar} B^{\Lambda_0} \right) \right] \\
&\quad= - \frac{1}{2\hbar} \nu^{0, \Lambda_0} \conv \left[ \left[ \left( S_0^{\Lambda_0} + B^{\Lambda_0}, S_0^{\Lambda_0} + B^{\Lambda_0} \right)^{\Lambda_0} + 2 \hbar \laplace^{\Lambda_0} \left( S_0^{\Lambda_0} + B^{\Lambda_0} \right) \right] \exp\left( - \frac{1}{\hbar} B^{\Lambda_0} \right) \right] \eqend{,}
\end{splitequation}
with the convolution $\conv$ defined in equation~\eqref{conv_def}, understood as a shorthand for the hierarchy of identities obtained when we expand the above equations in $\hbar$ and in external fields and antifields.
\end{proposition}
For $B^{\Lambda_0} = L^{\Lambda_0}$~\eqref{lint} equal to the interaction part of the Lagrangian, in the naive unregularised limit $\Lambda_0 \to \infty$ the right-hand side of equation~\eqref{anom_ward_general} would reduce to the convolution with what is known as the ``Quantum Master Equation'' (there is no factor of $\mathi$ because we are working in a Euclidean setting)
\begin{equation}
(S,S) + 2 \hbar \laplace S \eqend{,}
\end{equation}
whose vanishing is often given as a condition for the gauge-fixing independence of classically gauge-invariant correlation functions~\cite{weinberg_v2,costello2011}. In our framework, we show that for a suitable choice of boundary conditions the Ward identity~\eqref{ward_0op} holds, and for those boundary conditions, the ``regulated Quantum Master Equation'' (defined by the right-hand side of equation~\eqref{anom_ward_general}) thus vanishes in the unregularised limit. This makes the connection of our work to other approaches clear. Actually, as explained previously, our approach goes even further and also quantifies gauge invariance of renormalised composite operators and ``contact type terms'', so in this sense our approach is more general than the quantum master equation.

In the remaining subsections, we will apply the foregoing proposition to the following type of functional:
\begin{equation}
B^{\Lambda_0} = L^{\Lambda_0} + \sum_{k=1}^s \left\langle \chi_k, \op_{A_k} + \delta^{\Lambda_0} \op_{A_k} \right\rangle \eqend{,}
\end{equation}
where $L^{\Lambda_0}$ is the interaction Lagrangian with counterterms~\eqref{lint}, where $\op_A$ are local composite operators~\eqref{op_def} with the counterterm map $\delta^{\Lambda_0}$~\eqref{op_ct_def}, and where $\chi_k \in \mathcal{S}(\mathbb{R}^4)$.

\subsection{Proof of Proposition~\ref{thm_anomward}, Equation~\texorpdfstring{\eqref{anom_ward_0op}}{(\ref{anom_ward_0op})}}
\label{sec_brst_anomward1}

For finite UV cutoff $\Lambda_0$ we do not obtain equation~\eqref{anom_ward_0op}. Instead, we have
\begin{proposition}
\label{lemma_anomward1}
For a general choice of boundary conditions for the functionals without insertions $L^{\Lambda, \Lambda_0}$, the following holds:
\begin{enumerate}
\item There exists a functional $W^{\Lambda, \Lambda_0}$ obeying a linear flow equation (see equation~\eqref{anom_ward_0op_w_flow}) such that
\begin{equation}
\label{anom_ward_0op_w}
\st_0^{\Lambda_0} L^{0, \Lambda_0} = W^{0,\Lambda_0} \eqend{.}
\end{equation}
\item There exists a choice of boundary conditions defining the functional with one insertion of the anomaly $\mathsf{A}_0$ (and thus the anomaly itself), such that the decomposition
\begin{equation}
\label{anom_0op_w_split}
W^{\Lambda,\Lambda_0} = L^{\Lambda, \Lambda_0}\left( \mathsf{A}_0 \right) + N^{\Lambda, \Lambda_0}
\end{equation}
holds, with yet another functional $N^{\Lambda, \Lambda_0}$ satisfying $\lim_{\Lambda_0 \to \infty} N^{\Lambda, \Lambda_0} = 0$.
\item The anomaly $\mathsf{A}_0$ defined in this way satisfies the conditions of Proposition~\ref{thm_anomward}.
\end{enumerate}
\end{proposition}
Equation~\eqref{anom_ward_0op} then immediately follows from this proposition by taking the unregularised limit $\Lambda_0 \to \infty$, $\Lambda \to 0$. To prove the proposition, we first define the functional $W^{\Lambda, \Lambda_0}$ by a flow equation and boundary conditions, in such a way that equation~\eqref{anom_ward_0op_w} holds. Our aim is then to show that in the unregularised limit we obtain equation~\eqref{anom_ward_0op} from equation~\eqref{anom_ward_0op_w}, \ie, that $W^{0, \Lambda_0} \to L^{0, \Lambda_0}\left( \mathsf{A}_0 \right)$ as $\Lambda_0 \to \infty$ for an anomaly $\mathsf{A}_0$ that satisfies the conditions of Proposition~\ref{thm_anomward}. This is done via the decomposition of $W^{\Lambda, \Lambda_0}$ into two contributions~\eqref{anom_0op_w_split}, and proving that the unwanted contribution vanishes in the unregularised limit, again via a flow equation argument. Similar ideas appeared first in the context of a proof of the Ward identities for QED~\cite{kellerkopper1991}, and later for spontaneously broken $\mathrm{SU}(2)$ gauge theories~\cite{koppermueller2000a,koppermueller2000b,koppermueller2009}. As a last point, we have to verify the conditions on the anomaly stated in Proposition~\ref{thm_anomward}.

\begin{proof}[Proof of Equation~\texorpdfstring{\eqref{anom_ward_0op_w}}{(\ref{anom_ward_0op_w})}.]
In the case $B^{\Lambda_0} = L^{\Lambda_0}$, the regularised Ward identity, equation~\eqref{anom_ward_general}, reduces to
\begin{equation}
\label{anom_ward_0op_start}
\st_0^{\Lambda_0} \left[ \nu^{0, \Lambda_0} \conv \exp\left( - \frac{1}{\hbar} L^{\Lambda_0} \right) \right] = \st_0^{\Lambda_0} \exp\left( - \frac{1}{\hbar} L^{0, \Lambda_0} \right) = - \frac{1}{\hbar} \nu^{0, \Lambda_0} \conv \left[ W^{\Lambda_0} \exp\left( - \frac{1}{\hbar} L^{\Lambda_0} \right) \right] \eqend{,}
\end{equation}
where we used the definition~\eqref{l_0op_def} of the generating functional without operator insertions, and defined $W^{\Lambda_0}$ by
\begin{equation}
\label{anom_0op_w_irrelevant}
W^{\Lambda_0} \equiv \frac{1}{2} \left( S_0^{\Lambda_0} + L^{\Lambda_0}, S_0^{\Lambda_0} + L^{\Lambda_0} \right)^{\Lambda_0} + \hbar \laplace^{\Lambda_0} \left( S_0^{\Lambda_0} + L^{\Lambda_0} \right) \eqend{.}
\end{equation}
Since $\st_0^{\Lambda_0}$ is a derivation and thus obeys the Leibniz rule, we obtain
\begin{equation}
\st_0^{\Lambda_0} \exp\left( - \frac{1}{\hbar} L^{0, \Lambda_0} \right) = - \frac{1}{\hbar} \exp\left( - \frac{1}{\hbar} L^{0, \Lambda_0} \right) \st_0^{\Lambda_0} L^{0, \Lambda_0} \eqend{,}
\end{equation}
and by defining
\begin{splitequation}
\label{anom_0op_w_def}
W^{\Lambda, \Lambda_0} &\equiv - \hbar \frac{\total}{\total t} \ln \left[ \nu^{\Lambda, \Lambda_0} \conv \exp\left( - \frac{1}{\hbar} L^{\Lambda_0} - \frac{t}{\hbar} W^{\Lambda_0} \right) \right]_{t=0} \\
&= \exp\left( \frac{1}{\hbar} L^{\Lambda, \Lambda_0} \right) \nu^{\Lambda, \Lambda_0} \conv \left[ \exp\left( - \frac{1}{\hbar} L^{\Lambda_0} \right) W^{\Lambda_0} \right] \eqend{,}
\end{splitequation}
equation~\eqref{anom_ward_0op_start} gives equation~\eqref{anom_ward_0op_w}. In the same manner as for functionals with one operator insertion (equation~\eqref{l_sop_flow} for $s=1$), taking a $\Lambda$-derivative of the definition~\eqref{anom_0op_w_def} we obtain a flow equation for $W^{\Lambda,\Lambda_0}$, which reads
\begin{equation}
\label{anom_ward_0op_w_flow}
\partial_\Lambda W^{\Lambda, \Lambda_0} = \frac{\hbar}{2} \left\langle \frac{\delta}{\delta \phi_K}, \left( \partial_\Lambda C^{\Lambda, \Lambda_0}_{KL} \right) \ast \frac{\delta}{\delta \phi_L} \right\rangle W^{\Lambda, \Lambda_0} - \left\langle \frac{\delta}{\delta \phi_K} L^{\Lambda, \Lambda_0}, \left( \partial_\Lambda C^{\Lambda, \Lambda_0}_{KL} \right) \ast \frac{\delta}{\delta \phi_L} W^{\Lambda, \Lambda_0} \right\rangle \eqend{.}
\end{equation}
The boundary conditions for irrelevant functionals are then given by $W^{\Lambda_0,\Lambda_0} = W^{\Lambda_0}$~\eqref{anom_0op_w_irrelevant}, and are non-vanishing because of the regulator $R^{\Lambda_0}$ in the definition of the regularised antibracket~\eqref{bvreg_def} and the BV Laplacian~\eqref{bv_laplace_def}. For relevant and marginal functionals, in Section~\ref{sec_bounds} we always put boundary conditions at $\Lambda = \mu$ and zero momentum to make the proofs simpler. Nevertheless, as explained at the end of Subsection~\ref{sec_pert_theory}, we can alternatively also put conditions at $\Lambda = 0$ and some non-exceptional momenta, and there is a one-to-one correspondence between these possibilities. Since equation~\eqref{anom_ward_0op_w} only holds for $\Lambda = 0$, this is the only possible choice here, such that the boundary conditions for relevant and marginal functionals of $W^{\Lambda, \Lambda_0}$ are obtained from equation~\eqref{anom_ward_0op_w}.
\end{proof}

\begin{proof}[Proof of Equation~\texorpdfstring{\eqref{anom_0op_w_split}}{(\ref{anom_0op_w_split})}.]
To perform the split~\eqref{anom_0op_w_split}, we define both $N^{\Lambda, \Lambda_0}$ and $L^{\Lambda, \Lambda_0}\left( \mathsf{A}_0 \right)$ by a flow equation and boundary conditions in such a way that equation~\eqref{anom_0op_w_split} holds. Since the flow equation for $W^{\Lambda, \Lambda_0}$~\eqref{anom_ward_0op_w_flow} is linear, the flow equation for functionals with one insertion of a composite operator is also linear, and since both flow equations have exactly the same structure, $N^{\Lambda, \Lambda_0}$ must satisfy a linear flow equation of the same structure as well. Equation~\eqref{anom_0op_w_split} then holds if the boundary conditions are chosen such that the sum of the boundary conditions of $N^{\Lambda, \Lambda_0}$ and $L^{\Lambda, \Lambda_0}\left( \mathsf{A}_0 \right)$ is equal to the boundary conditions of $W^{\Lambda, \Lambda_0}$. Concretely, we impose vanishing boundary conditions for all marginal and relevant functionals of $N^{\Lambda, \Lambda_0}$, and vanishing boundary conditions for all irrelevant functionals of $L^{\Lambda, \Lambda_0}\left( \mathsf{A}_0 \right)$, such that $N^{\Lambda, \Lambda_0}$ collects all the non-zero boundary conditions of $W^{\Lambda, \Lambda_0}$ for irrelevant functionals, and $L^{\Lambda, \Lambda_0}\left( \mathsf{A}_0 \right)$ collects all the non-zero boundary conditions of $W^{\Lambda, \Lambda_0}$ for marginal and relevant functionals. In this way, $L^{\Lambda, \Lambda_0}\left( \mathsf{A}_0 \right)$ is really a functional with an insertion of an integrated composite operator $\mathsf{A}_0$ of dimension $5$: it has vanishing boundary conditions for all relevant functionals of dimension $<5$ at $\Lambda = 0$ and zero momentum (since $L^{0, \Lambda_0}$ vanishes there, $\brst_0^{\Lambda_0}$ increases the dimension by $1$, and the regulator contained in $\st_0^{\Lambda_0}$ only changes these conditions in higher orders of relevancy).

It remains to show that $N^{\Lambda, \Lambda_0}$ vanishes in the unregularised limit $\Lambda_0 \to \infty$, which follows from the bounds~\eqref{bound_l1i_delta} if we can show that the boundary conditions for irrelevant functionals are compatible with the bounds~\eqref{bound_l1i}. For this, it is convenient to define the matrix $\mathsf{M}_{KL}$ by
\begin{equation}
\label{mixing_def}
\brst_0 \phi_L \equiv \phi_K \ast \mathsf{M}_{KL} \eqend{.}
\end{equation}
For Yang-Mills theories with the field-antifield coupling~\eqref{action_field_antifield_coupling} expressed in component form, the matrix $\mathsf{M}_{KL}$ is (in momentum space representation) given by
\begin{equation}
\mathsf{M}_{KL} = \begin{pmatrix} 0 & 0 & - \xi^2 p^\mu & 0 \\ - \mathi p_\mu & 0 & 0 & - \mathi \xi p^2 \\ 0 & 0 & 0 & 0 \\ 0 & 0 & \xi & 0 \end{pmatrix} \delta_{ab} \eqend{,}
\end{equation}
and since $\brst_0$ increases the dimension by $1$, we have in general (using equation~\eqref{antifield_dim})
\begin{equation}
\label{mixing_est}
\abs{ \partial^w \mathsf{M}_{KL}(p) } \leq c \abs{p}^{4-\abs{w}-[\phi_K]-[\phi_L^\ddag]} = c \abs{p}^{1-\abs{w}-[\phi_K]+[\phi_L]} \eqend{,}
\end{equation}
as long as the exponent is non-negative (otherwise the left-hand side simply vanishes). Furthermore, since $S_0$ has dimension $4$, from the explicit expression~\eqref{free_action} for the free action $S_0$ we obtain the estimate
\begin{equation}
\label{propagator_est}
\abs{ \partial^w \left( C^{0, \infty} \right)^{-1}_{KL} } \leq c \abs{p}^{4-\abs{w}-[\phi_K]-[\phi_L]}
\end{equation}
again as long as the exponent is non-negative, and a vanishing result otherwise.

As explained on the last page, the boundary conditions for irrelevant functionals are given by $N^{\Lambda_0, \Lambda_0} = W^{\Lambda_0}$, and $W^{\Lambda_0}$ can be estimated by writing equation~\eqref{anom_0op_w_irrelevant} in the explicit form (using equations~\eqref{bvreg_s0_vanish,bv_laplace_s0_vanish} and the explicit form of $S_0^{\Lambda_0}$~\eqref{s0_lambda0_def,mixing_def})
\begin{splitequation}
&N^{\Lambda_0, \Lambda_0} = \left\langle \phi_N \ast \mathsf{M}_{NM}, R^{\Lambda_0} \ast \frac{\delta_\text{L} L^{\Lambda_0}}{\delta \phi_M\vphantom{\delta \phi_M^\ddag}} \right\rangle + \left\langle \frac{\delta_\text{R} L^{\Lambda_0}}{\delta \phi_M^\ddag} \ast R^{\Lambda_0}, \mathsf{M}_{MN} \ast \phi_N^\ddag \right\rangle \\
&\quad+ \left\langle \phi_N, \left( C^{0, \infty} \right)^{-1}_{NM} \ast \frac{\delta_\text{L} L^{\Lambda_0}}{\delta \phi_M^\ddag} \right\rangle + \left\langle \frac{\delta_\text{R} L^{\Lambda_0}}{\delta \phi_M\vphantom{\delta \phi_M^\ddag}}, R^{\Lambda_0} \ast \frac{\delta_\text{L} L^{\Lambda_0}}{\delta \phi_M^\ddag} \right\rangle + \hbar \left\langle \frac{\delta_\text{L}}{\delta \phi_M\vphantom{\delta \phi_M^\ddag}}, R^{\Lambda_0} \ast \frac{\delta_\text{R}}{\delta \phi_M^\ddag} \right\rangle L^{\Lambda_0} \eqend{.}
\end{splitequation}
Taking now additional derivatives with respect to fields (and antifields), performing a Fourier transform (with the overall $\delta$ which enforces momentum conservation taken out), expanding in $\hbar$ and taking some momentum derivatives to obtain an irrelevant functional (which we denote by $\mathcal{N}$), we obtain
\begin{splitequation}
&\partial^\vec{w} \mathcal{N}^{\Lambda_0, \Lambda_0, l}_{\vec{K} \vec{L}^\ddag}(\vec{q}) = \sum_{i=1}^m \sum_{\vec{u}+\vec{v}\leq \vec{w}} c_{uvw} \left( \partial^{\vec{w}-\vec{u}-\vec{v}} R^{\Lambda_0}(q_i) \right) \left( \partial^\vec{u} \mathsf{M}_{K_i M}(q_i) \right) \left( \partial^\vec{v} \mathcal{L}^{\Lambda_0, \Lambda_0, l}_{\vec{K}_{\setminus i} \vec{L}^\ddag M}(\vec{q}_{\setminus i}, q_i) \right) \\
&\quad+ \sum_{j=1}^n \sum_{\vec{u}+\vec{v}\leq \vec{w}} c_{uvw} \left( \partial^{\vec{w}-\vec{u}-\vec{v}} R^{\Lambda_0}(q_{m+j}) \right) \left( \partial^\vec{u} \mathsf{M}_{M L_j}(q_{m+j}) \right) \left( \partial^\vec{v} \mathcal{L}^{\Lambda_0, \Lambda_0, l}_{\vec{K} \vec{L}_{\setminus j}^\ddag M^\ddag}(\vec{q}_{\setminus (m+j)}, q_{m+j}) \right) \\
&\quad+ \sum_{i=1}^m \sum_{\vec{v}\leq \vec{w}} c_{vw} \left( \partial^{\vec{w}-\vec{v}} \left( C^{0,\infty} \right)^{-1}_{K_i M}(q_i) \right) \left( \partial^\vec{w} \mathcal{L}^{\Lambda_0, \Lambda_0, l}_{\vec{K}_{\setminus i} \vec{L}^\ddag M^\ddag}(\vec{q}_{\setminus i}, q_i) \right) \\
&\quad+ \int R^{\Lambda_0}(p) \partial^\vec{w} \mathcal{L}^{\Lambda_0, \Lambda_0, l-1}_{\vec{K} \vec{L}^\ddag M M^\ddag}(\vec{q},-p,p) \frac{\total^4 p}{(2\pi)^4} \\
&\quad+\!\! \sum_{\subline{\sigma \cup \tau = \{1, \ldots, m\} \\ \rho \cup \varsigma = \{1, \ldots, n\}}} \sum_{l'=0}^l \sum_{\vec{u}+\vec{v}\leq \vec{w}} \!\! c_{uvw} \left( \partial^\vec{u} \mathcal{L}^{\Lambda_0, \Lambda_0, l'}_{\vec{K}_\sigma \vec{L}_\rho^\ddag M}(\vec{q}_\sigma,\vec{q}_\rho,-k) \right) \left( \partial^{\vec{w}-\vec{u}-\vec{v}} R^{\Lambda_0}(k) \right) \\
&\hspace{12em}\times \left( \partial^\vec{v} \mathcal{L}^{\Lambda_0, \Lambda_0, l-l}_{M^\ddag \vec{K}_\tau \vec{L}_\varsigma^\ddag}(k,\vec{q}_\tau,\vec{q}_\varsigma) \right) \eqend{,}
\end{splitequation}
where the momentum $k$ is defined by~\eqref{k_def}, and where
\begin{equations}
\vec{K}_{\setminus i} &\equiv K_1 \cdots K_{i-1} K_{i+1} \cdots K_m \eqend{,} \\
\vec{L}^\ddag_{\setminus j} &\equiv L_1^\ddag \cdots L_{j-1}^\ddag L_{j+1}^\ddag \cdots L_m^\ddag \eqend{,} \\
\vec{q}_{\setminus i} &\equiv (q_1,\ldots,q_{i-1},q_{i+1},\ldots q_{m+n}) \eqend{.}
\end{equations}
We then insert the bounds~\eqref{bound_l0} evaluated at $\Lambda = \Lambda_0$ for the functionals (note that since the bare action is polynomial in momenta there are no logarithms in momenta at $\Lambda = \Lambda_0$) and the bounds on the regulator~\eqref{r_prop_bound}, the matrix $\mathsf{M}$~\eqref{mixing_est} and the inverse of the covariance~\eqref{propagator_est} to obtain
\begin{splitequation}
\label{anomaly_0op_n}
\abs{\partial^\vec{w} \mathcal{N}^{\Lambda_0, \Lambda_0, l}_{\vec{K} \vec{L}^\ddag}(\vec{q})} &\leq \Bigg[ \sum_{i=1}^m \sum_{\vec{u}+\vec{v}\leq \vec{w}} \sup(\abs{q_i},\Lambda_0)^{-\abs{\vec{w}}+\abs{\vec{u}}+\abs{\vec{v}}} \mathe^{-\frac{\abs{q_i}^2}{2 \Lambda_0^2}} \abs{q_i}^{1-\abs{\vec{u}}-[\phi_{K_i}]+[\phi_M]} \\
&\hspace{12em}\times \sum_{T \in \mathcal{T}_{m+n}} \mathsf{G}^{T,\vec{v}}_{\vec{K}_{\setminus i} \vec{L}^\ddag M}(\vec{q}_{\setminus i}, q_i; \mu, \Lambda_0) \\
&\quad+ \sum_{j=1}^n \sum_{\vec{u}+\vec{v}\leq \vec{w}} \sup(\abs{q_{m+j}},\Lambda_0)^{-\abs{\vec{w}}+\abs{\vec{u}}+\abs{\vec{v}}} \mathe^{-\frac{\abs{q_{m+j}}^2}{2 \Lambda_0^2}} \abs{q_{m+j}}^{4-\abs{\vec{u}}-[\phi_M]-[\phi_{L_j}^\ddag]} \\
&\hspace{12em}\times \sum_{T \in \mathcal{T}_{m+n}} \mathsf{G}^{T,\vec{v}}_{\vec{K} \vec{L}_{\setminus j}^\ddag M^\ddag}(\vec{q}_{\setminus (m+j)}, q_{m+j}; \mu, \Lambda_0) \\
&\quad+ \sum_{i=1}^m \sum_{\vec{v}\leq \vec{w}} \abs{q_i}^{4-\abs{\vec{w}}+\abs{\vec{v}}-[\phi_{K_i}]-[\phi_M]} \sum_{T \in \mathcal{T}_{m+n}} \mathsf{G}^{T,\vec{v}}_{\vec{K}_{\setminus i} \vec{L}^\ddag M^\ddag}(\vec{q}_{\setminus i}, q_i; \mu, \Lambda_0) \\
&\quad+ \int \mathe^{-\frac{\abs{p}^2}{2 \Lambda_0^2}} \sum_{T \in \mathcal{T}_{m+n+2}} \mathsf{G}^{T,\vec{w}}_{\vec{K} \vec{L}^\ddag M M^\ddag}(\vec{q},-p,p; \mu, \Lambda_0) \frac{\total^4 p}{(2\pi)^4} \\
&\quad+ \sum_{\subline{\sigma \cup \tau = \{1, \ldots, m\} \\ \rho \cup \varsigma = \{1, \ldots, n\}}} \sum_{l'=0}^l \sum_{\vec{u}+\vec{v}\leq \vec{w}} \sum_{T \in \mathcal{T}_{\abs{\sigma}+\abs{\rho}+1}} \mathsf{G}^{T,\vec{u}}_{\vec{K}_\sigma \vec{L}_\rho^\ddag M}(\vec{q}_\sigma,\vec{q}_\rho,-k; \mu, \Lambda_0) \, \mathe^{-\frac{\abs{k}^2}{2 \Lambda_0^2}} \\
&\qquad\times \sup(\abs{k},\Lambda_0)^{-\abs{\vec{w}}+\abs{\vec{u}}+\abs{\vec{v}}} \!\!\!\! \sum_{T' \in \mathcal{T}_{\abs{\tau}+\abs{\varsigma}+1}} \mathsf{G}^{T',\vec{v}}_{M^\ddag \vec{K}_\tau \vec{L}_\varsigma^\ddag}(k,\vec{q}_\tau,\vec{q}_\varsigma; \mu, \Lambda_0) \Bigg] \mathcal{P}\left( \ln_+ \frac{\Lambda_0}{\mu} \right) \eqend{.} \raisetag{1.7\baselineskip}
\end{splitequation}
Let us start with the first term. Since $\mathsf{M}$ vanishes when too many derivatives act~\eqref{mixing_est}, the power of $\abs{q_i}$ in that term is always positive, and we estimate
\begin{equation}
\label{anomaly_0op_est1}
\sup(\abs{q_i},\Lambda_0)^{-\abs{\vec{w}}+\abs{\vec{u}}+\abs{\vec{v}}} \mathe^{-\frac{\abs{q_i}^2}{2 \Lambda_0^2}} \abs{q_i}^{1-\abs{\vec{u}}-[\phi_{K_i}]+[\phi_M]} \leq \sup(\abs{q_i},\Lambda_0)^{1-\abs{\vec{w}}+\abs{\vec{v}}-[\phi_{K_i}]+[\phi_M]} \eqend{.}
\end{equation}
We then change the external vertex of each tree from $M$ to $K_i$, which according to Table~\ref{table_weights} gives an extra factor of
\begin{equation}
\label{anomaly_0op_est2}
\sup(\abs{q_i}, \Lambda_0)^{-[\phi_M]+[\phi_{K_i}]} \eqend{,}
\end{equation}
and use the estimate~\eqref{func_0op_estfuse_b} to convert the $\vec{v}$ derivatives acting on the tree into $\vec{w}$ derivatives. The second and third term are treated in the same way, using additionally that $[\phi_M] + [\phi_M^\ddag] = 3$~\eqref{antifield_dim}. In the fourth (quadratic) term we fuse the trees using the estimate~\eqref{gw_fused_2_est}, which does give an additional factor of $\sup(\abs{k}, \Lambda_0)$, and convert the $\vec{u}+\vec{v}$ derivatives acting on the fused tree into $\vec{w}$ derivatives using the estimate~\eqref{func_0op_estfuse_b}. The integral over $p$ can be done after the rescaling $p = x \Lambda_0$ using Lemma~\ref{lemma_pint2} with $\beta_i = \gamma_i = 1$, and we obtain
\begin{splitequation}
\abs{\partial^\vec{w} \mathcal{N}^{\Lambda_0, \Lambda_0, l}_{\vec{K} \vec{L}^\ddag}(\vec{q})} &\leq \Bigg[ \sum_{T \in \mathcal{T}_{m+n}} \mathsf{G}^{T,\vec{w}}_{\vec{K} \vec{L}^\ddag}(\vec{q}; \mu, \Lambda_0) \sum_{i=1}^{m+n} \sup(\abs{q_i},\Lambda_0) \\
&\quad+ \Lambda_0^4 \sum_{T \in \mathcal{T}_{m+n+2}} \mathsf{G}^{T,\vec{w}}_{\vec{K} \vec{L}^\ddag M M^\ddag}(\vec{q},0,0; \mu, \Lambda_0) \\
&\quad+ \sum_{T \in \mathcal{T}_{m+n}} \mathsf{G}^{T,\vec{w}}_{\vec{K} \vec{L}^\ddag}(\vec{q}; \mu, \Lambda_0) \sum_{\subline{\sigma \cup \tau = \{1, \ldots, m\} \\ \rho \cup \varsigma = \{1, \ldots, n\}}} \sup(\abs{k},\Lambda_0) \Bigg] \mathcal{P}\left( \ln_+ \frac{\Lambda_0}{\mu} \right) \eqend{.}
\end{splitequation}
From the trees in the second term, we have to amputate the external legs corresponding to $M$ and $M^\ddag$. The amputation gives a factor of~\eqref{amputate}
\begin{equation}
\Lambda_0^{-[\phi_M]-[\phi_M^\ddag]} = \Lambda_0^{-3} \leq \Lambda_0^{-4} \sup(\abs{\vec{q}},\Lambda_0) \eqend{,}
\end{equation}
and for the other terms we also estimate $\sup(\abs{k}, \Lambda_0), \sup(\abs{q_i}, \Lambda_0) \leq \sup(\abs{\vec{q}},\Lambda_0)$. This extra factor can be absorbed in the particular weight factor of the trees, such that
\begin{equation}
\abs{\partial^\vec{w} \mathcal{N}^{\Lambda_0, \Lambda_0, l}_{\vec{K} \vec{L}^\ddag}(\vec{q})} \leq \sum_{T \in \mathcal{T}_{m+n}} \mathsf{G}^{T,\vec{w}}_{\vec{K} \vec{L}^\ddag; 1}(\vec{q}; \mu, \Lambda_0) \, \mathcal{P}\left( \ln_+ \frac{\Lambda_0}{\mu} \right) \eqend{,}
\end{equation}
which is compatible with the bounds~\eqref{bound_l1i} for an integrated operator of dimension $5$. Thus, $N^{\Lambda, \Lambda_0}$ satisfies a linear flow equation with vanishing boundary conditions for all relevant and marginal functionals, and boundary conditions for the irrelevant functionals which vanish in the limit $\Lambda_0 \to \infty$. We can thus apply Proposition~\ref{thm_l1i_van} to $N^{\Lambda, \Lambda_0}$ (with the obvious change in notation), which gives the bounds~\eqref{bound_l1i_delta}. Since these bounds contain an explicit factor of $\left( \sup(\mu, \Lambda) / \Lambda_0 \right)^\frac{\Delta}{2}$ and are otherwise independent of $\Lambda_0$, we conclude that $\lim_{\Lambda_0 \to \infty} N^{\Lambda, \Lambda_0} = 0$.
\end{proof}

\begin{proof}[Proof of the properties of \texorpdfstring{$\mathsf{A}_0$}{A0}.]
It was already shown that $\mathsf{A}_0$ is an integrated composite operator of dimension $5$. Since $\st_0^{\Lambda_0}$ increases the ghost number by $1$, equations~\eqref{anom_ward_0op_w} and~\eqref{anom_0op_w_split} show that $\mathsf{A}_0$ has ghost number $1$. With our conventions, $\mathsf{A}_0$ is thus the integral of a $4$-form. Furthermore, since $L^{\Lambda, \Lambda_0}\left( \mathsf{A}_0 \right)$ satisfies a linear flow equation with vanishing boundary conditions for all relevant functionals, the marginal functionals at order $\hbar^0$ are independent of $\Lambda$ and thus equal to the marginal part of $W^{\Lambda_0}$~\eqref{anom_0op_w_irrelevant} at that order. However, as $\Lambda_0 \to \infty$, we obtain at order $\hbar^0$ that
\begin{equation}
W^{\Lambda_0} \to \frac{1}{2} (S,S) = 0 \eqend{,}
\end{equation}
since the classical theory is gauge invariant~\eqref{classical_brst_action}. Then all irrelevant functionals vanish as well, such that
\begin{equation}
\label{anom_0op_orderh}
L^{0, \infty}\left( \mathsf{A}_0 \right) = \bigo{\hbar^k}
\end{equation}
for some $k \geq 1$.
\end{proof}

\subsection{Proof of Proposition~\ref{thm_anomward}, Equations~\texorpdfstring{\eqref{anom_ward_1op_a}}{(\ref{anom_ward_1op_a})} and~\texorpdfstring{\eqref{anom_ward_1op}}{(\ref{anom_ward_1op})}}
\label{sec_brst_anomward2}

Similarly to the previous subsection, for finite UV cutoff $\Lambda_0$ we have
\begin{proposition}
\label{lemma_anomward2}
For a general choice of boundary conditions for the functionals without insertions $L^{\Lambda, \Lambda_0}$, the following holds:
\begin{enumerate}
\item There exists a functional $W^{\Lambda, \Lambda_0}\left( \op_A \right)$ satisfying an inhomogeneous flow equation such that
\begin{equation}
\label{anom_ward_1op_w}
\st_0^{\Lambda_0} L^{0, \Lambda_0}\left( \op_A \right) = L^{0, \Lambda_0}\left( \st\op_A \right) + W^{0, \Lambda_0}\left( \op_A \right) \eqend{.}
\end{equation}
\item There exists a choice of boundary conditions defining the functional with one insertion of the anomaly $\mathsf{A}_1$ (and thus the anomaly itself), such that the decomposition
\begin{equation}
\label{anom_1op_w_split}
W^{\Lambda, \Lambda_0}\left( \op_A \right) = L^{\Lambda, \Lambda_0}\left( \mathsf{A}_1\left( \op_A\right) \right) + N^{\Lambda, \Lambda_0}\left( \op_A \right)
\end{equation}
holds for $\op_A = \mathsf{A}_0$, with yet another functional $N^{\Lambda, \Lambda_0}\left( \op_A \right)$ satisfying $N^{\Lambda, \infty}\left( \op_A \right) = 0$. For a choice of boundary conditions for the functionals without insertions $L^{\Lambda, \Lambda_0}$ such that $\mathsf{A}_0 = 0$, equation~\eqref{anom_1op_w_split} holds for all composite operators $\op_A$.
\item The anomaly $\mathsf{A}_1$ defined in this way satisfies the conditions of Proposition~\ref{thm_anomward}.
\end{enumerate}
\end{proposition}
Equations~\eqref{anom_ward_1op_a} and~\eqref{anom_ward_1op} immediately follow from this proposition by taking the unregularised limit $\Lambda_0 \to \infty$, $\Lambda \to 0$. To prove it, we follow the same steps as in the previous subsection: Definition of $W^{\Lambda, \Lambda_0}\left( \op_A \right)$ by a flow equation and boundary condition such that equation~\eqref{anom_ward_1op_w} holds, decomposition of $W^{\Lambda, \Lambda_0}\left( \op_A \right)$ into two contributions of which one is vanishing in the unregularised limit $\Lambda_0 \to \infty$, and verifying the remaining properties of the anomaly $\mathsf{A}_1$.

\begin{proof}[Proof of Equation~\texorpdfstring{\eqref{anom_ward_1op_w}}{(\ref{anom_ward_1op_w})}.]
We take the functional $B^{\Lambda_0} = L^{\Lambda_0} + \left\langle \chi, \op_A + \delta^{\Lambda_0} \op_A \right\rangle$ in the regularised Ward identity~\eqref{anom_ward_general}. Taking a variational derivative with respect to $\chi$ and using equation~\eqref{anom_ward_0op_w} and the definition of the classical Slavnov-Taylor differential~\eqref{st_def}, one easily verifies that equation~\eqref{anom_ward_general} reduces to equation~\eqref{anom_ward_1op_w} with the functional $W^{\Lambda, \Lambda_0}\left( \op_A(x) \right)$ defined by
\begin{splitequation}
\label{anom_1op_w_def}
W^{\Lambda, \Lambda_0}\left( \op_A(x) \right) &\equiv - \hbar \frac{\total}{\total \chi(x)} \frac{\total}{\total t} \ln \left[ \nu^{\Lambda, \Lambda_0} \conv \exp\left( - \frac{1}{\hbar} L^{\Lambda_0} - \frac{1}{\hbar} \left\langle \chi, \op_A + \delta^{\Lambda_0} \op_A \right\rangle \right. \right. \\
&\hspace{14em}\left. \left. - \frac{t}{\hbar} W^{\Lambda_0} - \frac{t}{\hbar} \left\langle \chi, W^{\Lambda_0}\left( \op_A \right) \right\rangle \right) \right]_{t = \chi = 0}
\end{splitequation}
with
\begin{splitequation}
\label{anom_1op_w_irrelevant}
W^{\Lambda_0}\left( \op_A(x) \right) &\equiv \left( S_0^{\Lambda_0} + L^{\Lambda_0}, \op_A(x) + \delta^{\Lambda_0} \op_A(x) \right)^{\Lambda_0} + \hbar \laplace^{\Lambda_0} \left( \op_A(x) + \delta^{\Lambda_0} \op_A(x) \right) \\
&\qquad- (\st\op_A)(x) - \delta^{\Lambda_0} (\st\op_A)(x) \eqend{.}
\end{splitequation}
Taking a $\Lambda$ derivative of equation~\eqref{anom_1op_w_def} and using the flow equation for $W^{\Lambda, \Lambda_0}$~\eqref{anom_ward_0op_w_flow}, we obtain a flow equation for $W^{\Lambda, \Lambda_0}\left( \op_A(x) \right)$, which is the same as the flow equation for a functional with two operator insertions, equation~\eqref{l_sop_flow_hierarchy} with $s = 2$. The boundary conditions are given by $W^{\Lambda_0, \Lambda_0}\left( \op_A(x) \right) = W^{\Lambda_0}\left( \op_A(x) \right)$, but we may alternatively also use equation~\eqref{anom_ward_1op_w} to obtain boundary conditions at $\Lambda = 0$ and non-exceptional momenta.
\end{proof}

\begin{proof}[Proof of Equation~\texorpdfstring{\eqref{anom_1op_w_split}}{(\ref{anom_1op_w_split})}.]
We define the functional $L^{\Lambda, \Lambda_0}\left( \mathsf{A}_1\left( \op_A\right)(x) \right)$ by the linear flow equation~\eqref{l_sop_flow} with $s = 1$, vanishing boundary conditions at $\Lambda = \Lambda_0$ for irrelevant functionals and boundary conditions for the relevant and marginal functionals given by $L^{0, \Lambda_0}\left( \mathsf{A}_1\left( \op_A\right)(x) \right) = W^{0, \Lambda_0}\left( \op_A(x) \right)$ for non-exceptional momenta, which can be read off from equation~\eqref{anom_ward_1op_w}. This makes $\mathsf{A}_1\left( \op_A\right)$ a composite operator of dimension $[\mathsf{A}_1\left(\op_A\right)] = [\op_A]+1$ depending linearly on $\op_A$, such that $\mathsf{A}_1$ is a map as stated in Proposition~\ref{thm_anomward}.

We then define $N^{\Lambda, \Lambda_0}(\op_A)$ to be the difference
\begin{equation}
\label{anom_1op_w_split_op}
N^{\Lambda, \Lambda_0}\left( \op_A(x) \right) \equiv W^{\Lambda, \Lambda_0}\left( \op_A(x) \right) - L^{\Lambda, \Lambda_0}\left( \op_A(x) \otimes \mathsf{A}_0 \right) - L^{\Lambda, \Lambda_0}\left( \mathsf{A}_1\left( \op_A\right)(x) \right) \eqend{.}
\end{equation}
The second functional on the right-hand side $L^{\Lambda, \Lambda_0}\left( \op_A(x) \otimes \mathsf{A}_0 \right)$ is a functional with one insertion of a non-integrated and one insertion of an integrated composite operator, for which we did not derive bounds. However, since the anomaly $\mathsf{A}_0$ is Grassmann odd, $L^{\Lambda, \Lambda_0}\left( \mathsf{A}_0 \otimes \mathsf{A}_0 \right) = 0$ already for finite cutoffs $\Lambda$ and $\Lambda_0$, while for general composite operators $\op_A$ we only need to treat the case where $\mathsf{A}_0 = 0$ and this functional also vanishes. Thus, equation~\eqref{anom_1op_w_split_op} is the same as the decomposition~\eqref{anom_1op_w_split} in all cases relevant for us. Nevertheless, equation~\eqref{anom_1op_w_split_op} is important to obtain the correct flow equation for $N^{\Lambda, \Lambda_0}\left( \op_A(x) \right)$: by taking a $\Lambda$ derivative of this equation and using the flow equations~\eqref{anom_1op_w_def} for $W^{\Lambda, \Lambda_0}\left( \op_A \right)$ and~\eqref{l_sop_flow} for $L^{\Lambda, \Lambda_0}\left( \mathsf{A}_1\left( \op_A\right)(x) \right)$, we obtain
\begin{splitequation}
\label{n_op_flow}
\partial_\Lambda N^{\Lambda, \Lambda_0}(\op_A) &= \frac{\hbar}{2} \left\langle \frac{\delta}{\delta \phi_K}, \left( \partial_\Lambda C^{\Lambda, \Lambda_0}_{KL} \right) \ast \frac{\delta}{\delta \phi_L} \right\rangle N^{\Lambda, \Lambda_0}(\op_A) \\
&\quad- \left\langle \frac{\delta}{\delta \phi_K} L^{\Lambda, \Lambda_0}, \left( \partial_\Lambda C^{\Lambda, \Lambda_0}_{KL} \right) \ast \frac{\delta}{\delta \phi_L} N^{\Lambda, \Lambda_0}(\op_A) \right\rangle \\
&\quad- \left\langle \frac{\delta}{\delta \phi_K} L^{\Lambda, \Lambda_0}(\op_A), \left( \partial_\Lambda C^{\Lambda, \Lambda_0}_{KL} \right) \ast \frac{\delta}{\delta \phi_L} N^{\Lambda, \Lambda_0} \right\rangle \eqend{,}
\end{splitequation}
which is similar to the flow equation for a functional with two insertions, but contains the functional $N^{\Lambda, \Lambda_0}$ in the last line instead of $W^{\Lambda, \Lambda_0}$ (which naively would have been obtained from the decomposition~\eqref{anom_1op_w_split}). The boundary conditions for $N^{\Lambda, \Lambda_0}(\op_A)$ can be read off from the definition~\eqref{anom_1op_w_split_op} and the boundary conditions that we imposed on $L^{\Lambda, \Lambda_0}\left( \mathsf{A}_1\left( \op_A\right)(x) \right)$. They are of the form appropriate for a single insertion: for relevant and marginal functionals, we have vanishing boundary conditions at $\Lambda = 0$ and non-exceptional momenta, and for the irrelevant functionals they are equal to $W^{\Lambda_0}\left( \op_A \right)$~\eqref{anom_1op_w_irrelevant}. It remains to show that $N^{\Lambda, \Lambda_0}(\op_A)$ vanishes in the unregularised limit $\Lambda_0 \to \infty$, which follows almost immediately from the bounds~\eqref{bound_l1_delta} if we can show that the boundary conditions for irrelevant functionals are compatible with the bounds~\eqref{bound_l1}. The only obstacle is that the bounds~\eqref{bound_l1_delta} were derived for a linear flow equation, while $N^{\Lambda, \Lambda_0}(\op_A)$ satisfies the flow equation~\eqref{n_op_flow} with an additional term. However, since we proved in the last subsection that the bounds~\eqref{bound_l1i_delta} apply to $N^{\Lambda, \Lambda_0}$, this additional term can be estimated in exactly the same way as the second term on the right-hand side of the flow equation~\eqref{n_op_flow}, and we can apply the proof for functionals with one (non-integrated) operator insertion of Subsection~\ref{sec_bounds_additional}. Thus, to show that $N^{\Lambda, \infty}(\op_A) = 0$ we only have to prove that the boundary conditions for irrelevant functionals are compatible with the bounds~\eqref{bound_l1}, which is done below.

For irrelevant functionals, the last two terms of equation~\eqref{anom_1op_w_irrelevant} vanish by definition, and we have $\op_A(x) + \delta^{\Lambda_0} \op_A(x) = L^{\Lambda_0, \Lambda_0}(\op_A(x))$, such that for irrelevant functionals we get
\begin{splitequation}
&N^{\Lambda_0, \Lambda_0}(\op_A(x)) = \left( S_0^{\Lambda_0} + L^{\Lambda_0}, L^{\Lambda_0, \Lambda_0}(\op_A(x)) \right)^{\Lambda_0} + \hbar \laplace^{\Lambda_0} L^{\Lambda_0, \Lambda_0}(\op_A(x)) \\
&\quad= \left\langle \phi_N, \left( C^{0, \infty} \right)^{-1}_{NM} \ast \frac{\delta_\text{L} L^{\Lambda_0, \Lambda_0}(\op_A(x))}{\delta \phi_M^\ddag} \right\rangle + \left\langle \frac{\delta_\text{R} L^{\Lambda_0, \Lambda_0}(\op_A(x))}{\delta \phi_M^\ddag} \ast R^{\Lambda_0}, \mathsf{M}_{MN} \ast \phi_N^\ddag \right\rangle \\
&\qquad+ \left\langle \phi_N \ast \mathsf{M}_{NM}, R^{\Lambda_0} \ast \frac{\delta_\text{L} L^{\Lambda_0, \Lambda_0}(\op_A(x))}{\delta \phi_M} \right\rangle + \left\langle \frac{\delta_\text{R} L^{\Lambda_0}}{\delta \phi_K\vphantom{\delta \phi_K^\ddag}}, R^{\Lambda_0} \ast \frac{\delta_\text{L} L^{\Lambda_0, \Lambda_0}(\op_A(x))}{\delta \phi_K^\ddag} \right\rangle \\
&\qquad- \left\langle \frac{\delta_\text{R} L^{\Lambda_0}}{\delta \phi_K^\ddag}, R^{\Lambda_0} \ast \frac{\delta_\text{L} L^{\Lambda_0, \Lambda_0}(\op_A(x))}{\delta \phi_K\vphantom{\delta \phi_K^\ddag}} \right\rangle + \hbar \left\langle \frac{\delta_\text{L}}{\delta \phi_K}, R^{\Lambda_0} \ast \frac{\delta_\text{R}}{\delta \phi_K^\ddag} \right\rangle L^{\Lambda_0, \Lambda_0}(\op_A(x))
\end{splitequation}
using the definition of the regularised antibracket~\eqref{bvreg_def}, the regularised free action~\eqref{s0_lambda0_def,mixing_def} and the regularised BV Laplacian~\eqref{bv_laplace_def}. Overall translation invariance tells us that $N^{\Lambda, \Lambda_0}(\op_A)$ fulfils a shift property analogous to the one for functionals with one operator insertion (equation~\eqref{func_sop_shift} with $s=1$), such that we may restrict to $x = 0$. We then take some functional derivatives with respect to fields and antifields, perform a Fourier transform and take some momentum derivatives to obtain an irrelevant functional (where $m+n+\abs{\vec{w}} > [\op_A]+1$). Inserting the bounds~\eqref{bound_l0} and~\eqref{bound_l1} evaluated at $\Lambda = \Lambda_0$ for the functionals (again without the logarithms in momenta), the bounds on the regulator~\eqref{r_prop_bound}, the matrix $\mathsf{M}$~\eqref{mixing_est} and the inverse of the covariance~\eqref{propagator_est}, one obtains a bound for $\mathcal{N}^{\Lambda_0, \Lambda_0, l}$ similar to~\eqref{anomaly_0op_n}. The various terms can be estimated in the same way as in the last subsection, using the estimate~\eqref{anomaly_0op_est1}, changing vertices according to~\eqref{anomaly_0op_est2}, fusing the trees in the quadratic term according to the estimate~\eqref{gw_fused_2_est}, changing derivatives according to~\eqref{func_0op_estfuse_b}, performing the $p$ integral using Lemma~\ref{lemma_pint2} and amputating vertices using~\eqref{amputate}. We then obtain the bounds
\begin{splitequation}
\label{n_op_boundary}
\abs{\partial^\vec{w} \mathcal{N}^{\Lambda_0, \Lambda_0, l}_{\vec{K} \vec{L}^\ddag}\left( \op_A(0); \vec{q} \right)} &\leq \sup\left( 1, \frac{\abs{\vec{q}}}{\Lambda_0} \right)^{g^{(1)}([\op_A]+1,m+n+2l,\abs{\vec{w}})} \\
&\quad\times \sum_{T^* \in \mathcal{T}^*_{m+n}} \mathsf{G}^{T^*,\vec{w}}_{\vec{K} \vec{L}^\ddag; [\op_A]+1}(\vec{q}; \mu, \Lambda_0) \, \mathcal{P}\left( \ln_+ \frac{\Lambda_0}{\mu} \right) \eqend{.}
\end{splitequation}
The boundary conditions~\eqref{n_op_boundary} are then compatible with the bounds~\eqref{bound_l1} for functionals with one operator insertion of dimension $[\op_A]+1$, evaluated at $\Lambda = \Lambda_0$. Thus, $N^{\Lambda, \Lambda_0}\left( \op_A \right)$ fulfils a linear flow equation with vanishing boundary conditions for all relevant and marginal functionals, and boundary conditions for the irrelevant functionals which vanish in the limit $\Lambda_0 \to \infty$. We can thus apply Proposition~\ref{thm_l1_van} to $N^{\Lambda, \Lambda_0}\left( \op_A \right)$ (with the obvious change in notation), and obtain a bound of the form~\eqref{bound_l1_delta} for $N^{\Lambda, \Lambda_0}\left( \op_A \right)$. Since this bound contains an explicit factor of $\left( \sup(\mu, \Lambda) / \Lambda_0 \right)^\frac{\Delta}{2}$, and is otherwise independent of $\Lambda_0$, we conclude that $\lim_{\Lambda_0 \to \infty} N^{\Lambda, \Lambda_0}\left( \op_A \right) = 0$ as claimed.
\end{proof}

\begin{proof}[Proof of the properties of \texorpdfstring{$\mathsf{A}_1$}{A1}.]
It was already shown that $\mathsf{A}_1$ is a map as stated in Proposition~\ref{thm_anomward} of the right dimension, and it remains to show that $\mathsf{A}_1(\op_A)$ is of higher order in $\hbar$ than $\op_A$. Let us assume that the first non-vanishing contribution to $L^{\Lambda, \Lambda_0}(\op_A)$ is of order $\hbar^k$. Since $L^{\Lambda, \Lambda_0}\left( \mathsf{A}_1\left(\op_A\right) \right)$ satisfies a linear flow equation with vanishing boundary conditions for all relevant functionals and boundary conditions which are linear in $\op_A$, it is at least of order $\hbar^k$ as well. However, the marginal functionals at order $\hbar^k$ are independent of $\Lambda$ and equal to their value at $\Lambda_0$, which is given by the marginal part of $W^{\Lambda_0}\left( \op_A \right)$~\eqref{anom_1op_w_irrelevant}. In the limit $\Lambda_0 \to \infty$, the marginal part of $W^{\Lambda_0}(\op_A)$ vanishes at order $\hbar^k$, and then all irrelevant functionals $L^{\Lambda, \Lambda_0}\left( \mathsf{A}_1\left(\op_A\right) \right)$ vanish as well at that order, such that
\begin{equation}
\label{anom_1op_orderh}
L^{0, \infty}\left( \mathsf{A}_1\left(\op_A\right) \right) = \bigo{\hbar^{k+l}}
\end{equation}
with $l \geq 1$.
\end{proof}

\subsection{Proof of Proposition~\ref{thm_anomward}, Equation~\texorpdfstring{\eqref{anom_ward_sop}}{(\ref{anom_ward_sop})}}
\label{sec_brst_anomward3}

In complete analogy to the previous subsections, for finite cutoff $\Lambda_0$ we have
\begin{proposition}
\label{lemma_anomward3}
For a choice of boundary conditions for the functionals without insertions $L^{\Lambda, \Lambda_0}$ such that $\mathsf{A}_0 = 0$, the following holds:
\begin{enumerate}
\item There exist functionals $W^{\Lambda, \Lambda_0}\left( \bigotimes_{k=1}^s \op_{A_k} \right)$ obeying inhomogeneous flow equations such that
\begin{splitequation}
\label{anom_ward_sop_w}
\st_0^{\Lambda_0} L^{0, \Lambda_0}\left( \bigotimes_{k=1}^s \op_{A_k} \right) &= \sum_{1 \leq l < l' \leq s} L^{0, \Lambda_0}\left( \bigotimes_{k\in\{1,\ldots,s\}\setminus\{l,l'\}} \op_{A_k} \otimes \left( \op_{A_l},\op_{A_{l'}} \right) \right) \\
&\quad+ \sum_{l=1}^s L^{0, \Lambda_0}\left( \bigotimes_{k\in\{1,\ldots,s\}\setminus\{l\}} \op_{A_k} \otimes \st\op_{A_l} \right) + W^{0, \Lambda_0}\left( \bigotimes_{k=1}^s \op_{A_k} \right)
\end{splitequation}
with the single composite operator $(\op_{A_k}, \op_{A_l})$. This operator is defined by the classical expression which we decompose as
\begin{equation}
\label{op_bv_def}
\left( \op_{A_k}(x_k),\op_{A_l}(x_l) \right) = \sum_C \op_C(x_k) P^C_{AB}\left( \partial \right) \delta^4(x_k-x_l) \eqend{,}
\end{equation}
where $P^C_{AB}$ are homogeneous, $\mathrm{O}(4)$-covariant polynomials of order $[\op_A]+[\op_B]-[\op_C]-3$ which are uniquely determined by the left-hand side, together with counterterms $\delta^{\Lambda_0} (\op_{A_k}, \op_{A_l})$ given by
\begin{equation}
\delta^{\Lambda_0} \left( \op_{A_k}(x_k),\op_{A_l}(x_l) \right) = \sum_C \delta^{\Lambda_0} \op_C(x_k) P^C_{AB}\left( \partial \right) \delta^4(x_k-x_l)
\end{equation}
with the appropriate counterterms $\delta^{\Lambda_0} \op_C$ for the single operators $\op_C$.
\item There exists a choice of boundary conditions defining the functional with one insertion of the anomaly $\mathsf{A}_2$ (and thus the anomaly itself), such that the decomposition
\begin{splitequation}
\label{anom_sop_w_split}
W^{\Lambda, \Lambda_0}\left( \bigotimes_{k=1}^s \op_{A_k} \right) &= N^{\Lambda, \Lambda_0}\left( \bigotimes_{k=1}^s \op_{A_k} \right) + \sum_{l=1}^s L^{\Lambda, \Lambda_0}\left( \bigotimes_{k\in\{1,\ldots,s\}\setminus\{l\}} \op_{A_k} \otimes \mathsf{A}_1\left(\op_{A_l}\right) \right) \\
&\quad+ \sum_{1 \leq l < l' \leq s} L^{\Lambda, \Lambda_0}\left( \bigotimes_{k\in\{1,\ldots,s\}\setminus\{l,l'\}} \op_{A_k} \otimes \mathsf{A}_2\left( \op_{A_l} \otimes \op_{A_{l'}} \right) \right)
\end{splitequation}
holds, where the anomaly $\mathsf{A}_2\left( \op_{A_k} \otimes \op_{A_l} \right)$ is a composite operator of dimension $\leq [\op_{A_k}]+[\op_{A_l}]-3$ supported on the diagonal $x_k = x_l$, with yet other functionals $N^{\Lambda, \Lambda_0}\left( \bigotimes_{k=1}^s \op_{A_k}(x_k) \right)$ satisfying $N^{\Lambda, \infty}\left( \bigotimes_{k=1}^s \op_{A_k}(x_k) \right) = 0$.
\item The anomaly $\mathsf{A}_2$ defined in this way satisfies the conditions of Proposition~\ref{thm_anomward}.
\end{enumerate}
\end{proposition}
Equation~\eqref{anom_ward_sop} immediately follows from this proposition by taking the unregularised limit $\Lambda_0 \to \infty$, $\Lambda \to 0$. To prove it, we again follow the same steps as in the previous subsections.

\begin{proof}[Proof of Equation~\texorpdfstring{\eqref{anom_ward_sop_w}}{(\ref{anom_ward_sop_w})}.]
We now take $B^{\Lambda_0} = L^{\Lambda_0} + \sum_{k=1}^s \left\langle \chi_k, \op_{A_k} + \delta^{\Lambda_0} \op_{A_k} \right\rangle$ with $s > 1$ in the regulated anomalous Ward identity~\eqref{anom_ward_general}. Taking variational derivatives with respect to the $\chi_k$ and using equations~\eqref{anom_ward_0op_w} and~\eqref{anom_ward_1op_w}, equation~\eqref{anom_ward_general} reduces to equation~\eqref{anom_ward_sop_w} with the functional $W^{\Lambda, \Lambda_0}\left( \bigotimes_{k=1}^s \op_{A_k} \right)$ defined by
\begin{splitequation}
\label{anom_sop_w_def}
W^{\Lambda, \Lambda_0}\left( \bigotimes_{k=1}^s \op_{A_k} \right) &\equiv - \hbar \left( \prod_{k=1}^s \frac{\total}{\total \chi_k(x_k)} \right) \frac{\total}{\total t} \ln \Bigg[ \nu^{\Lambda, \Lambda_0} \conv \exp\Bigg( - \frac{1}{\hbar} L^{\Lambda_0} - \frac{t}{\hbar} W^{\Lambda_0} \\
&\qquad- \frac{t}{\hbar} \sum_{k=1}^s \left\langle \chi_k, W^{\Lambda_0}\left( \op_{A_k} \right) \right\rangle - \frac{1}{\hbar} \sum_{k=1}^s \left\langle \chi_k, \op_{A_k} + \delta^{\Lambda_0} \op_{A_k} \right\rangle \\
&\qquad- \frac{t}{\hbar} \sum_{1 \leq l < l' \leq s} \left\langle \chi_l, W^{\Lambda_0}\left( \op_{A_l} \otimes \op_{A_{l'}} \right) \ast \chi_{l'} \right\rangle \Bigg) \Bigg]_{t = \chi_k = 0} \eqend{,}
\end{splitequation}
with
\begin{splitequation}
\label{anom_2op_w_irrelevant}
W^{\Lambda_0}\left( \op_{A_k}(x_k) \otimes \op_{A_l}(x_l) \right) &\equiv \left( \op_{A_k}(x_k) + \delta^{\Lambda_0} \op_{A_k}(x_k), \op_{A_l}(x_l) + \delta^{\Lambda_0} \op_{A_l}(x_l) \right)^{\Lambda_0} \\
&\quad- \left( \op_{A_k}(x_k),\op_{A_l}(x_l) \right) - \delta^{\Lambda_0} \left( \op_{A_k}(x_k),\op_{A_l}(x_l) \right) \eqend{.}
\end{splitequation}
The flow equation for $W^{\Lambda, \Lambda_0}\left( \bigotimes_{k=1}^s \op_{A_k} \right)$ is obtained by taking a $\Lambda$ derivative of the definition~\eqref{anom_sop_w_def} and using the flow equation for $W^{\Lambda, \Lambda_0}\left( \op_{A_k} \right)$ as defined in the last subsection as well as the flow equations for $L^{\Lambda, \Lambda_0}\left( \op_{A_k} \right)$ (a linear flow equation) and $L^{\Lambda, \Lambda_0}$, equation~\eqref{l_0op_flow_hierarchy}. This then shows that $W^{\Lambda, \Lambda_0}\left( \bigotimes_{k=1}^s \op_{A_k} \right)$ fulfils a flow equation of the type~\eqref{l_sop_flow_hierarchy}, but with $s \to s+1$. The boundary conditions are given by evaluating the definition~\eqref{anom_sop_w_def} at $\Lambda = \Lambda_0$ and using that the measure $\nu^{\Lambda, \Lambda_0}$ gives a $\delta$ measure in this limit. This directly gives $W^{\Lambda_0, \Lambda_0}\left( \bigotimes_{k=1}^s \op_{A_k} \right) = 0$ for $s > 2$, and $W^{\Lambda_0, \Lambda_0}\left( \op_{A_k}(x_k) \otimes \op_{A_l}(x_l) \right) = W^{\Lambda_0}\left( \op_{A_k}(x_k) \otimes \op_{A_l}(x_l) \right)$~\eqref{anom_2op_w_irrelevant} for $s = 2$, for all functionals.
\end{proof}

\begin{proof}[Proof of Equation~\texorpdfstring{\eqref{anom_sop_w_split}}{(\ref{anom_sop_w_split})}.]
Again, we would like to define $N^{\Lambda, \Lambda_0}\left( \bigotimes_{k=1}^s \op_{A_k} \right)$ to be the difference between the functional $W^{\Lambda, \Lambda_0}\left( \bigotimes_{k=1}^s \op_{A_k} \right)$ and the remaining functionals on the right-hand side of equation~\eqref{anom_sop_w_split}, and then show that with this definition we have $N^{\Lambda, \infty}\left( \bigotimes_{k=1}^s \op_{A_k} \right) = 0$. Before we can do this, we first have to define the functionals with an insertion of the anomaly $\mathsf{A}_2\left( \op_k \otimes \op_l \right)$ by an appropriate flow equation and boundary conditions, which is a bit more complicated. $\mathsf{A}_2\left( \op_k(x_k) \otimes \op_l(x_k) \right)$ should be a composite operator of dimension $\leq [\op_{A_k}]+[\op_{A_l}]-3$ supported on the diagonal $x_k = x_l$ and depending bilinearly on $\op_{A_k}$ and $\op_{A_l}$. Thus, we only have to define the functionals with one insertion of $\mathsf{A}_2$ and no other composite operator by a linear flow equation and determine their boundary conditions. The functionals with an insertion of $\mathsf{A}_2$ and other composite operators are then automatically well-defined. These boundary conditions are obtained from an expansion of $W^{\Lambda_0, \Lambda_0}\left( \op_k \otimes \op_l \right)$, which we now perform in detail.

Since~\eqref{l_1op_lambda0}
\begin{equation}
\op_A + \delta^{\Lambda_0} \op_A = L^{\Lambda_0, \Lambda_0}(\op_A) \eqend{,}
\end{equation}
we obtain after taking some functional derivatives with respect to fields and antifields, performing a Fourier transform and using the shift property~\eqref{func_sop_shift} to bring the position of the operator insertions to $0$
\begin{splitequation}
\label{anom_w_expr}
\mathcal{W}^{\Lambda_0, \Lambda_0}_{\vec{K} \vec{L}^\ddag}\left( \op_{A_k}(x_k) \otimes \op_{A_l}(x_l); \vec{q} \right) &= \int \mathe^{-\mathi (x_k - x_l) p} F(p) \frac{\total^4 p}{(2\pi)^4} \\
&\quad- \sum_C \mathcal{L}^{\Lambda_0, \Lambda_0}_{\vec{K} \vec{L}^\ddag}\left( \op_C(x_k); \vec{q} \right) P^C_{A_k A_l}\left( \partial_{x_k} \right) \delta^4(x_k-x_l) \eqend{,}
\end{splitequation}
with
\begin{splitequation}
&F(p) \equiv \sum_{\subline{\sigma \cup \tau = \{1, \ldots, m\} \\ \rho \cup \varsigma = \{1, \ldots, n\}}} c_{\sigma\tau\rho\varsigma} \mathe^{-\mathi x_k k} \mathe^{-\mathi x_l k'} R^{\Lambda_0}(p) \sum_{l'=0}^l \bigg[ \mathcal{L}^{\Lambda_0, \Lambda_0, l'}_{\vec{K}_\sigma \vec{L}_\rho^\ddag M}\left( \op_{A_k}(0); \vec{q}_\sigma,\vec{q}_\rho,p \right) \\
&\times \mathcal{L}^{\Lambda_0, \Lambda_0, l-l'}_{\vec{K}_\tau \vec{L}_\varsigma^\ddag M^\ddag}\left( \op_{A_l}(0); \vec{q}_\tau,\vec{q}_\varsigma,-p \right) - \mathcal{L}^{\Lambda_0, \Lambda_0, l'}_{\vec{K}_\sigma \vec{L}_\rho^\ddag M^\ddag}\left( \op_{A_k}(0); \vec{q}_\sigma,\vec{q}_\rho,p \right) \mathcal{L}^{\Lambda_0, \Lambda_0, l-l'}_{\vec{K}_\tau \vec{L}_\varsigma^\ddag M}\left( \op_{A_l}(0); \vec{q}_\tau,\vec{q}_\varsigma,-p \right) \bigg] \eqend{,}
\end{splitequation}
the momentum $k$ defined in equation~\eqref{k_def}, and with
\begin{equation}
\label{ks_def}
k' \equiv \sum_{i \in \tau \cup \varsigma} q_i \eqend{.}
\end{equation}
The bounds~\eqref{bound_l1} on functionals with one operator insertion and the bounds~\eqref{r_prop_bound} imply that $F(p)$ is smooth, such that we can perform a Taylor expansion with remainder up to the finite order $r=[\op_{A_k}]+[\op_{A_l}]-3$
\begin{equation}
\label{anom_sop_taylor}
F(p) = \sum_{\abs{w} \leq r} \frac{p^w}{w!} \left[ \partial^w F(0) \right] + (r+1) \sum_{\abs{w} = r+1} \frac{p^w}{w!} \int_0^1 (1-t)^r \left[ \partial^w F(t p) \right] \total t \eqend{.}
\end{equation}
To obtain bounds for the various terms in this expansion, we take $w$ derivatives of $F(p)$ with respect to the momentum $p$ and use the bounds~\eqref{bound_l1} for the functionals with one operator insertion (noting that the polynomials in logarithms are absent for $\Lambda = \Lambda_0$) and the bounds~\eqref{r_prop_bound} for the regulator to obtain
\begin{splitequation}
&\abs{ \partial^w F(0) } \leq \sum_{\subline{\sigma \cup \tau = \{1, \ldots, m\} \\ \rho \cup \varsigma = \{1, \ldots, n\}}} \sum_{l'=0}^l \sum_{u+v \leq w} \Lambda_0^{-\abs{w}+\abs{u}+\abs{v}} \sup\left( 1, \frac{\abs{\vec{q}_\sigma,\vec{q}_\rho}}{\Lambda_0} \right)^{g^{(1)}([\op_{A_k}],\abs{\sigma}+\abs{\rho}+1+2l',\abs{u})} \\
&\quad\times \mathcal{P}\left( \ln_+ \frac{\Lambda_0}{\mu} \right) \sup\left( 1, \frac{\abs{\vec{q}_\tau,\vec{q}_\varsigma}}{\Lambda_0} \right)^{g^{(1)}([\op_{A_l}],\abs{\tau}+\abs{\varsigma}+1+2(l-l'),\abs{v})} \\
&\quad\times \sum_{T^* \in \mathcal{T}^*_{\abs{\sigma}+\abs{\rho}+1}} \mathsf{G}^{T^*,(\vec{0},u)}_{\vec{K}_\sigma \vec{L}_\rho^\ddag M; [\op_{A_k}]}(\vec{q}_\sigma,\vec{q}_\rho,0; \mu, \Lambda_0)  \sum_{T^* \in \mathcal{T}^*_{\abs{\tau}+\abs{\varsigma}+1}} \mathsf{G}^{T^*,(\vec{0},v)}_{\vec{K}_\tau \vec{L}_\varsigma^\ddag M^\ddag; [\op_{A_l}]}(\vec{q}_\tau,\vec{q}_\varsigma,0; \mu, \Lambda_0) \eqend{.}
\end{splitequation}
We then fuse the trees using the estimate~\eqref{gw_fused_1_est}, the large-momentum factors using the estimates~\eqref{func_2op_estlog} and~\eqref{func_sop_g1g1est}, remove the derivative weight factor corresponding to the $u+v$ derivatives that acted on the momentum $p$ from the tree, which gives a factor~\eqref{gw_def} $\Lambda_0^{-\abs{u}-\abs{v}}$, and amputate the external legs corresponding to $M$ and $M^\ddag$, which gives an additional factor of
\begin{equation}
\Lambda_0^{-[\phi_M]-[\phi_M^\ddag]} = \Lambda_0^{-3}
\end{equation}
according to the estimate~\eqref{amputate} and equation~\eqref{antifield_dim}. Finally, we change the particular dimension of the tree from $[\op_{A_k}]+[\op_{A_l}]$ to $[\op_{A_k}]+[\op_{A_l}]-3-\abs{w}$, which according to~\eqref{particular_weight} gives an extra factor
\begin{equation}
\sup(\abs{\vec{q}}, \Lambda_0)^{3+\abs{w}} = \Lambda_0^{3+\abs{w}} \sup\left( 1, \frac{\abs{\vec{q}}}{\Lambda_0} \right)^{3+\abs{w}} \eqend{,}
\end{equation}
and since $- ([\op_{A_k}]+[\op_{A_l}]+D)+3+\abs{w} \leq 0$ for all $\abs{w} \leq [\op_{A_k}]+[\op_{A_l}]-3$ we obtain
\begin{splitequation}
\label{anom_sop_f0_bound}
\abs{ \partial^w F(0) } &\leq \sup\left( 1, \frac{\abs{\vec{q}}}{\Lambda_0} \right)^{g^{(2)}([\op_{A_k}]+[\op_{A_l}],m+n+2l,0)} \\
&\quad\times \sum_{T^* \in \mathcal{T}^*_{m+n}} \mathsf{G}^{T^*,0}_{\vec{K} \vec{L}^\ddag; [\op_{A_k}]+[\op_{A_l}]-3-\abs{w}}(\vec{q}; \mu, \Lambda_0) \, \mathcal{P}\left( \ln_+ \frac{\Lambda_0}{\mu} \right) \eqend{.}
\end{splitequation}
Similarly, we obtain
\begin{splitequation}
\label{anom_sop_ftp_bound}
&\abs{ \partial^w F(tp) } \leq \sum_{u+v \leq w} \sup(t\abs{p},\Lambda_0)^{D+[\op_{A_k}]+[\op_{A_l}]-\abs{w}+\abs{u}+\abs{v}} \, \mathe^{- \frac{t^2 \abs{p}^2}{2 \Lambda_0}} \, \mathcal{P}\left( \ln_+ \frac{\Lambda_0}{\mu} \right) \\
&\times \sup\left( 1, \frac{\abs{\vec{q},tp,-tp}}{\Lambda_0} \right)^{g^{(2)}([\op_{A_k}]+[\op_{A_l}],m+n+2l,0)} \!\!\!\! \sum_{T^* \in \mathcal{T}^*_{m+n+2}} \mathsf{G}^{T^*,(\vec{0},u,v)}_{\vec{K} \vec{L}^\ddag M M^\ddag; -D}(\vec{q},tp,-tp; \mu, \Lambda_0) \eqend{.}
\end{splitequation}
This bound shows that the $p$ integral is absolutely convergent for the last term $\partial^w F(t p)$, and since the sums in the Taylor expansion~\eqref{anom_sop_taylor} are finite we can exchange summation and integrations to obtain (recall that $r = [\op_{A_k}]+[\op_{A_l}]-3$)
\begin{splitequation}
\label{anom_w_expr_2}
&\mathcal{W}^{\Lambda_0, \Lambda_0}_{\vec{K} \vec{L}^\ddag}\left( \op_{A_k}(x_k) \otimes \op_{A_l}(x_l); \vec{q} \right) \\
&\quad= \left[ \sum_{\abs{w} \leq r} \frac{\partial^w F(0)}{w!} \partial^w_{x_k} - \sum_C \mathcal{L}^{\Lambda_0, \Lambda_0}_{\vec{K} \vec{L}^\ddag}\left( \op_C(x_k); \vec{q} \right) P^C_{A_k A_l}\left( \partial_{x_k} \right) \right] \delta^4(x_k - x_l) \\
&\quad\qquad+ (r+1) \sum_{\abs{w} = r+1} \int_0^1 (1-t)^r \int \mathe^{-\mathi (x_k - x_l) p} \frac{p^w}{w!} \left[ \partial^w F(t p) \right] \frac{\total^4 p}{(2\pi)^4} \total t \eqend{.}
\end{splitequation}
The boundary conditions for the anomaly $\mathsf{A}_2\left( \op_{A_k} \otimes \op_{A_l} \right)$ are then given by the first line, and we see that it is supported on the diagonal $x_k = x_l$, and by construction depends bilinearly on $\op_{A_k}$ and $\op_{A_l}$, such that $\mathsf{A}_2$ is a map as stated in Proposition~\ref{thm_anomward}. Furthermore, the bounds~\eqref{anom_sop_f0_bound} show that these conditions are compatible with the bounds~\eqref{bound_l1} for the functionals with one operator insertion of dimension $[\op_{A_k}]+[\op_{A_l}]-3-\abs{w}$ evaluated at $\Lambda = \Lambda_0$, such that $\mathsf{A}_2\left( \op_{A_k} \otimes \op_{A_l} \right)$ is a composite operator of dimension $\leq [\op_{A_k}]+[\op_{A_l}]-3$. Note that while with these boundary conditions the functional with one insertion of $\mathsf{A}_2$ is well-defined for finite $\Lambda$ and $\Lambda_0$, they do not guarantee the existence of the unregularised limit $\Lambda \to 0$, $\Lambda_0 \to \infty$, since for the proofs we need boundary conditions for the relevant and marginal functionals at $\Lambda = \mu$ and vanishing momenta, or $\Lambda = 0$ and non-exceptional momenta. However, the existence of the functionals with insertions of $\mathsf{A}_2$ follows from the very decomposition~\eqref{anom_sop_w_split} once we have proven bounds for $N^{\Lambda, \Lambda_0}\left( \bigotimes_{k=1}^s \op_{A_k} \right)$, since all the other functionals appearing in the decomposition have already be proven to be finite in the unregularised limit.

As stated in the beginning, we now simply define
\begin{splitequation}
\label{anom_sop_w_split_op}
N^{\Lambda, \Lambda_0}\left( \bigotimes_{k=1}^s \op_{A_k}(x_k) \right) &\equiv W^{\Lambda, \Lambda_0}\left( \bigotimes_{k=1}^s \op_{A_k}(x_k) \right) \\
&\quad- \sum_{l=1}^s L^{\Lambda, \Lambda_0}\left( \bigotimes_{k\in\{1,\ldots,s\}\setminus\{l\}} \op_{A_k}(x_k) \otimes \mathsf{A}_1\left(\op_{A_l}\right)(x_l) \right) \\
&\quad- \sum_{1 \leq l < l' \leq s} L^{\Lambda, \Lambda_0}\left( \bigotimes_{k\in\{1,\ldots,s\}\setminus\{l,l'\}} \op_{A_k} \otimes \mathsf{A}_2\left( \op_{A_l} \otimes \op_{A_{l'}} \right) \right) \eqend{.}
\end{splitequation}
Taking a $\Lambda$ derivative of this definition and using the flow equations for $W^{\Lambda, \Lambda_0}\left( \bigotimes_{k=1}^s \op_{A_k} \right)$ and the functionals with operator insertions~\eqref{l_sop_flow_hierarchy}, we obtain a flow equation for $N^{\Lambda, \Lambda_0}\left( \bigotimes_{k=1}^s \op_{A_k} \right)$ similar to~\eqref{l_sop_flow}, which reads
\begin{splitequation}
\label{anom_sop_n_flow}
&\partial_\Lambda N^{\Lambda, \Lambda_0}\left( \bigotimes_{k=1}^s \op_{A_k} \right) = \frac{\hbar}{2} \left\langle \frac{\delta}{\delta \phi_K}, \left( \partial_\Lambda C^{\Lambda, \Lambda_0}_{KL} \right) \ast \frac{\delta}{\delta \phi_L} \right\rangle N^{\Lambda, \Lambda_0}\left( \bigotimes_{k=1}^s \op_{A_k} \right) \\
&\quad- \sum_{\alpha \cup \beta = \{1, \ldots, s\}} \left\langle \frac{\delta}{\delta \phi_K} L^{\Lambda, \Lambda_0}\left( \bigotimes_{k\in\alpha} \op_{A_k} \right), \left( \partial_\Lambda C^{\Lambda, \Lambda_0}_{KL} \right) \ast \frac{\delta}{\delta \phi_L} N^{\Lambda, \Lambda_0}\left( \bigotimes_{k\in\beta} \op_{A_k} \right) \right\rangle \eqend{,}
\end{splitequation}
with $N^{\Lambda, \Lambda_0}$ and $N^{\Lambda, \Lambda_0}(\op_{A_k})$ defined in equations~\eqref{anom_0op_w_split} and~\eqref{anom_1op_w_split}. For this flow equation to hold, it is important that $\mathsf{A}_0 = 0$, since otherwise the source terms would contain $W^{\Lambda, \Lambda_0}\left( \bigotimes_{k\in\beta} \op_{A_k} \right)$ instead of $N^{\Lambda, \Lambda_0}\left( \bigotimes_{k\in\beta} \op_{A_k} \right)$, and it would be impossible to prove that $N^{\Lambda, \infty}\left( \bigotimes_{k=1}^s \op_{A_k} \right) = 0$.

The boundary conditions can now be read off from the definition~\eqref{anom_sop_w_split_op} and the boundary conditions for the functionals with an insertion of $\mathsf{A}_2$ that we determined previously. First, all functionals $L^{\Lambda, \Lambda_0}\left( \bigotimes_{k\in\{1,\ldots,s\}\setminus\{l\}} \op_{A_k}(x_k) \otimes \mathsf{A}_1\left(\op_{A_l}\right)(x_l) \right)$ vanish at $\Lambda = \Lambda_0$. In the case $s = 2$, the functional $L^{\Lambda, \Lambda_0}\left( \mathsf{A}_2\left( \op_{A_k} \otimes \op_{A_l} \right) \right)$ has its boundary conditions defined at $\Lambda = \Lambda_0$ by the first part of the decomposition of $W^{\Lambda_0, \Lambda_0}\left( \op_{A_k} \otimes \op_{A_l} \right)$, the first line of equation~\eqref{anom_w_expr_2}, and the boundary conditions of $N^{\Lambda_0, \Lambda_0}\left( \op_{A_k} \otimes \op_{A_l} \right)$ are thus given by the second line of that equation. For $s > 2$, also the functionals $L^{\Lambda, \Lambda_0}\left( \bigotimes_{k\in\{1,\ldots,s\}\setminus\{l,l'\}} \op_{A_k} \otimes \mathsf{A}_2\left( \op_{A_l} \otimes \op_{A_{l'}} \right) \right)$ vanish at $\Lambda = \Lambda_0$, but since also $W^{\Lambda_0, \Lambda_0}\left( \bigotimes_{k=1}^s \op_{A_k} \right) = 0$ in this case, also $N^{\Lambda_0, \Lambda_0}\left( \bigotimes_{k=1}^s \op_{A_k} \right) = 0$ for $s > 2$.

It remains to prove that $N^{\Lambda, \infty}\left( \bigotimes_{k=1}^s \op_{A_k} \right) = 0$, which we do by induction in $s$, starting with $s = 2$. In this case, the boundary conditions are given by the second line of equation~\eqref{anom_w_expr_2}, and we convert the explicit factor of $p^w$ into $x_k$ derivatives according to
\begin{equation}
\mathe^{-\mathi (x_k - x_l) p} \frac{p^w}{w!} = \frac{\mathi^\abs{w}}{w!} \partial^w_{x_k} \mathe^{-\mathi (x_k - x_l) p} \eqend{.}
\end{equation}
We need to generate $D+3$ more $x_k$ derivatives using equation~\eqref{func_sop_genderivs} and integrate the resulting $D+3$ derivatives with respect to $p$ by parts, such that
\begin{equation}
\label{anom_sop_n_kernel}
\mathcal{N}^{\Lambda_0, \Lambda_0}_{\vec{K} \vec{L}^\ddag}\left( \op_{A_k}(x_k) \otimes \op_{A_l}(x_l); \vec{q} \right) = \sum_{\abs{w} = [\op_{A_k}]+[\op_{A_l}]-2} \partial^{w+u}_{x_k} \mathcal{K}^{\Lambda_0, \Lambda_0}_{\vec{K} \vec{L}^\ddag}\left( \op_{A_k}(x_k) \otimes \op_{A_l}(x_l); \vec{q} \right)
\end{equation}
with the kernel
\begin{splitequation}
\mathcal{K}^{\Lambda_0, \Lambda_0}_{\vec{K} \vec{L}^\ddag}\left( \op_{A_k}(x_k) \otimes \op_{A_l}(x_l); \vec{q} \right) &= ([\op_{A_k}]+[\op_{A_l}]-2) \frac{\mathi^\abs{w}}{w!} \int_0^1 (1-t)^{[\op_{A_k}]+[\op_{A_l}]-2} (\mathi t)^{D+3} \\
&\quad\times \int \mathcal{E}_{u^\alpha}\big[ \left( x_k - x_l \right)^\alpha p^\alpha \big] \left[ \partial^{w+u} F(t p) \right] \frac{\total^4 p}{(2\pi)^4} \total t \eqend{,}
\end{splitequation}
a direction $\alpha \in \{1,2,3,4\}$ such that $\abs{x_k^\alpha - x_l^\alpha} \geq \abs{x_k - x_l}/2$ and a multiindex $u = (u^1,u^2,u^3,u^4)$ with $\abs{u} = D+3$ and $u^\beta = \abs{u} \delta^\beta_\alpha$. To bound this kernel, we use the bounds~\eqref{anom_sop_ftp_bound} derived above,\linebreak the bounds~\eqref{expint_bounds} for the function $\mathcal{E}_k$ and the bounds~\eqref{func_sop_logest} for the appearing logarithm. Rescaling $p \to t p$ and setting afterwards $p = x \Lambda_0$, we can use Lemma~\ref{lemma_pint3} to perform the $p$ integral. The $t$ integral is then trivially done, and we obtain
\begin{splitequation}
&\abs{ \mathcal{K}^{\Lambda_0, \Lambda_0}_{\vec{K} \vec{L}^\ddag}\left( \op_{A_k}(x_k) \otimes \op_{A_l}(x_l); \vec{q} \right) } \leq \left( 1 + \ln_+ \frac{1}{\mu \abs{x_k - x_l}} \right) \sup\left( 1, \frac{\abs{\vec{q}}}{\Lambda_0} \right)^{g^{(2)}([\op_{A_k}]+[\op_{A_l}],m+n+2l,0)} \\
&\hspace{10em}\times \sum_{v \leq w+u} \Lambda_0^{3+\abs{v}} \sum_{T^* \in \mathcal{T}^*_{m+n+2}} \mathsf{G}^{T^*,(\vec{0},v)}_{\vec{K} \vec{L}^\ddag M M^\ddag; -D}(\vec{q},0,0; \mu, \Lambda_0) \, \mathcal{P}\left( \ln_+ \frac{\Lambda_0}{\mu} \right) \eqend{.}
\end{splitequation}
We now remove the derivative weight factor corresponding to the $v$ derivatives that acted on the momentum $p$ from the tree, which gives a factor~\eqref{gw_def} $\Lambda_0^{-\abs{v}}$, and amputate the external legs corresponding to $M$ and $M^\ddag$, which gives an additional factor of
\begin{equation}
\Lambda_0^{-[\phi_M]-[\phi_M^\ddag]} = \Lambda_0^{-3}
\end{equation}
according to the estimate~\eqref{amputate} and equation~\eqref{antifield_dim}. The resulting bound on $\mathcal{K}^{\Lambda_0, \Lambda_0}$ is compatible with the bound~\eqref{bound_ks_lambdamu} evaluated at $\Lambda = \Lambda_0$, the representation~\eqref{anom_sop_n_kernel} is also compatible with the representation~\eqref{bound_ls_lambdamu} for dimension $[\op_{A_k}]+[\op_{A_l}]+1$, and the flow equation~\eqref{anom_sop_n_flow} for $N^{\Lambda, \Lambda_0}\left( \op_{A_k} \otimes \op_{A_l} \right)$ contains only the functionals $N^{\Lambda, \Lambda_0}\left( \op_{A_k} \right)$ and $N^{\Lambda, \Lambda_0}$, which have been shown previously to vanish in the limit $\Lambda_0 \to \infty$. Thus, $N^{\Lambda, \Lambda_0}\left( \op_{A_k} \otimes \op_{A_l} \right)$ fulfils the premises of Proposition~\ref{thm_ls_van}, which gives a bound independent of $\Lambda_0$ except for an explicit factor of $\left( \sup(\mu, \Lambda) / \Lambda_0 \right)^\frac{\Delta}{2}$, such that $\lim_{\Lambda_0 \to \infty} N^{\Lambda, \Lambda_0}\left( \op_{A_k} \otimes \op_{A_l} \right) = 0$.

For the functionals $N^{\Lambda, \Lambda_0}\left( \bigotimes_{k=1}^s \op_{A_k} \right)$ with $s > 2$, we already have vanishing boundary conditions. The premises of Proposition~\ref{thm_ls_van} are thus directly fulfilled, such that also $N^{\Lambda, \Lambda_0}\left( \bigotimes_{k=1}^s \op_{A_k} \right) \to 0$ as $\Lambda_0 \to \infty$.
\end{proof}

\begin{proof}[Proof of the properties of \texorpdfstring{$\mathsf{A}_2$}{A2}.]
We have already shown that $\mathsf{A}_2$ is a map as stated in Proposition~\ref{thm_anomward} supported on the diagonal and with the appropriate dimension, and so we only have to show that it is of higher order in $\hbar$. This is done by the same arguments as before: the first non-vanishing contribution to the functionals with an insertion of $\mathsf{A}_2\left( \op_{A_l} \otimes \op_{A_{l'}} \right)$ is of order $\hbar^{k+k'+l}$ with $l \geq 1$ if the first non-vanishing contribution to the functionals with an insertion of $\op_{A_l}$ ($\op_{A_{l'}}$) is of order $\hbar^k$ ($\hbar^{k'}$), since the boundary conditions~\eqref{anom_2op_w_irrelevant} vanish at order $\hbar^{k+k'}$ as $\Lambda_0 \to \infty$ (the counterterms in these boundary conditions first appear at order $\hbar^{k+k'+1}$).
\end{proof}

\subsection{Proof of Propositions~\ref{thm_anomward} and~\ref{thm_brst}: consistency conditions}
\label{sec_brst_consistency}

To fully prove Proposition~\ref{thm_anomward} and from this Proposition~\ref{thm_brst}, it remains to derive consistency conditions on the three types of anomalies $\mathsf{A}_i$. This is done by applying $\st_0$ twice on functionals with and without operator insertions, and using its nilpotency $\st_0^2 = 0$. In the following three subsections, we distinguish the cases of no insertion (condition on $\mathsf{A}_0$), one insertion (condition on $\mathsf{A}_1$) and $s$ insertions (condition on $\mathsf{A}_2$). To shorten the notation, we will again use the abbreviation~\eqref{abbreviation}.

\begin{proof}[Consistency condition for \texorpdfstring{$\mathsf{A}_0$}{A0}.]
Applying $\st_0$ twice on the functionals without insertions, using the anomalous Ward identities~\eqref{anom_ward_0op} and~\eqref{anom_ward_1op_a} we obtain
\begin{equation}
\label{anom_0op_consistency}
0 = \st_0 L\left( \mathsf{A}_0 \right) = L\left( \st \mathsf{A}_0 \right) + L\left( \mathsf{A}_1\left( \mathsf{A}_0 \right) \right) \eqend{.}
\end{equation}
This is the Wess-Zumino consistency condition on the anomaly $\mathsf{A}_0$, which permits us to remove $\mathsf{A}_0$ by a suitable change in boundary conditions for the functionals without operator insertions~\cite{wesszumino1971,hollands2008}, as we now explain in detail.

Expanding the functionals on the right-hand side in a formal power series in $\hbar$, we know that $L\left( \mathsf{A}_0 \right)$ is of order $\hbar^k$ with $k \geq 1$~\eqref{anom_0op_orderh}, and thus also $L\left( \st \mathsf{A}_0 \right)$ is of order $\hbar^k$. Equation~\eqref{anom_1op_orderh} tells us that $L\left( \mathsf{A}_1\left( \mathsf{A}_0 \right) \right)$ is of order $\hbar^{k+l}$ with $l \geq 1$. To lowest non-vanishing order in $\hbar$, the first term on the right-hand side of the consistency condition~\eqref{anom_0op_consistency} thus vanishes, which means that we have
\begin{equation}
\st \mathsf{A}_0 = \bigo{\hbar^{k+1}} \eqend{.}
\end{equation}
Since $\mathsf{A}_0$ is an integrated operator of dimension $5$ and ghost number $1$, the relevant equivariant cohomology is $H^{1,4}_{\mathrm{E}(4)}(\st\vert\total)$ which is empty as explained in the introduction. Therefore, the only solution of this equation is
\begin{equation}
\mathsf{A}_0 = \st \mathsf{B}_0 + \bigo{\hbar^{k+1}}
\end{equation}
for some integrated composite operator $\mathsf{B}_0$ of dimension $4$ and ghost number $1$, which is also of order $\hbar^k$. We then go back to the regularised theory and perform the finite renormalisation
\begin{equation}
L^{\Lambda_0} \to \tilde{L}^{\Lambda_0} = L^{\Lambda_0} - \mathsf{B}_0 \eqend{.}
\end{equation}
This renormalisation changes the boundary conditions~\eqref{anom_0op_w_irrelevant} to
\begin{equation}
W^{\Lambda_0} \to \tilde{W}^{\Lambda_0} = W^{\Lambda_0} - \left( S_0^{\Lambda_0} + L^{\Lambda_0}, \mathsf{B}_0 \right)^{\Lambda_0} + \bigo{\hbar^{k+1}} = W^{\Lambda_0} - \st \mathsf{B}_0 - \delta^{\Lambda_0} \left( \st \mathsf{B}_0 \right) + \bigo{\hbar^{k+1}} \eqend{,}
\end{equation}
and since $\mathsf{B}_0$ is of dimension $4$, this is an allowed change of boundary conditions for marginal functionals. Furthermore, the flow equation for $L^{\Lambda, \Lambda_0}\left( \mathsf{A}_0 \right)$ (equation~\eqref{l_sop_flow_hierarchy} with $s=1$) is linear, and its right-hand side of the flow equation at order $\hbar^k$ only involves the functional $L^{\Lambda, \Lambda_0}$ at order $\hbar^0$ which is unchanged, such that we obtain
\begin{equation}
L\left( \mathsf{A}_0 \right) \to \tilde{L}\left( \tilde{\mathsf{A}}_0 \right) = L\left( \mathsf{A}_0 - \st \mathsf{B}_0 \right) + \bigo{\hbar^{k+1}} = \bigo{\hbar^{k+1}} \eqend{.}
\end{equation}
We have thus removed the anomaly $\mathsf{A}_0$ in order $\hbar^k$, and by repeating the procedure we can remove it to all orders in $\hbar$.
\end{proof}

\begin{proof}[Consistency condition for \texorpdfstring{$\mathsf{A}_1$}{A1}.]
Applying $\st_0$ twice on the functionals with one insertion, and using twice the anomalous Ward identity~\eqref{anom_ward_1op}, we obtain
\begin{equation}
\label{anom_1op_consistency}
0 = L\left( \left( \stq^2 \op_A \right)(x) \right) = L\left( \left( \st^2 \op_A \right)(x) \right) + L\left( \mathsf{A}_1\left( \st \op_A \right)(x) \right) + L\left( \st \mathsf{A}_1\left( \op_A \right)(x) \right) + L\left( \mathsf{A}_1\left( \mathsf{A}_1\left( \op_A \right) \right)(x) \right) \eqend{.}
\end{equation}
This shows that $\stq$ is nilpotent. We now prove that one can choose the boundary conditions for functionals with one insertion such that $\stq = \st$ for classically gauge-invariant operators. For this, assume that $\op_A$ is of this form: it is a local operator of ghost number $0$, form degree $p$ and fulfils $\st \op_A = 0$. Then only the last two functionals remain on the right-hand side of equation~\eqref{anom_1op_consistency}, and expanding in a formal power series in $\hbar$ we know that $\mathsf{A}_1\left( \op_A \right)$ is of order $\hbar^k$ with $k \geq 1$~\eqref{anom_1op_orderh}. Since $\mathsf{A}_1\left( \mathsf{A}_1\left( \op_A \right) \right)$ is then of order $\hbar^{k+l}$ with $l \geq 1$, to lowest non-vanishing order in $\hbar$ we must have
\begin{equation}
\st \mathsf{A}_1\left( \op_A \right) = \bigo{\hbar^{k+1}} \eqend{.}
\end{equation}
Since $\mathsf{A}_1$ augments the ghost number by $1$, the relevant cohomology is $H^{1,p}(\st)$, which is empty~\eqref{cohom_op}, and so the only solution of this equation is given by
\begin{equation}
\mathsf{A}_1\left( \op_A \right) = \st \mathsf{B}_1\left( \op_A \right) + \bigo{\hbar^{k+1}} \eqend{.}
\end{equation}
for some composite operator $\mathsf{B}_1\left( \op_A \right)$. Since the anomaly $\mathsf{A}_1\left( \op_A \right)$ is of dimension $[\mathsf{A}_1\left( \op_A \right)] = [\op_A]+1$ and $\st$ raises the dimension by $1$, it is of dimension $[\mathsf{B}_1\left( \op_A \right)] = [\op_A]$ and, since $\st$ raises also the ghost number by $1$, of the same ghost number as $\op_A$, depending linearly on $\op_A$. We then perform again a finite renormalisation (see equation~\eqref{l_1op_lambda0})
\begin{equation}
L^{\Lambda_0}(\op_A) = \op_A + \delta^{\Lambda_0} \op_A \to \tilde{L}^{\Lambda_0}(\op_A) = \op_A + \delta^{\Lambda_0} \op_A - \mathsf{B}_1\left( \op_A \right) \eqend{,}
\end{equation}
which in this case is a (finite) change of the counterterms of the composite operator $\op_A$ in the functional integral. This changes the boundary conditions~\eqref{anom_1op_w_irrelevant} to
\begin{splitequation}
W^{\Lambda_0}(\op_A) \to \tilde{W}^{\Lambda_0}(\op_A) &= W^{\Lambda_0}(\op_A) - \left( S_0^{\Lambda_0} + L^{\Lambda_0}, \mathsf{B}_1\left( \op_A \right) \right)^{\Lambda_0} + \bigo{\hbar^{k+1}} \\
&= W^{\Lambda_0} - \st \mathsf{B}_1\left( \op_A \right) + \bigo{\hbar^{k+1}} \eqend{,}
\end{splitequation}
and the dimension and ghost number of $\mathsf{B}_1\left( \op_A \right)$ are such that this change is a permissible change in boundary conditions for marginal functionals. This change entails
\begin{equation}
L\left( \mathsf{A}_1\left( \op_A \right) \right) \to \tilde{L}\left( \mathsf{A}_1\left( \tilde{\op}_A \right) \right) = L\Big( \mathsf{A}_1\left( \op_A \right) - \st \mathsf{B}_1\left( \op_A \right) \Big) + \bigo{\hbar^{k+1}} = \bigo{\hbar^{k+1}} \eqend{,}
\end{equation}
and we have removed the anomaly $\mathsf{A}_1$ in order $\hbar^k$. By repeating this procedure, we can thus remove the anomaly for all classically gauge-invariant operators of ghost number $0$, such that for those operators $\stq = \st$. In the general case, one proceeds similarly, but it may happen that the corresponding cohomology is not empty. In general, the BPHZ boundary conditions~\eqref{func_1op_bphz} we imposed for functionals with one operator insertion are thus modified in higher order in $\hbar$ by the removal of anomalies. However, it may be advantageous to stick with the BPHZ conditions, and use the differential $\stq$ even for classically gauge-invariant operators.
\end{proof}

\begin{proof}[Consistency condition for \texorpdfstring{$\mathsf{A}_2$}{A2}.]
Since we do not have any further freedom in changing boundary conditions, we cannot remove the anomaly $\mathsf{A}_2$, but the consistency condition for it gives us the properties of the quantum antibracket~\eqref{bvq_def} stated in Theorem~\ref{thm4}. The symmetry condition~\eqref{bvq_symm} of the quantum antibracket directly follows from the fact that the boundary conditions~\eqref{anom_2op_w_irrelevant} in the regulated theory, which involve the regulated antibracket, satisfy this symmetry. To show the compatibility condition~\eqref{bvq_stq_compat}, assume that $\op_A$ and $\op_B$ are two bosonic operators. For them, the Ward identity~\eqref{ward_sop} reads
\begin{equation}
\st_0 L\Big( \op_A \otimes \op_B \Big) = L\Big( \stq \op_A \otimes \op_B \Big) + L\Big( \op_A \otimes \stq \op_B \Big) + L\Big( \left( \op_A, \op_B \right)_\hbar \Big) \eqend{.}
\end{equation}
Since both $\stq$ and $\left( \cdot, \cdot \right)_\hbar$ are fermionic and the Ward identities~\eqref{ward_1op} and~\eqref{ward_sop} are valid in the given form for bosonic operators, we introduce an auxiliary constant fermion $\epsilon$ and obtain
\begin{equation}
\epsilon \st_0 L\Big( \op_A \otimes \op_B \Big) = L\Big( \epsilon \stq \op_A \otimes \op_B \Big) + L\Big( \op_A \otimes \epsilon \stq \op_B \Big) + L\Big( \epsilon \left( \op_A, \op_B \right)_\hbar \Big) \eqend{.}
\end{equation}
Now applying $\st_0$ another time on this equation and using that $\st_0^2 = 0$, that $\st_0$ and $\stq$ anticommute with $\epsilon$ and that $\stq^2 = 0$, it follows that
\begin{splitequation}
0 &= L\Big( \epsilon \stq \op_A \otimes \stq \op_B \Big) + L\Big( \stq \op_A \otimes \epsilon \stq \op_B \Big) \\
&\quad+ L\Big( \left( \epsilon \stq \op_A, \op_B \right)_\hbar \Big) + L\Big( \left( \op_A, \epsilon \stq \op_B \right)_\hbar \Big) - L\Big( \epsilon \stq \left( \op_A, \op_B \right)_\hbar \Big) \eqend{.}
\end{splitequation}
Since $\op_A$ was assumed to be bosonic, $\epsilon$ also anticommutes with $\stq \op_A$ and the two functionals in the first line cancel each other out. Since for bosonic operators the quantum antibracket is symmetric~\eqref{bvq_symm}, and constant factors can be taken out from the quantum antibracket in the first entry without additional signs (which again follows from the corresponding fact in the regulated theory), we get
\begin{equation}
0 = L\Big( \left( \stq \op_A, \op_B \right)_\hbar \Big) + L\Big( \left( \stq \op_B, \op_A \right)_\hbar \Big) - L\Big( \stq \left( \op_A, \op_B \right)_\hbar \Big) \eqend{,}
\end{equation}
and the compatibility condition~\eqref{bvq_stq_compat} follows by using the symmetry property~\eqref{bvq_symm} again to bring the antibracket in the second term in canonical order. Similar considerations apply if one or both of the operators are fermionic. Analogously, by acting twice with $\st_0$ on a functional with three insertions and using the nilpotency of $\stq$ and the compatibility condition~\eqref{bvq_stq_compat}, we obtain the graded Jacobi identity~\eqref{bvq_jacobi}. Acting twice with $\st_0$ on a functional with four or more insertions does not give any more conditions.
\end{proof}

\section{Discussion}

In this article, we have shown that Yang-Mills gauge theories based on compact semisimple Lie algebras can be consistently treated within the flow equation framework, and that stringent bounds can be obtained for the correlation functions of arbitrary fields and insertions of composite operators. To achieve this goal, we first derived bounds establishing the existence of the physical, unregularised limit $\Lambda \to 0$, $\Lambda_0 \to \infty$ (in perturbation theory) for correlation functions of arbitrary composite operators for an arbitrary, superficially renormalisable massless theory in four dimensions. (These bounds are also of interest in non-gauge theories, \eg, scalar-fermion theories with Yukawa couplings.) In a second step, we used the Batalin-Vilkovisky formalism for gauge-fixed theories to establish anomalous Ward identities for the correlation functions emerge in the unregularised limit. Based on cohomological methods, we showed that these anomalies can be removed by a finite change in the renormalisation conditions. The resulting, non-anomalous Ward identities are expressed in terms of a ``quantum differential'' $\stq$ and a ``quantum anti-bracket'' $( \cdot, \cdot )_\hbar$, which in general differ from the naive classical expressions by terms of order $\bigo{\hbar}$. We also identified subclasses of composite operators where these additional $\bigo{\hbar}$ terms can be made to vanish by another finite change in the renormalisation conditions, which especially includes all classically gauge-invariant operators. In any case, our Ward identities express the gauge invariance of the theory at the quantum level.

Our method of proof extends straightforwardly to other gauge theories which have a BV-extended action linear in the antifields, provided that one can remove the anomaly for the functionals without insertions. In our framework, as in other previous frameworks, this question is decided by the relevant equivariant cohomology of the corresponding classical BV differential $\st$ at dimension $4$ and ghost number $1$. If this cohomology is empty, then any anomaly can automatically be removed, but otherwise, a more refined analysis is needed. For instance, if chiral fermions are included the corresponding equivariant cohomology is not empty but includes the element $\mathcal{A}$~\eqref{cohom_parity}, often called ``gauge anomaly''. In such a case, our cohomological arguments do not work, and to show gauge invariance in this case one would have to invest additional work to trace the exact dependence of the anomaly on the boundary conditions, and show that the numerical coefficient in front of the anomaly vanishes for a certain field content and set of boundary conditions (as it does in the Standard Model~\cite{gengmarshak1989,minahanetal1990}). For theories where the action has a quadratic (or higher) dependence on antifields such as supergravity~\cite{vannieuwenhuizen1981}, some minor changes to our proof are necessary, but we believe that this does not constitute any problem since the proof works as long as one has the correct (naive) power-counting.

In a next step, we would like to extend the results on the Operator Product Expansion (OPE)~\cite{wilson1969,zimmermann1973} derived for scalar fields to gauge theories. The OPE is the statement that any product of local composite operators $\op_{A_1}, \ldots, \op_{A_n}$ can be expanded in the form
\begin{equation}
\op_{A_1}(x_1) \cdots \op_{A_n}(x_n) \sim \sum_B \mathcal{C}^B_{A_1 \cdots A_n}(x_1,\ldots,x_n) \op_B(x_n) \eqend{,}
\end{equation}
where the sum runs over all composite operators of the theory, indexed by the label $B$, and the coefficients $\mathcal{C}$ are distributions. This expansion is understood to hold in the weak sense, \ie, as an insertion into an arbitrary correlation function, and was conjectured to be asymptotic, \ie, the difference between the right-hand side and the left-hand side vanishes if all $x_i \to x_n$. The OPE has not only found important applications, \eg, in deep inelastic scattering~\cite{wilson1969}, but has also proven to be a valuable tool, especially in the analysis of conformal theories (see, \eg, Ref.~[\onlinecite{cft}]). Its main advantage, in our view, is that it is independent of any arbitrary choice of state (which is important, \eg, in curved spacetimes where no preferred vacuum state exists), and thus encodes the algebraic structure of the theory, while the only state-dependent information is contained in the one-point correlation functions. One can thus wonder whether it is possible to define the QFT by its OPE and the one-point correlation functions, similar to the bootstrap programme of two-dimensional conformal field theories where such a construction is possible~\cite{luescher1976,mack1976,belavinetal1984,fredenhagenjoerss1994}. There are two main obstacles to this approach in four dimensions: first, one would like to have a convergent expansion instead of an asymptotic one. Second, even if one imposes additional, natural conditions such as factorisation
\begin{equation}
\mathcal{C}^B_{A_1 \cdots A_n}(x_1,\ldots,x_n) \sim \sum_C \mathcal{C}^C_{A_1 \cdots A_m}(x_1,\ldots,x_m) \mathcal{C}^B_{C A_{m+1} \cdots A_n}(x_m,\ldots,x_n)
\end{equation}
(which formally results by applying the OPE twice, once for the first $m$ operators and then for the remaining ones) and associativity, \eg,
\begin{equation}
\sum_C \mathcal{C}^C_{A_1 A_2}(x_1, x_2) \mathcal{C}^B_{C A_3}(x_2, x_3) = \sum_C \mathcal{C}^C_{A_2 A_3}(x_2, x_3) \mathcal{C}^B_{A_1 C}(x_1, x_3)
\end{equation}
(which results by applying the OPE two times in two different orders), which give strong restrictions on the form of the OPE coefficients $\mathcal{C}$, it does not seem feasible to give a classification of solutions to these conditions in four dimensions. For two-dimensional models, such as the massless Thirring model (see Ref.~[\onlinecite{olbermann}] and references therein), and conformal field theories~\cite{mack1976,huangkong2005,pappadopuloetal2012} it has been proven that associativity holds and that the OPE is convergent. Both statements hold under certain restrictions on the positions of the operator insertions.

Nevertheless, it has been recently shown for Euclidean four-dimensional scalar field theory that a) an OPE exists and converges for arbitrary separations~\cite{hollandskopper2012,hollandetal2014}, b) factorisation and associativity hold as long as one performs the OPE first for the points which lie closest together~\cite{hollandhollands2015b}, and c) there exists an explicit, renormalised recursion formula for the OPE coefficients (written as a power series in the coupling constant), which does only involve the coefficients themselves, starting with the free theory~\cite{hollandhollands2015a}. These results have been derived in perturbation theory, and hold for an arbitrary, but fixed order (\ie, ``to all orders''), both for the massive and the massless scalar field. It would be very interesting to see if similar statements hold also for gauge theories. In such theories, it would be necessary, among other things, to check that the OPE closes on gauge-invariant operators, which means that if the product of operators on the left-hand side is annihilated by the quantum differential $\stq$ (up to contact terms), the expansion of the right-hand side should also only involve $\stq$-invariant operators. The analysis of the present paper provides a basis for a rigorous analysis of this problem.

\begin{acknowledgments}
The research leading to these results has received funding from the European Research Council under the European Union's Seventh Framework Programme (FP7/2007-2013) / ERC Starting Grant No. QC\&C 259562. M.~F. thanks M.~Taslimitehrani and T.~Jerabek for discussions. He also thanks Ch.~Kopper for discussions and the warm hospitality extended to him during a stay at the {\'E}cole polytechnique. The authors thank R.~Guida and Ch.~Kopper for making available results from the unpublished manuscript~\cite{guidakopper2015}, from which some Lemmas of the appendix have been adopted, and A.~Efremov for pointing out the indefiniteness of the free covariance in the standard $R_\xi$ gauges, and discussions on this important issue. We also thank A.~Efremov, R.~Guida and Ch.~Kopper for making available to us their forthcoming manuscript on Slavnov-Taylor identities for SU(2) gauge theory~\cite{efremovguidakopper2015}, and for discussions on many issues related to renormalisation.
\end{acknowledgments}

\appendix

\section{Lemmata}

\begin{lemma}[Supremum estimates]
\label{lemma_biggermomentum}
For $K \geq k \geq 0$, $c \geq 0$ and $a \geq b \geq 0$ we have
\begin{equation}
\left( \frac{\sup(a,k)}{\sup(b,k)} \right)^c \geq \left( \frac{\sup(a,K)}{\sup(b,K)} \right)^c \eqend{.}
\end{equation}
\end{lemma}

\begin{proof}
If $K \geq a,b$, the right-hand side is equal to $1$ and the inequality is obviously true. For the remaining cases, we make a case-by-case analysis:
\begin{itemize}
\item $a \geq b \geq K \geq k$: $\left( \frac{\sup(a,k)}{\sup(b,k)} \right)^c = \left( \frac{a}{b} \right)^c = \left( \frac{\sup(a,K)}{\sup(b,K)} \right)^c \eqend{,}$
\item $a \geq K \geq b \geq k$: $\left( \frac{\sup(a,k)}{\sup(b,k)} \right)^c = \left( \frac{a}{b} \right)^c \geq \left( \frac{a}{K} \right)^c = \left( \frac{\sup(a,K)}{\sup(b,K)} \right)^c \eqend{,}$
\item $a \geq K \geq k \geq b$: $\left( \frac{\sup(a,k)}{\sup(b,k)} \right)^c = \left( \frac{a}{k} \right)^c \geq \left( \frac{a}{K} \right)^c = \left( \frac{\sup(a,K)}{\sup(b,K)} \right)^c \eqend{.}$
\end{itemize}
\end{proof}

\begin{lemma}[\texorpdfstring{$p$}{p} integration]
\label{lemma_pint}
For any function $f(x) \geq 0$ such that
\begin{equation}
\int \mathe^{- \alpha \abs{x}^2} f(x) \total^4 x < \infty
\end{equation}
for all $\alpha > 0$, and for $\beta_i \geq 1$ we have
\begin{equation}
\int \mathe^{-\alpha \abs{x}^2} f(x) \prod_{i=1}^n \sup(\abs{x + a_i}, \beta_i)^{m_i} \total^4 x \leq c \prod_{i=1}^n \sup(\abs{a_i}, \beta_i)^{m_i}
\end{equation}
for some positive constant $c$.
\end{lemma}

\begin{proof}
We proceed by induction on the number of factors $n$. The result is obvious for $n = 0$. If the result has been proven for $n = 1$ and $n = N-1$, we calculate
\begin{splitequation}
&\left[ \int \mathe^{- \alpha \abs{x}^2} f(x) \prod_{i=1}^N \sup(\abs{x + a_i}, \beta_i)^{m_i} \total^4 x \right]^2 \\
&\quad= \left[ \int \left( \mathe^{-\frac{\alpha}{2} \abs{x}^2} f(x) \prod_{i=1}^{N-1} \sup(\abs{x + a_i}, \beta_i)^{m_i} \right) \left( \mathe^{-\frac{\alpha}{2} \abs{x}^2} \sup(\abs{x + a_N}, \beta_N)^{m_N} \right) \total x^4 \right]^2 \\
&\leq \int \left( \mathe^{-\frac{\alpha}{2} \abs{x}^2} f(x) \prod_{i=1}^{N-1} \sup(\abs{x + a_i}, \beta_i)^{m_i} \right)^2 \total^4 x \int \left( \mathe^{-\frac{\alpha}{2} \abs{x}^2} \sup(\abs{x + a_N}, \beta_N)^{m_N} \right)^2 \total^4 x \\
&= \int \mathe^{-\alpha \abs{x}^2} f^2(x) \prod_{i=1}^{N-1} \sup(\abs{x + a_i}, \beta_i)^{2m_i} \total^4 x \int \mathe^{-\alpha \abs{x}^2} \sup(\abs{x + a_N}, \beta_N)^{2m_N} \total^4 x \\
&\leq \left( c \prod_{i=1}^{N-1} \sup(\abs{a_i}, \beta_i)^{2m_i} \right) \left( c \sup(\abs{a_N}, \beta_N)^{2m_N} \right) \eqend{,}
\end{splitequation}
where we used the Cauchy-Schwarz inequality in the second step, and the result follows for $n = N$ by taking the square root.

To show the case $n = 1$, start with positive $m$. Then we have
\begin{splitequation}
&\int \mathe^{-\alpha \abs{x}^2} f(x) \sup(\abs{x + a}, \beta)^m \total^4 x \leq \int \mathe^{-\alpha \abs{x}^2} f(x) \sup(\abs{x} + \sup(\abs{a},\beta), \beta)^m \total^4 x \\
&\hspace{8em}= \sup(\abs{a},\beta)^m \int \mathe^{-\alpha \abs{x}^2} f(x) \left( \frac{\abs{x}}{\sup(\abs{a},\beta)} + 1 \right)^m \total^4 x \\
&\hspace{8em}\leq \sup(\abs{a},\beta)^m \int \mathe^{-\alpha \abs{x}^2} f(x) \left( \abs{x} + 1 \right)^m \total^4 x \leq c \sup(\abs{a},\beta)^m \eqend{.}
\end{splitequation}
For negative powers, we first write
\begin{equation}
\sup(\abs{x+a},\beta)^{-m} = \sup(\abs{x},1)^m \left[ \sup(\abs{x},1) \sup(\abs{x+a},\beta) \right]^{-m} \eqend{.}
\end{equation}
The last $\sup$ can be estimated against $\beta^{-m}$, and for the second one we have
\begin{equation}
\sup(\abs{x}, 1) \geq \sup\left( \frac{1}{2} \abs{a}, 1 \right) = \frac{1}{2} \sup\left( \abs{a}, 2 \right) \geq \frac{1}{2} \sup\left( \abs{a}, 1 \right)
\end{equation}
for $\abs{x} \geq \frac{1}{2} \abs{a}$, such that
\begin{equation}
\sup(\abs{x+a},\beta)^{-m} \leq \beta^{-m} \sup(\abs{x},1)^m \left[ \frac{1}{2} \sup(\abs{a},1) \right]^{-m}
\end{equation}
in this case. Thus it follows that
\begin{splitequation}
&\int_{\abs{x} \geq \frac{1}{2} \abs{a}} \mathe^{-\alpha \abs{x}^2} f(x) \sup(\abs{x + a}, \beta)^{-m} \total^4 x \\
&\qquad\leq \left[ \frac{\beta}{2} \sup(\abs{a},1) \right]^{-m} \int_{\abs{x} \geq \frac{1}{2} \abs{a}} \mathe^{-\alpha \abs{x}^2} f(x) \sup(\abs{x},1)^m \total^4 x \\
&\qquad\leq c \beta^{-m} \sup(\abs{a},1)^{-m} \leq c \sup(\abs{a},\beta)^{-m} \eqend{.}
\end{splitequation}
For $\abs{x} < \frac{1}{2} \abs{a}$, we have
\begin{splitequation}
\sup(\abs{x+a},\beta)^{-m} &\leq \sup(\abs{a}-\abs{x},\beta)^{-m} \leq \sup\left( \frac{1}{2} \abs{a}, \beta \right)^{-m} \\
&\leq 2^m \sup\left( \abs{a}, 2 \beta \right)^{-m} \leq c \sup\left( \abs{a}, \beta \right)^{-m}
\end{splitequation}
and thus
\begin{equation}
\int_{\abs{x} < \frac{1}{2} \abs{a}} \mathe^{-\alpha \abs{x}^2} f(x) \sup(\abs{x + a}, \beta)^{-m} \total x \leq c \sup\left( \abs{a}, \beta \right)^{-m} \eqend{.}
\end{equation}
By summing the two results the Lemma follows.
\end{proof}

\begin{lemma}[\texorpdfstring{$p$}{p} integration, part 2]
\label{lemma_pint2}
For $\delta_i > 0$, $\gamma_i \geq 1$ and with the other assumptions as in Lemma~\ref{lemma_pint}, we have
\begin{splitequation}
\label{pint2_hyp}
&\int \mathe^{-\alpha \abs{x}^2} f(x) \prod_{k=1}^s \sup\left( \abs{\vec{b}_k, x, -x}, \gamma_k \right)^{\delta_k} \prod_{j=1}^t \sup\left( \eta_{d_j}(\vec{d},x,-x), \gamma_{s+j} \right)^{-\delta_{s+j}} \prod_{i=1}^n \sup(\abs{x + a_i}, \beta_i)^{m_i} \total^4 x \\
&\qquad\leq c \prod_{k=1}^s \sup\left( \abs{\vec{b}_k}, \gamma_k \right)^{\delta_k} \prod_{j=1}^t \sup\left( \eta_{d_j}(\vec{d}), \gamma_{s+j} \right)^{-\delta_{s+j}} \prod_{i=1}^n \sup(\abs{a_i}, \beta_i)^{m_i} \eqend{,}
\end{splitequation}
where $\eta_j$ is defined in equation~\eqref{eta_i_def}. The same estimate is valid if we replace $\eta_j$ by $\bar{\eta}_j$, defined in equation~\eqref{bareta_i_def}.
\end{lemma}

\begin{proof}
We do induction on the number of factors $s+t$, starting with $t = 0$. For $s = 0$, the bound is proven by Lemma~\ref{lemma_pint}. Assume thus that it has been shown for $(s-1)$ factors. Denote the number of elements in $\vec{b}_s$ by $r_s$, define $b_{s,r_s+1} \equiv x$, $b_{s,r_s+2} \equiv -x$ and estimate
\begin{equation}
\sup\left( \abs{\vec{b}_s, x, -x}, \gamma_s \right)^{\delta_s} = \sup_{S \subseteq \{1,\ldots,r_s+2\}} \sup\left( \abs{\sum_{i \in S} b_{s,i}}, \gamma_s \right)^{\delta_s} \leq \sum_{S \subseteq \{1,\ldots,r_s+2\}} \sup\left( \abs{\sum_{i \in S} b_{s,i}}, \gamma_s \right)^{\delta_s} \eqend{.}
\end{equation}
Applying the induction hypothesis~\eqref{pint2_hyp} to each of the summands gives
\begin{splitequation}
&\int \mathe^{-\alpha \abs{x}^2} f(x) \prod_{j=1}^s \sup\left( \abs{\vec{b}_j, x, -x}, \gamma_j \right)^{\delta_j} \prod_{i=1}^n \sup(\abs{x + a_i}, \beta_i)^{m_i} \total^4 x \\
&\qquad\leq \sum_{S \subseteq \{1,\ldots,r_s+2\}} c_S \sup\left( \abs{\sum_{i \in S} b_{s,i}}, \gamma_s \right)^{\delta_s} \prod_{j=1}^{s-1} \sup\left( \abs{\vec{b}_j}, \gamma_j \right)^{\delta_j} \prod_{i=1}^n \sup(\abs{a_i}, \beta_i)^{m_i} \eqend{,}
\end{splitequation}
where now $b_{s,r_s+1} = b_{s,r_s+2} = 0$. Thus the sum can be restricted to $S \subseteq \{1,\ldots,r_s\}$, and the Lemma follows with $c = \sum_{S \subseteq \{1,\ldots,r_s+2\}} c_S$. For $t > 0$, we again do induction on the number of factors $t$, and can assume that the Lemma has been shown for $(t-1)$ factors. We then denote the number of elements in $\vec{d}$ by $r$, define $d_{r+1} \equiv x$, $d_{r+2} \equiv -x$ and estimate
\begin{equation}
\begin{split}
\sup\left( \eta_{d_t}(\vec{d},x,-x), \gamma_{s+t} \right)^{-\delta_{s+t}} &= \sup\left( \inf_{S \subseteq \{1,\ldots,r+2\}\setminus\{t\}} \abs{d_t + \sum_{i \in S} d_i}, \gamma_{s+t} \right)^{-\delta_{s+t}} \\
&\leq \sum_{S \subseteq \{1,\ldots,r+2\}\setminus\{t\}} \sup\left( \abs{d_t + \sum_{i \in S} d_i}, \gamma_{s+t} \right)^{-\delta_{s+t}} \eqend{.}
\end{split}
\end{equation}
Applying the induction hypothesis~\eqref{pint2_hyp} to each of the summands gives
\begin{splitequation}
&\int \mathe^{-\alpha \abs{x}^2} f(x) \prod_{k=1}^s \sup\left( \abs{\vec{b}_k, x, -x}, \gamma_k \right)^{\delta_k} \prod_{j=1}^t \sup\left( \eta_{d_j}(\vec{d},x,-x), \gamma_{s+j} \right)^{-\delta_{s+j}} \\
&\quad\times \prod_{i=1}^n \sup(\abs{x + a_i}, \beta_i)^{m_i} \total^4 x \leq \sum_{S \subseteq \{1,\ldots,r+2\}\setminus\{t\}} c_S \sup\left( \abs{d_t + \sum_{i \in S} d_i}, \gamma_{s+t} \right)^{-\delta_{s+t}} \\
&\hspace{8em}\times \prod_{k=1}^s \sup\left( \abs{\vec{b}_k}, \gamma_k \right)^{\delta_k} \prod_{j=1}^{t-1} \sup\left( \eta_{d_j}(\vec{d}), \gamma_{s+j} \right)^{-\delta_{s+j}} \prod_{i=1}^n \sup(\abs{a_i}, \beta_i)^{m_i}\eqend{,}
\end{splitequation}
where now $d_{r+1} = d_{r+2} = 0$. Thus the sum can be restricted to $S \subseteq \{1,\ldots,r\}\setminus\{t\}$, and since
\begin{equation}
\sup\left( \abs{d_t + \sum_{i \in S} d_i}, \gamma_{s+t} \right)^{-\delta_{s+t}} \leq \sup\left( \eta_t(\vec{d}), \gamma_{s+t} \right)^{-\delta_{s+t}}
\end{equation}
the Lemma follows with $c = \sum_{S \subseteq \{1,\ldots,r+2\}\setminus\{t\}} c_S$.
\end{proof}

\begin{lemma}[\texorpdfstring{$p$}{p} integration, part 3]
\label{lemma_pint3}
For any $y$, any $A^i \geq 1$, any function $f(x) \geq 0$ such that
\begin{equation}
\int \mathe^{- \alpha \abs{x}^2} f^2(x) \total^4 x < \infty
\end{equation}
for all $\alpha > 0$, and with the other assumptions as in Lemma~\ref{lemma_pint2}, we have
\begin{splitequation}
\label{pint3_hyp}
&\int \mathe^{-\alpha \abs{x}^2} f(x) \prod_{a=1}^4 \left( A^i + \ln_+ \frac{1}{\abs{x^i+y^i}} \right) \prod_{k=1}^s \sup\left( \abs{\vec{b}_k, x, -x}, \gamma_k \right)^{\delta_k} \\
&\qquad\qquad\times \prod_{j=1}^t \sup\left( \eta_{d_j}(\vec{d},x,-x), \gamma_{s+j} \right)^{-\delta_{s+j}} \prod_{i=1}^n \sup(\abs{x + a_i}, \beta_i)^{m_i} \total^4 x \\
&\quad\leq c A^1 A^2 A^3 A^4 \prod_{k=1}^s \sup\left( \abs{\vec{b}_k}, \gamma_k \right)^{\delta_k} \prod_{j=1}^t \sup\left( \eta_{d_j}(\vec{d}), \gamma_{s+j} \right)^{-\delta_{s+j}} \prod_{i=1}^n \sup(\abs{a_i}, \beta_i)^{m_i} \eqend{.}
\end{splitequation}
\end{lemma}

\begin{proof}
We again use the Cauchy-Schwarz inequality, which gives
\begin{splitequation}
&\int \mathe^{-\alpha \abs{x}^2} f(x) \prod_{a=1}^4 \left( A^i + \ln_+ \frac{1}{\abs{x^i+y^i}} \right) \prod_{k=1}^s \sup\left( \abs{\vec{b}_k, x, -x}, \gamma_k \right)^{\delta_k} \\
&\qquad\quad\times \prod_{j=1}^t \sup\left( \eta_{d_j}(\vec{d},x,-x), \gamma_{s+j} \right)^{-\delta_{s+j}} \prod_{i=1}^n \sup(\abs{x + a_i}, \beta_i)^{m_i} \total^4 x \\
&\quad\leq \left[ \int \mathe^{-\alpha \abs{x}^2} \prod_{a=1}^4 \left( A^i + \ln_+ \frac{1}{\abs{x^i+y^i}} \right)^2 \total^4 x \right]^\frac{1}{2} \Bigg[ \int \mathe^{-\alpha \abs{x}^2} f^2(x) \prod_{k=1}^s \sup\left( \abs{\vec{b}_k, x, -x}, \gamma_k \right)^{2\delta_k} \\
&\qquad\quad\times \prod_{j=1}^t \sup\left( \eta_{d_j}(\vec{d},x,-x), \gamma_{s+j} \right)^{-2\delta_{s+j}} \prod_{i=1}^n \sup(\abs{x + a_i}, \beta_i)^{2m_i} \total^4 x \Bigg]^\frac{1}{2} \eqend{.}
\end{splitequation}
The second integral can be done using Lemma~\ref{lemma_pint2} (noting that when $f(x)$ is square integrable for all $\alpha > 0$, it is integrable for all $\alpha > 0$). The first integral can be done in the same manner for all four components, and we calculate for the component $i$
\begin{splitequation}
\int \mathe^{-\alpha (x^i)^2} \left( A^i + \ln_+ \frac{1}{\abs{x^i+y^i}} \right)^2 \total x^i &\leq \sqrt{\frac{\pi}{\alpha}} \left( A^i \right)^2 + 2 A^i \int \ln_+ \frac{1}{\abs{x^i+y^i}} \total x^i \\
&\hspace{8em}+ \int \left( \ln_+ \frac{1}{\abs{x^i+y^i}} \right)^2 \total x^i \\
&= \sqrt{\frac{\pi}{\alpha}} \left( A^i \right)^2 + 2 A^i + 2 \leq \left( \sqrt{\frac{\pi}{\alpha}} + 4 \right) \left( A^i \right)^2 \eqend{,}
\end{splitequation}
where we shifted the integration variable $x^i \to x^i - y^i$ in the second step, and used the fact that $A^i \geq 1$ to obtain the last inequality. Multiplying both results and taking the square root the Lemma follows.
\end{proof}

\begin{lemma}[\texorpdfstring{$\Lambda$}{Λ} integration]
\label{lemma_lambdaint}
For $a_0, A, K, L \geq 0$, $k,l \in \mathbb{N}_0$ and $m < -1$, we have
\begin{equation}
\int_{a_0}^{a_1} \sup(A, x)^m \ln_+^k \frac{K}{x} \ln_+^l \frac{x}{L} \total x \leq \sup(A, a_0)^{m+1} \, \mathcal{P}\left( \ln_+ \frac{\sup(K, A)}{\sup(a_0, \inf(A,L))}, \ln_+ \frac{a_0}{L} \right) \eqend{.}
\end{equation}
\end{lemma}

\begin{proof}
First note that since $x \geq a_0$ we can write $\sup(A, x) = \sup(A, a_0, x)$. We then split the integral at $x = \sup(A, a_0)$, and obtain for the first part
\begin{splitequation}
&\int_{a_0}^{\sup(A, a_0)} \sup(A, a_0, x)^m \ln_+^k \frac{K}{x} \ln_+^l \frac{x}{L} \total x = \sup(A, a_0)^m \int_{a_0}^{\sup(A, a_0)} \ln_+^k \frac{K}{x} \ln_+^l \frac{x}{L} \total x \\
&\qquad\leq \sup(A, a_0)^m \ln_+^l \frac{\sup(A, a_0)}{L} \int_0^{\sup(A, a_0)} \left( \ln_+ \frac{K}{\sup(A, a_0)} + \ln_+ \frac{\sup(A, a_0)}{x} \right)^k \total x \\
&\qquad\leq \sup(A, a_0)^{m+1} \ln_+^l \frac{\sup(A, a_0)}{L} \, \mathcal{P}\left( \ln_+ \frac{K}{\sup(A, a_0)} \right) \\
&\qquad= \sup(A, a_0)^{m+1} \, \mathcal{P}\left( \ln_+ \frac{K}{\sup(A, a_0)}, \ln_+ \frac{\sup(A, a_0)}{L} \right) \eqend{.}
\end{splitequation}
For the second part, we have similarly
\begin{splitequation}
&\int_{\sup(A, a_0)}^{a_1} \sup(A, a_0, x)^m \ln_+^k \frac{K}{x} \ln_+^l \frac{x}{L} \total x = \int_{\sup(A, a_0)}^{a_1} x^m \ln_+^k \frac{K}{x} \ln_+^l \frac{x}{L} \total x \\
&\qquad\leq \ln_+^k \frac{K}{\sup(A, a_0)} \int_{\sup(A, a_0)}^{\infty} x^m \left( \ln_+ \frac{\sup(A, a_0)}{L} + \ln_+ \frac{x}{\sup(A, a_0)} \right)^l \total x \\
&\qquad\leq \ln_+^k \frac{K}{\sup(A, a_0)} \sup(A, a_0)^{m+1} \, \mathcal{P}\left( \ln_+ \frac{\sup(A, a_0)}{L} \right) \\
&\qquad = \sup(A, a_0)^{m+1} \, \mathcal{P}\left( \ln_+ \frac{K}{\sup(A, a_0)}, \ln_+ \frac{\sup(A, a_0)}{L} \right)
\end{splitequation}
(if $a_1 < \sup(A,a_0)$ this second part vanishes, and trivially fulfils the bound). The subsequent estimates
\begin{equations}
\ln_+ \frac{\sup(A, a_0)}{L} &\leq \ln_+ \frac{\sup(A, a_0)}{\sup(a_0, L)} + \ln_+ \frac{\sup(a_0, L)}{L} \leq \ln_+ \frac{A}{\sup(a_0, L)} + \ln_+ \frac{a_0}{L} \eqend{,} \\
\ln_+ \frac{A}{\sup(a_0, L)} &\leq \ln_+ \frac{\sup(K, A)}{\sup(a_0, \inf(A,L))} \eqend{,} \\
\ln_+ \frac{K}{\sup(A, a_0)} &\leq \ln_+ \frac{\sup(K, A)}{\sup(a_0, \inf(A,L))}
\end{equations}
then give the desired result.
\end{proof}

\begin{lemma}[Estimating logarithms]
\label{lemma_largerlog}
For all $y \geq x \geq 0$, $K \geq 0$, $k \in \mathbb{N}_0$ and $\alpha > 0$ we have
\begin{equation}
x^\alpha \ln_+^k \frac{K}{x} \leq y^\alpha \, \mathcal{P}\left( \ln_+ \frac{K}{y} \right) \eqend{,}
\end{equation}
where the coefficients in the polynomial only depend on $\alpha$.
\end{lemma}

\begin{proof}
We use induction. For $k = 0$ the statement is obvious. For $k = 1$ we want to show that
\begin{equation}
x^\alpha \ln_+ \frac{K}{x} \leq y^\alpha \left( c + \ln_+ \frac{K}{y} \right) \eqend{,}
\end{equation}
which defining $A \equiv K/x$ and $a \equiv K/y$ (such that $A \geq a$) can be rewritten as
\begin{equation}
\left( \frac{a}{A} \right)^\alpha \ln_+ A - \ln_+ a \leq c \eqend{.}
\end{equation}
If $A,a \leq 1$ the logarithms vanish and the statement is true. If $A \geq 1 \geq a$, we estimate
\begin{equation}
\left( \frac{a}{A} \right)^\alpha \ln_+ A - \ln_+ a \leq A^{-\alpha} \ln A \eqend{.}
\end{equation}
The right-hand side vanishes for $A = 1$ and $A \to \infty$ and has one local maximum at $A = \exp(1/\alpha)$ where it evaluates to $1/(\alpha \mathe)$, such that the statement is true in this case as well. In the last case where $a,A \geq 1$ we write $A = a + \delta$ and view
\begin{equation}
\left( \frac{a}{a + \delta} \right)^\alpha \ln (a + \delta) - \ln a
\end{equation}
as a function of $\delta$, which vanishes for $\delta = 0$ and takes the value $- \ln a$ as $\delta \to \infty$. If $a \geq \exp(1/\alpha)$, the first derivative at $\delta = 0$ is negative and no local extrema occur for $\delta \geq 0$, such that we may bound the function by $0$. If $a < \exp(1/\alpha)$, the function has a local maximum at $\delta = \exp(1/\alpha) - a$ where it takes the value
\begin{equation}
\frac{a^\alpha}{\alpha \mathe} - \ln a \leq \frac{a^\alpha}{\alpha \mathe} \leq \frac{1}{\alpha} \eqend{,}
\end{equation}
and so the statement is true also in this case. For $k > 1$, we write
\begin{equation}
x^\alpha \ln_+^k \frac{K}{x} = \left( x^\frac{\alpha}{2} \ln_+^{k-1} \frac{K}{x} \right) \left( x^\frac{\alpha}{2} \ln_+ \frac{K}{x} \right)
\end{equation}
and apply the induction hypothesis to each of the two terms, thus proving the Lemma.
\end{proof}

\begin{lemma}[\texorpdfstring{$\Lambda$}{Λ} integration, part 2]
\label{lemma_lambdaint2}
For $m > -1$, $a_1 \geq a_0 \geq 0$, $b,K,L \geq 0$ and $k,l \in \mathbb{N}_0$ we have
\begin{equation}
\int_{a_0}^{a_1} \sup(x,b)^m \ln_+^k \frac{K}{x} \ln_+^l \frac{x}{L} \total x \leq \sup(c,a_1)^{m+1} \, \mathcal{P}\left( \ln_+ \frac{K}{\sup(c,a_1)}, \ln_+ \frac{a_1}{L} \right) \eqend{.}
\end{equation}
for any $c \geq b$.
\end{lemma}

\begin{proof}
We first extract the second logarithm using that $\ln_+ x/L \leq \ln_+ a_1/L$, extend the integration range to $[0,\sup(c,a_1)]$ and split the integral into two parts
\begin{equation}
\int_{a_0}^{a_1} \sup(x,b)^m \ln_+^k \frac{K}{x} \ln_+^l \frac{x}{L} \total x \leq b^m \ln_+^l \frac{a_1}{L} \int_0^b \ln_+^k \frac{K}{x} \total x + \ln_+^l \frac{a_1}{L} \int_b^{\sup(c,a_1)} x^m \ln_+^k \frac{K}{x} \total x \eqend{.}
\end{equation}
For the first integral we obtain
\begin{equation}
b^m \int_0^b \ln_+^k \frac{K}{x} \total x \leq b^{m+1} \, \mathcal{P}\left( \ln_+ \frac{K}{b} \right) \leq \sup(c,a_1)^{m+1} \, \mathcal{P}\left( \ln_+ \frac{K}{\sup(c,a_1)} \right) \eqend{,}
\end{equation}
where the second inequality follows from Lemma~\ref{lemma_largerlog} because of $m+1 > 0$ and $c \geq b$, and for the second integral we get
\begin{equation}
\int_b^{\sup(c,a_1)} x^m \ln_+^k \frac{K}{x} \total x \leq \int_0^{\sup(c,a_1)} x^m \ln_+^k \frac{K}{x} \total x \leq \sup(c,a_1)^{m+1} \, \mathcal{P}\left( \ln_+ \frac{K}{\sup(c,a_1)} \right) \eqend{.}
\end{equation}
Combining both results we obtain the inequality.
\end{proof}

\begin{lemma}[Taylor expansion]
\label{lemma_taylor}
For $m > -1$, $a_1 \geq a_0 \geq 0$, $K,L \geq 0$ and $k \in \mathbb{N}_0$ we have
\begin{equation}
\int_{a_0}^{a_1} \sup(x,L)^m \ln_+^k \frac{\sup(x, K)}{\sup(\inf(K, x), L)} \total x \leq \sup(a_1, L)^{m+1} \, \mathcal{P}\left( \ln_+ \frac{\sup(a_1, K)}{\sup(\inf(K, a_1), L)} \right) \eqend{.}
\end{equation}
\end{lemma}

\begin{proof}
We first estimate
\begin{equation}
\ln_+ \frac{\sup(x, K)}{\sup(\inf(K, x), L)} = \begin{cases} \ln_+ \dfrac{K}{\sup(x,L)} \leq \ln_+ \dfrac{\sup(a_1,K)}{x} & x \leq K \\ \ln_+ \dfrac{x}{\sup(K,L)} \leq \ln_+ \dfrac{a_1}{\sup(K,L)} \leq \ln_+ \dfrac{\sup(a_1, K)}{\sup(\inf(K, a_1),L)} & x \geq K \end{cases}
\end{equation}
and get
\begin{splitequation}
&\int_{a_0}^{a_1} \sup(x,L)^m \ln_+^k \frac{\sup(x, K)}{\sup(\inf(K, x), L)} \total x \\
&\qquad\leq \int_{a_0}^{a_1} \sup(x,L)^m \left( \ln_+ \frac{\sup(a_1,K)}{x} + \ln_+ \frac{\sup(a_1, K)}{\sup(\inf(K, a_1),L)} \right)^k \total x \eqend{.}
\end{splitequation}
An application of Lemma~\ref{lemma_lambdaint2} with $L = a_1$ and $c = L$ and the subsequent estimate
\begin{equation}
\ln_+ \frac{\sup(a_1,K)}{\sup(L,a_1)} \leq \ln_+ \frac{\sup(a_1,K)}{\sup(\inf(K, a_1),L)}
\end{equation}
yield the required bound.
\end{proof}

\begin{lemma}[\texorpdfstring{$\Lambda$}{Λ} integration, part 3]
\label{lemma_lambdaint3}
For $b \geq a \geq 0$, $c,K \geq 0$ and $k \in \mathbb{N}_0$ we have
\begin{equation}
\int_a^b \frac{1}{\sup(c,x)} \ln_+^k \frac{K}{x} \total x \leq \mathcal{P}\left( \ln_+ \frac{\sup(K, b)}{\sup(c,a)} \right) \eqend{.}
\end{equation}
\end{lemma}

\begin{proof}
We use as usual that $\sup(c,x) = \sup(\tilde{c},x)$ with $\tilde{c} \equiv \sup(c,a)$. Now using
\begin{equation}
\ln_+ \frac{K}{x} \leq \ln_+ \frac{K}{\tilde{c}} + \ln_+ \frac{\tilde{c}}{x}
\end{equation}
we obtain
\begin{equation}
\int_a^b \frac{1}{\sup(c,x)} \ln_+^k \frac{K}{x} \total x \leq \int_a^b \frac{1}{\sup(\tilde{c},x)} \left( \ln_+ \frac{K}{\tilde{c}} + \ln_+ \frac{\tilde{c}}{x} \right)^k \total x \eqend{.}
\end{equation}
We split the remaining integral in two; in the first part we change variables to $t = x / \tilde{c}$ to obtain
\begin{equation}
\int_a^{\tilde{c}} \frac{1}{\sup(\tilde{c},x)} \left( \ln_+ \frac{K}{\tilde{c}} + \ln_+ \frac{\tilde{c}}{x} \right)^k \total x \leq \int_0^1 \left( \ln_+ \frac{K}{\tilde{c}} + \ln \frac{1}{t} \right)^k \total t = \mathcal{P}\left( \ln_+ \frac{K}{\tilde{c}} \right) \eqend{.}
\end{equation}
In the second part, we have
\begin{equation}
\int_{\tilde{c}}^b \frac{1}{\sup(\tilde{c},x)} \left( \ln_+ \frac{K}{\tilde{c}} + \ln_+ \frac{\tilde{c}}{x} \right)^k \total x = \ln_+^k \frac{K}{\tilde{c}} \int_{\tilde{c}}^b \frac{1}{x} \total x = \ln_+^k \frac{K}{\tilde{c}} \ln_+ \frac{b}{\tilde{c}} \leq \mathcal{P}\left( \ln_+ \frac{\sup(K,b)}{\tilde{c}} \right) \eqend{,}
\end{equation}
and the Lemma follows by combining both parts.
\end{proof}

\begin{definition}[Exponential integrals]
Let us define the sequence of functions $\mathcal{E}_k(z)$ with $z \in \mathbb{R}$ and $k \geq 0$ by
\begin{equations}[expint_func_def]
\mathcal{E}_0(z) &= \mathe^{-\mathi z} \eqend{,} \\
\mathcal{E}_k(z) &= \frac{\mathe^{-\mathi z}}{\Gamma^2(k)} \int_{\frac{1}{2} - \mathi \infty}^{\frac{1}{2} + \mathi \infty} \Gamma^2(s) \Gamma(k-s) (\mathi z)^{-s} \frac{\total s}{2\pi\mathi} \eqend{,}
\end{equations}
where the integration contour is a straight line parallel to the imaginary axis, and the integral is absolutely convergent because of the exponential decay of the $\Gamma$ functions in imaginary directions.
\end{definition}

\begin{lemma}[Exponential integral estimates]
\label{lemma_expint}
The functions $\mathcal{E}_k(z)$ defined by equation~\eqref{expint_func_def} are bounded by
\begin{equation}
\label{expint_bounds}
\abs{\mathcal{E}_k(z)} \leq c \left( 1 + \ln_+ \abs{z}^{-1} \right)
\end{equation}
for some constant $c$ (depending on $k$).
\end{lemma}

\begin{proof}
Closing the contour to the left we obtain the convergent sum (with the Polygamma function $\psi$)
\begin{equation}
\mathcal{E}_k(z) = \frac{\mathe^{-\mathi z}}{\Gamma^2(k)} \sum_{m=0}^\infty \frac{\Gamma(k+m)}{(m!)^2} \left( 2 \psi(m+1) - \psi(m+k) - \ln(\mathi z) \right) (\mathi z)^m \eqend{.}
\end{equation}
The sum is also absolutely convergent and has a well-defined limit as $k \to 0$ (in which only the $m = 0$ term makes a contribution, leaving only the exponential). By manipulation of the summand and using the recursion relations for $\psi$, one easily obtains the important recursion relation
\begin{equation}
\label{expint_abl}
\partial_z \left( z \partial_z \mathcal{E}_{k+1}(z) \right) = \mathi \mathcal{E}_k(z) \eqend{.}
\end{equation}
For $\abs{z} \leq 1$, from the sum we directly obtain the bound
\begin{equation}
\abs{\mathcal{E}_k(z)} \leq c + c' \ln \abs{z}^{-1}
\end{equation}
for some positive constants $c$ and $c'$ depending on $k$. For $\abs{z} \geq 1$, we shift the integration contour over the first $n$ poles to the right to obtain
\begin{equation}
\mathcal{E}_k(z) = \frac{\mathe^{-\mathi z}}{\Gamma^2(k)} \left[ (-1)^k \sum_{m=1}^n \frac{\Gamma^2(m)}{\Gamma(m-k+1)} (-\mathi z)^{-m} + \int_{n+\frac{1}{2} - \mathi \infty}^{n+\frac{1}{2} + \mathi \infty} \Gamma^2(s) \Gamma(k-s) (\mathi z)^{-s} \frac{\total s}{2\pi\mathi} \right] \eqend{.}
\end{equation}
For all $m < k$, the sum vanishes because of the poles of the $\Gamma$ function in the numerator, such that the first non-vanishing term is obtained for $n = k$. For $\abs{z} \geq 1$, we can thus estimate
\begin{equation}
\abs{\mathcal{E}_k(z)} \leq \abs{z}^{-k} + \frac{\abs{z}^{-k-\frac{1}{2}}}{\Gamma^2(k)} \int_{k+\frac{1}{2} - \mathi \infty}^{k+\frac{1}{2} + \mathi \infty} \abs{\Gamma^2(s) \Gamma(k-s)} \mathe^{\pm \frac{\pi}{2} \Im s} \frac{\total s}{2\pi\mathi} \leq c'' \abs{z}^{-k}
\end{equation}
for some constant $c''$ depending on $k$, such that in total we obtain the bound~\eqref{expint_bounds}.
\end{proof}

\begin{lemma}[Slaloms]
\label{lemma_slalom}
For $a \geq 0$ and $u,v \in \mathbb{R}^4$, we have
\begin{equation}
\int_0^1 \frac{\abs{v}}{a + \abs{u + t v}} \total t \leq \mathcal{P}\left( \ln_+ \frac{\abs{v}}{a} \right) \eqend{.}
\end{equation}
\end{lemma}

\begin{proof}
If $\abs{u} \geq \abs{v}$, we use the triangle inequality
\begin{equation}
\abs{u + t v} \geq \abs{u} - t \abs{v}
\end{equation}
and estimate
\begin{splitequation}
\int_0^1 \frac{\abs{v}}{a + \abs{u + t v}} \total t &\leq \int_0^1 \frac{\abs{v}}{a + \abs{u} - t \abs{v}} \total t = \ln\left( 1 + \frac{\abs{v}}{a + \abs{u} - \abs{v}} \right) \\
&\leq \ln\left( 1 + \frac{\abs{v}}{a} \right) \leq \mathcal{P}\left( \ln_+ \frac{\abs{v}}{a} \right) \eqend{,}
\end{splitequation}
while for $\abs{u} \leq \abs{v}$ we split the integral and get (again with the triangle inequality)
\begin{splitequation}
\int_0^1 \frac{\abs{v}}{a + \abs{u + t v}} \total t &\leq \int_0^\frac{\abs{u}}{\abs{v}} \frac{\abs{v}}{a + \abs{u} - t \abs{v}} \total t + \int_\frac{\abs{u}}{\abs{v}}^1 \frac{\abs{v}}{a + t \abs{v} - \abs{u}} \total t \\
&= \ln\left( 1 + \frac{\abs{u}}{a} \right) + \ln\left( 1 + \frac{\abs{v}-\abs{u}}{a} \right) \leq 2 \ln\left( 1 + \frac{\abs{v}}{a} \right) \leq \mathcal{P}\left( \ln_+ \frac{\abs{v}}{a} \right) \eqend{.}
\end{splitequation}
\end{proof}

\begin{lemma}[Fractional derivatives]
\label{lemma_frac}
For rapidly decreasing $f(p)$ (such that we can perform integration by parts without boundary terms) with $\abs{\partial^w f(p)} \leq M^{-\abs{w}} \abs{f(p)}$, for $k \in \mathbb{N}_0$, $0 < \epsilon < 1$ and an arbitrary direction $\alpha \in \{1,2,3,4\}$, we have
\begin{equation}
\abs{\int \mathe^{-\mathi x p} f(p) \total^4 p} \leq ( \abs{x^\alpha} M )^{-k} \int \abs{f(p)} \total^4 p
\end{equation}
and
\begin{equation}
\abs{\int \mathe^{-\mathi x p} f(p) \total^4 p} \leq \frac{4}{1-\epsilon} ( \abs{x^\alpha} M )^{-k+\epsilon} \int \left( \frac{\abs{p^\alpha}}{M} \right)^{-1+\epsilon} \left( 1 + \frac{\abs{p^\alpha}}{M} \right) \abs{f(p)} \total^4 p \eqend{.}
\end{equation}
\end{lemma}

\begin{proof} We first introduce $k-1$ additional derivatives by
\begin{equation}
\mathe^{- \mathi x p} = \frac{1}{(-\mathi)^{k-1} (x^\alpha)^{k-1}} \partial^{k-1}_{p^\alpha} \mathe^{- \mathi x p} \eqend{,}
\end{equation}
and integrate these derivatives by parts such that
\begin{equation}
\int \mathe^{-\mathi x p} f(p) \total^4 p = \frac{\mathi}{(x^\alpha)^{k-1}} \int \mathe^{- \mathi x p} \partial^{k-1}_{p^\alpha} f(p) \total^4 p \eqend{.}
\end{equation}
Taking the absolute value, the first inequality follows. To show the second one, let us define for $\nu > 1$
\begin{equation}
\tilde{p}^\alpha \equiv \sgn p^\alpha \, \abs{p^\alpha}^\frac{1}{\nu} \eqend{,} \qquad \tilde{p}^k \equiv p^k \quad (k \in \{1,2,3,4\}\setminus\{\alpha\}) \eqend{,}
\end{equation}
such that
\begin{splitequation}
\int \mathe^{-\mathi x p} \partial^{k-1}_{p^\alpha} f(p) \total^4 p &= \nu \int \mathe^{-\mathi x^\alpha \sgn \tilde{p}^\alpha \, \abs{\tilde{p}^\alpha}^\nu} \partial^{k-1}_{p^\alpha} f(p) \abs{\tilde{p}^\alpha}^{\nu-1} \total^4 \tilde{p} \\
&= \nu \int \left( \partial_{\tilde{p}^\alpha} \int_0^{\tilde{p}^\alpha} \mathe^{-\mathi x^\alpha \sgn \tau \, \abs{\tau}^\nu} \total \tau \right) \partial^{k-1}_{p^\alpha} f(p) \abs{\tilde{p}^\alpha}^{\nu-1} \total^4 \tilde{p} \\
&= - \nu \int \int_0^{\tilde{p}^\alpha} \mathe^{-\mathi x^\alpha \sgn \tau \, \abs{\tau}^\nu} \total \tau \, \partial_{\tilde{p}^\alpha} \left( \partial^{k-1}_{p^\alpha} f(p) \abs{\tilde{p}^\alpha}^{\nu-1} \right) \total^4 \tilde{p} \\
&= - \nu^2 \int \int_0^{\tilde{p}^\alpha} \mathe^{-\mathi x^\alpha \sgn \tau \, \abs{\tau}^\nu} \total \tau \abs{\tilde{p}^\alpha}^{\nu-1} \partial^k_{p^\alpha} f(p) \abs{\tilde{p}^\alpha}^{\nu-1} \total^4 \tilde{p} \\
&\quad- \nu (\nu-1) \int \int_0^{\tilde{p}^\alpha} \mathe^{-\mathi x^\alpha \sgn \tau \, \abs{\tau}^\nu} \total \tau (\tilde{p}^\alpha)^{-1} \partial^{k-1}_{p^\alpha} f(p) \abs{\tilde{p}^\alpha}^{\nu-1} \total^4 \tilde{p} \eqend{.}
\end{splitequation}
We then calculate
\begin{equation}
\int_0^{\tilde{p}^\alpha} \mathe^{-\mathi x^\alpha \sgn \tau \, \abs{\tau}^\nu} \total \tau = (x^\alpha)^{-\frac{1}{\nu}} \int_0^{(x^\alpha)^\frac{1}{\nu} \tilde{p}^\alpha} \mathe^{-\mathi \sgn t \, \abs{t}^\nu} \total t \eqend{.}
\end{equation}
If $\abs{(x^\alpha)^\frac{1}{\nu} \tilde{p}^\alpha} \leq 1$, we can bound the right-hand side by $\abs{x^\alpha}^{-\frac{1}{\nu}}$. If $\abs{(x^\alpha)^\frac{1}{\nu} \tilde{p}^\alpha} \geq 1$, we use the van Corput Lemma
\begin{equation}
\abs{\int_a^b \mathe^{- \mathi \phi(t)} \total t} \leq \frac{3}{\inf_{[a,b]} \phi'(t)}
\end{equation}
if $\phi''(t) > 0$ (or $\phi''(t) < 0$) on $[a,b]$, and thus have in any case (since $\nu > 1$)
\begin{equation}
\abs{\int_0^{\tilde{p}^\alpha} \mathe^{-\mathi x^\alpha \sgn \tau \, \abs{\tau}^\nu} \total \tau} \leq \abs{x^\alpha}^{-\frac{1}{\nu}} \left( 1 + \frac{3}{\nu} \right) \leq 4 \abs{x^\alpha}^{-\frac{1}{\nu}} \eqend{.}
\end{equation}
Thus it follows that
\begin{splitequation}
\abs{\int \mathe^{-\mathi x p} f(p) \total^4 p} &\leq 4 \abs{x^\alpha}^{-(k-1)-\frac{1}{\nu}} \nu^2 \int \abs{\tilde{p}^\alpha}^{\nu-1} M^{-k} \abs{f(p)} \abs{\tilde{p}^\alpha}^{\nu-1} \total^4 \tilde{p} \\
&\quad+ 4 \abs{x^\alpha}^{-(k-1)-\frac{1}{\nu}} \nu^2 \int \abs{(\tilde{p}^\alpha)^{-1}} M^{-(k-1)} \abs{f(p)} \abs{\tilde{p}^\alpha}^{\nu-1} \total^4 \tilde{p} \\
&\leq 4 \nu ( \abs{x^\alpha} M )^{-(k-1)-\frac{1}{\nu}} \int \left( \frac{\abs{p^\alpha}}{M} \right)^{-\frac{1}{\nu}} \left( 1 + \frac{\abs{p^\alpha}}{M} \right) \abs{f(p)} \total^4 p \eqend{.}
\end{splitequation}
which gives the Lemma with $\epsilon = 1 - 1/\nu$.
\end{proof}

\begin{lemma}[Smearings]
\label{lemma_smearing}
With the Schwartz norms defined in equation~\eqref{schwartz_norm}, we have for any multiindex $w$
\begin{equation}
\int \Big[ 1 - \ln \inf(1, \norm{x}) \Big] \abs{ \partial^w_x f(x) } \total^4 x \leq 2^{12} \norm{f}_{\abs{w}} \eqend{.}
\end{equation}
\end{lemma}

\begin{proof} We first estimate
\begin{equation}
1 - \ln \inf(1, \norm{x}) \leq 2^4 \prod_{\alpha=1}^4 \Big[ 1 - \ln \inf(1, \abs{x^\alpha}) \Big] \eqend{.}
\end{equation}
Defining $g(x) \equiv \partial^w f(x)$, the above inequality then follows from
\begin{equation}
\int \prod_{\alpha=1}^4 \Big[ 1 - \ln \inf(1, \abs{x^\alpha}) \Big] \abs{ g(x) } \total^4 x \leq 2^8 \sup_{x \in \mathbb{R}^4} \abs{ (1+x^2)^4 g(x) } \eqend{.}
\end{equation}
We thus set $h(x) \equiv \prod_{\alpha=1}^3 \Big[ 1 - \ln \inf(1, \abs{x^\alpha}) \Big] \abs{ g(x) }$, and estimate
\begin{splitequation}
&\int \prod_{\alpha=1}^4 \Big[ 1 - \ln \inf(1, \abs{x^\alpha}) \Big] \abs{ g(x) } \total^4 x = \int_{\abs{x^4} \leq 1} \Big[ 1 - \ln \abs{x^4} \Big] h(x) \total^4 x + \int_{\abs{x^4} > 1} h(x) \total^4 x \\
&\hspace{8em}\leq \int \left[ \sup_{x^4 \in \mathbb{R}} h(x) \int_{\abs{x^4} \leq 1} \Big[ 1 - \ln \abs{x^4} \Big] \total x^4 + \int h(x) \total x^4 \right] \total^3 x \\
&\hspace{8em}\leq \int \left[ 4 \sup_{x^4 \in \mathbb{R}} h(x) + \sup_{x^4 \in \mathbb{R}} \left[ \left( 1+(x^4)^2 \right) h(x) \right] \int \frac{1}{1+(x^4)^2} \total x^4 \right] \total^3 x \\
&\hspace{8em}\leq 4 \int \sup_{x^4 \in \mathbb{R}} \left[ \left( 1+(x^4)^2 \right) h(x) \right] \total^3 x \eqend{.}
\end{splitequation}
Performing the same estimates for the integrals over $x^1$ to $x^3$, the Lemma follows.
\end{proof}

\section*{Bibliography}

\providecommand{\href}[2]{#2}\begingroup\raggedright\endgroup

%\bibliography{literature}

\end{document}